\documentclass{amsart}
\input{packages}
\input{def}
\newtheorem{introthm}{Theorem}

\title[Symmetric compatibility of theta structures, change of level and isogenies]{Symmetric compatibility of theta structures, with applications to change of level and isogenies}
\makeatletter
\@ifclassloaded{amsart}{%
\author{Antoine Dequay}
\address{Univ Rennes, CNRS, IRMAR, UMR 6625,F-35000 Rennes, France}
\email{antoine.dequay@ens-rennes.fr}

\author{David Lubicz}
\address{Univ Rennes, CNRS, IRMAR, UMR 6625,F-35000 Rennes, France}
\address{DGA Maîtrise de l'information, BP 7419, F-35174 Bruz, France}
\email{david.lubicz@univ-rennes.fr}
}{
\iflanguage{french}{
\author{David Lubicz et Antoine Dequay}
}{
\author{David Lubicz and Antoine Dequay}
}
}
\makeatother

\begin{document}
\begin{abstract}
Let $m,n,d >1$ be integers such that $n=md$ and let $(B, \bpol, \Thetabpol)$ be a marked abelian
variety of level $m$ over a field of odd characteristic. We study the theta structures of level
$n$ for $\bpol^d$, and the theta structures on the abelian varieties $d$-isogenous to $B$, which are
compatible with $\Thetabpol$. When $d$ is odd and coprime to $m$, they correspond to the
isotropic subgroups of $B[n]$, and this correspondence underlies the known algorithms in theta
coordinates. In general, this is no longer true. We introduce the notion of symmetric
compatibility of a torsion point with a symmetric level subgroup, and we show that a subgroup of
$B[n]$ isomorphic to $(\Z/n\Z)^g$ comes from a compatible theta structure if and only if it is
isotropic and all its points are symmetric compatible. We extend to every $d$ Mumford's notion of
compatible pairs of theta structures, defined by Mumford for $d=2$, and prove that it is
equivalent to our notion, which can be tested on theta coordinates; we also describe exactly the
remaining ambiguity, which explains the sign choices met when the duplication formula is used to
change the level. As applications, we obtain change of level and isogeny algorithms without
assuming that $d$ is odd and coprime to $m$, using $O(n^{3g} \log(d))$ and $O(n^{2g} \log(d))$ operations in the
base field respectively. In particular, they give change of level algorithms for even $d$, and
$2$-isogenies in level $4$ theta coordinates with an explicit control of the theta structure of
the codomain.
\end{abstract}
\maketitle

\section{Introduction}\label{sec:intro}
Mumford's theory of algebraic theta functions \cite{MumfordOEDAV1,MumfordOEDAV2,MumfordOEDAV3}
gives, for an abelian variety together with an ample symmetric line bundle and a theta structure,
a canonical basis of sections, hence canonical projective coordinates, together with explicit
equations, the Riemann relations. Its cornerstone is the duplication formula, which relates the
theta functions of level $m$ and $2m$: it is used to express the product of two theta functions in
the canonical basis, to derive the Riemann relations, which give, when the level is divisible by
$4$, a complete set of equations for the abelian variety and for the closure of the moduli space
of abelian varieties with a theta structure. The duplication formula can also be viewed as a
change of level formula from level $m$ to level $2m$, or as a formula for $2$-isogenies. Read from
level $2m$ to level $m$, it is a mere computation. Read in the other direction, it does not
determine the result: part of the target theta structure cannot be recovered, several choices are
possible, and this translates into sign choices in square root computations (see Example
\ref{ex:dup601}). This paper explains this phenomenon in general, for any change of level by a
factor $d$ and any $d$-isogeny, by identifying which data determine the target theta structure, and
uses this description to obtain completely general algorithms.

\subsection*{Setting}
Let $(B, \bpol)$ be a $g$-dimensional abelian variety together with an ample symmetric line bundle
over an algebraically closed field $\overk$ of odd characteristic; to simplify the notation, we
suppose that $\bpol$ is a power of a principal line bundle. For $n \in \N^*$, let
$Z(n)=(\frac1n\Z/\Z)^g$, let $\dZ(n)$ be its dual and $K(n)=Z(n)\times\dZ(n)$. Let $K(\bpol)$ be the
kernel of the isogeny $B \rightarrow \widehat{B}$, $x \mapsto \tau^*_x \bpol \otimes \bpol^{-1}$; we say
that $\bpol$ is of level $n$ if $K(\bpol) \simeq K(n)$. The theta group $G(\bpol)$ is the group of
pairs $(\tau_x, \psi_x)$ where $x \in K(\bpol)$ and $\psi_x: \bpol \rightarrow \tau^*_x\bpol$ is an
isomorphism; it is a central extension of $K(\bpol)$ by $\overk^*$, and its commutator pairing is
the pairing $e_\bpol$ on $K(\bpol)$. A level subgroup is a subgroup of $G(\bpol)$ meeting $\overk^*$
trivially, and a theta structure $\Thetabpol$ of level $n$ is an isomorphism between $G(\bpol)$ and
the Heisenberg group $G(n)$; equivalently, it is given by a decomposition $K(\bpol)=K_1(\bpol)
\times K_2(\bpol)$ into maximal isotropic subgroups for $e_\bpol$ together with level subgroups
$\tildeK_1(\bpol), \tildeK_2(\bpol)$ above them. As explained by Mumford \cite{MumfordOEDAV1}, it
determines a canonical basis $(\theta_i^\Thetabpol)_{i \in \Zn}$ of $\Gamma(B, \bpol)$, and thus,
if $n \geq 3$, an embedding $B \rightarrow \proj^{Z(n)}$; the image of $0_B$ is the theta null
point of $(B, \bpol, \Thetabpol)$.

Let $m,n,d>1$ be integers such that $n=md$. A ($d$-)change of level algorithm takes as input the
theta null point of $(B, \bpol, \Thetabpol)$ of level $m$ and the theta coordinates
$(\theta_i^\Thetabpol(x))_{i \in Z(m)}$ of a point $x \in B(\overk)$, and outputs
$(\theta_j^\Thetabpold(x))_{j \in \Zn}$ for a theta structure $\Thetabpold$ of level $n$ for
$\bpol^d$. An isogeny algorithm takes as input the theta null point of $(B, \bpol, \Thetabpol)$, a
subgroup $K \subset B[d]$ isomorphic to $Z(d)$ and isotropic for the Weil pairing $e_{B,d}$, and the
theta coordinates of $x \in B(\overk)$, and outputs the theta coordinates of $f(x)$ for a theta
structure of level $m$ on $A=B/K$, where $f: B \rightarrow A$ is the quotient isogeny. In both
cases, the output theta structure has to be related to $\Thetabpol$. The case where $d$ is odd and
coprime to $m$ is treated in \cite{DRmodular,lubicz:hal-03738315}. The aim of this paper is to
lift these restrictions. This is not a technicality: since the number of coordinates is $m^g$ and
the Riemann relations require an even level, one computes in practice with $m=2$ or $4$, and the
smallest, hence most useful, isogenies are $2$-isogenies.

\subsection*{The obstruction}
Assume $d=2$ to simplify the presentation. The symmetry of $\bpol$ allows one to define, above each
$x \in K(\bpol)$, exactly two symmetric lifts $\pm\tildex \in G(\bpol)$ \cite{MumfordOEDAV1}, and a
theta structure is symmetric if the level subgroups $\tildeK_1(\bpol)$ and $\tildeK_2(\bpol)$ are
made of symmetric elements. Mumford's morphism of theta groups $\epsilon_d(\bpol): G(\bpol)
\rightarrow G(\bpol^d)$, $(\tau_x, \psi_x) \mapsto (\tau_x, \psi_x^{\otimes d})$, allows one to say
when a theta structure $\Thetabpold$ for $\bpol^d$ extends $\Thetabpol$: if $\tildeK_i(\bpol)$
(resp. $\tildeK_i(\bpol^d)$) are the level subgroups defining $\Thetabpol$ (resp. $\Thetabpold$), we
ask that
\begin{equation}\label{eq:ext}
\epsilon_d(\bpol)(\tildeK_i(\bpol)) \subset
\tildeK_i(\bpol^d).
\end{equation}
When $d$ is odd and coprime to $m$, there exists a unique symmetric theta structure for $\bpol^d$
extending $\Thetabpol$ \cite{lubicz:hal-03738315}, and it is used to define the change of level
algorithm. When $d$ is even, such a theta structure does not always exist, and when it exists, it
is not unique. This was already observed by Mumford \cite{MumfordOEDAV1}, who proves that a
symplectic structure of level $2m$ for $\bpol^2$ induces, by means of a morphism of theta groups
$\eta_2(\bpol^2): G(\bpol^2) \rightarrow G(\bpol)$, a unique theta structure of level $m$ for
$\bpol$. The definition of $\eta_2(\bpol^2)$ is more subtle than that of $\epsilon_2(\bpol)$, since
it involves the symmetry of $\bpol$, and the theta structure obtained in this way need not be the
theta structure $\Thetabpol$ given as input. The underlying problem is that $\epsilon_2(\bpol)$
forgets signs: if $x \in K(\bpol)$, the two symmetric lifts $\pm\tildex \in G(\bpol)$ of $x$ have
the same image by $\epsilon_2(\bpol)$. As a consequence, Condition (\ref{eq:ext}) does not
determine $\Thetabpol$: for any character $\chi: K_i(\bpol) \rightarrow \{\pm 1\}$, the level
subgroup made of the elements $\chi(x)\tildex$, for $\tildex \in \tildeK_i(\bpol)$ above $x \in
K_i(\bpol)$, has the same image by $\epsilon_2(\bpol)$ as $\tildeK_i(\bpol)$, whereas
$\eta_2(\bpol^2)$ only gives one of these level subgroups. Mumford's definition of a compatible
pair of theta structures for $(\bpol, \bpol^2)$ \cite[p.~317]{MumfordOEDAV1} requires both
$\epsilon_2(\bpol)(\tildeK_i(\bpol))\subset \tildeK_i(\bpol^2)$ and
$\eta_2(\bpol^2)(\tildeK_i(\bpol^2))=\tildeK_i(\bpol)$. This definition is not easily amenable to
computation, and it is only given for $d=2$.

\subsection*{Main results}
Our main tool is the following notion. Let $\tildeH$ be a symmetric level subgroup of $G(\bpol)$
and $x \in K(\bpol)$ such that some multiple of $x$ lies in $H=\pigbpol(\tildeH)$; let $\ell$ be the
smallest positive integer such that $\ell x \in H$. We say that $x$ is symmetric compatible with
$\tildeH$ if there exists a symmetric lift $\tildex \in G(\bpol)$ of $x$ such that $\tildex^\ell
\in \tildeH$. For a torsion point $x \notin K(\bpol)$, we pull back $\bpol$ by an isogeny to reduce
to this case (Definition~\ref{def:symcompat}). When $\ell$ is odd, this condition is always
satisfied (Proposition~\ref{sec2:proptors}): symmetric compatibility is a condition on the
$2$-part. It is not additive, since $\ell$ depends on the point, but it becomes additive once the
exponent is fixed: if $N$ is even and $Nx \in H$, then for any symmetric lift $\tildex$ of $x$, we
have $\tildex^N=\varphi_N(x, \tildeH)\,h$ with $h \in \tildeH$ and $\varphi_N(x, \tildeH) \in \{\pm
1\}$, and if $x_1, x_2 \in K(\bpol)$ are such that $e_\bpol(x_1, x_2)=1$, then $\varphi_N(x_1+x_2,
\tildeH)=\varphi_N(x_1, \tildeH)\varphi_N(x_2, \tildeH)$ (Proposition~\ref{sec2:propgroup}). In particular, it suffices to test symmetric compatibility on a
basis of a subgroup $G \simeq Z(n)$ (Corollary~\ref{sec2:corgroup}).

Let $(B, \bpol, \Thetabpol)$ be a marked abelian variety of level $m$. If $(A, \pol, \Thetapol)$ is
a marked abelian variety of level $n$ and $f: A \rightarrow B$ an isogeny such that the two theta
structures are compatible through $f$ (isog-$f$-compatible, Definition~\ref{sec1:def1}; in this
case, the theta null point of $B$ is obtained by keeping some of the coordinates of the theta null
point of $A$), then $G=f(\Thetapolbar(Z(n)\times\{0\}))$ is a subgroup of $B[n]$ isomorphic to
$Z(n)$, containing $\Thetabpolbar(Z(m)\times\{0\})$, and $A \simeq B/G[d]$ (Proposition
\ref{prop:charac}). When $d$ is coprime to $m$, every such subgroup which is isotropic for the Weil
pairing arises in this way \cite[Theorem~12]{DRmodular}. In general, we prove:
\begin{introthm}[Theorem~\ref{prop:cond} and Proposition~\ref{prop:phs}]\label{th:introA}
 Let $G \subset B[n]$ be a subgroup isomorphic to $Z(n)$ containing $\Thetabpolbar(Z(m)\times
 \{0\})$. Then $G$ arises from a marked abelian variety $(A, \pol, \Thetapol)$ of level $n$
 isog-$f$-compatible with $(B, \bpol, \Thetabpol)$ if and only if $G$ is isotropic for $e_{B,n}$
 and all the points of $G$ are symmetric compatible with $\Thetabpol(\{1\} \times Z(m) \times
 \{0\})$. Moreover, the marked abelian varieties arising from subgroups $G$ with a given $G[d]$
 form an orbit under the group of symmetric automorphisms of the Heisenberg group $G(n)$
 compatible with $G(m)$, and if $d \mid m$, all of them form a single orbit.
\end{introthm}
We extend Mumford's notion of a compatible pair of theta structures to any $d$ (Definition
\ref{def:mumfordcomp}), and we show that the condition involving $\eta_d$ is exactly a symmetric
compatibility condition:
\begin{introthm}[Proposition~\ref{prop:mumfordequiv} and Theorem~\ref{thm:compatthetas}]
 Let $\Thetabpol$ and $\Thetabpold$ be symmetric theta structures of levels $m$ and $n$ for $\bpol$
 and $\bpol^d$ satisfying Condition (\ref{eq:ext}). Then $\eta_d(\bpol^d)$ maps the level subgroups
 of $\Thetabpold$ onto those of $\Thetabpol$ if and only if all the points of
 $\Thetabpoldbar(Z(n)\times\{0\})$ (resp. $\Thetabpoldbar(\{0\}\times\dZ(n))$) are symmetric
 compatible with $\Thetabpol(\{1\}\times Z(m)\times\{0\})$ (resp. $\Thetabpol(\{1\}\times\{0\}
 \times\dZ(m))$). Moreover, Mumford compatibility is equivalent to $f$-compatibility, that is to the
 existence of an isogenous marked abelian variety from which both theta structures are induced.
\end{introthm}
Mumford's condition, which involves the abstract map $\eta_d$, thus becomes a pointwise condition
which can be tested on theta coordinates. Symmetric compatibility is even finer, point by point,
than the test given by $\eta_N$ for a fixed $N$ (Remark~\ref{rk:symvsmumford}). Finally, we describe
exactly the remaining ambiguity:
\begin{introthm}[Proposition~\ref{prop:uniqueness} and Theorem~\ref{th:main1}]
 Let $G$ be as in Theorem~\ref{th:introA}, isotropic and with all its points symmetric compatible, and
 fix a numbering $G=\{g(i), i \in \Zn\}$ compatible with $\Thetabpol$. The partial theta
 structures of level $n$ for $\bpol^d$, $f$-compatible with $\Thetabpol$ and inducing this
 numbering, form a nonempty set, and any two of them differ by the action of an element of
 $\dZ(d)[2]$. Each good lift of $G$, that is each affine lift of the points of $G$ satisfying the
 Riemann relations, the symmetry relations and the action of $\Thetabpol$, determines one of them.
 Good lifts are exactly the affine lifts given by the theta null points of isogenous marked
 abelian varieties; two good lifts differ by roots of unity, whose $d$-th powers are $\pm 1$ and
 determine the theta structure.
\end{introthm}
When $d$ is odd, $\dZ(d)[2]$ is trivial and we recover the uniqueness of \cite{lubicz:hal-03738315};
when $d$ is even, the action of $\dZ(d)[2] \simeq (\Z/2\Z)^g$ is the sign ambiguity of the
duplication formula. Examples~\ref{ex:dup601} and~\ref{ex:isog601} illustrate this on an elliptic
curve over $\F_{601}$: the eight solutions of the duplication formula for a change of level from
$4$ to $8$ are exactly the outputs of our algorithm for the possible choices of numberings and of
good lifts. The obstruction is concentrated in the $2$-primary part, where Mumford's theory of
$2$-adic theta structures \cite[\S~7]{MumfordOEDAV2} describes the compatible systems of symmetric
theta structures along a tower of isogenies of degree $2^e$: they correspond to the symplectic
isomorphisms of the $2$-adic Tate module, so that a symmetric theta structure of level $n$ is
determined by points of order $2n$. Symmetric compatibility is a finite level counterpart of this
description, which can be tested on theta coordinates: the theta structure is fixed by points of
order $n$ and by choices of good lifts, which only involve square roots, instead of points of order
$2n$, which may require a field extension (Example~\ref{ex:dup601iter}).

The definitions above involve theta groups, which are not part of the input of an algorithm: the
theta structure is only given by its theta null point. Following \cite{carlsPointCounting} and
\cite{DRmodular}, we represent level subgroups by their action on affine lifts of points, and we use
the arithmetic of affine points given by the Riemann relations, the symmetry relations and the
action of the theta group \cite{DRarithmetic} (see also the formalism of cubical torsors developed
by Robert in \cite{cryptoeprint:2024/517}). We give a definition of symmetric compatibility for
affine points (Definition~\ref{def:symcompat2}), which can be tested with a few $\scalmult$
computations, and we prove that it is equivalent to the definition with level subgroups
(Proposition~\ref{prop:propclef}). The additivity of the latter then yields the additivity of the
former (Corollary~\ref{cor:symbasis}), which we could not prove directly, since only a
pseudo-addition is available on affine points (to compute $\widetilde{x+y}$, one needs $\tildex$,
$\tildey$ and $\widetilde{x-y}$). When the compatibility fails, we change the theta structure or
the symplectic structure to restore it (Propositions~\ref{prop:changesym1} and
\ref{prop:changesym2}, Algorithms~\ref{algo:changesym1} and~\ref{algo:changesym2}); to this end, we
use the transformation formula for algebraic theta functions under the action of the symplectic
group, which we give in terms of generators (Proposition~\ref{prop:trans}, equivalent to
\cite[Theorem~12]{dartois_fast}), and we describe the action of the automorphisms of $G(n)$
compatible with $G(m)$, with the normalization of signs required by the compatibility with
$\Thetabpol$ (Proposition~\ref{prop:action2}).

From these results, we obtain change of level and isogeny theorems without assuming that $d$
is odd and coprime to $m$, which we state here informally (see Theorems~\ref{th:main7} and~\ref{th:mainisog}). Write
$d=\sum_{j=1}^r a_j^2$ with $\gcd(a_j,n)=1$.
\begin{introthm}[informal version of Theorem~\ref{th:main7}]
 Let $B[n]=G_1 \oplus G_2$ be a decomposition into subgroups isomorphic to $Z(n)$, isotropic for
 $e_{B,n}$, containing respectively $\Thetabpolbar(Z(m)\times\{0\})$ and
 $\Thetabpolbar(\{0\}\times\dZ(m))$, and whose elements are symmetric compatible with the
 corresponding level subgroups of $\Thetabpol$. Then there exists a theta structure $\Thetabpold$
 of type $K(n)$ for $\bpol^d$, compatible with $\Thetabpol$, such that for all $x \in B(\overk)$,
 the coordinates $(\theta_j^{\Thetabpold}(x))_{j \in \Zn}$ are given, up to a common factor, by
 sums over $Q \in G_2$ of products $\prod_{i=1}^r (a_i(\widetilde{x+g_1(j)+Q}))_{\beta_i}$ computed
 with $\diffadd$ and $\scalmult$ from good lifts of the points of $x+G_1+G_2$.
\end{introthm}
\begin{introthm}[informal version of Theorem~\ref{th:mainisog}]
 Let $K \subset B[d]$ be a subgroup isomorphic to $Z(d)$ and isotropic for $e_{B,d}$, let $A=B/K$
 and $f: B \rightarrow A$ be the quotient isogeny. After possibly changing $\Thetabpol$, there
 exists a subgroup $G_1 \subset B[n]$ as above such that $K=G_1[d]$, and a theta structure
 $\Thetapol$ of type $K(m)$ for $A$ such that for all $x \in B(\overk)$, the coordinates
 $(\theta_{j_0}^{\Thetapol}(f(x)))_{j_0 \in Z(m)}$ are given, up to a common factor, by sums over
 $P \in K$ of products of coordinates of the affine points $a_i(\widetilde{x+P+g_1(j_1)})$.
\end{introthm}
In both theorems, the products involve indices $\beta \in Z(m)^r$ which must be chosen so that a
certain constant $C_\beta$ does not vanish; this can be tested on the theta null point, and it is
always possible when $m \geq 3$ and $a_1=\cdots=a_r=1$, and for the change of level when
$m=d=2$. We deduce the change of level and
isogeny Algorithms~\ref{algo:changelevel} and~\ref{algo:isogcomp}: they take as input the theta
null point of $(B, \bpol, \Thetabpol)$, a basis of $B[n]$ (and $K$, without any further condition
relating $K$ and $\Thetabpol$, for the isogeny) and the theta coordinates of $x$, and compute the
theta coordinates of $x$ for $\Thetabpold$ (resp. of $f(x)$ for $\Thetapol$) in
$O(Nn^{3g}\log(d))$ (resp. $O(Nn^{2g}\log(d))$) operations in the base field, where $N$ is the
number of values of $\beta$ tried (Corollaries~\ref{cor:changeoflevel} and~\ref{cor:isogcomp}).

All these results include the level $m=2$, which is the one used in practice for the theta null
points of Jacobians (Thomae's formula) and for Kummer arithmetic. The theory of Sections
\ref{sec:isogcompat} to~\ref{sec:struct} applies without change. The differences are
algorithmic and are gathered in Remark~\ref{rk:level2}: in level $2$, the theta coordinates
embed the Kummer variety of $B$ rather than $B$, so that points are only known up to sign, the
normal addition is not available and subgroups are given by chain bases, while symmetric
compatibility becomes an eigenvector condition for the action of the theta group. Moreover,
the existence of $\beta$ with $C_\beta \neq 0$ cannot be deduced from projective normality when
$m=2$; for the change of level from $2$ to $4$, it follows instead from the duplication formula
(Theorem~\ref{th:main7} (3)). Example~\ref{ex:dup24} shows that the sign choices met when the
duplication formula is solved from level $2$ to level $4$ are the choices of good lifts.

\subsection*{Relation to previous work}
The case where $d$ is odd and coprime to $m$ is treated in \cite{DRmodular} for the computation of
modular correspondences and in \cite{lubicz:hal-03738315} for change of level and isogenies, which
are also extended there to abelian varieties with real multiplication. For $2$-isogenies, the
ambiguity of the duplication formula is resolved in \cite{dartois_fast} with the data of a
subgroup of the $8$-torsion isotropic for the Weil pairing, which, by classical results of Mumford,
defines a unique symmetric level subgroup above twice this subgroup. Our results explain this
mechanism and extend it to every $d$ and every even $m$, giving a single framework for change of
level and isogeny computations in theta coordinates, independent of the degree and of the level,
which encompasses the known theta function based algorithms. On the practical side, $2$-isogenies
are very useful in higher dimensional isogeny-based cryptography
\cite{cryptoeprint:2024/760,SIDH,CasDe2022,MMP_SIDH,Robert_SIDH,Robert22applications}; other
applications of changes of level and of $2$-isogenies in theta coordinates include the
generalization of the AGM method to compute period matrices \cite{dupont2006moyenne} and point
counting \cite{carlsPointCounting}. Some questions raised by the case of real multiplication are
discussed in Section~\ref{sec:concl}.

From a purely computational point of view, an isogeny of degree $d=2^e d'$ with $d'$ odd can be
obtained by composing an isogeny of degree $d'$, computed with \cite{lubicz:hal-03738315}, with a
chain of $2$-isogenies computed with \cite{dartois_fast}; for large $e$, this is more efficient
than our direct algorithm, whose cost grows like $n^{2g}$. Our aim is different. First, change of
level by an even factor, in particular when $d \mid m$, does not reduce to isogeny computations
and is not covered by previous work. Second, the data used to fix the target theta structure
differ: in \cite{dartois_fast}, the theta structure of level $m$ on the codomain is the one
induced, by Mumford's results on compatible pairs, by a symplectic basis of $B[2n]$; here, the
input is a subgroup of $B[n]$, we characterize exactly which subgroups can be used (Theorem
\ref{th:introA}), and we describe the remaining ambiguity, the action of $\dZ(d)[2]$, which
corresponds to the choices of good lifts and accounts for the sign choices met in the
duplication formula. Third, our results hold in any even level, in particular in level $4$, where
the images are points of the abelian variety, and not only of its Kummer variety as in level $2$.
Finally, both \cite{dartois_fast} and the present paper use the transformation formula for
algebraic theta functions in Mumford's formalism: \cite[Theorem~12]{dartois_fast} gives it in
closed form, for a symmetric automorphism determined by a symplectic matrix modulo $2n$, whereas
Proposition~\ref{prop:trans} decomposes a symplectic matrix modulo $n$ into generators, which
makes the residual action of $K(2)$ explicit and is cheaper for large $n$
(Remark~\ref{rk:transcost}).

\subsection*{Organization of the paper}
In Section~\ref{sec:notations}, we gather the main results and notation that we use. In Section
\ref{sec:isogcompat}, we introduce the notion of isogeny compatibility of theta structures. In
Section~\ref{sec:trans}, we study the action of the metaplectic group on theta null points, recall
the transformation formula for algebraic theta functions and describe the action of the
automorphisms compatible with a theta structure of level $m$. In Section~\ref{sec:struct}, we
introduce symmetric compatibility and classify the marked abelian varieties which are compatible
with a given one up to an isogeny (Theorem~\ref{prop:cond}). In Section~\ref{sec:isogchange}, we
compute their theta null points from good lifts of torsion points. In Section~\ref{sec:main}, we
compare our notion of compatibility with Mumford's (Theorem~\ref{thm:compatthetas}), and we give the
change of level and isogeny theorems and algorithms, with examples. Section~\ref{sec:concl}
discusses the extension to abelian varieties with real multiplication and some open questions.

\section*{Acknowledgments}

The authors would like to thank Pierrick Dartois and Damien Robert for insightful
discussions concerning the first version of this article.

\section{Notation and basic facts}\label{sec:notations}

In this section, we recall the notation and well-known facts used in this paper. The main general
references for this section are \cite{BiLaCAV,MumfordOEDAV1,milne1986abelian,MumfordAV}.

Let $A$ be a $g$-dimensional abelian variety over a field $k$ of characteristic $p \ne 2$.
For $x$ a geometric point of $A$, we denote by $\tau_x$ the translation by $x$ map on
$A$. If $\pol$ is an ample line bundle on $A$, we let $\phi_\pol: A \rightarrow \hatA, x\mapsto
\tau_x^* \pol \otimes \pol^{-1}$. It is well known that $\phi_\pol$ characterizes the
algebraic class of $\pol$, which is a polarization of $A$. We denote by $\Kpol$ the kernel
of $\phi_\pol$, which is finite (because $\pol$ is ample), and assume that $\phi_\pol$ is separable. The degree of
$\pol$ is the degree of $\phi_\pol$. A principal line bundle is a line bundle of degree $1$.

For $\ell \in \Z^*$, we denote by $[\ell]: A \rightarrow A, x\mapsto \ell x$ the
multiplication by $\ell$ isogeny.
A line bundle $\pol$ on $A$ is symmetric if there exists an isomorphism $[-1]^*\pol \simeq \pol$. 
If so, there exists a unique normalized isomorphism $\psii : \pol \rightarrow
[-1]^*(\pol)$ such that $\psii : \pol(0) \rightarrow [-1]^*(\pol(0))=\pol(0)$ is the
identity morphism. 
If $n$
is a positive integer, we say that $\pol$ is a level $n$ line bundle if $\pol= \pol_0^n$
for $\pol_0$ a principal line bundle. From now on, we assume that $\pol$ is a level $n$ symmetric line bundle defined over $k$. 
If $\Autg_A$ is a subgroup of the automorphism group of $A$ as an algebraic variety
such that for every $\tau \in \Autg_A$ there exists an isomorphism $\psi_\tau: \pol \rightarrow \tau^*\pol$,
we may form the set of such pairs $(\tau, \psi_\tau)$. Equipped with the composition law
$$(\tau, \psi_\tau) \circ (\tau', \psi_{\tau'})=(\tau \circ \tau',
{\tau'}^*(\psi_{\tau})\circ \psi_{\tau'}),$$
this set becomes a group.
Taking $\Autg_A = \{\tau_x, x \in \Kpol\}$,
where each $\tau_x$ is the translation-by-$x$ automorphism of $A$,
yields the theta group $\Gpol$ associated to $\pol$.
It is known that $\Gpol$ is a central extension of $\Kpol$ by $\overk^*$ (see \cite{MumfordOEDAV1}).

The commutator
pairing $\Gpol \times \Gpol \rightarrow \overk^*, (g_x,g_y) \mapsto g_xg_y g_x^{-1} g_y^{-1}$ descends
to a skew-symmetric pairing $\epol: \Kpol \times \Kpol \rightarrow \overk^*$, which is perfect.
A level subgroup
above $K\subset \Kpol$ is a subgroup $\tildeK$ of $\Gpol$ such that $\pigpol: \tildeK \rightarrow K$ is
an isomorphism; it exists if and only if $K$ is isotropic for the commutator pairing.

We collect in the following proposition the results on Weil and commutator pairings that we use (\cite[p.~228]{MumfordAV}):
\begin{proposition}\label{def:weil}
 Let $\ell$ be a positive integer, 
 denote by $e_{A, \ell}: A[\ell]^2 \rightarrow \overk^*$ the Weil pairing (associated to the principal
 polarization of Remark~\ref{rk:weilconv}) and by
 $e_\pol: \Kpol^2 \rightarrow \overk^*$ the commutator pairing. Suppose that there
 exists $\ell_0$ a positive integer such that $K(\pol^{\ell_0})= A[\ell]$.
We have:
 \begin{enumerate}
	\item For $(x_1, x_2)
 \in A[\ell]^2$:
 \begin{equation}
	e_{A, \ell}(x_1, x_2)= e_{\pol^{\ell_0}}(x_1, x_2).
 \end{equation}
	\item
If $f: B \rightarrow A$ is an isogeny, for all $x,y \in f^{-1}(\Kpol)$:
\begin{equation}\label{eq:coupcompat}
e_{f^*(\pol)}(x,y) =e_\pol(f(x), f(y)).
\end{equation}
\item For $\kappa$ a positive integer, all $x \in \Kpol$ and $y \in [\kappa]^{-1}(\Kpol)$:
 \begin{equation}
 e_{\pol^\kappa}(x,y)=e_{\pol}(x, \kappa y).
 \end{equation}
\item If $\pol_1$ and $\pol_2$ are algebraically equivalent, then
$e_{\pol_1}=e_{\pol_2}$.
 \end{enumerate}
\end{proposition}

\begin{remark}\label{rk:weilconv}
 Suppose that $\pol=\pol_0^{n'}$ is a level $n'$ line bundle, with $\pol_0$ principal. For every positive
 integer $\ell$, we denote by $e_{A,\ell}: A[\ell]^2 \rightarrow \overk^*$ the Weil pairing
 associated to the principal polarization $\phi_{\pol_0}$ (which only depends on $\pol$
 since $\phi_\pol=n'\phi_{\pol_0}$). When $n' | \ell$, we have $K(\pol^{\ell/n'})=A[\ell]$ and
 $e_{A,\ell}$ is the pairing of Proposition~\ref{def:weil} (1); but $e_{A,\ell}$ is defined
 for all $\ell$, in particular on $A[d]$ even if $d$ does not divide $n'$. Recall (see
 \cite[Sections 20 and 23]{MumfordAV}) that $e_{A,\ell}$ is a perfect alternating pairing with
 values in $\mu_\ell$ and that for all positive integers $k, \ell$, $x \in A[k\ell]$ and $y \in
 A[\ell]$:
 \begin{equation}\label{eq:weilcompat}
	e_{A, k\ell}(x,y)=e_{A,\ell}(kx,y).
 \end{equation}
 Finally, for a group $H$ and a positive integer $\ell$, we denote by $H[\ell]$ the subgroup
 of elements of $H$ killed by $\ell$.
\end{remark}

\begin{remark}\label{rm:intro}
 Let $x \in A(\overk)$ be a torsion point and $\pol$ an ample symmetric line bundle on $A$. We show that there always
 exists an isogeny $f_0: B \rightarrow A$ and $y \in B(\overk)$ such that $f_0(y)=x$
 and $y \in K(f_0^*(\pol))$. Indeed, there exists a positive integer $\ell$ such
 that $\ell x \in \Kpol$. As
 $\pol$ is symmetric, $[\ell]^*(\pol)=\pol^{\ell^2}$ (see \cite{MumfordAV}). So we have
 $K([\ell]^*(\pol)) = [\ell^2]^{-1}(\Kpol)$ (see \cite[Proposition~4]{MumfordOEDAV1}).
 Let $f_0=[\ell]: A \rightarrow A$ and take $y \in f_0^{-1}(x)(\overk)$, then $y \in [\ell^2]^{-1}(\Kpol) \subset K([\ell]^*(\pol))$. 
\end{remark}

Since $\pol$ is symmetric, let $\psii: \pol \rightarrow [-1]^* (\pol)$ be the normalized isomorphism.
Denote by
$G_0(\pol)$ the group generated by $\Gpol$ and $([-1], \psii)$. It is clear that
$G_0(\pol)=\Gpol \rtimes \Z / 2 \Z$.
Following Mumford \cite{MumfordOEDAV1}, we define the group morphism:
\begin{equation}
 \begin{split}
	\delta_{-1}(\pol): \Gpol & \rightarrow \Gpol \\
	(\tau_x, \psi_x) & \mapsto ([-1], \psii) \circ (\tau_x, \psi_x) \circ ([-1],
	\psii).
 \end{split}
\end{equation}
When no confusion can arise, we abbreviate $\delta_{-1}(\pol)$ by $\delta_{-1}$.
With an easy computation, we get that $\delta_{-1}((\tau_x, \psi_x))=(\tau_{-x},
\tau^*_{-x}(\psii)^{-1} \circ [-1]^*(\psi_x) \circ
\psii)$.

\begin{definition}
An element $g_z \in \Gpol$ is symmetric if $\delta_{-1}(g_z)=g_z^{-1}$.
A level subgroup is symmetric if all its elements are symmetric. 
\end{definition}

The following important lemma is due to Mumford (see \cite[p.~308]{MumfordOEDAV1}):
\begin{lemma}\label{intro:lemma1}
If $x \in \Kpol$, there always exist exactly two $\psi_x$ such that $(\tau_x, \psi_x)$ is
symmetric: if $g_x=(\tau_x, \psi_x) \in \Gpol$ is symmetric, then $-g_x=(\tau_x, -\psi_x)$ is the other symmetric
element over $x$.
\end{lemma}
We will see in Lemma~\ref{lem:symlevel} below that, when $\pol$ is totally symmetric, there
always exists a symmetric level subgroup above any finite subgroup of $\Kpol$ isotropic for $e_\pol$.

\begin{definition}
 Let 
 $\mathrm{Km}(A)$ be the Kummer variety of $A$, that is, the quotient of $A$ by $[-1]$. Let
 $\pi_{\mathrm{Km}(A)}: A \rightarrow \mathrm{Km}(A)$ be the canonical projection.
 We say that $\pol$ is totally symmetric if there exists an ample line bundle $\pol_0$
 on $\mathrm{Km}(A)$ such that $\pi_{\mathrm{Km}(A)}^*(\pol_0)=\pol$.
\end{definition}
If $\pol$ is totally symmetric, then $[-1]^*(\pol)=\pol$ and if $([-1], \psii) \in
G_0(\pol)$, $\psii$ is the identity morphism.

\begin{lemma}\label{lem:symlevel}
 Suppose that $\pol$ is totally symmetric and let $H$ be a finite subgroup of $\Kpol$ isotropic for
 $e_\pol$. Then there exists a symmetric level subgroup $\tildeH$ of $\Gpol$ such that
 $\pigpol(\tildeH)=H$. This is not true in general if $\pol$ is only symmetric.
\end{lemma}
\begin{proof}
 We first show that if $x \in \Kpol \cap A[2]$ and $g=(\tau_x, \psi) \in \Gpol$ is symmetric, then
 $g^2=e_*(x)$, where $e_*(x) \in \{-1,1\}$ is the scalar by which $\psii$ acts on the fibre $\pol(x)$
 (note that $-x=x$) (see \cite[Definition~p.~304]{MumfordOEDAV1}). Writing $\psi_z: \pol(z) \rightarrow \pol(z+x)$ the map induced by $\psi$ on
 the fibre at $z$, the element $g^2$ is the scalar $\psi_x \circ \psi_0$. The relation
 $\delta_{-1}(g)=g^{-1}$, evaluated on the fibre at $0$, reads $(\psii)_x \circ \psi_0 \circ (\psii)_0
 =\psi_x^{-1}$ (recall that $x=-x$). As $\psii$ is normalized, $(\psii)_0=1$, so that $e_*(x)\psi_0=\psi_x^{-1}$
 and $g^2=\psi_x\circ\psi_0=e_*(x)^{-1}=e_*(x)$. If $\pol$ is totally symmetric, $\psii$ is the
 identity, so that $e_*(x)=1$ and $g^2=1$.

 Now suppose that $\pol$ is totally symmetric. Let $x \in H$ be of order $r$ and $g$ be a symmetric
 element above $x$ (Lemma~\ref{intro:lemma1}). As $\delta_{-1}$ is a morphism, the powers of $g$ are
 symmetric; $g^r$ is a symmetric element above $0$, that is a scalar $c$ with $c=c^{-1}$, so that
 $g^r=\pm 1$. If $r$ is odd, replacing $g$ by $-g$ if necessary, we get $g^r=1$. If $r$ is even,
 $g^{r/2}$ is a symmetric element above $(r/2)x \in \Kpol \cap A[2]$, so that $g^r=(g^{r/2})^2=1$ by the
 first part of the proof. Let $(h_1, \ldots, h_s)$ be elements of $H$ such that $H=\oplus_i (h_i)$,
 let $r_i$ be the order of $h_i$, and let $g_i$ be a symmetric element above $h_i$ such that
 $g_i^{r_i}=1$. As $H$ is isotropic, the $g_i$ commute with each other, so that the subgroup $\tildeH$
 generated by $g_1, \ldots, g_s$ consists of the $\prod_i g_i^{a_i}$, which are symmetric (as products of
 commuting symmetric elements) and lie above $\sum_i a_ih_i$. If $\sum_i a_ih_i=0$, then $r_i | a_i$
 for all $i$ and $\prod_i g_i^{a_i}=1$: $\tildeH$ is a symmetric level subgroup above $H$.

 Finally, if $\pol$ is only symmetric and $x \in H \cap A[2]$ satisfies $e_*(x)=-1$, an element of a
 symmetric level subgroup above $x$ would be a symmetric $g$ with $g^2=1$, which contradicts the first
 part of the proof.
\end{proof}

\begin{definition}\label{def:notationmaps}
For all $n \in \N^*$, let $\Zn$ be the group $(\frac1n \Z /\Z)^g$ and denote by $\dZn$ its
dual group, that is the group of characters of $\Zn$. If $m,n,d>1$ are integers such that
$n=md$, we denote by $\mu_{m,n}: Z(m) \rightarrow \Zn$ the canonical
injection. There is also a surjection $\nu_{n,m}: \Zn \rightarrow Z(m), x \mapsto dx$. We
denote by $\dnu_{m,n}: \dZ (m) \rightarrow \dZ (n)$ the injective dual of $\nu_{n,m}$
and by $\dmu_{n,m}: \dZ (n) \rightarrow \dZ (m)$ the surjective dual of $\mu_{m,n}$.
Finally, let
$\rho_{n,m}: \Zn \rightarrow Z(m)\simeq \Zn/\mu_{d,n}(Z(d))$ (resp. $\drho_{n,m} : \dZ(n)
\rightarrow \dZ(m) \simeq \dZ(n)/\dnu_{d,n}(\dZ(d))$)
be the canonical projection. The isomorphism $Z(m)\simeq \Zn/\mu_{d,n}(Z(d))$ is induced by
$\nu_{n,m}$, so that $\rho_{n,m}=\nu_{n,m}$ (and similarly $\drho_{n,m}=\dmu_{n,m}$).

In the following, to ease the notation when there is no ambiguity, we will often consider
$Z(m)$ (resp. $\dZ(m)$) as a subgroup of $Z(n)$ via $\mu_{m,n}$ (resp. $\dZ (n)$ via
$\dnu_{m,n}$).

Throughout the paper, for any pair of positive integers $s, t$ with $s \mid t$,
the notation $\mu_{s,t} : Z(s) \rightarrow Z(t)$, $\nu_{t,s} : Z(t) \rightarrow
Z(s)$, $\dnu_{s,t} : \dZ(s) \rightarrow \dZ(t)$, $\dmu_{t,s} : \dZ(t)
\rightarrow \dZ(s)$, $\rho_{t,s} : Z(t) \rightarrow Z(s)$ and $\drho_{t,s} :
\dZ(t) \rightarrow \dZ(s)$ refer to the canonical injection, the multiplication by
$t/s$, the dual injection, the dual surjection and the canonical projections defined by the
same formulae as above for the pair $(m,n)$. We use this notation freely for any
such pair without further comment.
\end{definition}

Let $\Kn=\Zn \times \dZn$. Note that $K(n)$ should not be confused with $\Kpol$, $K(\bpol)$, etc., which
denote throughout the paper the kernel of the isogeny associated to a line bundle.
In the following, we will consider $Z(n)$ and $\dZn$ as subgroups of $\Kn$ in the
obvious manner.
Denote by $G(n)$ the
level $n$ Heisenberg group, that is the set $\overk^* \times \Zn \times \dZn$ together with the
group law given by $(\alpha_1, x_1, y_1).(\alpha_2, x_2, y_2)=(\alpha_1 \alpha_2 y_2(x_1),
x_1+x_2, y_1+y_2)$. With the canonical projection $\pi_{G(n)}:G(n)\rightarrow \Kn$, $G(n)$ is a
central extension of $\Kn$ by $\overk^*$. We denote by $e_n: K(n) \times K(n) \rightarrow \overk^*$
the pairing induced by the commutator pairing on $G(n)$.
One can see that, for $g_1=(\lambda_1, \alpha_1, \beta_1)$ and $g_2=(\lambda_2, \alpha_2, \beta_2)$ in
$G(n)$, $g_1g_2g_1^{-1}g_2^{-1}=\beta_2(\alpha_1)/\beta_1(\alpha_2)$, that is
\begin{equation}\label{eq:commutator}
e_n((\alpha_1, \beta_1), 
(\alpha_2, \beta_2)) = \beta_2(\alpha_1) / \beta_1(\alpha_2). 
\end{equation}
Denote by $D_{-1}: G(n) \rightarrow G(n)$, $(\alpha, x, y) \mapsto (\alpha, -x, -y)$ the
morphism of Heisenberg groups.
\begin{definition}
A theta structure for $(A, \pol)$ of type $K(n)$ is an isomorphism $\Thetapol: G(n)
\rightarrow \Gpol$ compatible with the structures of central extension of $G(n)$ and
$\Gpol$. 
 A symplectic structure for $(A, \pol)$ of type $K(n)$ is a symplectic isomorphism (for
 $e_\pol$) $\Spol: K(n)
 \rightarrow \Kpol$. 

 A theta structure $\Thetapol$ is said to be symmetric if $\delta_{-1} \circ \Thetapol
 = \Thetapol \circ D_{-1}$. It is equivalent to the condition that
$\Thetapol(\{1\}\times\Zn\times\{0\})$ and 
$\Thetapol(\{1\}\times\{0\}\times\dZn)$ are symmetric level subgroups of $\Gpol$. 

A triple $(A, \pol,
\Thetapol)$ given by an abelian variety together with an ample totally symmetric line bundle and a symmetric
 theta structure of type $K(n)$ is called a marked abelian variety of type $K(n)$.

 If $G^* \subset G(n)$ is a subgroup containing $\overk^*$, a partial theta structure (resp.
 partial symmetric theta structure) of type $G^*$ is an injective group morphism $\Thetapol^*:
 G^* \rightarrow \Gpol$ (resp. satisfying $\delta_{-1} \circ \Thetapol^* = \Thetapol^*
 \circ D_{-1}$).
\end{definition}
\begin{remark}\label{rk:partialtheta}
 A theta structure $\Thetapol: G(n) \rightarrow \Gpol$ is
 equivalent to the data of partial theta structures $\Theta^1_\pol: \overk^*\times \Zn
 \rightarrow \Gpol$ and $\Theta^2_\pol: \overk^*\times\dZn \rightarrow \Gpol$.
\end{remark}
It is clear that a theta structure induces via the canonical projections $\pi_{G(n)}:G(n) \rightarrow
\Kn$ and $\pigpol:\Gpol \rightarrow \Kpol$ a symplectic structure $\Thetapolbar: K(n)
\rightarrow \Kpol$. Note
that if $\pol$ is totally symmetric of type $K(n)$, then $2|n$. By \cite[Proposition~2.4.2]{DRphd}, if $(A,
\pol)$ is of type $K(n)$ with $2|n$, there always exists a unique totally symmetric line
bundle in the algebraic class of $\pol$.

Following \cite{MumfordOEDAV1}, for $g=(\tau_x, \psi_x) \in \Gpol$ and $s \in \Gamma(A,\pol)$, let
$U_g(s)=\tau^*_{-x}(\psi_x(s))$; this defines a left action of $\Gpol$ on $\Gamma(A,\pol)$, for which
$U_g(s)$ ``at $z$'' corresponds to $s$ ``at $z-x$''. In order to translate points by $+x$, we
use the following left action of $\Gpol$ on $\Gamma(A,\pol)$: for $g \in \Gpol$ and $s \in
\Gamma(A,\pol)$, we set $g.s=U_{\delta_{-1}(g)}(s)$. If $g=(\tau_x, \psi_x)$ is symmetric, then
$\delta_{-1}(g)=g^{-1}=(\tau_{-x}, \tau^*_{-x}(\psi_x)^{-1})$ and $g.s=\psi_x^{-1}(\tau^*_x(s))$, where
$\psi_x^{-1}: \tau_x^*\pol \rightarrow \pol$ is the inverse of $\psi_x$. In all cases, the elements
of $\overk^* \subset \Gpol$ act by multiplication. An important property of a theta structure is that it defines a basis of
$\Gamma(A,\pol)$.

\begin{proposition}[Basis of $\Gamma(A,\pol)$ associated to a theta structure]\label{prop:base}
To each theta structure for $(A, \pol)$ one can associate a basis of $\Gamma(A,\pol)$, defined up to a constant multiple, as follows. Consider the linear endomorphism $\pi_\Thetapol: \Gamma(A,\pol) \rightarrow
\Gamma(A,\pol)$ defined
by $s \mapsto \sum_{\lambda \in \Thetapol(\{1\}\times\{0\}\times\dZn)} \lambda.s$. Then $\pi_\Thetapol^2=n^g\pi_\Thetapol$, so that $\pi_\Thetapol$ is, up to the factor $n^g$, a
projection, whose image is a $1$-dimensional subspace of $\Gamma(A,\pol)$. Let
$\theta_0^\Thetapol$ be any generator of this subspace, then $(\theta_i^\Thetapol)_{i \in \Zn}:=
(\Thetapol((1,i,0)).\theta_0^\Thetapol)_{i \in \Zn}$ forms a basis of $\Gamma(A,\pol)$ canonically associated to the
theta structure $\Thetapol$.
\end{proposition}

\begin{definition}\label{def:thetanull}
 Let $(A, \pol, \Thetapol)$ be an abelian variety together with a theta structure of
 type $K(n)$. We define $\proj^{Z(n)}$ as $\projf(k[X_i, i \in \Zn])$ the projective
 space associated to the graded ring $k[X_i, i \in \Zn]$ (see \cite[Section~II-2]{hartshorne2000algebraic}).

 The canonical basis of $\Gamma(A, \pol)$ defines a morphism, which is an embedding if $n \geq 3$
 (if $n=2$, its image is the Kummer variety of $A$), 
 \begin{equation}
	e_\Thetapol: A \rightarrow \proj^{Z(n)},
 \end{equation}
 such that for all $i \in \Zn$, $e_\Thetapol^*(X_i)= \theta_i^\Thetapol$. 
 Let $0_A \in A(k)$ be the neutral point of $A$. The projective point
 $e_\Thetapol(0_A)$,
 with projective coordinates $(\theta_i^\Thetapol(0_A))_{i \in \Zn} \in \proj^{Z(n)}$,
is called the theta null point of $(A, \pol, \Thetapol)$ and is denoted $0_\Thetapol$. By abuse
of notation, we shall also write $0_\Thetapol$ for the neutral point $0_A$ (for instance in
$\rho^\pol_{0_\Thetapol}$).
\end{definition}
\begin{remark}\label{rq:rqgi}
Let $\gi=([-1], \psii) \in G_0(\pol)$. For $\lambda \in \Gpol$, $\gi \lambda \gi^{-1}=\delta_{-1}(\lambda)$,
and $\delta_{-1}(\lambda)=\lambda^{-1}$ for $\lambda \in \Thetapol(\{1\}\times\{0\}\times\dZn)$ since this
level subgroup is symmetric. Hence, if $s \in \Gamma(A,\pol)$ is invariant under
$\Thetapol(\{1\}\times\{0\}\times\dZn)$, then so is $\gi.s$. By Proposition~\ref{prop:base}, the
subspace of invariant sections is the line spanned by $\theta_0^\Thetapol$, so that
$\gi.\theta_0^\Thetapol=\mu \theta_0^\Thetapol$ for some $\mu \in \overk$, and $\mu \neq 0$ since $\gi$ acts
bijectively on $\Gamma(A,\pol)$. As $\psii$ is normalized, $\gi^2=([-1]^2, [-1]^*(\psii)\circ\psii)$ is an
automorphism of $\pol$ above the identity which is the identity on the fibre at $0_A$, so that
$\gi^2=1$ and $\mu=\pm 1$. Moreover, $\gi. \theta_i^\Thetapol = \gi.\Thetapol((1,i,0)).\theta_0^\Thetapol=
\delta_{-1}(\Thetapol((1,i,0))).\gi.\theta_0^\Thetapol=\mu\,\Thetapol((1,-i,0)).\theta_0^\Thetapol=\mu\,\theta_{-i}^\Thetapol$.
We have obtained that for all $i \in \Zn$:
 \begin{equation}
\gi. \theta_i^\Thetapol =\mu\, \theta_{-i}^\Thetapol, \quad \mu \in \{-1, 1\}.
 \end{equation}
 The sign $\mu$ cannot be determined by this argument alone: if $\mu=-1$, then $\pi_\Thetapol(s)=0$
 for every section $s$ invariant under $\gi$, and $\mu=-1$ can happen when $\pol$ is only symmetric (odd theta
 characteristics). When $\pol$ is totally symmetric, $\mu=1$: this is the inverse formula of
 \cite[p.~331]{MumfordOEDAV1}.
\end{remark}
 When $\pol$ is totally symmetric, $\psii$ is the identity (see above) and $\mu=1$, so that the
 preceding equation becomes the inverse formula of \cite[p.~331]{MumfordOEDAV1}:
\begin{lemma}\label{sec1:leminv}
 Let $(A, \pol, \Thetapol)$ be a marked abelian variety, we have:
 \begin{equation}\label{eq:inverse}
	[-1]^* \theta_i^\Thetapol =\theta_{-i}^\Thetapol.
 \end{equation}
\end{lemma}

Using the theta structure and the action of $\Gpol$ on $\theta_i^\Thetapol$, we get
an action of $G(n)$ on $\theta_i^\Thetapol$. For $(\alpha,x,y) \in G(n)$, it is
given by:
\begin{gather}
 \begin{aligned}\label{eq:trans}
	(\alpha,x,y).\theta_i^\Thetapol & = y(-x) (\alpha,x,0) (1, 0, y)(1, i, 0).
	\theta_0^\Thetapol\\ & =y(-x) y(-i) (\alpha, i+x, 0) (1,0,y).\theta_0^\Thetapol \\
			 & =\alpha y(-i-x) \theta_{i+x}^\Thetapol.
\end{aligned}
\end{gather}
Clearly, by letting $G(n)$ act on $(\theta_i^\Thetapol(0_A))_{i \in \Zn}$, we recover all points of $A[n]$.

\begin{definition}\label{def:affinepoint}
Recall that $\pi_{\proj^{Z(n)}}: \Aff^{Z(n)}-\{0\}
\rightarrow \proj^{Z(n)}$ is the canonical projection. If we identify $x\in A(\overk)$ with
 $e_\Thetapol(x)$, we say that $\tildexthetapol$ is an affine lift of $x$ if
 $\tildexthetapol \in
 \Aff^{Z(n)}(\overk)$ and
$\pi_{\proj^{Z(n)}}(\tildexthetapol)=x$. When no confusion can arise, we abbreviate
 $\tildexthetapol$ to $\tildex$. If $\tildexthetapol$ is an affine point, for $i \in
 \Zn$, we denote by
 $(\tildexthetapol)_i$ its $i$-th coordinate.
\end{definition}

Let $(A, \pol, \Thetapol)$ be a marked abelian variety. Let $\stsheaf_A$ be the
structural sheaf of $A$ and $x \in A(\overk)$ be a geometric point. A rigidification of
$\pol$ at $x$ is a choice of an isomorphism $\rho^\pol_x: \pol(x)=\pol \otimes
\stsheaf_A(x)
\rightarrow \stsheaf_A(x)$. Any such morphism can be obtained by taking a local
trivialization $\pol_x
\rightarrow \stsheaf_{A,x}$ and doing a base change by the canonical evaluation morphism
$\stsheaf_{A,x} \rightarrow \stsheaf_{A}(x)$. Thus it can be seen as a way to evaluate $s
\in \pol_x$ in $x$ to obtain $\rho^\pol_x(s) \in \overk$. A morphism $\psi: (\pol, \rho^\pol_x) \rightarrow 
(\bpol, \rho^\bpol_x)$ of rigidified line bundles is a morphism $\psi: \pol \rightarrow
\bpol$ such that 
$\rho^\bpol_x \circ \psi(x) \circ (\rho^\pol_x)^{-1}$ is the identity of $\stsheaf_A(x)$.
Such a morphism, if it exists, is unique, even though the set of morphisms $\psi:
\pol \rightarrow \bpol$ is a principal homogeneous space over $\overk^*$.
If
${\rho'}^\pol_x$ is any other rigidification of $\pol$ at $x$, then there exists $\lambda \in
\overk^*$ such that ${\rho'}^\pol_x=\lambda \rho^\pol_x$.

\begin{definition}\label{def:marked2} 
 The data of $(A, \pol, \Thetapol, \theta_0^\Thetapol, \rho^\pol_{x})$, a marked abelian variety of type
 $K(n)$ together with:
 \begin{itemize}
	\item a generator $\theta_0^\Thetapol$ of the image subspace of endomorphism
	 $\pi_\Thetapol: \Gamma(A,\pol) \rightarrow \Gamma(A,\pol)$ of Proposition
	~\ref{prop:base};
	\item a rigidification $\rho^\pol_x$ of $\pol$ at a point $x \in A(\overk)$ 
 \end{itemize}
 is called a marked rigidified abelian variety or more simply a rigidified abelian variety (of type $K(n)$).
\end{definition}

\begin{remark}\label{sec:rmequiv}
 From $(A, \pol, \Thetapol, \theta_0^\Thetapol, \rho^\pol_{x})$, following
 Proposition~\ref{prop:base}, one recovers the unique basis $({\theta_i}^\Thetapol)_{i \in \Zn}$ of
 $\Gamma(A,\pol)$ defined by the theta structure and $\theta_0^\Thetapol$ and then the
 affine lift $(\rho^\pol_{x}(\theta^\Thetapol_i(x)))_{i \in \Zn}$ of $x
 \in A(\overk)$. Conversely, the data of the affine lift
 $\tildex$ of $x \in A(\overk)$ is equivalent to the data of a rigidification
 $\rho^\pol_x$ once we have fixed $\theta_0^\Thetapol$ a generator of the image of
 $\pi_\Thetapol$ (defined in Proposition~\ref{prop:base}). Note however that if $(A, \pol, \Thetapol, \theta_0^\Thetapol,
 \rho^\pol_{x})$ gives the affine lift $\tildex$, then any other rigidified
 abelian variety of the form $(A, \pol,
 \Thetapol, \lambda.\theta_0^\Thetapol,
 \frac1\lambda \rho^\pol_x)$ for $\lambda \in \overk^*$ will give the same
 affine lift $\tildex$. Two rigidified abelian
 varieties at $x$ having the same affine lift $\tildex$ are called equivalent.
\end{remark}

Let $\rho^\pol_x: \pol(x) \rightarrow \stsheaf_A(x) \subset \overk$ be a rigidification of $\pol$ in $x \in
A(\overk)$. Let $g_y \in \Gpol$ be above $y$ and write $\delta_{-1}(g_y)=(\tau_{-y}, \psi'_{-y})$, with
$\psi'_{-y}: \pol \rightarrow \tau^*_{-y}\pol$, whose fibre at $x+y$ is an isomorphism $\psi'_{-y}(x+y):
\pol(x+y) \rightarrow \pol(x)$ (if $g_y=(\tau_y, \psi_y)$ is symmetric, $\psi'_{-y}(x+y)=\psi_y(x)^{-1}$).
We obtain a rigidification $\rho^\pol_{x+y}$ of $\pol$ in $x+y$:
\begin{equation}\label{eq:rigid}
 \begin{tikzpicture}
 \matrix [column sep={1cm}, row sep={1cm}]
 { 
 \node(a){$\pol(x+y)$}; & \node(b){$\pol(x)$}; & \node(c){$\overk$.}; \\
 };
 \draw [->] (a) -- (b) node[above, midway]{$\psi'_{-y}(x+y)$};
 \draw [->] (b) -- (c) node[above, midway]{$\rho^\pol_x$};
 \end{tikzpicture}
\end{equation}
In particular, we have an action of $\Gpol$ on affine lifts of geometric points of $A$:
if $\tildex=(\rho^\pol_x(s_i))_{i \in \Zn}$, for
$(s_i)_{i\in \Zn}$ a basis of the $k$-vector space $\Gamma(A,\pol)$,
is an affine lift of $x \in
A(\overk)$ and
$g_y \in \Gpol$ is above $y$, we get an affine lift of $x+y$:
\begin{equation}\label{eq:actthetagroup}
\tildexpy=
(\rho^\pol_x(g_y.s_i))_{i \in \Zn},
\end{equation}
since $g_y.s=U_{\delta_{-1}(g_y)}(s)=\tau^*_y(\psi'_{-y}(s))$ takes at $x$ the value $\psi'_{-y}(x+y)(s(x+y))$.

We can gather all these remarks in the following lemma:
\begin{lemma}\label{sec1:lemact}
 Let $\rho^\pol_{x}: \pol(x) \rightarrow \stsheaf_A(x) \subset \overk$ be a
 rigidification of $\pol$ in $x \in
 A(\overk)$.
For $g_y \in \Gpol$ above $y$, Diagram (\ref{eq:rigid}) gives a
 rigidification 
 $\rho^\pol_{x+y}$ of $\pol$ in $x+y$ such that for all $s \in \Gamma(A,\pol)$:
 \begin{equation}
 \rho^\pol_{x+y}(s)=\rho^\pol_x(g_y.s).
 \end{equation}
\end{lemma}
\begin{remark}\label{rk:inverse}
    Note that because $\psii : \pol \rightarrow [-1]^*(\pol)$ is normalized, we have for
    all $s \in \Gamma(A, \pol)$, $\rho_0(g_{-1}(s))=\rho_0(s)$.
\end{remark}

\begin{remark}\label{def:descendrigid}
Suppose now that there is an isogeny $f: A \rightarrow B$. Let $\bpol$ be an ample line
bundle on $B$ such that $\pol = f^*(\bpol)$. Let $x \in A(\overk)$ and choose
$\rho^\pol_{x}:
\pol_x \rightarrow \overk$ a rigidification. Let $x_0 =f(x)$.
Note that as $\pol(x) = f^*(\bpol(x_0))$,
$\rho^\pol_{x}$ defines a rigidification that we denote by $f(\rho^\pol_{x})$ of $\bpol$
in $x_0$ such that
$f(\rho^\pol_{x})(s)=\rho^\pol_x(f^*(s))$ for all $s \in \bpol_{x_0}$. 
With this definition, we note that Diagram (\ref{eq:rigid}) gives an action of $\Gpol$
not only on affine lifts of $A$, but also on affine lifts of $B$.

 For $x \in A(\overk)$ and $y \in \Kpol$, there is a unique $g_y=(\tau_y, \psi_y) \in \Gpol$ such
 that $g_y.\tildex = \tildexpy$. This means that computing with affine lifts on $B$
 allows one to fix elements of $\Gpol$. This idea, which was already present in \cite{DRmodular}, is also
 one of the main techniques used in the present paper.
\end{remark}

From the knowledge of $(A, \pol, \Thetapol, \theta_0^\Thetapol, \rho^\pol_{0_A})$ a rigidified abelian variety, we get
$\tildenullpol$ its affine theta null point but also for all $i \in \Zn$ (resp. $j \in \dZn$)
the affine lift $\Thetapol((1, i,0)).\tildenullpol$ (resp.
$\Thetapol((1,0,j)).\tildenullpol$) of
$\Thetapolbar((i,0))$ (resp. $\Thetapolbar((0,j))$).

In what follows, for $s \in \Gamma(A,\pol)$ and
points $x_1, \ldots, x_k$ in an affine open subset $U$ of $A$, we denote by $s(x_1), \ldots, s(x_k)$ the evaluations
of $s$ at $x_1, \ldots, x_k$ obtained by choosing a local trivialization $\pol|_U \simeq
\stsheaf_A|_U$, where $\stsheaf_A$ is the structure sheaf of $A$.
In order to state the Riemann equations, we introduce the following definition.
\begin{definition}\label{sec1:defrieman}
 Let $G$ be an abelian group. We say that elements $(x_1, \ldots, x_4; y_1, \ldots, y_4)$ are in Riemann position if
 there exists $z \in G$ such that $-x_1 + x_2 + x_3 + x_4 = 2z$ and $y_1 = x_1+z$,
 $y_2=x_2-z$, $y_3=x_3-z$, $y_4=x_4-z$.
\end{definition}
\begin{theorem}{\cite[Theorem~1]{DRpairing}}\label{sec:thriemann}
 Let $(A, \pol, \Thetapol)$ be a marked abelian variety of type $K(n)$ with $n$ a
 positive even integer. We consider $Z(2)$ as a subgroup of $Z(n)$ via $\mu_{2,n}$. Let
 $(x_1, \ldots, x_4; x_5, \ldots, x_8)$ be elements of
 $A(\overk)$ (resp. let $(i_1, \ldots, i_4; i_5, \ldots, i_8)$ be elements of $Z(n)$)
 in Riemann position. For any $\chi \in \dZ(2)$, $i,j \in \Zn$, $x,y \in A(\overk)$,
 we set:
 \begin{align*}
	L(\Thetapol,\chi,i,j,x,y) &= \sum_{\eta \in Z(2)} \chi(\eta) \theta_{i +
	\eta}^\Thetapol(x) \theta_{j + \eta}^\Thetapol(y).
 \end{align*}
Then we have:
 \begin{equation}\label{eq:riemaneq}
	L(\Thetapol,\chi,i_1, i_2, x_1, x_2) L(\Thetapol, \chi, i_3, i_4, x_3, x_4)
	=L(\Thetapol,\chi,i_5, i_6, x_5, x_6)
	L(\Thetapol, \chi, i_7, i_8, x_7, x_8).
 \end{equation}
 By summing (\ref{eq:riemaneq}) over all $\chi \in \dZ(2)$, we obtain another form
 of the Riemann relations:
 \begin{equation}\label{eq:riemaneq1}
	\sum_{\eta \in Z(2)} \prod_{j=1}^4 \theta_{i_j+\eta}^\Thetapol(x_j)=
\sum_{\eta \in Z(2)} \prod_{j=5}^8 \theta_{i_j+\eta}^\Thetapol(x_j).
 \end{equation}
\end{theorem}

\begin{remark}
 The Riemann relations are an easy consequence of the duplication formulae (see \cite{MumfordOEDAV1}),
 as shown in \cite[Theorem~1]{DRpairing}. Several formulations appear in the literature. Equation~(C') of \cite{MumfordOEDAV1} is a variant of our version which is only valid for theta
 null values, as it uses the symmetry relations; our version coincides with Equation~(C) of \cite{MumfordOEDAV1} and with Theorem~1 of \cite{DRpairing}.

 The two sides of (\ref{eq:riemaneq}) involve different points, so that the
 $\theta_i^\Thetapol(x_j)$ cannot be evaluations with respect to arbitrary rigidifications: the
 relations hold for suitably normalized affine lifts of the $x_j$ (over $\C$, for the values of
 the analytic theta functions at lifts $z_1, \ldots, z_8 \in \C^g$ of the $x_j$ which are in
 Riemann position in $\C^g$, see \cite{DRpairing}), and this is how they are used to define
 $\diffadd$ and $\threeway$ below. When both sides involve the same points, as in
 (\ref{eq:riemaneqv}) below, the relations are homogeneous in each point and hold for any choice
 of affine lifts.
\end{remark}

Denote by $\Modu_n$ the locus of theta null points of marked abelian varieties $(A, \pol, \Thetapol)$
of a fixed type $K(n)$, for $n$ a positive integer. Setting $x_i = 0$ in the Riemann
relations yields equations satisfied by $\Modu_n$. By Lemma~\ref{sec1:leminv}, theta null points also
satisfy the symmetry relations:
\begin{proposition}\label{sec:propsym}
For all $i \in \Zn$, we have $\theta_i^\Thetapol(0_A)=\theta_{-i}^\Thetapol(0_A)$.
\end{proposition}

Denote by $\bModu_n$ the closed subvariety of $\proj^{Z(n)}$ given in the projective
coordinates $(\theta_i^\Thetapol(0_A))_{i \in \Zn}$ by Riemann and symmetry relations. By
\cite{MumfordOEDAV2,kempf1989linear}, if $4\mid n$, then $\Modu_n$ is a
quasi-projective variety that is an open dense subset of $\bModu_n$.

A point $x \in \bModu_n(\overk)$ determines a theta null point $(\theta_i(0))_{i \in \Zn}$,
and Theorem~\ref{sec:thriemann} provides a set of homogeneous equations satisfied by
the variety $e_\Thetapol(A)$. Indeed, for $(i_1, \ldots, i_4; i_5, \ldots, i_8)$ elements
of $Z(n)$ in Riemann position, we have the relation:
 \begin{equation}\label{eq:riemaneqv}
	L(\Thetapol, \chi,i_1, i_2, x, x) L(\Thetapol, \chi, i_3, i_4, 0, 0) =
	L(\Thetapol,\chi,i_5, i_6, x, x)
	L(\Thetapol, \chi, i_7, i_8, 0, 0).
 \end{equation}
 We have the following result of Mumford \cite{MumfordOEDAV1}:
\begin{theorem}{\cite{MumfordOEDAV1}}
 If $4\mid n$, the relations (\ref{eq:riemaneqv}) form a complete set of homogeneous equations for $e_\Thetapol(A)$.
\end{theorem}

Let $(\univAb_n, \poluniv, \Thetapoluniv)$ be the universal
marked abelian variety of type $K(n)$ which is an abelian variety over $\Modu_n$ whose
fiber over any point $x$ of $\Modu_n$ is the marked abelian variety of type $K(n)$ defined
by $x$. The preceding theorem tells us that the relations (\ref{eq:riemaneqv}) form a complete set of
equations for $e_\Thetapol(\univAb_n)$.

Let $(A, \pol)$ be an abelian variety together with an ample line bundle.
Let $f: A \rightarrow B$ be a separable isogeny with kernel $K \subset \Kpol$ that we
suppose isotropic for $\epol$. Let $\pigpol: \Gpol \rightarrow \Kpol$ be the canonical
map. By the descent theory of Grothendieck, there is a bijection between the set of level
subgroups $\tildeK$ such that $\pigpol(\tildeK)=K$ and the set of pairs $(\bpol, \psi)$
where $\bpol$ is an ample line bundle on $B$ and $\psi: f^*(\bpol) \rightarrow \pol$ is
an isomorphism, taken up to isomorphism. It sends $\tildeK$ to the unique pair $(\bpol, \psi)$ such that for all
$(\tau_x, \psi_x) \in \tilde K$ the following diagram commutes:
\begin{equation}\label{eq:descentdata}
 \begin{tikzpicture}
 \matrix [column sep={1cm}, row sep={1cm}]
 { 
 \node(a){$f^*(\bpol)$}; & \node(b){$\pol$}; \\
 \node(c){$\tau_x^*(f^*(\bpol))$}; & \node(d){$\tau_x^*(\pol)$}; \\
 };
 \draw [->] (a) -- (b) node[above, midway]{$\psi$};
 \draw [->] (c) -- (d) node[above, midway]{$\tau^*_x(\psi)$};
 \draw [->] (b) -- (d) node[right, midway] {$\psi_x$};
 \draw [double distance = 2 pt] (a) -- (c);
 \end{tikzpicture}
\end{equation}
Conversely, if $(\bpol, \psi)$ is such a pair, we recover $\tildeK$ by saying that for
all $x \in K$, $(\tau_x, \psi_x) \in \tildeK$ if $\psi_x$ is the isomorphism making Diagram
(\ref{eq:descentdata}) commutative.

In other words, $(\bpol, \psi)$ is the quotient of $\pol$ by $\tildeK$.
In what follows, we refer to $\tildeK$ as descent data of $\pol$ to $\bpol$. 
If $(\bpol, =)$ is the quotient of $\pol$ by $\tildeK$, we will write $\bpol = \pol /
\tildeK$.

\begin{remark}\label{rk:descent}
 Using Diagram (\ref{eq:descentdata}), we see that
 for $s \in \Gamma(A,\pol)$, there exists $s_0 \in \Gamma(B,\bpol)$ such that $ \psi(f^*(s_0))=s$ 
 if and only if
 $x.s = s$ for all $x \in \tildeK$.
\end{remark}

\begin{definition}\label{def:fsharp}
 Let $(A, \pol)$ be an abelian variety together with an ample line bundle.
 Let $f: A \rightarrow B$ be a separable isogeny with kernel $K$. 
Let $\psi: f^*(\bpol) \rightarrow
 \pol$ be an isomorphism where $\bpol$ is an ample line bundle on $B$ and denote by
 $\tildeK \subset \Gpol$ the descent data of $\pol$ to
 $\bpol$. Let $\Gspol$ be the centralizer of $\tildeK \subset
 \Gpol$.
 We define the group morphism:
 \begin{equation}
	\begin{split}
	 \fsharp(\pol): \Gspol & \rightarrow G(\bpol) \\
	 (\tau_x, \psi_x) & \mapsto (\tau_y, \psi_y)
	\end{split}
 \end{equation}
 where $(\tau_y, \psi_y)$ is such that $f(x)=y$ and 
$\psi_y
 : \bpol \rightarrow \tau^*_y \bpol$ is the unique isomorphism satisfying
 $f^*(\psi_y)=\tau^*_x(\psi^{-1}) \circ \psi_x \circ\psi$.

 When no confusion is possible, we will replace $\fsharp(\pol)$ by $\fsharp$.
\end{definition}

To see that there exists $\psi_y: \bpol \rightarrow \tau^*_y
\bpol$ such that $f^*(\psi_y)=\psi'_y:=\tau_x^*(\psi^{-1}) \circ \psi_x \circ\psi$, it suffices to show that for all $z$ closed
 point of $K$, $\tau^*_z(\psi'_y)=\psi'_y$, which 
 is an immediate consequence of the fact
 that $\Gspol$ commutes with $\tildeK$ and the following diagram, where $(\tau_z,
 \psi_z) \in \tildeK$:
\begin{equation}\label{diag:descent}
 \begin{tikzpicture}
 \matrix [column sep={1.5cm}, row sep={1cm}]
 { 
	\node(a){$f^*(\bpol)$}; & \node(b){$\pol$}; & \node(c) {$\tau^*_x \pol$}; & \node(d)
	{$\tau^*_x f^*(\bpol)$}; \\
	\node(ap){$\tau^*_z f^*(\bpol)$}; & \node(bp){$\tau^*_z\pol$}; & \node(cp)
	{$\tau^*_{x+z} \pol$}; &
	\node(dp)
	{$\tau^*_{x+z} f^*(\bpol)$}; \\
 };
	\draw[->] (a) -- (b) node[above, midway]{$\psi$};
	\draw[->] (b) -- (c) node[above, midway]{$\psi_x$};
	\draw[->] (c) -- (d) node[above, midway]{$\tau^*_x(\psi^{-1})$};
	\draw[->] (ap) -- (bp) node[above, midway]{$\tau^*_z(\psi)$};
	\draw[->] (bp) -- (cp) node[above, midway]{$\tau^*_z(\psi_x)$};
	\draw[->] (cp) -- (dp) node[above, midway]{$\tau^*_{x+z}(\psi^{-1})$};
	\draw[double distance = 2 pt] (a) -- (ap);
	\draw[double distance = 2 pt] (d) -- (dp);
	\draw[->] (b) -- (bp) node[right, midway]{$\psi_z$};
	\draw[->] (c) -- (cp) node[right, midway]{$\tau^*_x(\psi_z)$};
 \end{tikzpicture}
\end{equation}

\begin{remark}\label{rm:sym}
 Keeping the notation of Definition~\ref{def:fsharp}, we note that if $\bpol$ is
 a symmetric ample line bundle, then $f^*(\bpol)$ is a symmetric line bundle. But it is not
 true that if $\pol$ is symmetric then its descent by $\tildeK$ is also symmetric. Let
 $\psii: \pol
 \rightarrow [-1]^*(\pol)$ be an isomorphism. Then $\bpol$ is symmetric if and only if
 $\psii$ descends by $\tildeK$ to an isomorphism $\bpol \rightarrow [-1]^*(\bpol)$. This
 is equivalent to $\psii$ commuting with the descent data $\tildeK$, which means that for
 all $(\tau_x, \psi_x), (\tau_{-x}, \psi_{-x}) \in \tildeK$ the following diagram is
 commutative:
\begin{equation}\label{diag:descentsym}
 \begin{tikzpicture}
 \matrix [column sep={1.5cm}, row sep={1cm}]
 { 
     \node(a){$\pol$}; & \node(b){$[-1]^*(\pol)$}; & \\
     \node(c){$\tau^*_{-x} (\pol)$}; & \node(d){$[-1]^*
     \tau^*_x(\pol)=\tau^*_{-x}[-1]^*(\pol)$}; \\
 };
	\draw[->] (a) -- (b) node[above, midway]{$\psii$};
	\draw[->] (a) -- (c) node[left, midway]{$\psi_{-x}$};
	\draw[->] (b) -- (d) node[left, midway]{$[-1]^*(\psi_x)$};
	\draw[->] (c) -- (d) node[above, midway]{$\tau^*_{-x}(\psii)$};
 \end{tikzpicture}
\end{equation}
Thus $\bpol$ is symmetric if and only if $\tildeK$ is symmetric.
\end{remark}
\begin{proposition}[Compatibility of the action with isogeny,
 \cite{MumfordOEDAV1}]\label{prop:compat} 
 Keeping the notation of Definition~\ref{def:fsharp}, $\fsharp(\pol)$ induces a canonical group morphism:
 \begin{equation}\label{eq:quotientthetagroup}
	\fsharp(\pol): \Gspol/\tildeK \rightarrow \Gbpol,
 \end{equation}
 which is an isomorphism. In particular, we have $\pigpol(\Gspol)=f^{-1}(K(\bpol))$ and
 \begin{equation}
	f(K^{\perp_{\epol}})= K(\bpol),
 \end{equation}
 where $K^{\perp_{\epol}}$ is the $\epol$-orthogonal of $K$ in $\Kpol$.


\end{proposition}
The following corollary results from the fact that $\fsharp(\pol)$ is
canonical:
\begin{corollary}\label{cor:canonical}
 Keeping the notation of \cref{prop:compat}, for all $s \in \Gamma(B,\bpol)$ and $t \in \Gspol$, we have:
 \begin{equation}
	t.\psi(f^*(s))= \psi(f^* (\fsharp(\pol)(t).s)).
 \end{equation}
\end{corollary}

The two preceding results, although elementary, have as an immediate consequence the isogeny
theorem \cite{MumfordOEDAV1}, which is a cornerstone of Mumford's theory of algebraic theta
functions. They will be used below to prove further results of the same kind.

\section{Isogeny compatibility}\label{sec:isogcompat}
Let $n,d,m>1$ be integers such that $n=dm$.
Let $(A, \pol, \Thetapol)$ be a marked abelian variety of 
type $K(n)$. Let $f: A
\rightarrow B$ be an isogeny with kernel $K\subset \Kpol$ isotropic for $\epol$ and
isomorphic as a group to $\dZ(d)$. Let $\tildeK$ be a level subgroup above $K$
and let $\bpol = \pol / \tildeK$. By \cref{prop:compat}, we have $K(\bpol) \simeq
K^{\perp_{\epol}}/K$. If necessary, by changing $\Thetapol$, we can suppose that $K=
\Thetapolbar(\{0\}\times \dnu_{d,n}(\dZ(d)))$ or that $K=\Thetapolbar(\mu_{d,n}(Z(d))
\times \{0 \})$. In the first case, we have $K^{\perp_{\epol}} \simeq
Z(m) \times \dZn$, so
that $K(\bpol) \simeq (Z(m) \times \dZn)/(\{0\}\times \dnu_{d,n}(\dZ(d))) \simeq Z(m)
\times \dZ(m)=K(m)$ and in the second case, $K^{\perp_{\epol}} \simeq
Z(n) \times \dZm$, so
that $K(\bpol) \simeq (Z(n) \times \dZm)/(\mu_{d,n}(Z(d))\times \{0 \}) \simeq Z(m)
\times \dZ(m)=K(m)$. This motivates the following definition.

\begin{definition}\label{sec1:def1}
 Let $n,d,m>1$ be integers such that $n=dm$.
Let $(A, \pol, \Thetapol)$ and $(B, \bpol, \Thetabpol)$ be marked abelian varieties of
 respective types $K(n)$ and $K(m)$. Let $f:
A \rightarrow B$ be an isogeny with kernel $K\simeq \dZ(d)$ isotropic for $\epol$.
 Denote by $\Spol: \Kn \rightarrow \Kpol$ the
 symplectic structure defined by $\Thetapol$.

 We say that $(A, \pol, \Thetapol)$ and $(B, \bpol, \Thetabpol)$ are isog-$f$-compatible
 (resp. dual-isog-$f$-compatible) if:
 \begin{enumerate}
	\item $K =\Spol(\{0\} \times \dnu_{d,n}(\dZ(d)))$ (resp. $K=\Spol(\mu_{d,n}(Z(d))
	 \times \{0 \})$);
	\item $f^*(\bpol) = \pol$ and the level subgroup above $K$ associated to $(f^*(\bpol), =)$
	 is $\tildeK=\Thetapol(\{1\}\times\{0\}\times\dnu_{d,n}(\dZ(d)))$ (resp.
	 $\tildeK= \Thetapol(\{1\} \times \mu_{d,n}(Z(d)) \times \{0 \})$);
	\item If $\fsharp$ is the isomorphism of \cref{prop:compat}, for all $x \in Z(m)$,
	 $\fsharp(\Thetapol((1,\mu_{m,n}(x),0))) = \Thetabpol((1,x,0))$ (resp. for all
	 $x \in \Zn$, $\fsharp(\Thetapol((1, x,0)))=\Thetabpol((1, \rho_{n,m}(x), 0))$, where
	 $\rho_{n,m}$ is given by Definition~\ref{def:notationmaps});	 
	\item For all $y \in \dZn$, $\fsharp(\Thetapol((1,0,y)))=
	 \Thetabpol((1,0,\drho_{n,m}(y)))$, where $\drho_{n,m}$ is given by Definition
	~\ref{def:notationmaps} (resp. for all $y \in \dZm$, $\fsharp(\Thetapol((1,
	 0, \dnu_{m,n}(y))))=\Thetabpol((1, 0, y))$).
 \end{enumerate}
\end{definition}

The following proposition shows that, in the coordinate system given by
isog-$f$-compatible theta structures, the image of a point is obtained simply by dropping certain
coordinates, while a slightly more involved formula holds in the dual-isog-$f$-compatible
case. Note that Definition~\ref{sec1:def1} and Proposition~\ref{sec1:prop1} are analogous to
\cite[Section~3.1]{DRmodular} and \cite[Proposition~7]{DRmodular} respectively, except that we do not assume $d$ coprime to $m$.
\begin{proposition}\label{sec1:prop1}
 Let $n,m,d >1$ be integers such that $n=md$.
 Let $(A, \pol, \Thetapol)$ and $(B, \bpol, \Thetabpol)$ be marked abelian
 varieties with respective theta null points $(\theta_i^\Thetapol(0_A))_{i
 \in \Zn}$
 and
 $(\theta_i^\Thetabpol(0_B))_{i \in Z(m)}$. If they are isog-$f$-compatible, there exists a constant factor
 $\lambda \in \overk^*$ such that we have for all $i \in Z(m)$, 
 \begin{equation}
	f^*(\theta^\Thetabpol_i)=\lambda \theta^\Thetapol_{\mu_{m,n}(i)}.
 \end{equation}
 If they are dual-isog-$f$-compatible, there exists a constant factor $\lambda \in \overk^*$
 such that for all $i \in Z(m)$,
 \begin{equation}
	f^*(\theta^\Thetabpol_i)= \lambda \sum_{j \in\rho_{n,m}^{-1}(i)}
	\theta^\Thetapol_{j}.
 \end{equation}
\end{proposition}
\begin{proof}
 The result is an immediate consequence of Mumford's isogeny theorem \cite[Theorem
 4]{MumfordOEDAV1}.
\end{proof}
The preceding proposition allows us to extend Definition~\ref{sec1:def1} to rigidified abelian varieties.
\begin{definition}\label{sec1:def2}
 Let $n,d,m>1$ be integers such that $n=dm$.
 Let $(A, \pol, \Thetapol, \theta_0^\Thetapol, \rho^\pol_{0_A})$ and $(B, \bpol,
 \Thetabpol, \theta_0^\Thetabpol, \rho^\bpol_{0_B})$ be rigidified abelian varieties of
 respective types $K(n)$ and $K(m)$. Let $f:
A \rightarrow B$ be an isogeny with kernel $K\simeq Z(d)$ isotropic for $\epol$.
 We say that $(A, \pol, \Thetapol, \theta_0^\Thetapol, \rho^\pol_{0_A})$ and 
$(B, \bpol,
 \Thetabpol, \theta_0^\Thetabpol, \rho^\bpol_{0_{\Thetabpol}})$ are isog-$f$-compatible
 (resp. dual-isog-$f$-compatible) if they satisfy the conditions
 of Definition~\ref{sec1:def1} and furthermore:
 \begin{enumerate}
	 \setcounter{enumi}{4}
	\item $f^*(\theta_0^\Thetabpol)=\theta_0^\Thetapol$ (resp.
	 $f^*(\theta_0^\Thetabpol)= \sum_{j \in \rho_{n,m}^{-1}(0)} \theta_{j}^\Thetapol$) (by Proposition
	~\ref{sec1:prop1} this
	 equality is always true up to a constant factor);
	\item
	 $\rho^\bpol_{0_{\Thetabpol}}(\theta_0^\Thetabpol)=\rho^\pol_{0_A}(\theta_0^\Thetapol)$
	 (resp.
	 $\rho^\bpol_{0_{\Thetabpol}}(\theta_0^\Thetabpol)=\rho^\pol_{0_A}(\sum_{j\in\rho_{n,m}^{-1}(0)}\theta_j^\Thetapol)$).
 \end{enumerate}
\end{definition}
\begin{remark}\label{rk:affinefcompat}
 Suppose that $(A, \pol, \Thetapol)$ and $(B, \bpol, \Thetabpol)$ are isog-$f$-compatible
 marked abelian varieties. One can always endow $(B, \bpol, \Thetabpol)$ with a
 rigidification $(B, \bpol, \Thetabpol, \theta_0^\Thetabpol, \rho^\bpol_{0_{\Thetabpol}})$ and
 then choose
 $\theta_0^\Thetapol$ and $\rho^\pol_{0_A}$ such that
$(B, \bpol, \Thetabpol, \theta_0^\Thetabpol, \rho^\bpol_{0_{\Thetabpol}})$ and $(A, \pol,
 \Thetapol, \theta_0^\Thetapol, \rho^\pol_{0_A})$ are rigidified isog-$f$-compatible
 abelian varieties.
\end{remark}
\begin{corollary}\label{sec1:cor1}
 Let $n,d,m>1$ be integers such that $n=dm$.
If $(A, \pol, \Thetapol, \theta_0^\Thetapol, \rho^\pol_{0_A})$ and $(B, \bpol,
 \Thetabpol, \theta_0^\Thetabpol, \rho^\bpol_{0_{B}})$ are isog-$f$-compatible, then for all $i \in Z(m)$,
 \begin{equation}
	f^*(\theta^\Thetabpol_i)=\theta^\Thetapol_{\mu_{m,n}(i)},\quad
	\rho^\bpol_{0_{\Thetabpol}}(\theta_i^\Thetabpol(0_B))=\rho^\pol_{0_A}(\theta_{\mu_{m,n}(i)}^\Thetapol(0_A)).
 \end{equation}
If $(A, \pol, \Thetapol, \theta_0^\Thetapol, \rho^\pol_{0_A})$ and $(B, \bpol,
 \Thetabpol, \theta_0^\Thetabpol, \rho^\bpol_{0_{\Thetabpol}})$ are dual-isog-$f$-compatible, then for all $i \in Z(m)$,
 \begin{equation}
	f^*(\theta^\Thetabpol_i)=\sum_{j \in \rho_{n,m}^{-1}(i)}
	\theta^\Thetapol_{j},\quad
	\rho^\bpol_{0_{\Thetabpol}}(\theta_i^\Thetabpol(0_B))=\rho^\pol_{0_A}(\sum_{j
	\in \rho_{n,m}^{-1}(i)} \theta_{j}^\Thetapol(0_A)).
 \end{equation}

\end{corollary}
\begin{proof}
 This is an immediate consequence of Proposition~\ref{sec1:prop1} and Definition
~\ref{sec1:def2}.
\end{proof}

We introduce the following definition.
 \begin{definition}\label{def:affinemap}
	If $(B, \bpol, \Thetabpol)$ and $(A, \pol, \Thetapol)$ are isog-$f$-compatible (resp.
	dual-isog-$f$-compatible) abelian
	varieties, we denote by 
	$\tildef: \Aff^{Z(n)}=\Spec(k[X_i, i \in \Zn]) \rightarrow
 \Aff^{Z(m)} =\Spec(k[Y_i, i \in Z(m)])$ the affine map
 such that $\tildef^*(Y_i) =X_{\mu_{m,n}(i)}$ (resp. such that $\tildef^*(Y_i)
 = \sum_{j \in \rho_{n,m}^{-1}(i)} X_{j}$).
\end{definition}
 With this definition, Corollary~\ref{sec1:cor1} can be rephrased as follows: if
 $\tildenullpol$ and
 $\tildenullbpol$ are the
 affine theta null points of, respectively, $(A, \pol, \Thetapol, \theta_0^\Thetapol,
 \rho^\pol_{0_A})$ and $(B, \bpol,
 \Thetabpol, \theta_0^\Thetabpol)$, then
 $\tildef(\tildenullpol)=\tildenullbpol$.

Keeping the notation of \cref{sec1:prop1}, we have a map: 
\begin{align*}
 \pibar^0_{n,m}: \bModu_n \subset \proj([k[x_i], i \in \Zn]) & \rightarrow \bModu_m
 \subset \proj([k[y_i], i \in Z(m)]),\\
 (\pibar^0_{n,m})^*(y_i) & = x_{\mu_{m,n}(i)}.
\end{align*}
The map $\pibar^0_{n,m}$ extends the map $\pi^0_{n,m}: \Modu_n \rightarrow \Modu_m$,
which is the $0$-section of the map:
\begin{align*}
 \pi_{n,m}: \univAb_n \subset \proj^{Z(n)} & \rightarrow \univAb_m \subset \proj^{Z(m)},\\
 \pi_{n,m}^*((\theta_i)_{i \in Z(m)}) & = (\theta_{\mu_{m,n}(i)})_{i \in Z(m)}.
\end{align*}

Let $(B, \bpol, \Thetabpol)$ be a marked abelian variety of type $K(m)$ with theta null point
$$\xtheta=(\theta^\Thetabpol_i(0_\Thetabpol))_{i\in Z(m)} \in \Modu_m(\overk) \subset
\bModu_m(\overk).$$ We wish to compute the fiber of $\pibar^0_{n,m}$ over $\xtheta$. The case $d$ odd and $d$ coprime to $m$ was treated
in \cite{DRmodular}. In this paper, we address the general case. 

From now on, we
assume only that $2\mid m$ and we write $d=d_0 d_1$ where $d_0 = \gcd(d,m)$.
We consider the ideal $\Jx$ of the polynomial ring $k[X_i,i \in \Zn]$
generated:
\begin{itemize}
 \item by the Riemann and symmetry relations of \cref{sec:thriemann} and
	\cref{sec:propsym};
 \item by the specialization relations
	$X_{\mu_{m,n}(i)}=\theta_{i}^\Thetabpol(0_\Thetabpol)$ for all $i \in Z(m)$, where
	$(\theta_{i}^\Thetabpol(0_\Thetabpol))_{i \in Z(m)}$ is a fixed affine lift of $\xtheta$.
\end{itemize}
As the specialization relations are not homogeneous, $\Jx$ is not a homogeneous ideal: let
$V_{\Jx}$ be the closed subscheme of the affine space $\Aff^{Z(n)}$ defined by $\Jx$. Since the
Riemann and symmetry relations are homogeneous, a point of $\proj^{Z(n)}$ satisfying them lies in
the fiber $(\pibar^0_{n,m})^{-1}(\xtheta)$ if and only if it has an affine lift in $V_{\Jx}$,
and this lift is then unique. Thus the projection $\Aff^{Z(n)}-\{0\} \rightarrow \proj^{Z(n)}$
identifies $V_{\Jx}(\overk)$ with the fiber $(\pibar^0_{n,m})^{-1}(\xtheta)$, and we shall not
distinguish between them. When $d$ is odd and prime to $m$, it is shown in \cite{DRmodular} that
$V_{\Jx}$ is a reduced zero dimensional scheme; we shall not use this result. Denote by
$V_{\Jx}^0(\overk)$ the set of $y \in V_{\Jx}(\overk)$ which are valid level $n$ theta null
points, that is which belong to $\Modu_n(\overk)$; it is a finite set by Propositions
\ref{prop:charac} and~\ref{prop:phs} below. The following proposition shows that each point of $V_{\Jx}^0(\overk)$
corresponds to a choice of $A=B/K_1$, where $K_1 \subset B[d]$ is isomorphic to $Z(d)$ and
$K_1[d_0]=K_0 = \Thetabpolbar(\mu_{d_0, m}(Z(d_0)) \times \{0\})$, so that $A$ is
$d_1$-isogenous to $B/K_0$, and a theta structure for $(A, \pol)$
isog-$f$-compatible
with $(B, \bpol, \Thetabpol)$:
\begin{proposition}\label{prop:charac}
 We suppose that $4|n$.
 Let $y \in \Modu_n(\overk)$ be a point representing a level $n$ marked abelian variety $(A, \pol,
 \Thetapol)$. Then $y \in \vjxz (\overk)$ if and only if there exists an isogeny
 $f:
 A\rightarrow B$ such that $(A, \pol, \Thetapol)$ and
 $(B, \bpol, \Thetabpol)$ are isog-$f$-compatible.

Suppose that $y \in \vjxz (\overk)$, set $K_0=\Sbpol(\mu_{d_0,m}(Z(d_0))\times\{0\})$ and
$K_1 = f(\Spol(\mu_{d,n}(Z(d))\times\{0\}))\subset B[d]$. Then $K_1$ is isomorphic to $Z(d)$,
$K_0=K_1[d_0] \subset K_1$ and $A \simeq B/K_1$. In particular, $A$ is $d_1$-isogenous to $B/K_0$.
\end{proposition}
\begin{proof}
The first claim is just \cite[Proposition~11]{DRmodular}. 
 Suppose
 that $y \in \vjxz (\overk)$. As $(A, \pol, \Thetapol)$ and
 $(B, \bpol, \Thetabpol)$ are isog-$f$-compatible by definition, we know that $\Ker f= \Spol(\{0\}
 \times \dnu_{d,n}(\dZ(d)))$.
 Moreover, since $\Ker f \cap \Spol(\mu_{d,n}(Z(d))\times\{0\}) = \{ 0 \}$ and $d_0| m$, by Proposition
~\ref{prop:compat}, $f(\Spol(\mu_{d_0,n}(Z(d_0)) \times
 \{0\}))=\Sbpol(\mu_{d_0,m}(Z(d_0))\times\{0\})=K_0$.
 As $\Ker f \cap \Spol(\mu_{d,n}(Z(d))\times\{0\}) = \{ 0 \}$, $f$ induces an isomorphism
 between $\Spol(\mu_{d,n}(Z(d))\times\{0\}) \simeq Z(d)$ and $K_1$, which maps
 $\Spol(\mu_{d_0,n}(Z(d_0))\times\{0\})$, the subgroup of elements killed by $d_0$, onto $K_0$.
 Thus $K_1 \simeq Z(d)$, $K_0=K_1[d_0]$, and $K_1/K_0\simeq Z(d)/Z(d_0) \simeq Z(d_1)$.
 Moreover, $A[d]=\Spol(\mu_{d,n}(Z(d)) \times \dnu_{d,n}(\dZ(d)))$ contains $\Ker f$ and
 $f(A[d])=K_1$, so that $A \simeq A/A[d] \simeq B/K_1$. Finally, $K_1/K_0 \simeq
 \hatf_0(K_1)$ where $\hatf_0 : B \rightarrow A_0=B/K_0$ is the quotient map, so that
 $A \simeq A_0/\hatf_0(K_1)$ is $d_1$-isogenous to $A_0$.
 Moreover, if $G=f(\Thetapolbar(\Zn\times\{0\}))$, then, as $\mu_{d,n}(Z(d))=m\Zn$ and $G \simeq
 \Zn$, $K_1=f(\Thetapolbar(\mu_{d,n}(Z(d))\times\{0\}))=mG=G[d]$.
\end{proof}
\begin{remark}\label{rk:charac}
 The preceding proposition shows a first striking difference between the case $d$ prime
 to $m$ and the general case. When $d$ is odd and prime to $m$, \cite[Proposition~20]{DRmodular}
 shows that any abelian variety of the form $A=B/K$ where $K$ is isomorphic to $Z(d)$
 can be equipped with a theta structure
 $(A, \pol, \Thetapol)$ such that the associated theta null point $y$ is in
 $\vjxz(\overk)$. In general, by Proposition~\ref{prop:charac} and Theorem~\ref{prop:cond}
 below, the abelian varieties appearing in $\vjxz(\overk)$ are the $B/G[d]$ where $G$ runs
 through the subgroups of $B[n]$ isomorphic to $Z(n)$, isotropic for $e_{B,n}$, containing
 $\Thetabpolbar(Z(m)\times\{0\})$ and whose elements are symmetric compatible with
 $\Thetabpol(\{1\}\times Z(m) \times\{0\})$. For such a $G$, we always have $G[d][d_0]=K_0$,
 but this condition does not determine $G[d]$ when $d \nmid m$. On the contrary, when $d | m$,
 we have $d_0=d$ and $G[d]=K_0=\Sbpol(\mu_{d,m}(Z(d))\times \{0\})$, so that the only abelian
 variety appearing in $\vjxz(\overk)$ is $B/K_0$.
\end{remark}
We seek a better understanding of the set $\vjxz(\overk)$. For this, we show that there is
an action of the metaplectic group on it, which we characterize.
In the rest of this section and in the following one, we closely follow \cite[Section~5.2]{DRmodular}.

Recall that an automorphism of $G(n)$, for $n$ a positive integer, is a
group automorphism that preserves the structure of central extension of $G(n)$. We
denote by $\Aut(G(n))$ the set of such automorphisms, called the level-$n$ metaplectic group. We say that $g_s \in \Aut(G(n))$
is symmetric if $g_s \circ D_{-1}=D_{-1} \circ g_s$. We denote by $\Auts(G(n))$ the set of symmetric automorphisms of
$G(n)$. As theta structures are isomorphisms of central extensions $G(n) \rightarrow \Gpol$, $\Aut(G(n))$ (resp.\ $\Auts(G(n))$) acts freely and transitively on the set of
theta structures (resp.\ symmetric theta structures) of type $K(n)$ via $g' \mapsto
\Thetapol \circ g'$. In particular, $\Auts(G(n))$ acts on $\Modu_n$. 

Recall \cite[Definition~16]{DRmodular},
\begin{definition}\label{def:compat}
 We say that $g_c \in \Auts(G(n))$ is compatible with $G(m)$ if and only if:
 \begin{enumerate}
	\item $g_c(\{1\}\times\{0\} \times \dnu_{d,n}(\dZ(d)))=\{1\}\times\{0\} \times
	 \dnu_{d,n}(\dZ(d))$;
	\item for all $x \in \{1\}\times\{0\} \times \dZn$, $g_c(x) = x \mod (\{1\}\times\{0\}
	 \times \dnu_{d,n}(\dZ(d)))$;
	\item for all $x \in \{1\}\times \mu_{m,n}(Z(m)) \times\{0\}$, $g_c(x) = x \mod (\{1\}\times\{0\}
	 \times \dnu_{d,n}(\dZ(d)))$.
 \end{enumerate}
\end{definition}

\begin{lemma}{\cite[Lemma~17]{DRmodular}}\label{lem:action}
Let $\subG(n)$ be the biggest subgroup of $\Auts(G(n))$ which acts on
$\vjxz(\overk)$. Then $\subG(n)$ is exactly the subgroup of elements of
 $\Auts(G(n))$ which are compatible with $G(m)$.
In particular, $\subG(n)$ is independent of $(B, \bpol, \Thetabpol)$.
\end{lemma}
 Recall that $\pi_{G(n)}:G(n)\rightarrow K(n)$ is the canonical projection.

For $y \in \vjxz(\overk)$, we denote by $(A_y, \pol_y, \Theta_{\pol_y})$ a marked abelian
variety of type $K(n)$ with theta null point $y$, and by $F(y)$ the set of isogenies $f: A_y
\rightarrow B$ such that $(A_y, \pol_y, \Theta_{\pol_y})$ and $(B, \bpol, \Thetabpol)$ are
isog-$f$-compatible; it is not empty by Proposition~\ref{prop:charac}. We set
\begin{equation}
	\mathcal{K}(y)=\{ f(A_y[d]), f \in F(y) \}.
\end{equation}
By Proposition~\ref{prop:charac}, every element of $\mathcal{K}(y)$ is a subgroup $K_1
\subset B[d]$ isomorphic to $Z(d)$ such that $K_1[d_0]=K_0$ and $A_y \simeq B/K_1$; moreover, if
$f \in F(y)$, then $f(A_y[d])=G[d]$ where $G=f(\overline{\Theta}_{\pol_y}(\Zn\times\{0\}))$.

Let $\phi: A' \rightarrow A$ be an isomorphism of abelian varieties and $\pol$ be a line
bundle on $A$. We denote by $\phi^*: \Gpol \rightarrow G(\phi^*\pol)$ the group isomorphism
$(\tau_x, \psi_x) \mapsto (\tau_{\phi^{-1}(x)}, \phi^*(\psi_x))$ (recall that $\phi \circ
\tau_{\phi^{-1}(x)}=\tau_x \circ \phi$). If $(A, \pol, \Thetapol)$ is a marked abelian variety,
then $(A', \phi^*\pol, \phi^* \circ \Thetapol)$ is a marked abelian variety isomorphic to $(A,
\pol, \Thetapol)$: indeed, $\phi$ commutes with $[-1]$ and maps the normalized isomorphism
$\pol \rightarrow [-1]^*\pol$ to the normalized isomorphism for $\phi^*\pol$, so that $\phi^*$
commutes with $\delta_{-1}$, and the theta functions defined by $\phi^*\circ\Thetapol$ are the
$\phi^*(\theta_i^\Thetapol)$, so that the two theta null points are equal. Moreover, by
transport of structure through $\phi$, if $(A, \pol, \Thetapol)$ and $(B, \bpol, \Thetabpol)$
are isog-$f$-compatible, then $(A', \phi^*\pol, \phi^*\circ\Thetapol)$ and $(B, \bpol,
\Thetabpol)$ are isog-$(f \circ \phi)$-compatible: the four conditions of Definition
\ref{sec1:def1} are transported by $\phi$, since $\Ker(f\circ\phi)=\phi^{-1}(\Ker f)$,
$(f\circ\phi)^*\bpol=\phi^*\pol$ and the map $\fsharp$ of Proposition~\ref{prop:compat}
associated to $f \circ \phi$ is $\fsharp \circ (\phi^*)^{-1}$.

\begin{proposition}\label{prop:phs}~
 \begin{enumerate}
	\item For all $y \in \vjxz(\overk)$ and $g \in \subG(n)$, we have $F(y.g)=F(y)$, where
	 $y.g$ is represented by $(A_y, \pol_y, \Theta_{\pol_y}\circ g)$. In particular,
	 $\mathcal{K}(y.g)=\mathcal{K}(y)$.
	\item Let $y, y' \in \vjxz(\overk)$. The following conditions are equivalent:
	 \begin{enumerate}
		\item there exists $g \in \subG(n)$ such that $y'=y.g$;
		\item $\mathcal{K}(y)=\mathcal{K}(y')$;
		\item $\mathcal{K}(y)\cap\mathcal{K}(y')\neq\emptyset$.
	 \end{enumerate}
 \end{enumerate}
 Consequently, for $K_1 \subset B[d]$, the set $V^0_{K_1}=\{ y \in \vjxz(\overk), K_1 \in
 \mathcal{K}(y)\}$ is either empty or an orbit of $\subG(n)$, and $\vjxz(\overk)$ is the union
 of the $V^0_{K_1}$, $K_1 \subset B[d]$. If $d \mid m$, then $\mathcal{K}(y)=\{K_0\}$ for all $y \in \vjxz(\overk)$, so
 that $\vjxz(\overk)$ is a single orbit of $\subG(n)$.
\end{proposition}
\begin{remark}\label{rk:Fy}
 Automorphisms of $B$ and of $A_y$ play a role in Proposition~\ref{prop:phs}.

 (1) Let $f \in F(y)$. The elements of $F(y)$ have the same kernel
 $\Thetapolbar(\{0\} \times \dnu_{d,n}(\dZ(d)))$ by condition (1) of Definition
~\ref{sec1:def1}, so that they are of the form $\alpha \circ f$, where $\alpha$ is an
 automorphism of $B$. By conditions (3) and (4) of Definition~\ref{sec1:def1}, $\alpha \circ f
 \in F(y)$ if and only if $\alpha$ is an automorphism of the marked abelian variety $(B, \bpol,
 \Thetabpol)$ (the subgroups $\{1\}\times Z(m)\times\{0\}$ and $\{1\}\times\{0\}\times\dZ(m)$
 generate $G(m)$). Such an $\alpha$ acts trivially on $K(\bpol)=B[m]$. If $m \geq 3$, that is
 $m \geq 4$ since $m$ is even, $\alpha=\mathrm{Id}_B$ by Serre's lemma (see
 \cite[Section~21]{MumfordAV}), so that $F(y)$ and $\mathcal{K}(y)$ have only one element. If
 $m=2$, $D_{-1}$ is the identity of $G(2)$, so that $[-1]_B$ is an automorphism of $(B, \bpol,
 \Thetabpol)$ and $[-1]_B \circ f \in F(y)$; more generally, $F(y)$ contains the $\alpha \circ
 f$ for all the automorphisms $\alpha$ of $(B, \bpol, \Thetabpol)$, and the subgroups
 $\alpha(f(A_y[d]))$ may be different when $d \nmid m$ (for instance if $B$ is a product and
 $\alpha$ is $\mathrm{Id}$ on one factor and $[-1]$ on the other one). This is why we work
 with the sets $F(y)$ and $\mathcal{K}(y)$; when $d \mid m$, $\mathcal{K}(y)=\{K_0\}$ in any
 case (Proposition~\ref{prop:phs}). In particular, when $m \geq 4$ or $d \mid m$, conditions (b) and (c) of
 Proposition~\ref{prop:phs} (2) are trivially equivalent; condition (c) is only
 relevant when $m=2$ and $d \nmid m$.

 (2) The action of $\subG(n)$ on the orbits $V^0_{K_1}$ is not free in general: the stabilizer of
 $y$ contains the elements $g$ such that $\Theta_{\pol_y} \circ g=\alpha^* \circ \Theta_{\pol_y}$ for
 an automorphism $\alpha$ of $(A_y, \pol_y)$. For instance, $\Theta_{\pol_y} \circ
 D_{-1}=\delta_{-1}\circ \Theta_{\pol_y}$ corresponds to $\alpha=[-1]$, and one checks with
 Definition~\ref{def:compat} that $D_{-1} \in \subG(n)$ if and only if $m=2$, in accordance with
 (1), since $f \circ [-1]_{A_y}=[-1]_B \circ f$.
\end{remark}
\begin{proof}
	(1) Let $f \in F(y)$ and $g \in \subG(n)$; write $(A, \pol, \Thetapol)=(A_y, \pol_y,
	\Theta_{\pol_y})$ and $\Theta'_\pol=\Thetapol \circ g$. By Lemma~\ref{lem:action}, $g$
	is compatible with $G(m)$ (Definition~\ref{def:compat}). We check the conditions of
	Definition~\ref{sec1:def1} for $\Theta'_\pol$ and the same isogeny $f$. By condition (1) of Definition
	\ref{def:compat}, $\Theta'_\pol(\{1\}\times\{0\}\times \dnu_{d,n}(\dZ(d)))=
	\Thetapol(\{1\}\times\{0\}\times \dnu_{d,n}(\dZ(d)))=\tildeK$, which gives conditions (1) and
	(2) of Definition~\ref{sec1:def1}. Let $t \in \{1\}\times \mu_{m,n}(Z(m)) \times \{0\}$ (resp. $t \in
	\{1\} \times \{0\} \times \dZn$). By condition (3) (resp. (2)) of Definition~\ref{def:compat},
	$g(t)=t k$ with $k \in \{1\}\times\{0\}\times \dnu_{d,n}(\dZ(d))$, so that
	$\Theta'_\pol(t)=\Thetapol(t)\Thetapol(k)$ with $\Thetapol(k) \in \tildeK$. Since $\fsharp$
	is trivial on $\tildeK$ (see Proposition~\ref{prop:compat}), we get
	$\fsharp(\Theta'_\pol(t))=\fsharp(\Thetapol(t))$, so that conditions (3) and (4) of
	Definition~\ref{sec1:def1} for $\Theta'_\pol$ follow from those for $\Thetapol$. Thus
	$f \in F(y.g)$ and $F(y) \subset F(y.g)$. Applying this to $y.g$ and $g^{-1}$, we get
	$F(y.g)=F(y)$. As $A_y[d]$ does not depend on the theta structure, $\mathcal{K}(y.g)=
	\mathcal{K}(y)$.

	(2) (a) implies (b) by (1), and (b) implies (c) since $\mathcal{K}(y) \neq \emptyset$.
	Suppose (c): there exist $f \in F(y)$ and $f' \in F(y')$ such that $f(A_y[d])=f'(A_{y'}[d])
	=K_1$. Write $(A, \pol, \Thetapol)=(A_y, \pol_y, \Theta_{\pol_y})$ and $(A', \pol',
	\Theta_{\pol'})=(A_{y'}, \pol_{y'}, \Theta_{\pol_{y'}})$. As $\Ker f=\Thetapolbar(\{0\} \times
	\dnu_{d,n}(\dZ(d))) \subset A[d]$, there exists a unique isogeny $\hatf: B \rightarrow A$ such
	that $\hatf \circ f=[d]_A$. Then $f \circ \hatf \circ f=[d]_B \circ f$, so that $f \circ
	\hatf=[d]_B$ since $f$ is surjective, and $\Ker \hatf=f(A[d])=K_1$: indeed, for $a \in
	A(\overk)$, $\hatf(f(a))=da$ vanishes if and only if $a \in A[d]$. In the same way, we
	define $\hatf': B \rightarrow A'$ such that $\hatf' \circ f'=[d]_{A'}$, $f' \circ \hatf' =
	[d]_B$ and $\Ker \hatf'=K_1$. As $\hatf$ and $\hatf'$ are separable isogenies with the same
	kernel (they are separable since the characteristic of $k$ does not divide $n$), there exists an isomorphism $\phi: A \rightarrow A'$ such that $\hatf'=\phi \circ
	\hatf$. Then $f' \circ \phi \circ \hatf=f' \circ \hatf'=[d]_B=f \circ \hatf$, so that $f'
	\circ \phi=f$ since $\hatf$ is surjective. In particular, $\phi^*\pol'=\phi^* f'^*
	\bpol=f^*\bpol=\pol$.

	By the preliminaries above, $(A, \pol, \Theta''_\pol)$, where $\Theta''_\pol= \phi^* \circ
	\Theta_{\pol'}$, is a marked abelian variety isomorphic to $(A', \pol', \Theta_{\pol'})$, so
	that its theta null point is $y'$, and it is isog-$f$-compatible with $(B, \bpol,
	\Thetabpol)$, since $f' \circ \phi=f$. Let $g_c=\Thetapol^{-1} \circ \Theta''_\pol$. As
	$\Thetapol$ and $\Theta''_\pol$ are symmetric theta structures for the same $(A, \pol)$,
	$g_c \in \Auts(G(n))$. Thanks to condition (1) of Definition~\ref{sec1:def1}, we have
	$\Thetapolbar(\{0\} \times \dnu_{d,n}(\dZ(d)))=\overline{\Theta''_\pol}(\{0\} \times
	\dnu_{d,n}(\dZ(d)))=\Ker f$, and condition (2) of Definition~\ref{sec1:def1} gives
	$g_c(\{1\}\times\{0\} \times \dnu_{d,n}(\dZ(d)))=\{1\}\times\{0\} \times \dnu_{d,n}(\dZ(d))$.
	Let $t \in \{1\}\times \mu_{m,n}(Z(m)) \times \{0\}$ (resp. $t \in \{1\}\times\{0\} \times
	\dZn$). By conditions (3) (resp. (4)) of Definition~\ref{sec1:def1} for $\Thetapol$ and
	$\Theta''_\pol$, $\fsharp(\Thetapol(g_c(t)))=\fsharp(\Theta''_\pol(t))=\fsharp(\Thetapol(t))$,
	so that $\Thetapol(g_c(t) t^{-1})$ is in the kernel $\tildeK$ of $\fsharp$, that is
	$g_c(t)=t \mod (\{1\}\times\{0\}\times \dnu_{d,n}(\dZ(d)))$. We have proved that $g_c$ is
	compatible with $G(m)$, so that $g_c \in \subG(n)$ by Lemma~\ref{lem:action}, and
	$y'=y.g_c$ since $\Theta''_\pol=\Thetapol \circ g_c$.

	For the consequences: by (2), two points $y, y'$ of $V^0_{K_1}$ satisfy (c), so that they
	are in the same orbit, and conversely, by (1), $V^0_{K_1}$ is stable under $\subG(n)$. Every
	$y \in \vjxz(\overk)$ belongs to $V^0_{K_1}$ for $K_1 \in \mathcal{K}(y) \neq \emptyset$. Finally, if $d | m$, then $d_0=d$ and every $K_1 \in
	\mathcal{K}(y)$ is isomorphic to $Z(d)$ and satisfies $K_1=K_1[d]=K_1[d_0]=K_0$.
\end{proof}

\section{Action of the metaplectic group}\label{sec:trans}
Denote by $\Sp(K(n))$ the group of symplectic (for the commutator pairing) automorphisms of $K(n)$.
Recall that from \cite{DRmodular}, there is an exact sequence:
\begin{equation}\label{eq:exactaut}
 \begin{tikzpicture}
 \matrix [column sep={1cm}, row sep={1cm}]
 { 
 \node(a){$0$}; & \node(b){$K(n)$}; & \node(c){$\Aut(G(n))$}; & \node(d){$\Sp(K(n))$};
 & \node(e){$0$}; \\
 };
 \draw [->] (a) -- (b);
 \draw [->] (b) -- (c) node[above, midway] {$\Phi$};
 \draw [->] (c) -- (d) node[above, midway] {$\Psi$};
 \draw [->] (d) -- (e);
 \end{tikzpicture}
\end{equation}
where $\Phi$ is defined as
\begin{equation}\label{eq:defphi}
 \begin{split}
 \Phi: K(n) & \rightarrow \Aut(G(n)) \\
	c=(c_1,c_2) & \mapsto g_c: (\alpha, x, y) \mapsto (\alpha\, c_2(x)y(c_1)^{-1}, x,y)=(\alpha
	e_n((x,y),c), x,y),
 \end{split}
\end{equation}
and $\Psi(g')$ is the unique symplectic automorphism $\bar{g}$ of $K(n)$ such that $\pi_{G(n)}
\circ g'=\bar{g} \circ \pi_{G(n)}$. Moreover, $\Phi^{-1}(\Auts(G(n)))=K(2)$.

From the preceding result, the computation of the action of $\Auts(G(n))$ on $\Modu_n$
reduces to the computation of $\Phi(K(2))$ and the action of a (set) section of $\Psi$.
The following simple lemma takes care of the action of $\Phi(K(2))$.
\begin{lemma}\label{lem:actionK}
 Let $(A, \pol, \Thetapol)$ be a marked abelian variety of type $K(n)$.
 Let $(\theta^\Thetapol_i)_{i \in \Zn}$ be the associated basis of $\Gamma(A,\pol)$
 according to Proposition~\ref{prop:base}. Let $g_c = \Phi(c)$ for $c=(c_1, c_2) \in
 K(n)=Z(n) \times \dZn$. Then, for all $i \in \Zn$, we have:
 \begin{equation}
	\theta^{\Thetapol \circ g_c}_i = c_2(i) \theta^\Thetapol_{i -c_1}.
 \end{equation}
\end{lemma}
\begin{proof}
 Let $\Theta'_\pol=\Thetapol \circ g_c$.
 We use the construction of the basis of $\Gamma(A,\pol)$ provided by Proposition
~\ref{prop:base}. For $(\alpha, x,y) \in G(n)$, using (\ref{eq:defphi}), we have $\Theta'_\pol ((1,x,y))= 
 \Thetapol((c_2(x)/y(c_1), x,y))$. Thus,
 $\Theta'_\pol((\alpha,x,y)).\theta_i^\Thetapol=\alpha c_2(x) y(-x-i-c_1)
 \theta_{i+x}^\Thetapol$ following the computation of (\ref{eq:trans}).
 We have:
 \begin{equation}
 \begin{split}
	\theta_0^{\Theta'_\pol} & = \sum_{y \in \dZn} \Theta'_\pol((1,0,y)). \sum_{i \in
 \Zn} \theta_i^\Thetapol \\
	& =\sum_{i \in \Zn} \sum_{y \in \dZn} y(-i-c_1) \theta_i^\Thetapol \\
	& = \theta_{-c_1}^\Thetapol.
 \end{split}
 \end{equation}
 Then for $x \in \Zn$, 
 \begin{equation}
	\begin{split}
	 \theta_x^{\Theta'_\pol} &= \Theta'_\pol((1,x,0)).\theta_0^{\Theta'_\pol} \\
	 & = c_2(x) \theta_{x-c_1}^\Thetapol.
	\end{split}
 \end{equation}
\end{proof}
We would like to compute the action of a section of $\Psi$ on the basis of $\Gamma(A,\pol)$
defined by $\Thetapol$. For the convenience of the
reader, we recall the definition and result from
\cite{DRmodular} which state that the sections of $\Psi$ are in one-to-one correspondence with
semi-characters.

\begin{definition}\label{def:semicharacter}
 Let $\overpsi \in \Sp(K(n))$. A $\overpsi$-semi-character is a map $\chipsi: 
 K(n) \rightarrow \overk^*$ such that for $(x_1,x_2),(x_1',x_2') \in K(n)=Z(n) \times
 \dZn$, 
 \begin{equation}\label{eq:defchi}
 \chipsi( (x_1+x_1',x_2+x_2') )=\chipsi( (x_1,x_2) ).\chipsi(
 (x_1',x_2')
).[\overpsi((x_1',x_2'))_2(\overpsi((x_1,x_2))_1)].x'_2(x_1)^{-1},
\end{equation}
where $\overpsi((x_1,x_2))=(\overpsi((x_1,x_2))_1,
\overpsi((x_1,x_2))_2)$ (resp.
$\overpsi((x_1',x_2'))=(\overpsi((x_1',x_2'))_1, \overpsi((x_1',x_2'))_2)$) in the canonical
 decomposition of $K(n)=Z(n) \times \dZn$.
A semi-character $\chipsi$ is said to
 be symmetric if for all $(x_1,x_2) \in K(n)$,
$\chipsi(-(x_1,x_2))=\chipsi ( (x_1,x_2) )$.
\end{definition}

The preceding definition is motivated by the following lemma:
\begin{lemma}
 Let $\psi \in \Aut(G(n))$ and let $\overpsi=\Psi(\psi)$. There exists a unique
 semi-character $\chipsi$ such that for all
 $(\alpha,(x_1,x_2)) \in G(n)$,
 \begin{equation}\label{eq@moneq1}
 \psi: (\alpha,(x_1,x_2) ) \mapsto
(\alpha \chipsi( (x_1,x_2) ), \overpsi( (x_1,x_2 ))).
\end{equation}
 Moreover, $\psi \in \Auts(G(n))$ if and only if $\chipsi$ is a symmetric
 semi-character.

 As a consequence, if $\overpsi \in \Sp(K(n))$, there is a one-to-one
 correspondence between the set of extensions of $\overpsi$ to $\Aut(G(n))$ (resp.
 to $\Auts(G(n))$)
 and the set of semi-characters (resp. symmetric semi-characters).
\end{lemma}
Actually, the condition (\ref{eq:defchi}) just means that $\psi$ as defined by (\ref{eq@moneq1}) is a group morphism for the group law of $G(n)$.

Using the preceding lemma, one can show that there always exists a section of $\Psi$. We
have to slightly adapt \cite[Lemma~15]{DRmodular} to the case of even
level.

\begin{proposition}\label{lem:semi}
 Let $n \in \N^*$, let $\mB=\{v_i, v_{g+i}\}_{i=1, \ldots, g}$ be a basis of $K(n)=Z(n)\times \dZn$. 
 Let $\overpsi \in \Sp(K(n))$. For $i=1, \ldots, 2g$, write $v_i=((v_i)_1, (v_i)_2) \in \Zn \times
 \dZn$ and let $a_i=\overpsi(v_i)_2(\overpsi(v_i)_1)$ and $b_i=(v_i)_2((v_i)_1)$, which are
 $n$-th roots of unity. There exists a
 unique $\overpsi$-semi-character $\chipsi$ such that
 $\chipsi(v_i)=t_i$ for $i=1, \ldots, 2g$, where
 \begin{itemize}
	\item
 $t_i^n=(a_ib_i)^{n/2} \in \{-1, 1\}$ if $n$ is even,
\item $t_i^n=1$ if $n$ is odd,
 \end{itemize}
	 and all
 $\overpsi$-semi-characters are obtained in that way. Moreover $\chipsi$ is symmetric if and
 only if $t_i^2=a_ib_i^{-1}$ for $i=1, \ldots, 2g$. When $\mB$ is the canonical basis of
 $K(n)$, we have $b_i=1$ for $i=1, \ldots, 2g$.
\end{proposition}
\begin{proof}
 For $i=1, \ldots, 2g$, let $t_i$ be any element of $\overk^*$. We prove the
 uniqueness of a $\overpsi$-semi-character such that $\chipsi(v_i)=t_i$ exactly as in
 \cite[Lemma~15]{DRmodular} by using the relation (\ref{eq:defchi}) to deduce the value of
 $\chipsi(v)$ for all $v \in K(n)$.

 We remark that $\chipsi(0)=1_{\overk^*}$ and that for all $v \in K(n)$,
 $\chipsi(-v)=\chipsi(v)^{-1} \overpsi(v)_2(\overpsi(v)_1) v_2(v_1)^{-1}$ (where
 $v=(v_1, v_2) \in \Zn \times \dZn$). Moreover an easy induction shows that for
 all $\ell \in \N$ and $v \in K(n)$, 
 \begin{equation}
	\chipsi(\ell.v)=\chipsi(v)^\ell \overpsi(v)_2(\overpsi(v)_1)^{\ell (\ell-1)/2}v_2(v_1)^{-\ell
	(\ell-1)/2}.
 \end{equation}
 We deduce that there exists a $\overpsi$-semi-character $\chipsi$ such that
 $\chipsi(v_i)=t_i$ provided that $\chipsi(n.v_i)=1_{\overk^*}$, that is $t_i^n=(a_ib_i^{-1})^{-n(n-1)/2}$. As
 $a_i$ and $b_i$ are $n$-th roots of unity, $(a_ib_i^{-1})^{n(n-1)/2}$ is equal to $1$ if $n$ is
 odd, and to $(a_ib_i^{-1})^{n/2}=(a_ib_i)^{n/2} \in \{-1,1\}$ if $n$ is even, so that this means that
 \begin{equation}\label{eq:lift1}
	\begin{split}
	t_i^n = (a_ib_i)^{n/2}\ \text{if}\ n\ \text{is even},\\
	t_i^n=1\ \text{if}\ n\ \text{is odd}.
	\end{split}
 \end{equation}

 Moreover, $\chipsi$ is symmetric if and only if for all $v \in K(n)$,
 $\chipsi(-v)=\chipsi(v)^{-1} \overpsi(v)_2(\overpsi(v)_1)v_2(v_1)^{-1}=\chipsi(v)$. Let
 $\delta(v)=\chipsi(-v)\chipsi(v)^{-1}=\chipsi(v)^{-2}\overpsi(v)_2(\overpsi(v)_1)v_2(v_1)^{-1}$. Using
 (\ref{eq:defchi}), a direct computation shows that for all $v, v' \in K(n)$,
 \begin{equation}
	\delta(v+v')=\delta(v)\delta(v')\,\frac{\overpsi(v)_2(\overpsi(v')_1)}{\overpsi(v')_2(\overpsi(v)_1)}\,
	\frac{v'_2(v_1)}{v_2(v'_1)}=\delta(v)\delta(v'),
 \end{equation}
 since $\overpsi$ is symplectic (the two fractions are the commutator pairings $e_n$ of
 $(\overpsi(v'), \overpsi(v))$ and of $(v, v')$, which are inverse of each other). Thus $\delta$ is a
 character of $K(n)$, and $\chipsi$ is symmetric if and only if $\delta(v_i)=1$ for $i=1,
 \ldots, 2g$, that is:
 \begin{equation}\label{eq:lift2}
	t_i^2= a_ib_i^{-1}.
 \end{equation}
 Note that (\ref{eq:lift2}) implies (\ref{eq:lift1}) if $n$ is even.
\end{proof}
\begin{remark}
In the previous proposition, choices of $\overpsi$-semi-characters representing two
 elements of $\Psi^{-1}(\overpsi)$ differ by choices of $n$-th roots of
 $(a_ib_i)^{n/2}$ (or of $1$ if $n$ is odd), thus by an action of $\Phi(K(n))$, which is
 consistent with the exact sequence (\ref{eq:exactaut}). In the same way, two choices
 of symmetric $\overpsi$-semi-characters representing two elements of
 $\Psi^{-1}(\overpsi)$ differ by an action of $\Phi(K(2))$.
\end{remark}
As seen previously, there is an action of $\Auts(G(n))$ on $\Modu_n$ which decomposes via
the exact sequence (\ref{eq:exactaut}) into an action of $K(2)$ described in Lemma~\ref{lem:actionK} and an action of
$\Sp(K(n))$, which is defined up to the action of $K(2)$. In the classical analytic theory
of theta functions, there is an action of $\Sp_{2g}(\Z)$ on $\mathcal{H}_g$, the
$g$-dimensional Siegel 
upper half space, which induces an action on classical theta functions. This action is given by the so-called
transformation formula (see \cite{BiLaCAV, MR85h:14026}). The following result makes explicit the action of
$\Sp(K(n))$ on $\Modu_n$ and thus, can be seen as an analog of the transformation formula.
    We note that Proposition~\ref{prop:trans} is equivalent to
    \cite[Theorem~12]{dartois_fast}, which gives a closed-form expression
    for the action of a general $M \in \Sp_{2g}(\ZZn)$ without decomposing
    into generators. The proofs of the two results closely follow the approach of
    \cite[Section~5.2]{DRmodular}.

 Let $\mB=\{v_k, v_{k+g} \}_{k=1, \ldots, g}$ be the canonical basis of $K(n)=Z(n)\times
 \dZn$. Denote by $\Sp_{2g}(\Z/n\Z)$ the group of symplectic matrices of dimension $2g$ with coefficients
 in $\ZZn$. If $M \in \Sp_{2g}(\ZZn)$, we denote by $\psimb(M)$ the element of
 $\Sp(K(n))$ whose matrix is $M$ in the basis $\mB$.
 For $A \in GL_g(\ZZn)$, we denote by $B_g(A) \in \Sp_{2g}(\ZZn)$ the matrix of the form
 $\begin{pmatrix} A & 0_g \\ 0_g & ^t A^{-1} \end{pmatrix}$. 
	For $C \in M_g(\ZZn)$ a
	symmetric matrix, we denote by $S_g(C) \in \Sp_{2g}(\ZZn)$ the matrix of the form
	$\begin{pmatrix} 1_g & 0_g \\ C & 1_g \end{pmatrix}$. Finally, let
	 $H_g=\begin{pmatrix} 0_g & 1_g
	 \\ -1_g & 0_g \end{pmatrix} \in \Sp_{2g}(\ZZn)$.

\begin{proposition}\label{prop:trans}
 Let $(A, \pol, \Thetapol)$ be a marked abelian variety of type $K(n)$. The group
 $\Sp_{2g}(\ZZn)$ is generated by the matrices $B_g(A)$ for $A \in GL_g(\ZZn)$,
 $S_g(C)$ for $C \in M_g(\ZZn)$ symmetric and $H_g$. In order to describe the action
 of $\psimb(\Sp_{2g}(\ZZn))$ on $\Gamma(A,\pol)$, it is thus enough to do it for matrices of
 the form $B_g(A)$, $S_g(C)$ and $H_g$. This is given by:
 \begin{enumerate}
	\item Let $A \in GL_g(\ZZn)$, denote by $\gamma_A \in GL(Z(n))$ the automorphism with
	 matrix $A$ in the basis $\{v_1, \ldots, v_g\}$. We can choose $\gamma \in
	 \Psi^{-1}(\psimb(B_g(A)))$ so that there exists $\lambda \in \overk^*$ such
	 that for all $i \in \Zn$:
	 \begin{equation}\label{eq:prop10:1}
		\theta_i^{\Thetapol \circ \gamma}=\lambda \theta^\Thetapol_{\gamma_A(i)}.
	 \end{equation}
	\item Let $C \in M_g(\ZZn)$ be a symmetric matrix, denote by $\gamma_C: \Zn
	 \rightarrow \dZn$ the morphism given in the basis $\{v_1, \ldots, v_g \}$
	 of $Z(n)$ and $\{ v_{g+1}, \ldots, v_{2g} \}$ of $\dZn$ by the matrix $C$. We can
	 choose $\gamma \in \Psi^{-1}(\psimb(S_g(C)))$ so that there exists $\lambda \in
	 \overk^*$ such that for $i \in \Zn$:
	 \begin{equation}\label{eq:prop10:2}
 \theta_i^{\Thetapol \circ \gamma}= \lambda {(\gamma_C(i)(i))}^{-1/2}	
 \theta^\Thetapol_{i}.
	 \end{equation}
	 Let $\ell$ be such that $\gamma_C(i)(i)$ is a primitive $\ell$-th root of
	 unity. If $\ell$ is odd, $\sqrt{\gamma_C(i)(i)}=\gamma_C(i)(i)^{(\ell +1)/2}$
	 is uniquely determined. If $\ell$ is even, 
	 the signs of the square roots are chosen in the following manner: we choose
	 signs for $\sqrt{\gamma_C(v_k)(v_k)}$ for $k=1, \ldots, g$ arbitrarily and we compute the
	 other signs using the composition law:
	 \begin{equation}\label{eq:choicesign}
		\sqrt{\gamma_C(i+j)(i+j)}=\sqrt{\gamma_C(i)(i)} \sqrt{\gamma_C(j)(j)}
		\gamma_C(i)(j),
	 \end{equation}
	 for all $i,j \in \Zn$.
	\item Denote by $\gamma_{H_g}: \Zn \rightarrow \dZn$ the morphism such that
	 $\gamma_{H_g}(v_k)= v_{k+g}$ for $k=1,\ldots,g$. We can choose $\gamma \in
	 \Psi^{-1}(\psimb(H_g))$ so that there exists
	 $\lambda \in \overk^*$ such that for $i \in
	 \Zn$:
	 \begin{equation}\label{eq:prop10:3}
		\theta_i^{\Thetapol \circ \gamma}= \lambda \sum_{j \in \Zn}
		\gamma_{H_g}(i)(j)
		\theta_j^\Thetapol.
	 \end{equation}
 \end{enumerate}
\end{proposition}
\begin{proof}
 We first show that $\Sp_{2g}(\ZZn)$ is generated by the matrices $B_g(A)$ for $A \in GL_g(\ZZn)$,
 $S_g(C)$ for $C \in M_g(\ZZn)$ symmetric, and $H_g$. For $X \in M_g(\ZZn)$ symmetric, let
 $T_g(X)=\begin{pmatrix} 1_g & X \\ 0_g & 1_g \end{pmatrix}$; a direct computation gives
 $T_g(X)=H_g S_g(-X) H_g^{-1}$. As $B_g(A)^{-1}=B_g(A^{-1})$, $S_g(C)^{-1}=S_g(-C)$ and
 $H_g^{-1}=H_g^3$, it suffices to show that, starting from $M = \begin{pmatrix} A & B \\ C & D
 \end{pmatrix} \in \Sp_{2g}(\ZZn)$ and multiplying $M$ on the left by such matrices and by matrices
 $T_g(X)$, we can reach $1_{2g}$. Recall that $M$ is symplectic if and only if ${}^tAC$ and ${}^tBD$ are
 symmetric and ${}^tAD-{}^tCB=1_g$, and that a matrix with coefficients in $\ZZn$ is invertible if and
 only if its reduction modulo $p$ is invertible for every prime divisor $p$ of $n$.

 \emph{Step 1: there exists a symmetric $X \in M_g(\ZZn)$ such that $A+XC$ is invertible.} Let $p$
 be a prime divisor of $n$ and denote by $\bar{A}$, $\bar{C}$ the reductions of $A$, $C$ modulo $p$.
 Then ${}^t\bar{A}\bar{C}$ is symmetric and, $\bar{M}$ being invertible, the $2g \times g$ matrix
 $\begin{pmatrix} \bar{A} \\ \bar{C} \end{pmatrix}$ has rank $g$ over the field $\mathbb{F}_p$. Let $r$ be the rank
 of $\bar{A}$; by Gaussian elimination over $\mathbb{F}_p$, there exist $U, V \in GL_g(\mathbb{F}_p)$ such
 that $A_0:=U\bar{A}V=\begin{pmatrix} 1_r & 0 \\ 0 & 0_{g-r} \end{pmatrix}$. Let $C_0={}^tU^{-1}\bar{C}V$,
 which we partition as $A_0$: $C_0=\begin{pmatrix} C_1 & C_2 \\ C_3 & C_4 \end{pmatrix}$. As
 ${}^tA_0C_0={}^tV\,{}^t\bar{A}\bar{C}\,V$ is symmetric and equal to $\begin{pmatrix} C_1 & C_2 \\ 0 & 0
 \end{pmatrix}$, we get $C_2=0$. As $\begin{pmatrix} A_0 \\ C_0 \end{pmatrix}$ has rank $g$ and its last $g-r$
 columns are $\begin{pmatrix} 0 \\ 0 \\ C_2 \\ C_4\end{pmatrix}=\begin{pmatrix} 0 \\ 0 \\ 0 \\ C_4\end{pmatrix}$, $C_4$ is invertible. Let $X_0=\begin{pmatrix} 0_r
 & 0 \\ 0 & 1_{g-r} \end{pmatrix}$; then $A_0+X_0C_0=\begin{pmatrix} 1_r & 0 \\ C_3 & C_4 \end{pmatrix}$ is
 invertible. Let $\bar{X}_p=U^{-1}X_0\,{}^tU^{-1}$, which is symmetric; then
 $\bar{A}+\bar{X}_p\bar{C}=U^{-1}(A_0+X_0C_0)V^{-1}$ is invertible. By the Chinese remainder theorem, there
 exists a symmetric $X \in M_g(\ZZn)$ whose reduction modulo $p$ is $\bar{X}_p$ for every prime divisor
 $p$ of $n$, and $A+XC$ is invertible.

 \emph{Step 2.} Let $X$ be as in Step 1 and $M'=T_g(X)M=\begin{pmatrix} A' & B' \\ C & D \end{pmatrix}$, with
 $A'=A+XC$ invertible. Then $M_1=B_g(A'^{-1})M'$ is a symplectic matrix of the form $\begin{pmatrix} 1_g &
 B_1 \\ C_1 & D_1 \end{pmatrix}$, so that $C_1={}^t1_gC_1$ is symmetric and $M_2=S_g(-C_1) M_1=\begin{pmatrix} 1_g & B_1
 \\ 0_g & D_2 \end{pmatrix}$. As $M_2$ is symplectic, $D_2=1_g$ and $B_1={}^tB_1D_2$ is symmetric, so that
 $T_g(-B_1)M_2=1_{2g}$. This proves the claim.

Next, we use Proposition~\ref{prop:base} to compute
$\theta_i^{\Thetapol \circ \gamma}$. We remark that if $\psi'=\psimb(B_g(A))$ or
$\psi'=\psimb(S_g(C))$ then $\psi'(\dZn)=\dZn$. We deduce that there exists $\lambda \in \overk^*$
such that $\theta_0^{\Thetapol \circ \gamma} = \lambda
\theta_0^\Thetapol$ if $\gamma \in \Psi^{-1}(\psi')$.

For case (1), let $\overpsi'= \psimb(B_g(A))$, following Proposition~\ref{lem:semi}, we choose $\gamma \in \Psi^{-1}(\overpsi')$ defined by the
$\overpsi'$-semi-character such that $\chi_{\overpsi'}(v_k)=1$ for $k=1,\ldots,2g$.
By Proposition~\ref{prop:base}, we have for $i \in \Zn$, $\theta_i^{\Thetapol \circ
\gamma}= (\Thetapol \circ \gamma(i)). \theta_0^{\Thetapol\circ \gamma}=\Thetapol((1,
\gamma_A(i), 0)). \lambda \theta_0^\Thetapol= \lambda \theta_{\gamma_A(i)}^\Thetapol$ which is
(\ref{eq:prop10:1}).

 For case (2), let $\overpsi'= \psimb(S_g(C))$, according to Proposition~\ref{lem:semi},
 we choose $\gamma \in \Psi^{-1}(\overpsi')$ defined by the
 $\overpsi'$-semi-character such that $\chi_{\overpsi'}(v_k)=\sqrt{
 \overpsi'(v_k)_2(\overpsi'(v_k)_1)}$ and $\chi_{\overpsi'}(v_{g+k})=1$ for $k=1, \ldots, g$. Let $\ell$
 be such that
$\overpsi'(v_k)_2(\overpsi'(v_k)_1)$ is a primitive $\ell$-th root of unity. If $\ell$ is odd,
 $\sqrt{\overpsi'(v_k)_2(\overpsi'(v_k)_1)}$ 
 is uniquely determined; if $\ell$ is even, we choose the
signs of the square roots arbitrarily.
 By definition, for $k=1, \ldots, g$, we have $\overpsi'(v_k)_2=\gamma_C(v_k)$ and
 $\overpsi'(v_k)_1=v_k$, thus we have $\gamma((1,
v_k,0))=(\sqrt{\gamma_C(v_k)(v_k)}, v_k, \gamma_C(v_k))$ for $k=1, \ldots, g$. Then, 
 for all $i \in \Zn$, one can set $\gamma((1,
i,0))=(\sqrt{\gamma_C(i)(i)}, i, \gamma_C(i))$ where the sign of the square root is
 computed using the group law of $G(n)$ (in case there is an ambiguity).
Indeed,
 for $i,j \in \Zn$, we have 
 \begin{equation}
	\begin{split}
 (\sqrt{\gamma_C(i)(i)}, i, \gamma_C(i))
	 (\sqrt{\gamma_C(j)(j)}, j, \gamma_C(j)) & =
 (\sqrt{\gamma_C(i)(i) \gamma_C(j)(j)} \gamma_C(j)(i), i+j, \gamma_C(i+j)) \\
	 & = (\sqrt{\gamma_C(i+j)(i+j)}, i+j, \gamma_C(i+j)).
	\end{split}
	\end{equation}
 For the last equality, we use the fact that for all $i, j \in \Zn$,
 $\gamma_C(j)(i)=\gamma_C(i)(j)$ because of the
symmetry of $C$.
Thus, we have, for $i \in \Zn$, $\theta_i^{\Thetapol \circ \gamma}=(\Thetapol\circ
\gamma(i)). \lambda \theta_0^{\Thetapol \circ \gamma}=\lambda \Thetapol((\sqrt{\gamma_C(i)(i)},
 i, \gamma_C(i))). \theta_0^\Thetapol= \lambda \sqrt{\gamma_C(i)(i)} \gamma_C(i)(-i) \theta_i^\Thetapol$
 for $\lambda \in \overk^*$ which
 is (\ref{eq:prop10:2}).

 For case (3), let $\overpsi'= \psimb(H_g)$. Following Proposition~\ref{lem:semi}, we
 choose $\gamma \in \Psi^{-1}(\overpsi')$ defined by the $\overpsi'$-semi-character such
 that $\chi_{\overpsi'}(v_k)=1$ for $k=1, \ldots, 2g$. Recall from Proposition
~\ref{prop:base} that $\mathrm{Vect}(\theta_0^{\Thetapol \circ \gamma}) \subset \Gamma(A,\pol)$ is
 the image of the projector $s \mapsto \sum_{j \in \dZn}
 \Thetapol \circ \gamma((1,0,j))(s)$. But we have $\sum_{j \in \dZn} \Thetapol \circ
 \gamma((1,0,j))= \sum_{i \in \Zn} \Thetapol((1,i,0))$. Thus 
 we have $\theta_0^{\Thetapol \circ \gamma}= \lambda
 \sum_{i \in \Zn}
 \theta_i^\Thetapol$ for $\lambda \in \overk^*$. Then, $\theta_i^{\Thetapol\circ \gamma}=\lambda \Thetapol\circ
 \gamma((1,i,0)). (\sum_{j \in \Zn} \theta_j^\Thetapol)= \lambda \sum_{j \in \Zn}
 \Thetapol \circ \gamma_{H_g}(i) \theta_j^\Thetapol=\lambda \sum_{j \in \Zn} \gamma_{H_g}(i)(j)
 \theta_j^\Thetapol$ which is (\ref{eq:prop10:3}).
\end{proof}
\begin{remark}
 The sign choices of the square roots in case~(2) are unavoidable, as they correspond precisely to the action of $K(2)$ in the exact sequence
 (\ref{eq:exactaut}). Cases~(1) and~(3) are similarly defined up to an action of $K(2)$.
\end{remark}
From the preceding proposition, we derive Algorithm~\ref{algo:trans} and the following corollary:
\begin{corollary}
 There exists an algorithm that takes as input $(B, \bpol, \Thetabpol)$ a marked
 abelian variety of type $K(m)$ given by its theta null point $0_{\Thetabpol}$ and a
 symplectic matrix $M \in \Sp_{2g}(\Z/m\Z)$ and outputs $0_{\Thetabpol \circ \gamma}$
 where $\gamma \in \Auts(G(m))$ is such that $\Psi(\gamma)=M$ with running time $O(g^3)$
 operations in $\Z/m\Z$ (once $m$ is factored) and $O(g^2m^g+g\,m^{g+1})$ base field
 operations.
\end{corollary}
\begin{proof}
 By the proof of Proposition~\ref{prop:trans}, the decomposition of $M$ is obtained by Gaussian
 eliminations modulo the prime divisors of $m$ and the Chinese remainder theorem, hence in
 $O(g^3)$ operations in $\Z/m\Z$, and it involves a bounded number of generators. The action of
 $B_g(A)$ is a permutation of the coordinates and the action of $S_g(C)$ is a multiplication of
 each coordinate by a root of unity; once the values $\gamma_A(i)$ and $\gamma_C(i)(i)$ are
 computed, both cost $O(g^2m^g)$ operations. The action of $H_g$ is a discrete Fourier transform
 on $(\Z/m\Z)^g$; computed one coordinate at a time, it costs $O(g\,m^{g+1})$ operations (and
 $O(g\,m^g\log(m))$ operations with a fast Fourier transform).
\end{proof}

\begin{remark}\label{rk:transcost}
    Algorithm~\ref{algo:trans} decomposes $M$ into generators and applies
    Proposition~\ref{prop:trans} sequentially, at a cost of $O(g^2m^g + g\,m^{g+1})$
    base field operations. Alternatively, one can apply the closed-form formula of
    \cite[Theorem~12]{dartois_fast} directly, computing for each $i \in
    Z(m)$ a sum over all $j \in Z(m)$, at a cost of $O(m^{2g})$ operations.
    For $g=1$, or for fixed small $m$ (e.g.\ $m=2$ or $m=4$) and moderate $g$, the two
    approaches have comparable cost; when $g \geq 2$ and $m$ grows, the decomposition into
    generators is preferable.
\end{remark}

\begin{algorithm}
\SetKwInOut{Input}{input}\SetKwInOut{Output}{output}

\SetKwComment{Comment}{/* }{ */}
\Input{
 \begin{itemize}
	\item A marked abelian variety $(B, \bpol, \Thetabpol)$ given by its theta null
	 point $0_{\Thetabpol}$;
	\item $M \in \Sp_{2g}(\Z/m\Z)$.
 \end{itemize}
}
\Output{
 \begin{itemize}
	\item $0_{\Thetapbpol}$ with $\Thetapbpol = \Thetabpol \circ \gamma$ for $\gamma
	 \in \Auts(G(m))$ such that $\Psi(\gamma)=M$.
 \end{itemize}
}

\BlankLine
Compute a decomposition of $M$ as a product of generators $B_g(A)$, $S_g(C)$, $H_g$
for $A \in GL_g(\Z/m\Z)$ and $C \in M_g(\Z/m\Z)$ as in the proof of Proposition~\ref{prop:trans}\;
Use the formulas of Proposition~\ref{prop:trans} to compute $0_{\Thetapbpol}$.\\
\caption{Algorithm to compute the action of $\Sp_{2g}(\Z/m\Z)$ on theta null points.}
 \label{algo:trans}
\end{algorithm}

By Lemma~\ref{lem:action}, $\subG(n)$ is exactly the subgroup of $\Auts(G(n))$
consisting of elements compatible with $G(m)$. Using Proposition~\ref{prop:trans}, we can
now compute their action on $\vjxz(\overk)$.

\begin{proposition}\label{prop:action2}
 Let $\mB=\{v_k, v_{k+g}\}$ be the canonical basis of $K(n)$, and let $(A, \pol, \Thetapol)$
 be a marked abelian variety of type $K(n)$.
 The group $\subG(n) \subset \Auts(G(n))$ is generated by the following subgroups:
 \begin{enumerate}
	\item $\subG_0= \Phi(K(n)[2]) \cap \subG(n)$.
	 If $d$ is even then $\subG_0=\Phi(\{0\}\times\dZn[2])$ and the action is
	 given by Lemma~\ref{lem:actionK}, otherwise $\subG_0=\{ \mathrm{Id} \}$.
	\item $\subG_1$ composed of $g_1 \in \subG(n)$ such that $\psimb^{-1}(\Psi(g_1)) \in \Sp_{2g}(\ZZn)$ is a matrix
	 of the form $\begin{pmatrix} A & 0_g \\ 0_g & ^t A^{-1} \end{pmatrix}$ in the
		basis $\mB$ with $A \in GL_g(\Z/n \Z)$ such that $A = 1_g \mod m$. If
		$g_1 \in \subG_1$, its action on $\theta_i^\Thetapol$ is given by Proposition
		\ref{prop:trans} (1).
	\item $\subG_2$ composed of $g_2 \in \subG(n)$ such that $\psimb^{-1}(\Psi(g_2)) \in \Sp_{2g}(\ZZn)$ is a matrix
	 of the form $\begin{pmatrix} 1_g & 0_g \\ C & 1_g \end{pmatrix}$ in the
		basis $\mB$ with $C$ a symmetric matrix with coefficients in $\Z / n \Z$
		such that $C = 0_g \mod m'$ where $m'$ is the positive integer such
		that $m=d_0 m'$. If $g_2 \in \subG_2$, its action on
		$\theta_i^\Thetapol$ is given by Proposition~\ref{prop:trans} (2) where
		the signs of $\gamma_C(i)(i)^{-1/2}$ are chosen so that for all $i \in
		Z(m)$:
		\begin{equation}
		 \gamma_C(\mu_{m,n}(i))(\mu_{m,n}(i))^{-1/2}=1.
		\end{equation}
 \end{enumerate}
 \end{proposition}
\begin{proof}
 Let $g' \in \subG(n)$. Let $\psi_{g'}= \Psi(g') \in Sp(K(n))$ and denote
 $M_{\psi_{g'}}=\psimb^{-1}(\Psi(g'))$ its
 matrix in the basis $\mB$ of $K(n)$. By Lemma~\ref{lem:action} and conditions (1) and
 (2) of Definition~\ref{def:compat}, $M_{\psi_{g'}}$ is a matrix of the form
 $\begin{pmatrix} A & 0 \\ C_0 & ^{t}A^{-1} \end{pmatrix}$. Condition (2) means
	furthermore that for $k=1, \ldots, g$, $\psi_{g'}(v_{g+k})=v_{g+k} + \sum_{i=1}^{g}
	a_{k,i} v_{g+i}$, where $a_{k,i} \in \Z / n \Z$ and $d a_{k,i}=0$, that is $m | a_{k,i}$, which means that
	$A = 1_g\mod m$. Set $M_1 = B_g(A^{-1}) M_{\psi_{g'}}$, then $M_1=S_g(C)$
	for $C=(c_{ij})$ a symmetric matrix.

 Let $\psi_{M_1}=\psi(M_1)$ be the morphism whose matrix in $\mB$ is $M_1$. For $k=1, \ldots, g$,
 we have:
 \begin{equation}
	\psi_{M_1}(v_k)=v_k+ \sum_{i=1}^g c_{k,i} v_{g+i}.
 \end{equation}
 We remark that if $\{ v_1, \ldots, v_g \}$ is a basis of $Z(n)$ then $\{ dv_1, \ldots,
 d v_g \}$ is a basis of $\mu_{m,n}(Z(m))$. So condition (3) of Definition
~\ref{def:compat} is fulfilled if and only if for all $k=1,\ldots, g$:
 $d\sum_{i=1}^g c_{k,i} v_{g+i} \in \dnu_{d,n}(\dZ(d))$. But this means that $m' | c_{k,i}$
 for all $i,k=1, \ldots, g$.

When $\Psi(g')$ is the identity, $g'$ is of the form $\Phi(c)$ with $c \in K(n)[2]$, since $g'$ is symmetric. Let
 $g_0 \in \subG_0$ and let $c=(c_1, c_2) \in K(n)[2] \subset K(n)= \Zn\times \dZn$ be such that $g_0 = \Phi(c)$.
 We know that $\pi_{G(n)}\circ g_0 = \pi_{G(n)}$. Moreover, by Lemma~\ref{lem:action}, and
 conditions (1) and (2) of Definition~\ref{def:compat}, we deduce that for all $x \in
 \{1\}\times\{0\} \times \dZn$, $g_0(x)=x$. This means that $e_n(c,(0,y))=1$ for all $y \in
 \dZn$, thus $c_1 = 0$. By condition (3) of Definition~\ref{def:compat}, we have
 $e_n(c,(x,0))=1$ for all $x \in \mu_{m,n}(Z(m))=d\Zn$, that is $c_2(x)^d=1$ for all $x \in
 \Zn$. As $c_2 \in \dZn[2]$, this condition is fulfilled for all $c_2$ if $d$ is even, and
 implies $c_2=0$ if $d$ is odd.

 We have just proved that $\subG(n)$ is generated by the groups $\subG_0$, $\subG_1$ and $\subG_2$.

\end{proof}

\section{Structure of \texorpdfstring{$\vjxz(\overk)$}{V0Jx(k)}}\label{sec:struct}
Keeping the notation of Proposition~\ref{sec1:prop1},
we introduce the main object of study of this section.
\begin{definition}\label{sec1:defG}
Let $f: A \rightarrow B$ be an isogeny such that $(A, \pol, \Thetapol)$
 and $(B, \bpol, \Thetabpol)$ are isog-$f$-compatible. Let $x \in A(\overk)$,
 we denote by $G(x)$ the set
 $f(x+\Thetapolbar(Z(n) \times\{0\}))$.
\end{definition}

It follows immediately that $G(0_A)$ is a subgroup of $B[n]$
 isomorphic to $Z(n)$: indeed,
 $\Thetapolbar(Z(n) \times\{0\}) \cap \Ker(f)=\{0\}$ by Definition~\ref{sec1:def1}.
 By the same definition, $\Thetabpolbar(Z(m) \times\{0\}) \subset G(0_A)$.
 More generally, for all $x \in A(\overk)$, $f(x) + \Thetabpolbar(Z(m) \times\{0\})
 \subset G(x) \subset f(x) + B[n]$.
Using \cref{sec1:prop1}, we can compute the theta coordinates of the geometric points
of $G(x)$ for all $x \in A(\overk)$:

\begin{proposition}\label{sec3:prop9}
 Let $e_{\Thetabpol}: B \rightarrow \proj^{Z(m)}$ be the embedding of Definition
~\ref{def:thetanull}. Let $x \in A(\overk)$ and suppose that we have chosen 
 $\rho^\pol_x: \pol(x)=\pol \otimes
\stsheaf_A(x)
\rightarrow \stsheaf_A(x)$ a rigidification of $\pol$ at $x$. For all $s \in \Gamma(A,\pol)$,
 we denote by $s(x)$ the evaluation of $\rho_x^\pol(s)$ in $x$. Let $\rho^\bpol_{f(x)}: \bpol(f(x))
 \rightarrow \stsheaf_B(f(x))$ be the unique rigidification such that
 $f^*(\rho^\bpol_{f(x)})=\rho^\pol_x$, so that we have $f^*(s)(x)=s(f(x))$ for all $s \in
 \Gamma(B,\bpol)$. With these settings, let $\lambda_x: \Zn \rightarrow G(x)$ be the map such that for $j \in \Zn$,
 $\lambda_x(j)$ is the point of $B \subset \proj^{Z(m)}$ with coordinates:
 \begin{equation}\label{eq:sec3:prop9}
	(\theta^\Thetapol_{\mu_{m,n}(i)+j}(x))_{i \in Z(m)}.
 \end{equation}
Then $\lambda_{0_A}: \Zn \rightarrow G(0_A)$ is a group isomorphism and $\lambda_x(j)=f(x)+\lambda_{0_A}(j)$
 for all $j \in \Zn$, so that $\lambda_x$ is a bijection. For all $j \in Z(m)$,
 $\lambda_x(\mu_{m,n}(j))=f(x)+\Thetabpolbar((j, 0))$, thus
for all $x, y \in A(\overk)$ and $j \in Z(m)$, we have
 $\lambda_x(\mu_{m,n}(j))=\lambda_y(\mu_{m,n}(j))+f(x-y)$.\end{proposition}
\begin{proof}
 The point $e_\Thetapol(x)$ of $A$ has projective coordinates
 $(\theta^\Thetapol_i(x))_{i \in \Zn}$.
 Following (\ref{eq:trans}), for $j \in \Zn$, $x + \Thetapolbar((j, 0))$ has projective
 coordinates $(\theta^\Thetapol_{i+j}(x))_{i \in \Zn}$ and then by \cref{sec1:prop1}
 we have that $f(x+\Thetapolbar((j, 0)))$ has projective coordinates
 $(\theta^\Thetapol_{\mu_{m,n}(i)+j}(x))_{i \in Z(m)}$. This shows that $\lambda_x(j)=f(x+\Thetapolbar((j, 0)))=f(x)+\lambda_{0_A}(j)$, and
 $\lambda_{0_A}=f \circ \Thetapolbar$ restricted to $\Zn \times \{0\}$ is a group isomorphism onto
 $G(0_A)$.
 Now for $j \in Z(m)$, we have
 $\lambda_x(\mu_{m,n}(j))=f(x +\Thetapolbar((\mu_{m,n}(j), 0)))=f(x)+\Thetabpolbar((j, 0))$ by
 point (3) of Definition~\ref{sec1:def1}.
\end{proof}

We would like to characterize the subgroups of $B[n]$ that can arise
as $G(0_{A_y})$ for $y \in \vjxz(\overk)$ and $f \in F(y)$ (with the notation introduced before
Proposition~\ref{prop:phs}). For this, we need the following definition:

\begin{definition}\label{def:symcompat}
 Let $(A, \pol, \Thetapol)$ be a level $n$ marked abelian variety, and recall that $\pigpol:\Gpol
 \rightarrow \Kpol$ is the canonical projection.
 Let $\tildeH$ be a
 symmetric level subgroup of $\Gpol$ and $H = \pigpol(\tildeH)$. Let $x \in
 A(\overk)$ be a torsion point.
 If $x \in \Kpol( \overk)$, we say that $x$ is
 symmetric compatible with $\tildeH$
 if there exists a symmetric $g_x \in \Gpol$ such that the subgroup $(g_x)$ generated by
 $g_x$ is a level subgroup above the subgroup $(x)$ generated by $x$ and $\pigpol((g_x) \cap \tildeH)=(x) \cap H$. 

 In general (we do not suppose that $x \in \Kpol(\overk)$), we say that $x$ is symmetric
 compatible with $\tildeH$ if there exist: 
 \begin{enumerate}
	\item 
 a separable isogeny $f_0: B_0 \rightarrow A$ with kernel $K_0$ and $y \in
	 K(f_0^*(\pol))(\overk)$ such that $f_0(y)=x$. Note that $H'=f_0^{-1}(H)$ is
	 isotropic for $e_{f_0^*(\pol)}$ since $H$ is isotropic for $e_\pol$, and denote by
	 $\tildeK_0$ the (symmetric by Remark~\ref{rm:sym}) descent data of $f_0^*(\pol)$ to $\pol$; 
\item a symmetric level subgroup $\tildeH'$ of $G(f_0^*(\pol))$ above $H'$
 such that $\tildeK_0 \subset \tildeH'$ and
	 $f_0^\sharp(\tildeH')=\tildeH$;
 \end{enumerate}
	 such that $y$ is symmetric compatible with $\tildeH'$.
\end{definition}
For $x \in A(\overk)$ a torsion point and $\tildeH$ a level subgroup of $\Gpol$, we let
$\expo(x,\tildeH)=\min\{ \lambda \in \N^* | \lambda x \in \pi_{\Gpol}(\tildeH) \}$; the set of
integers $\lambda$ such that $\lambda x \in \pi_{\Gpol}(\tildeH)$ is $\expo(x, \tildeH)\Z$. We use the
same notation for any abelian variety, line bundle and level subgroup, for instance
$\expo(y, \tildeH')$ with the notation of Definition~\ref{def:symcompat}.

The following lemma collects the properties of the maps $\fsharp$ of Definition
\ref{def:fsharp} that we need. In it, isogenies are separable, line bundles are ample and,
for an isogeny $f: A \rightarrow B$ and a line bundle $\bpol$ on $B$, we take $\pol=f^*(\bpol)$
and $\psi$ the identity in Definition~\ref{def:fsharp}; we denote by $\tildeK_f \subset
\Gpol$ the corresponding descent data, by $G^*_f(\pol)$ its centralizer and by
$\fsharp: G^*_f(\pol) \rightarrow G(\bpol)$ the map of Definition~\ref{def:fsharp}. Recall that
$\fsharp$ is surjective with kernel $\tildeK_f$ (Proposition~\ref{prop:compat}), that it is the
identity on $\overk^*$ and that $\pi_{G(\bpol)} \circ \fsharp=f \circ \pigpol$. Moreover, an element
$(\tau_x, \psi_x) \in \Gpol$ belongs to $G^*_f(\pol)$ if and only if $\psi_x=f^*(\psi_y)$ for
an isomorphism $\psi_y: \bpol \rightarrow \tau_{f(x)}^*\bpol$, and then $\fsharp((\tau_x,
\psi_x))=(\tau_{f(x)}, \psi_y)$; the elements of $\tildeK_f$ are the $(\tau_z, f^*(\mathrm{Id}_\bpol))$
for $z \in \Ker f$. Finally, the commutator of two elements of $\Gpol$ only depends on their
images in $\Kpol$, so that $G^*_f(\pol)=\pigpol^{-1}(\Ker(f)^{\perp_{\epol}})$, where
$\Ker(f)^{\perp_{\epol}}=f^{-1}(K(\bpol))$ by Proposition~\ref{prop:compat}.

\begin{lemma}~
 \label{lem:fsharpprops}
 \begin{enumerate}
	\item Suppose that $\bpol$ is symmetric and let $\delta_{-1}$ be defined with the normalized
	 isomorphisms $\psii^\bpol$ and $\psii^\pol=f^*(\psii^\bpol)$. Then $\tildeK_f$ is
	 symmetric, $\delta_{-1}(G^*_f(\pol))=G^*_f(\pol)$ and $\fsharp \circ
	 \delta_{-1}=\delta_{-1}\circ \fsharp$ on $G^*_f(\pol)$.
	\item With the hypothesis of (1), if $g \in G^*_f(\pol)$ is such that $\fsharp(g)$ is
	 symmetric, then $g$ is symmetric.
	\item Let $f_1: A \rightarrow A_1$ and $f_2: A_1 \rightarrow A_2$ be isogenies, $f=f_2 \circ
	 f_1$, $\bpol_2$ a line bundle on $A_2$, $\bpol_1=f_2^*(\bpol_2)$ and
	 $\pol=f_1^*(\bpol_1)=f^*(\bpol_2)$. Then $\tildeK_{f_1} \subset \tildeK_f$,
	 $G^*_f(\pol)=(f_1^\sharp)^{-1}(G^*_{f_2}(\bpol_1))$, $f_1^\sharp(\tildeK_f)=\tildeK_{f_2}$ and
	 $\fsharp=f_2^\sharp \circ f_1^\sharp$ on $G^*_f(\pol)$.
 \end{enumerate}
\end{lemma}
\begin{proof}
 (1) Note that $\psii^\pol=f^*(\psii^\bpol)$ is normalized since $f(0_A)=0_B$. Let
 $g=(\tau_x, f^*(\psi_y)) \in G^*_f(\pol)$ and $y=f(x)$. As $f \circ [-1]=[-1]\circ f$ and $f
 \circ \tau_{-x}=\tau_{-y} \circ f$, we have
 \begin{equation}
	\delta_{-1}(g)=(\tau_{-x}, \tau^*_{-x}(f^*\psii^\bpol)^{-1} \circ [-1]^*(f^*\psi_y) \circ
	f^*\psii^\bpol)=(\tau_{-x}, f^*(\tau^*_{-y}(\psii^\bpol)^{-1} \circ [-1]^*(\psi_y) \circ
	\psii^\bpol)),
 \end{equation}
 so that $\delta_{-1}(g) \in G^*_f(\pol)$ and $\fsharp(\delta_{-1}(g))=\delta_{-1}(\fsharp(g))$. For $g
 \in \tildeK_f$, $\fsharp(g)=1$, so that $\fsharp(\delta_{-1}(g))=1$ and $\delta_{-1}(g) \in \tildeK_f$
 lies above $-\pigpol(g)$; as $\tildeK_f$ is a level subgroup, $\delta_{-1}(g)=g^{-1}$.

 (2) Let $k=\delta_{-1}(g)g$. By (1), $k \in G^*_f(\pol)$ and $\fsharp(k)=\delta_{-1}(\fsharp(g))\fsharp(g)=1$,
 so that $k \in \tildeK_f$. As $k$ lies above $-\pigpol(g)+\pigpol(g)=0$ and $\tildeK_f$ is a
 level subgroup, $k=1$, that is $\delta_{-1}(g)=g^{-1}$.

 (3) For $z \in \Ker f_1 \subset \Ker f$, $f_1^*(\mathrm{Id}_{\bpol_1})=f^*(\mathrm{Id}_{\bpol_2})$
 as isomorphisms $\pol \rightarrow \tau_z^*\pol$, so that $\tildeK_{f_1} \subset \tildeK_f$. For $z
 \in \Ker f$, $(\tau_z, f^*(\mathrm{Id}_{\bpol_2}))=(\tau_z, f_1^*(f_2^*(\mathrm{Id}_{\bpol_2})))$, so
 that it belongs to $G^*_{f_1}(\pol)$ and $f_1^\sharp$ maps it to $(\tau_{f_1(z)},
 f_2^*(\mathrm{Id}_{\bpol_2})) \in \tildeK_{f_2}$; as $f_1(\Ker f)=\Ker f_2$, we get
 $f_1^\sharp(\tildeK_f)=\tildeK_{f_2}$. Let $g \in G^*_{f_1}(\pol)$. As $f_1^\sharp$ is a morphism
 which is the identity on $\overk^*$ and $\tildeK_{f_1}=\Ker f_1^\sharp \subset \tildeK_f$, for
 $k \in \tildeK_f$ the commutator $gkg^{-1}k^{-1} \in \overk^*$ is equal to
 $f_1^\sharp(g)f_1^\sharp(k)f_1^\sharp(g)^{-1}f_1^\sharp(k)^{-1}$, so that $g$ commutes with
 $\tildeK_f$ if and only if $f_1^\sharp(g)$ commutes with $f_1^\sharp(\tildeK_f)=\tildeK_{f_2}$:
 this means that $G^*_f(\pol)=(f_1^\sharp)^{-1}(G^*_{f_2}(\bpol_1))$ (note that $G^*_f(\pol)
 \subset G^*_{f_1}(\pol)$). Finally, if $g=(\tau_x, \psi_x) \in G^*_f(\pol)$,
 $f_1^\sharp(g)=(\tau_{f_1(x)}, \psi')$ and $f_2^\sharp(f_1^\sharp(g))=(\tau_{f(x)}, \psi'')$, then
 $f^*(\psi'')=f_1^*(f_2^*(\psi''))=f_1^*(\psi')=\psi_x$, so that $f_2^\sharp(f_1^\sharp(g))=\fsharp(g)$.
\end{proof}

\begin{lemma}\label{lem:symcompatchar}
 Let $\pol'$ be a line bundle on an abelian variety $B_0$, $\tildeH'$ a symmetric level
 subgroup of $G(\pol')$ above $H'$ and $y \in K(\pol')$. Let $\expo=\expo(y, \tildeH')$. Then $y$ is symmetric compatible with $\tildeH'$ if and only
 if there exists a symmetric element $g \in G(\pol')$ above $y$ such that $g^\expo \in
 \tildeH'$. Moreover, if $\expo$ is odd, $y$ is always symmetric compatible with $\tildeH'$.
\end{lemma}
\begin{proof}
 Recall that the set of integers $\lambda$ such that $\lambda y \in H'$ is $\expo\Z$.
 If $y$ is symmetric compatible with $\tildeH'$, let $g$ be as in Definition
~\ref{def:symcompat}; as $\expo y \in (y) \cap H'$, an element of $(g)$ above $\expo y$ lies
 in $\tildeH'$, and since $(g)$ is a level subgroup, this element is $g^\expo$. Conversely, let
 $g$ be symmetric above $y$ with $g^\expo \in \tildeH'$, and let $r$ be the order of $y$. As $ry=0
 \in H'$, $\expo | r$ and $g^r=(g^\expo)^{r/\expo}$ is an element of the level subgroup $\tildeH'$
 above $0$, so that $g^r=1$. Thus, if $g^k \in \overk^*$, then $ky=0$, $r | k$ and $g^k=1$: $(g)$
 is a level subgroup above $(y)$. If $g^k \in \tildeH'$ then $ky \in H'$; conversely, if $ky \in
 H'$, then $\expo | k$ and $g^k=(g^\expo)^{k/\expo} \in \tildeH'$. Hence $\pi_{G(\pol')}((g) \cap
 \tildeH')=(y)\cap H'$. Finally, let $g$ be any symmetric element above $y$; $g^\expo$ is a
 symmetric element above $\expo y \in H'$, so that $g^\expo=\pm h$ where $h$ is the element of
 $\tildeH'$ above $\expo y$ (Lemma~\ref{intro:lemma1}). If $\expo$ is odd, $(\pm g)^\expo=\pm
 g^\expo$, so that one of the two symmetric elements $\pm g$ above $y$ satisfies the condition.
\end{proof}

The following proposition shows that the property ``$x \in A(\overk)$ is symmetric
compatible with $\tildeH$'' does not depend on the choices of $y$ and $f_0$ in Definition
\ref{def:symcompat}.
\begin{proposition}\label{prop:compat1}
 In the preceding definition:
 \begin{enumerate}
	\item there always exists a separable isogeny $f_0: B_0 \rightarrow A$ and $y \in B_0(\overk)$
	 such that $f_0(y)=x$ and $y \in K(f_0^*(\pol))(\overk)$;
	\item if $f_0: B_0 \rightarrow A$ is a separable isogeny with kernel $K_0$, there
	 always exists a unique level subgroup $\tildeH'$ over $H'=f_0^{-1}(H)$ such
	 that $\tildeK_0 \subset \tildeH'$ and
	 $f_0^\sharp(\tildeH')=\tildeH$, namely $\tildeH'=(f_0^\sharp)^{-1}(\tildeH)$, and it is
	 symmetric;
	\item the validity of the condition ``$y$ is symmetric compatible with $\tildeH'$'' in
	 Definition~\ref{def:symcompat} does not depend on the choices of $f_0$ and $y$ satisfying (1).
 \end{enumerate}
\end{proposition}
\begin{proof}
 The first claim is just Remark~\ref{rm:intro}. In the rest of the proof, we use the notation
 introduced before Lemma~\ref{lem:fsharpprops}; recall that $\pol$ is symmetric.

 (2) Let $\pol'=f_0^*(\pol)$ and $\tildeH'=(f_0^\sharp)^{-1}(\tildeH) \subset G^*_{f_0}(\pol')$. It is a
 subgroup containing $\tildeK_0=\Ker f_0^\sharp$, and $f_0^\sharp(\tildeH')=\tildeH$ since
 $f_0^\sharp$ is surjective. As $\pi_{G(\pol)} \circ f_0^\sharp=f_0 \circ \pi_{G(\pol')}$,
 $\pi_{G(\pol')}(\tildeH') \subset H'$; conversely, if $h' \in H'$, there exists $g \in \tildeH'$
 above a point of $h'+K_0$, and multiplying $g$ by an element of $\tildeK_0$, we get an element of
 $\tildeH'$ above $h'$. As $f_0^\sharp$ is the identity on $\overk^*$ and $\tildeH \cap
 \overk^*=\{1\}$, $\tildeH' \cap \overk^*=\{1\}$, so that $\tildeH'$ is a level subgroup above
 $H'$. By Lemma~\ref{lem:fsharpprops} (2), as $\tildeH$ is symmetric, $\tildeH'$ is symmetric.
 If $\tildeH''$ is another level subgroup above $H'$ such that $\tildeK_0 \subset \tildeH''$ and
 $f_0^\sharp(\tildeH'')=\tildeH$, then $\tildeH''$ is commutative (its commutators are in
 $\tildeH'' \cap \overk^*=\{1\}$) and contains $\tildeK_0$, so that $\tildeH'' \subset
 G^*_{f_0}(\pol')$ and $\tildeH'' \subset (f_0^\sharp)^{-1}(\tildeH)=\tildeH'$; as both have $|H'|$
 elements, $\tildeH''=\tildeH'$. Note also that $H'$ is isotropic for $e_{\pol'}$ by
 Proposition~\ref{def:weil} (2), since $H$ is isotropic for $\epol$.

 Let $\expo=\expo(x, \tildeH)$; as $H'=f_0^{-1}(H)$, for
 all $y \in f_0^{-1}(x)$, $\expo(y, \tildeH')=\expo$. If $\expo$ is odd, then by Lemma~\ref{lem:symcompatchar}, $y$ is symmetric compatible with
 $\tildeH'$ for all choices of $f_0$ and $y$, and there is nothing to prove. We suppose from now on
 that $\expo$ is even.

 (3) \emph{Independence of the choice of $y$.} Let $f_0$ be as in (2), $\pol'=f_0^*(\pol)$, and
 $y_1, y_2 \in K(\pol')$ such that $f_0(y_1)=f_0(y_2)=x$. Suppose that $y_1$ is symmetric
 compatible with $\tildeH'$ and let $g_1$ be a symmetric element above $y_1$ such that
 $g_1^\expo \in \tildeH'$ (Lemma~\ref{lem:symcompatchar}). Let $h'=y_2-y_1 \in K_0$ and $k \in
 \tildeK_0$ be the element above $h'$; it is symmetric by Lemma~\ref{lem:fsharpprops} (1). Let
 $\zeta=g_1kg_1^{-1}k^{-1} \in \overk^*$, so that $kg_1=\zeta^{-1}g_1k$, and let $\lambda \in
 \overk^*$ be such that $\lambda^2=\zeta^{-1}$. Let $s=\lambda g_1k$, which lies above $y_2$. We have
 $\delta_{-1}(s)=\lambda g_1^{-1}k^{-1}$ and $s^{-1}=\lambda^{-1}k^{-1}g_1^{-1}=\lambda^{-1}\zeta^{-1}
 g_1^{-1}k^{-1}$, so that $s$ is symmetric. Moreover, moving the factors $k$ to the right,
 \begin{equation}
	s^\expo=\lambda^\expo (g_1k)^\expo=\lambda^\expo \zeta^{-\expo(\expo-1)/2}g_1^\expo
	k^\expo=\zeta^{-\expo^2/2}g_1^\expo k^\expo.
 \end{equation}
 Now $\zeta=e_{\pol'}(y_1,h')$ and $\zeta^\expo=e_{\pol'}(\expo y_1, h')=1$ since $\expo y_1, h' \in
 H'$ and $H'$ is isotropic, so that $\zeta^{\expo^2/2}=(\zeta^\expo)^{\expo/2}=1$. As $g_1^\expo \in
 \tildeH'$ and $k^\expo \in \tildeK_0 \subset \tildeH'$, we get $s^\expo \in \tildeH'$, and $y_2$ is
 symmetric compatible with $\tildeH'$ by Lemma~\ref{lem:symcompatchar}.

 \emph{Independence of the choice of the isogeny.} Let $f_0: B_0 \rightarrow A$ and $f_1: B_1
 \rightarrow A$ be two separable isogenies with kernels $K_0$ and $K_1$, and $y_0 \in
 K(f_0^*(\pol))$, $y_1 \in K(f_1^*(\pol))$ such that $f_0(y_0)=f_1(y_1)=x$. For $i=0,1$, let
 $\pol'_i=f_i^*(\pol)$, $H'_i=f_i^{-1}(H)$ and $\tildeH'_i=(f_i^\sharp)^{-1}(\tildeH)$. We suppose
 that $y_0$ is symmetric compatible with $\tildeH'_0$ and show that $y_1$ is symmetric compatible
 with $\tildeH'_1$. The fibre product $B_0 \times_A B_1$ is a group scheme whose projection to
 $B_0$ is surjective with kernel $\{0\}\times K_1$, which is finite and \'etale since $f_1$ is
 separable; hence it is smooth of dimension $g$ and its neutral component $C$ is an abelian
 variety. The projections restrict to surjective morphisms $f'_1: C \rightarrow B_0$ and $f'_0:
 C \rightarrow B_1$ (their images are closed connected subgroups of dimension $g$) whose kernels
 are contained in $\{0\}\times K_1$ and $K_0 \times \{0\}$: they are separable isogenies, and the
 following diagram commutes:
\begin{equation}\label{diag:isogcompat}
 \begin{tikzpicture}
 \matrix [column sep={1.5cm}, row sep={1cm}]
 {
	\node(a){$C$}; & \node(b){$B_0$}; & \\
	\node(c){$B_1$}; & \node(d){$A$}; \\
 };
	\draw[->] (a) -- (b) node[above, midway]{$f'_1$};
	\draw[->] (a) -- (c) node[left, midway]{$f'_0$};
	\draw[->] (b) -- (d) node[left, midway]{$f_0$};
	\draw[->] (c) -- (d) node[above, midway]{$f_1$};
 \end{tikzpicture}
\end{equation}
 Let $h=f_0 \circ f'_1=f_1 \circ f'_0$, $\polC=h^*(\pol)={f'}_1^*(\pol'_0)={f'}_0^*(\pol'_1)$ and
 $\tildeH_C=(h^\sharp)^{-1}(\tildeH)$. By Lemma~\ref{lem:fsharpprops} (3) applied to $h=f_0 \circ
 f'_1$, we have $G^*_h(\polC)=({f'}_1^\sharp)^{-1}(G^*_{f_0}(\pol'_0))$ and $h^\sharp=f_0^\sharp \circ
 {f'}_1^\sharp$ on $G^*_h(\polC)$, so that
 \begin{equation}\label{eq:compat1:1}
	\tildeH_C=({f'}_1^\sharp)^{-1}(\tildeH'_0) \text{ and, in the same way, }
	\tildeH_C=({f'}_0^\sharp)^{-1}(\tildeH'_1).
 \end{equation}
 Let $g_0$ be a symmetric element above $y_0$ such that $g_0^\expo \in \tildeH'_0$. As ${f'}_1^\sharp$
 is surjective, there exists $g_C \in G^*_{f'_1}(\polC)$ such that ${f'}_1^\sharp(g_C)=g_0$, and
 multiplying it by an element of $\tildeK_{f'_1}$, we can suppose that $g_C$ lies above a point
 $y_C$ such that $f'_1(y_C)=y_0$. By Lemma~\ref{lem:fsharpprops} (2), $g_C$ is symmetric, and by
 (\ref{eq:compat1:1}), $g_C^\expo \in \tildeH_C$ since ${f'}_1^\sharp(g_C^\expo)=g_0^\expo \in
 \tildeH'_0$. Let $y'_1=f'_0(y_C)$. As $f_1(y'_1)=h(y_C)=x=f_1(y_1)$, $y'_1-y_1 \in K_1 \subset
 K(\pol'_1)$, so that $y'_1 \in K(\pol'_1)$ and $y_C \in {f'}_0^{-1}(K(\pol'_1))$: this means that
 $g_C \in G^*_{f'_0}(\polC)$. Let $g'_1={f'}_0^\sharp(g_C)$; it lies above $y'_1$, it is symmetric
 by Lemma~\ref{lem:fsharpprops} (1) and, by (\ref{eq:compat1:1}), ${g'_1}^\expo={f'}_0^\sharp(g_C^\expo)
 \in \tildeH'_1$. By Lemma~\ref{lem:symcompatchar}, $y'_1$ is symmetric compatible with $\tildeH'_1$,
 and so is $y_1$ by the independence of the choice of $y$.
\end{proof}

We put the following definition:

\begin{definition}\label{def:kappa}
 Let $x \in A(\overk)$ be a torsion point and $\tildeH$ a level subgroup of $\Gpol$. We
 let $\kappa_0(x, \tildeH)=1$ if $x$ is symmetric compatible with $\tildeH$ and otherwise we let
 $\kappa_0(x, \tildeH)=-1$.
\end{definition}
If $x \in \Kpol$ and $\expo(x, \tildeH)$ is even, then for any symmetric $g_x \in \Gpol$ above $x$,
$g_x^{\expo(x, \tildeH)}=\kappa_0(x, \tildeH)\,h$ with $h \in \tildeH$ (Lemmas~\ref{intro:lemma1}
and~\ref{lem:symcompatchar} with $f_0=\mathrm{Id}_A$).
The sign $\kappa_0(x, \tildeH)$ is attached to the minimal exponent $\expo(x, \tildeH)$,
which depends on $x$. As a consequence, $\kappa_0$ is not additive in general: for
instance, if $\tildeH$ is a symmetric level subgroup above $H \subset \Kpol$ and $x \in
\Kpol$ is such that $\expo(x, \tildeH)$ is divisible by $4$ and $\kappa_0(x, \tildeH)=-1$, then
$\expo(2x, \tildeH)=\expo(x, \tildeH)/2$, $2x$ admits the symmetric lifts $\pm g_x^2$ whose
$(\expo(x,\tildeH)/2)$-th power is $g_x^{\expo(x,\tildeH)} \notin \tildeH$, so that $\kappa_0(2x,
\tildeH)=-1 \neq \kappa_0(x, \tildeH)^2$. In order to obtain an additive invariant, we fix
the exponent.

\begin{definition}\label{def:phiN}
 Let $\tildeH$ be a symmetric level subgroup of $\Gpol$ above $H=\pigpol(\tildeH)$, let $N$ be an even
 positive integer and let $x \in A(\overk)$ be such that $Nx \in H$. Let $f_0: B_0
 \rightarrow A$, $y \in K(f_0^*(\pol))(\overk)$ and $\tildeH'$ be as in Definition
~\ref{def:symcompat} (they exist by Proposition~\ref{prop:compat1}), and let $g_y \in
 G(f_0^*(\pol))$ be a symmetric element above $y$. As $Ny \in H'=f_0^{-1}(H)$, $g_y^N$ is a
 symmetric element above $Ny$, so that, by Lemma~\ref{intro:lemma1}, there exists a unique
 $\varphi_N(x, \tildeH) \in \{-1, 1\}$ such that $\varphi_N(x, \tildeH)\, g_y^N \in
 \tildeH'$. Note that $\varphi_N(x, \tildeH)$ does not depend on the choice of $g_y$ since the
 other symmetric element above $y$ is $-g_y$ and $N$ is even.
\end{definition}

\begin{lemma}\label{lem:phikappa}
 With the notation of Definition~\ref{def:phiN}, $\expo(x, \tildeH)$ divides $N$ and
 \begin{equation}
	\varphi_N(x, \tildeH)=\kappa_0(x, \tildeH)^{N/\expo(x, \tildeH)}.
 \end{equation}
 In particular, $\varphi_N(x, \tildeH)$ does not depend on the choices of $f_0$ and $y$, and if
 $N/\expo(x, \tildeH)$ is odd, then $\varphi_N(x, \tildeH)=\kappa_0(x, \tildeH)$.
\end{lemma}
\begin{proof}
 The set of integers $\lambda$ such that $\lambda x \in H$ is the ideal of $\Z$ generated by
 $\expo=\expo(x, \tildeH)$, so $\expo$ divides $N$. Note that $\expo(y, \tildeH')=\expo$ since
 $H'=f_0^{-1}(H)$. Let $g_y$ be a symmetric element above $y$. As $g_y^{\expo}$ is a symmetric
 element above $\expo y \in H'$, there exist $s \in \{-1, 1\}$ and $h \in \tildeH'$ such
 that $g_y^{\expo}=s h$. If $\expo$ is even, $s$ does not depend on the choice of $g_y$ and
 $s=\kappa_0(x, \tildeH)$ (see Definition~\ref{def:kappa} and the remark which follows it). If $\expo$
 is odd, then $\kappa_0(x, \tildeH)=1$ by Proposition~\ref{sec2:proptors} (1), and, replacing
 $g_y$ by $-g_y$ if necessary, we can suppose that $s=1$. In both cases, $g_y^N=s^{N/\expo}
 h^{N/\expo}$ with $h^{N/\expo} \in \tildeH'$, so that $\varphi_N(x, \tildeH)=s^{N/\expo}=\kappa_0(x,
 \tildeH)^{N/\expo}$. The independence of the choices then follows from Proposition
~\ref{prop:compat1}.
\end{proof}

The following proposition says that, once the exponent is fixed, the symmetric compatibility
property is additive:
\begin{proposition}\label{sec2:propgroup}
 Let $\tildeH$ be a symmetric level subgroup of $\Gpol$ above $H=\pigpol(\tildeH)$, let $N$ be an even
 positive integer and let $x_1, x_2 \in A(\overk)$ be torsion points such that $Nx_1,
 Nx_2 \in H$. We suppose that there exist a separable isogeny $f_0: B_0 \rightarrow A$ and
 $y_1, y_2 \in K(f_0^*(\pol))(\overk)$ such that $f_0(y_i)=x_i$ for $i=1,2$ and
 $e_{f_0^*(\pol)}(y_1, y_2)^N=1$ (if $x_1, x_2 \in \Kpol$, one can take $f_0=\mathrm{Id}_A$, and this
 condition means that $e_\pol(x_1, x_2)^N=1$).
 Then we have
 \begin{equation}
	\varphi_N(x_1 +x_2, \tildeH)=\varphi_N(x_1, \tildeH)\, \varphi_N(x_2, \tildeH).
 \end{equation}
\end{proposition}
\begin{proof}
 Let $\tildeH'$ be the symmetric level subgroup of $G(f_0^*(\pol))$ given by Proposition
~\ref{prop:compat1} (2). As $y_1+y_2 \in K(f_0^*(\pol))$ and $f_0(y_1+y_2)=x_1+x_2$, we can use
 $f_0$, $y_1$, $y_2$, $y_1+y_2$ and $\tildeH'$ to compute $\varphi_N(x_1, \tildeH)$,
 $\varphi_N(x_2, \tildeH)$ and $\varphi_N(x_1+x_2, \tildeH)$ (Lemma~\ref{lem:phikappa}).
 For $i=1,2$, let $g_i \in G(f_0^*(\pol))$ be a symmetric element above $y_i$, and let $c \in
 \overk^*$ be such that $g_1g_2=c\,g_2g_1$; $c$ is $e_{f_0^*(\pol)}(y_1, y_2)^{\pm 1}$, so that
 $c^N=1$. Let $r \in \overk^*$ be such that $r^2=c$ and $t=r^{-1}g_1g_2$. As $\delta_{-1}(g_1g_2)=g_1^{-1}
 g_2^{-1}=(g_2g_1)^{-1}=c\,(g_1g_2)^{-1}$, we have $\delta_{-1}(t)=t^{-1}$, so that $t$ is a symmetric
 element above $y_1+y_2$. By induction, $(g_1g_2)^N=c^{-N(N-1)/2}g_1^Ng_2^N$, and, with
 $\zeta=c^{N/2}=r^N \in \{-1,1\}$, $c^{N(N-1)/2}=\zeta^{N-1}=\zeta$ since $N$ is even. Hence
 \begin{equation}
	t^N=\zeta^{-2}g_1^N g_2^N=g_1^N g_2^N=\varphi_N(x_1, \tildeH)\varphi_N(x_2, \tildeH)\, h_1 h_2,
 \end{equation}
 where $h_i=\varphi_N(x_i, \tildeH) g_i^N \in \tildeH'$ for $i=1,2$, so that $h_1 h_2 \in \tildeH'$.
 This proves the claim.
\end{proof}

\begin{remark}
 The hypothesis that the exponent $N$ is the same for $x_1$, $x_2$ and $x_1+x_2$ is
 essential: Proposition~\ref{sec2:propgroup} does not hold if one replaces $\varphi_N$ by
 $\kappa_0$, even if $\expo(x_1, \tildeH)=\expo(x_2, \tildeH)$. For instance, take $g=2$,
 $n=8$, $H=\{0,4\}^2 \times \{0\} \subset K(8)$ and $\tildeH=\{(s(h),h,0), h \in H\} \subset
 G(8)$, where $s$ is the character of $H$ such that $s((4,0))=s((0,4))=-1$. Then $e_1=(1,1)$
 and $e_2=(1,-1)$ (in $Z(8)=(\Z/8\Z)^2$) satisfy $\expo(e_i, \tildeH)=4$ and are symmetric compatible
 with $\tildeH$, but $e_1+e_2=(2,0)$ satisfies $\expo(e_1+e_2, \tildeH)=2$ and is not
 symmetric compatible with $\tildeH$.
\end{remark}

\begin{corollary}\label{sec2:corgroup}
 Let $\tildeH$ be a symmetric level subgroup of $\Gpol$ above $H=\pigpol(\tildeH)$ and
 suppose that $\Kpol=A[m]$. Let $d$ be a positive integer, $n=md$, and let $G \subset A[n]$ be a
 subgroup isomorphic to $Z(n)$, isotropic for the Weil pairing $e_{A,n}$, such that $H=dG$. Let
 $(e_i)_{i \in I}$ be a set of generators of
 $G$ such that $d/\expo(e_i, \tildeH)$ is odd for all $i \in I$ (for instance a basis of
 $G$). Then every $x \in G$ is symmetric compatible with $\tildeH$ if and only if $e_i$ is
 symmetric compatible with $\tildeH$ for all $i \in I$.
\end{corollary}
\begin{proof}
 If $d$ is odd, then $\expo(x, \tildeH)$ is odd for all $x \in G$ and the result follows from
 Proposition~\ref{sec2:proptors} (1). We suppose that $d$ is even. For all $x \in G$,
 $dx \in H$, so that $\varphi_d(x, \tildeH)$ is defined. Let $f_0=[n]_A$, which is separable, and
 $\pol'=f_0^*(\pol)$; as $\pol$ is symmetric, $\pol' \simeq \pol^{n^2}$, so that $K(\pol')=A[mn^2]$
 contains every $y \in A(\overk)$ such that $f_0(y) \in G$. By Lemma~\ref{lem:phikappa}, we can compute
 $\varphi_d(x, \tildeH)$ with $f_0$ and any $y$ such that $f_0(y)=x$. If $y_1, y_2 \in f_0^{-1}(G)$, we
 have, by Proposition~\ref{def:weil}, $e_{\pol'}(y_1, y_2)^d=e_{\pol'^d}(y_1,
 y_2)=e_{\pol^d}(f_0(y_1), f_0(y_2))=e_{A,n}(f_0(y_1), f_0(y_2))=1$ since $K(\pol^d)=A[n]$ and $G$ is
 isotropic. Proposition~\ref{sec2:propgroup} then shows that $x \mapsto \varphi_d(x, \tildeH)$ is a
 group morphism $G \rightarrow \{-1, 1\}$.

 Suppose that $e_i$ is symmetric compatible with $\tildeH$ for all $i \in I$. As
 $d/\expo(e_i, \tildeH)$ is odd, Lemma~\ref{lem:phikappa} gives
 $\varphi_d(e_i, \tildeH)=\kappa_0(e_i, \tildeH)=1$, so that $\varphi_d(x, \tildeH)=1$ for all
 $x \in G$. Let $x \in G$ and let $c=d/\expo(x, \tildeH)$. As $(d/c) x \in H=dG$, there exists
 $z \in G$ such that $(d/c)(x-cz)=0$, so that $x-cz \in G[d/c]$. As $G \simeq Z(n)$, we have
 $G[d/c]=(nc/d) G=(mc) G \subset cG$, hence there exists $x' \in G$ such that
 $x=cx'$. Let $y' \in A(\overk)$ be such that $f_0(y')=x'$ and $g_{y'}$ be a symmetric element above
 $y'$. Then $y=cy'$ is above $x$ and $g_{y'}^c$ is a symmetric element above $y$ (as a
 power of a symmetric element). If $\expo(x, \tildeH)$ is odd, $x$ is symmetric compatible
 with $\tildeH$ by Proposition~\ref{sec2:proptors} (1). Otherwise, $\kappa_0(x, \tildeH)$ is
 determined by $(g_{y'}^c)^{\expo(x, \tildeH)}=g_{y'}^d$ (see the proof of Lemma
~\ref{lem:phikappa}), and we get $\kappa_0(x, \tildeH)=\varphi_d(x', \tildeH)=1$. The converse
 is obvious.
\end{proof}

The following proposition shows that the property of $x$ being symmetric compatible with
$\tildeH$ is non-trivial only when $\expo(x, \tildeH)$ is even. In this case, it explains how
to change $\tildeH$ so as to make $x$ symmetric compatible with $\tildeH$.
\begin{proposition}\label{sec2:proptors}
 Let $x \in A(\overk)$ be a torsion point and $\tildeH$ a symmetric level subgroup of
 $\Gpol$ over $H= \pigpol(\tildeH)$. We claim that:
 \begin{enumerate}
	\item
 If $\lambda = \expo(x, \tildeH)$ is odd, then $x$ is symmetric compatible with $\tildeH$.
\item 
 Suppose that $\tildeH=(\tildee_1, \ldots, \tildee_\kappa)$, where $(\pi_{\Gpol}(\tildee_i))_{i=1,
 \ldots, \kappa}$ is a basis of $H$ (i.e. $H=\oplus_i (\pi_{\Gpol}(\tildee_i))$), with
 $\pi_{\Gpol}(\tildee_1)=\lambda x$. Write $\tildee_1=(\tau_{e_1}, \psi_{e_1}) \in \Gpol$. If $x$ is not
 symmetric compatible with $\tildeH$ then $x$ is symmetric compatible with $(\tildee'_1,
 \ldots, \tildee_\kappa)$ where $\tildee'_1=(\tau_{e_1}, -\psi_{e_1})$.
 \end{enumerate}
\end{proposition}
\begin{proof}
 Using Proposition~\ref{prop:compat1}, we can suppose, if necessary by taking an isogeny, that $x
 \in \Kpol$. We have seen that there are exactly two symmetric elements of $\Gpol$
 above $x$: if $z=(\tau_x, \psi_x) \in \Gpol$ is one of them, then
 $(\tau_x, -\psi_x)$ is the other one. Let $\lambda = \expo(x, \tildeH)$. As
 $\delta_{-1}$ and the inverse map are respectively
 a morphism and an anti-morphism of $\Gpol$, it is clear that the map $\lambda^*: \pi_{\Gpol}^{-1}(x)
 \rightarrow \pi_{\Gpol}^{-1}(\lambda x)$, $z \mapsto \lambda z$, maps a symmetric
 element to a symmetric element and is onto if $\lambda$ is odd. This implies
 that $x$ is symmetric compatible with $\tildeH$, which is the first claim.

 For the second claim, we can still suppose that $x \in \Kpol$. If $\lambda$ is odd, there is
 nothing to prove, so we suppose that $\lambda$ is even. Let $g_x$ be a symmetric lift of $x$.
 Then $g_x^\lambda$ is a symmetric element above $\lambda x=e_1$. So if $g_x^\lambda \ne
 \tildee_1=(\tau_{e_1}, \psi_{e_1})$, then $g_x^\lambda=(\tau_{e_1}, -\psi_{e_1})$ because of Lemma
~\ref{intro:lemma1}, and $x$ is symmetric compatible with $(\tildee'_1, \ldots,
 \tildee_\kappa)$.
\end{proof}

The following theorem characterizes the subgroups of $B[n]$ that can be realized as
$G(0_{A_y})$ for $y \in \vjxz(\overk)$ and $f \in F(y)$. It should be compared with \cite[Theorem~12]{DRmodular}, which shows
that, when $d$ is coprime to $m$, every isotropic subgroup arises in this way.
In contrast, when $\gcd(d,m)=d_0 > 1$, the symmetric compatibility condition appears.
\begin{theorem}\label{prop:cond}
 Let $(B, \bpol, \Thetabpol)$ be a marked abelian variety of type $K(m)$ with theta
 null point $\xtheta$. Let $G$ be a subgroup of $B[n]$ isomorphic to $Z(n)$
 containing
 $\Thetabpolbar(Z(m)\times\{0\})$.
 Then
 $G=G(0_{A_y})$ (see Definition~\ref{sec1:defG}) for some $y \in \vjxz(\overk)$ and $f \in F(y)$ if and
 only if $G$ is isotropic for $e_{B, n}$ (the Weil pairing, see Proposition~\ref{def:weil}) and for all $x \in G(\overk)$, $x$ is
 symmetric compatible with
 $\Thetabpol( \{1\}\times Z(m) \times\{0\} )$.
\end{theorem}
\begin{proof}
 If $G=G(0_{A_y})$ for $y \in \vjxz(\overk)$ and $f \in F(y)$, write $(A, \pol,
 \Thetapol)=(A_y, \pol_y, \Theta_{\pol_y})$, so that $(A, \pol, \Thetapol)$ and $(B, \bpol,
 \Thetabpol)$ are isog-$f$-compatible.

 Let $x_1, x_2 \in G(0_A)$. We want to show that $e_{B, n}(x_1, x_2)=1$. 
 Let $x'_1, x'_2 \in A(\overk)$ be such that $f(x'_i)=x_i$ for $i=1,2$.
 By definition, $G(0_A)=f(\Thetapolbar(Z(n) \times \{0\}))$ and by
 definition of isog-$f$-compatibility (see Definition~\ref{sec1:def1}),
 $Ker(f)=\Thetapolbar(\{0\} \times \dnu_{d,n}(\dZ(d)))$ so that
 $x'_i \in \Thetapolbar(Z(n) \times \dnu_{d,n}(\dZ(d)))$. We deduce that $e_{\pol}(x'_1,
 x'_2)^d=1$.
 Using Proposition~\ref{def:weil}, we have: $e_{B, n}(x_1, x_2)=e_{\bpol^d}(x_1,
 x_2)=e_{\pol^d}(x'_1, x'_2)=e_{\pol}(x'_1, x'_2)^d=1$.

 Let $x \in G(\overk)$. We show that $x$ is symmetric
 compatible with $\tildeH=\Thetabpol( \{1\}\times Z(m) \times\{0\} )$. 
 Let $K= \Ker f=\Thetapolbar(\{0\}\times \dnu_{d,n}(\dZ(d)))$ and denote by $\tildeK$
 the level subgroup which is the descent data of $\pol = f^*(\bpol)$ to
 $\bpol$. 
 Set $H' = f^{-1}(\Thetabpolbar(Z(m)\times\{0\}))$ and $x' \in A(\overk)$ such that
 $f(x')=x$. We can choose $x'$ such that $x'=\Thetapolbar((\lambda, 0))$ for some $\lambda \in \Zn$ by
 definition of $G$. Let $g_x'=
 \Thetapol((1, \lambda, 0))$. Then $(g_x')$ is a level subgroup of
 $\Gpol$. Let $\tildeH'= \Thetapol(\{1\}\times \mu_{m,n}(Z(m)) \times\{0\}) \cdot
 \tildeK$, then we have $H'=\pi_{\Gpol}(\tildeH')$. As $\Thetapolbar(\mu_{m,n}(Z(m))\times\{0\}) \subset
 f^{-1}(K(\bpol))$, by Proposition~\ref{prop:compat}, $\Thetapolbar(\mu_{m,n}(Z(m))\times\{0\}) \subset
 K^{\perp_{e_{\pol}}}$. Thus, $\tildeH'$ is a level subgroup.
 We have $\fsharp(\tildeH')=\tildeH$ (see (3) of Definition
~\ref{sec1:def1}), so $\tildeH'$ satisfies all the conditions (2) of Definition
~\ref{def:symcompat}. Moreover, $\pigpol((g_x') \cap
 \tildeH')=(x') \cap H'$ which means that $x$ is symmetric compatible with $\tildeH$.

 Conversely, let $G$ be a subgroup of $B[n]$ isotropic for $e_{B, n}$, isomorphic to
 $Z(n)$, containing
 $Z_B=\Thetabpolbar(Z(m) \times\{0\})$ and such that for all $x \in G(\overk)$, $x$ is symmetric
 compatible with $\tildeH=\Thetabpol( \{1\}\times Z(m) \times\{0\} )$. Let $K=G[d]=mG$, which is
 isomorphic to $Z(d)$. As $G \simeq \Zn$ contains $Z_B \simeq Z(m)$,
 we have $G[m]=dG=Z_B$, so that $K[d_0]=G[d_0]=
 \Thetabpolbar(\mu_{d_0,m}(Z(d_0)) \times\{0\})$. Moreover, $K$ is isotropic for $e_{B,d}$:
 for $x,y \in G$, by (\ref{eq:weilcompat}), $e_{B,d}(mx,my)=e_{B,n}(x,my)=e_{B,n}(x,y)^m=1$.
 Let $A=B/K$, $\hatf: B \rightarrow A$ be the
 quotient isogeny by $K$ and $f: A \rightarrow B$ be the contragredient isogeny, so that $f
 \circ \hatf=[d]_B$ and $\hatf \circ f=[d]_A$, and denote by $K_{f}$ the kernel of $f$; it
 has $d^g$ elements. Let $\pol = f^*(\bpol)$. As $K$ is a maximal isotropic subgroup of $B[d]$
 for $e_{B,d}$, the principal polarization $\lambda$ of $B$ such that $\phi_\bpol=m\lambda$
 satisfies $d\lambda=\hatf^\vee \circ \lambda_A \circ \hatf$ for a principal polarization
 $\lambda_A$ of $A$, so that $\phi_\pol=f^\vee\circ \phi_\bpol \circ f=n\lambda_A$ (see
 \cite[Section~23]{MumfordAV}). In particular, $\Kpol=A[n]$ and $e_\pol=e_{A,n}$ (Remark
~\ref{rk:weilconv}).
 We define a theta structure $\Thetapol$ for $(A, \pol)$
 such that $(A, \pol, \Thetapol)$ and $(B, \bpol, \Thetabpol)$ are isog-$f$-compatible. Let
 $\tildeK_f$ be the symmetric level subgroup above $K_f$ which is the descent data of $\pol$
 to $\bpol$ (Remark~\ref{rm:sym}), so that $K_f$ is isotropic for $e_\pol$.

 \emph{Step 0: a basis of $G$ adapted to $Z_B$.} For $j=1, \ldots, g$, let $\epsilon_j$ be the
 $j$-th element of the canonical basis of $Z(m)$ and $z_j=\Thetabpolbar((\epsilon_j,0))$. We
 show that there exists a basis $(e_1, \ldots, e_g)$ of $G$ such that $de_j=z_j$ for $j=1,
 \ldots, g$. As $G[m]=dG=Z_B$, there exist $e^0_1, \ldots, e^0_g \in G$ such that $de^0_j=z_j$,
 and since the multiplication by $d$ induces an isomorphism $G/mG \simeq dG$, $(e^0_j \bmod
 mG)_j$ is a basis of $G/mG$. Write $n=st$ where $s$ is the largest divisor of $n$ coprime to
 $m$; then $s | d$, every prime divisor of $t$ divides $m$, and $G=G[s]\oplus G[t]$. Write
 $e^0_j=a_j+b_j$ with $a_j \in G[s]$ and $b_j \in G[t]$, let $(\kappa_1, \ldots, \kappa_g)$
 be a basis of $G[s] \simeq Z(s)$ and set $e_j=\kappa_j+b_j$. As $s | d$, $de_j=db_j=de^0_j=z_j$.
 For every prime divisor $p$ of $t$, $p | m$ so that $(e^0_j \bmod pG)_j$ is a basis of
 $G/pG=G[t]/pG[t]$, and so is $(b_j \bmod pG[t])_j$. Hence the matrix of $(b_1, \ldots, b_g)$
 in a basis of $G[t] \simeq (\Z/t\Z)^g$ has a determinant invertible modulo every prime
 divisor of $t$, so that $(b_j)$ is a basis of $G[t]$ and $(e_j)$ is a basis of $G$.

 \emph{Step 1: the subgroups $\dZ_A$ and $Z'_A$ and their level subgroups.}
 As $\dZ_B=\Thetabpolbar(\{0\}\times \dZ(m))$ is isotropic for $e_{\bpol}$, by
 Proposition~\ref{def:weil} (2),
 $\dZ_A=f^{-1}(\dZ_B)$
 is isotropic for $e_\pol$. It has $n^g$ elements and, for $a \in \dZ_A$, $da=\hatf(f(a)) \in
 \hatf(B[m]) \subset A[m]$, so that $\dZ_A \subset A[n]$. As in Proposition
~\ref{prop:compat1} (2), there is a unique
 symmetric level
 subgroup $\widetilde{\dZ}_A$ over $\dZ_A$ such that $\tildeK_{f} \subset \widetilde{\dZ}_A$
 and
 \begin{equation}\label{eq:prop8:1}
 f^{\sharp}
 (\widetilde{\dZ}_A)=\Thetabpol(\{1\}\times \{0\} \times \dZ(m)).
 \end{equation}
 In the same way, let $Z'_A=f^{-1}(Z_B)$, which contains $K_f$ and is isotropic for $e_\pol$.
 There is a unique symmetric level subgroup $\tildeZ'_A$ over $Z'_A$ such that
 $\tildeK_{f} \subset \tildeZ'_A$ and
 \begin{equation}\label{eq:prop8:2}
 \fsharp(\tildeZ'_A)=\tildeH=\Thetabpol(\{1\}\times Z(m) \times \{0\}).
 \end{equation}

 \emph{Step 2: lifts of the basis of $G$ and a decomposition of $A[n]$.} For $j=1, \ldots,
 g$, let $x_j \in A(\overk)$ be such that $f(x_j)=e_j$. Then $dx_j=\hatf(e_j)$ and
 $mdx_j=\hatf(me_j)=0$ since $me_j \in K$, so that $x_j \in A[n]$. Let $Z_A$ be the subgroup
 generated by $x_1, \ldots, x_g$. As $f(Z_A)=G$ has $n^g$ elements and $Z_A$ is generated by
 $g$ elements of $A[n]$, $f$ induces an isomorphism $Z_A \simeq G$; in particular $Z_A \simeq
 \Zn$ with basis $(x_j)$ and $Z_A \cap K_f=\{0\}$. If $x \in Z_A \cap \dZ_A$, then $f(x) \in G
 \cap \dZ_B \subset G[m] \cap \dZ_B=Z_B \cap \dZ_B=\{0\}$, so that $x \in Z_A \cap
 K_f=\{0\}$. Counting elements, we get $A[n]=Z_A \oplus \dZ_A$. As $\dZ_A$ is isotropic and has
 $n^g$ elements, its orthogonal for the perfect pairing $e_\pol$ on $A[n]$ is $\dZ_A$ itself, so
 that the map $Z_A \rightarrow \mathrm{Hom}(\dZ_A, \overk^*)$, $x \mapsto e_\pol(x,\cdot)$, is
 injective, hence bijective by counting: the pairing $e_\pol: Z_A \times \dZ_A \rightarrow
 \overk^*$ is perfect. In particular $\dZ_A \simeq \mathrm{Hom}(Z_A, \overk^*) \simeq \dZn$ and,
 since $K_f \subset \dZ_A$ is killed by $d$ and has $d^g$ elements, $K_f=\dZ_A[d]$. Note
 that these properties only depend on the fact that $f(x_j)=e_j$ for $j=1, \ldots, g$.

 \emph{Step 3: choice of the lifts so that $Z_A$ is isotropic.}
 For $i,j=1, \ldots, g$, let $c_{ij}=e_\pol(x_i,x_j)$. By Proposition~\ref{def:weil} (3), (2)
 and (1), $c_{ij}^d=e_\pol(x_i,dx_j)=e_{\pol^d}(x_i,x_j)=e_{\bpol^d}(e_i,e_j)=e_{B,n}(e_i,e_j)=1$
 since $G$ is isotropic for $e_{B,n}$. Thus $c_{ij} \in \mu_d$, $c_{ii}=1$ and
 $c_{ji}=c_{ij}^{-1}$. By Step 2, for $j=1, \ldots, g$, there exists $\kappa_j \in \dZ_A$ such
 that $e_\pol(x_i,\kappa_j)=c_{ij}^{-1}$ if $i<j$ and $e_\pol(x_i,\kappa_j)=1$ if $i \geq j$; as
 these values are in $\mu_d$, $e_\pol(x, d\kappa_j)=1$ for all $x \in Z_A$, so that
 $d\kappa_j=0$ and $\kappa_j \in \dZ_A[d]=K_f$. As $K_f$ is isotropic, for $i<j$ we get
 \begin{equation}
	e_\pol(x_i+\kappa_i,x_j+\kappa_j)=c_{ij}\,e_\pol(x_i,\kappa_j)\,e_\pol(x_j,\kappa_i)^{-1}=1,
 \end{equation}
 and, $e_\pol$ being alternating, $e_\pol(x_i+\kappa_i, x_j+\kappa_j)=1$ for all $i,j$. From now
 on, we replace $x_j$ by $x_j+\kappa_j$, which does not change $f(x_j)=e_j$, so that the
 properties of Step 2 still hold and moreover $Z_A$ is isotropic for $e_\pol$.

 \emph{Step 4: the symplectic structure.} Let $(\epsilon'_1, \ldots, \epsilon'_g)$ be the
 canonical basis of $\Zn$. By Step 2, there is a unique symplectic isomorphism
 $\Thetapolbar: K(n) \rightarrow A[n]=\Kpol$ such that $\Thetapolbar((\epsilon'_j,0))=x_j$ for $j=1,
 \ldots, g$ and $\Thetapolbar(\{0\}\times \dZn)=\dZ_A$: for $y \in \dZn$, $\Thetapolbar((0,y))$ is
 the unique $a \in \dZ_A$ such that $e_\pol(\Thetapolbar((i,0)), a)=e_n((i,0),(0,y))$ for all $i \in
 \Zn$, and $\Thetapolbar$ is symplectic since $Z_A$ and $\dZ_A$ are isotropic. We have:
 \begin{equation}\label{eq:prop8:3}
	\begin{split}
	 \Thetapolbar(\{0\}\times \dnu_{d,n}(\dZ(d)))=\dZ_A[d]=K_{f}, \\
	 \Thetapolbar(\mu_{m,n}(Z(m)) \times\{0\})=Z_A[m]=dZ_A=Z_A \cap Z'_A,\\
	 f(\Thetapolbar((\mu_{m,n}(x),0)))=\Thetabpolbar((x,0)) \text{ for all } x \in Z(m).
	\end{split}
 \end{equation}
 Indeed, $\dnu_{d,n}(\dZ(d))=\dZn[d]$ and $\mu_{m,n}(Z(m))=\Zn[m]$; for $x=\sum_j x_j' \epsilon_j \in
 Z(m)$, $\mu_{m,n}(x)$ is the element $\sum_j d x'_j\epsilon'_j$ (with $x'_j$ seen as integers), so that
 $f(\Thetapolbar((\mu_{m,n}(x),0)))=\sum_j x'_j de_j=\sum_j x'_j z_j=\Thetabpolbar((x,0))$; finally, $Z_A
 \cap Z'_A=f_{|Z_A}^{-1}(G \cap Z_B)=f_{|Z_A}^{-1}(dG)=dZ_A$. Moreover, for all $y \in \dZn$:
 \begin{equation}\label{eq:prop8:4}
	f(\Thetapolbar((0,y)))=\Thetabpolbar((0,\drho_{n,m}(y))).
 \end{equation}
 Indeed, let $a=\Thetapolbar((0,y)) \in \dZ_A$, so that $f(a) \in \dZ_B$. For $x \in Z(m)$, let
 $u=\Thetapolbar((\mu_{m,n}(x),0))$; as $u, a \in f^{-1}(B[m])$, by Proposition~\ref{def:weil} (2),
 (\ref{eq:prop8:3}) and (\ref{eq:commutator}),
 $e_\bpol(\Thetabpolbar((x,0)), f(a))=e_\pol(u,a)=e_n((\mu_{m,n}(x),0),(0,y))=y(\mu_{m,n}(x))$.
 As $\drho_{n,m}(y)$ is the restriction of the character $y$ to $\mu_{m,n}(Z(m))$, this is also
 $e_m((x,0),(0,\drho_{n,m}(y)))=e_\bpol(\Thetabpolbar((x,0)), \Thetabpolbar((0,\drho_{n,m}(y))))$.
 As the pairing $e_\bpol: Z_B \times \dZ_B \rightarrow \overk^*$ is perfect, we get
 (\ref{eq:prop8:4}).

 \emph{Step 5: a symmetric level subgroup over $Z_A$.} Let $j \in \{1, \ldots, g\}$. As $e_j$ is
 symmetric compatible with $\tildeH$ and $x_j \in \Kpol=K(f^*(\bpol))$ satisfies $f(x_j)=e_j$, by
 Proposition~\ref{prop:compat1} (applied to the isogeny $f$, the point $x_j$ and the level
 subgroup $\tildeZ'_A$ of Step 1), $x_j$ is symmetric compatible with $\tildeZ'_A$: there
 exists a symmetric $g_j \in \Gpol$ above $x_j$ such that $(g_j)$ is a level subgroup and
 $\pigpol((g_j) \cap \tildeZ'_A)=(x_j) \cap Z'_A$. As $dx_j \in Z'_A$ (since $f(dx_j)=z_j$) and
 $g_j^d$ is the only element of $(g_j)$ above $dx_j$, we get $g_j^d \in \tildeZ'_A$; moreover
 $g_j^n=1$ since $(g_j)$ is a level subgroup and $nx_j=0$. For $i,j=1, \ldots, g$, $g_ig_jg_i^{-1}g_j^{-1}
 =e_\pol(x_i,x_j)=1$ by Step 3, so that the subgroup $\tildeZ_A$ generated by $g_1, \ldots, g_g$ is
 abelian and its elements are the $\prod_j g_j^{a_j}$, which lie above $\sum_j a_jx_j$. If
 $\sum_j a_jx_j=0$, then $n | a_j$ for all $j$ and $\prod_j g_j^{a_j}=1$, so that $\tildeZ_A$
 is a level subgroup above $Z_A$. It is symmetric since a product of commuting symmetric
 elements is symmetric. Finally, an element $\prod_j g_j^{a_j}$ of $\tildeZ_A$ above a point of
 $Z_A \cap Z'_A=dZ_A$ satisfies $d | a_j$ for all $j$, so that it is equal to $\prod_j
 (g_j^d)^{a_j/d} \in \tildeZ'_A$:
 \begin{equation}\label{eq:prop8:5}
	\tildeZ_A \cap \pigpol^{-1}(Z_A \cap Z'_A) \subset \tildeZ'_A.
 \end{equation}

 \emph{Step 6: the theta structure.} By Remark~\ref{rk:partialtheta}, there is a unique theta
 structure $\Thetapol: G(n) \rightarrow \Gpol$ for $(A, \pol)$ inducing $\Thetapolbar$ and such that:
\begin{equation}
	\begin{split}
	 \Thetapol(\{1\}\times\{0\} \times \dZn)=\widetilde{\dZ}_A;\\
	 \Thetapol(\{1\}\times \Zn \times\{0\})=\tildeZ_A,
	\end{split}
 \end{equation}
 and it is symmetric since $\widetilde{\dZ}_A$ and $\tildeZ_A$ are symmetric level subgroups. We
 check the conditions of Definition~\ref{sec1:def1}. Condition (1) is the first line of
 (\ref{eq:prop8:3}). As $\tildeK_f \subset \widetilde{\dZ}_A$, $\Thetapol(\{1\}\times\{0\}\times
 \dnu_{d,n}(\dZ(d)))$ is the subgroup of $\widetilde{\dZ}_A$ above $K_f$, that is $\tildeK_f$,
 which gives condition (2). Let $x \in Z(m)$. By (\ref{eq:prop8:3}) and (\ref{eq:prop8:5}),
 $\Thetapol((1,\mu_{m,n}(x),0))$ is an element of $\tildeZ'_A$ above a point $u$ such that
 $f(u)=\Thetabpolbar((x,0))$; by (\ref{eq:prop8:2}) and since $\pi_{\Gbpol} \circ \fsharp=f \circ
 \pigpol$, $\fsharp(\Thetapol((1,\mu_{m,n}(x),0)))$ is the element of $\tildeH$ above
 $\Thetabpolbar((x,0))$, that is $\Thetabpol((1,x,0))$: this is condition (3). In the same way,
 condition (4) follows from (\ref{eq:prop8:1}) and (\ref{eq:prop8:4}). Hence $(A, \pol, \Thetapol)$
 and $(B, \bpol, \Thetabpol)$ are isog-$f$-compatible: if $y$ is the theta null point of $(A,
 \pol, \Thetapol)$, then $y \in \vjxz(\overk)$ by Proposition~\ref{prop:charac}, $f \in F(y)$, and
 $G(0_A)=f(\Thetapolbar(\Zn\times\{0\}))=f(Z_A)=G$.
\end{proof}
The elements of $\mathcal{K}(y)$ are the $G(0_{A_y})[d]$, where $G(0_{A_y})$ is computed with the
isogenies $f \in F(y)$ (see the description of $\mathcal{K}(y)$ before Proposition
\ref{prop:phs}). Hence Theorem~\ref{prop:cond} shows that the subgroups $K_1 \subset B[d]$ such
that the orbit $V^0_{K_1}$ of Proposition~\ref{prop:phs} is not empty are the $G[d]$, for $G$ as
in Theorem~\ref{prop:cond}.

Recall that $d_0=\gcd(d,m)$ and that the Weil pairings $e_{B,\ell}$ are defined for all
positive integers $\ell$ (see Remark~\ref{rk:weilconv}). Let $Z_B=\Thetabpolbar(Z(m)\times\{0\})$
and $\dZ_B=\Thetabpolbar(\{0\}\times \dZ(m))$.

\begin{lemma}\label{lem:constructG}
	Let $K \subset B[d]$ be a subgroup isomorphic to $Z(d)$ and isotropic for $e_{B,d}$.
	The following conditions are equivalent:
	\begin{enumerate}
		\item there exists a subgroup $G \subset B[n]$ isomorphic to $\Zn$, isotropic for $e_{B,n}$,
			such that $G[m]=Z_B$ and $G[d]=K$;
		\item $K[d_0]=Z_B[d_0]=\Thetabpolbar(\mu_{d_0,m}(Z(d_0))\times\{0\})$.
	\end{enumerate}
	Moreover, in this case, $G$ can be chosen together with a group isomorphism $g_f: \Zn
	\rightarrow G$ such that $g_f(\mu_{m,n}(i))=\Thetabpolbar((i,0))$ for all $i \in Z(m)$.
\end{lemma}
\begin{proof}
	If $G$ satisfies (1), then $K[d_0]=G[d][d_0]=G[d_0]=G[m][d_0]=Z_B[d_0]$, and since
	$\mu_{d_0,m}$ is the multiplication by $m/d_0$, we have $Z_B[d_0]=(m/d_0)Z_B=
	\Thetabpolbar(\mu_{d_0,m}(Z(d_0))\times\{0\})$.

	Conversely, suppose that (2) holds. For $i=1, \ldots, g$, let $\epsilon_i$ be the $i$-th
	element of the canonical basis of $Z(m)$ and $z_i=\Thetabpolbar((\epsilon_i,0))$, so that
	$(z_1, \ldots, z_g)$ is a basis of $Z_B \simeq Z(m)$.

	\emph{Step 1: a basis of $K$ adapted to $Z_B$.} We show that there exists a basis $(k_1,
	\ldots, k_g)$ of $K \simeq Z(d)$ such that for $i=1, \ldots, g$:
	\begin{equation}\label{eq:constructG1}
		(d/d_0)k_i=(m/d_0)z_i.
	\end{equation}
	As $K \simeq Z(d)$, the multiplication by $d/d_0$ induces an isomorphism $K/d_0K \simeq
	K[d_0]$. By (2), $(m/d_0)z_i \in K[d_0]$ for $i=1, \ldots, g$, so that there exist $k^0_1,
	\ldots, k^0_g \in K$ such that $(d/d_0)k^0_i=(m/d_0)z_i$. Moreover, as $((m/d_0)z_i)_i$ is a
	basis of $Z_B[d_0]=K[d_0]$, $(k^0_i \bmod d_0K)_i$ is a basis of $K/d_0K$.

	Write $d=st$ where $s$ is the largest divisor of $d$ coprime to $d_0$. Then $s$ and
	$t$ are coprime, every prime divisor of $t$ divides $d_0$, $d_0 | t$ and $K=K[s]\oplus
	K[t]$ with $K[s] \simeq Z(s)$ and $K[t] \simeq Z(t)$. Write $k^0_i=a_i+b_i$ with
	$a_i \in K[s]$ and $b_i \in K[t]$, let $(\kappa_1, \ldots, \kappa_g)$ be any basis of
	$K[s]$ and set $k_i=\kappa_i+b_i$. As $d/d_0=s(t/d_0)$ is a multiple of $s$, we have
	$(d/d_0)k_i=(d/d_0)b_i=(d/d_0)k^0_i$, so that (\ref{eq:constructG1}) holds. It remains to
	show that $(b_1, \ldots, b_g)$ is a basis of $K[t]$. Let $p$ be a prime divisor of $t$.
	Since $p | d_0$, $(k^0_i \bmod pK)_i$ is a basis of $K/pK$, and as $K/pK=K[t]/pK[t]$, $(b_i
	\bmod pK[t])_i$ is a basis of the $\Z/p\Z$-vector space $K[t]/pK[t]$. Hence the matrix of
	$(b_1, \ldots, b_g)$ in a basis of $K[t] \simeq (\Z/t\Z)^g$ has a determinant which is
	invertible modulo every prime divisor of $t$, so it is invertible in $\Z/t\Z$ and $(b_1,
	\ldots, b_g)$ is a basis of $K[t]$. Thus $(k_1, \ldots, k_g)$ is a basis of $K$.

	\emph{Step 2: construction of a basis of $G$.} We show that for $i=1, \ldots, g$, there
	exists $u_i \in B[n]$ such that
	\begin{equation}\label{eq:constructG2}
		m u_i = k_i, \quad d u_i = z_i.
	\end{equation}
	Fix an isomorphism $B[n] \simeq (\Z/n\Z)^{2g}$. It induces isomorphisms $B[d] = mB[n]
	\simeq (m\Z/n\Z)^{2g}$ and $B[m]=dB[n] \simeq (d\Z/n\Z)^{2g}$, so that we can work
	coordinatewise. Let $mk'$ (resp. $dz'$) with $k' \in \Z/d\Z$ (resp. $z' \in \Z/m\Z$) be a coordinate of $k_i$
	(resp. the same coordinate of $z_i$). By (\ref{eq:constructG1}), $(d/d_0)mk'=(m/d_0)dz' \mod n$, that
	is $(n/d_0)(k'-z')=0 \mod n$, so that $k'=z' \mod d_0$. By the Chinese remainder theorem (in its
	general form), there exists $u \in \Z$ such that $u=k' \mod d$ and $u=z' \mod m$, hence
	$mu=mk' \mod n$ and $du=dz' \mod n$. Doing this for each coordinate gives $u_i$.

	\emph{Step 3: isotropy.} For $i,j \in \{1, \ldots, g\}$, let $c_{ij}=e_{B,n}(u_i,u_j)$. By
	(\ref{eq:weilcompat}) and (\ref{eq:constructG2}), we have
	$c_{ij}^d=e_{B,n}(u_i, z_j)=e_{B,m}(du_i, z_j)=e_{B,m}(z_i,z_j)=1$ since $Z_B$ is isotropic for
	$e_{B,m}$, and $c_{ij}^m=e_{B,n}(u_i, k_j)=e_{B,d}(mu_i, k_j)=e_{B,d}(k_i,k_j)=1$ since $K$
	is isotropic for $e_{B,d}$. Thus $c_{ij} \in \mu_{d_0}$. Moreover $c_{ii}=1$ and
	$c_{ji}=c_{ij}^{-1}$ since $e_{B,n}$ is alternating.

	Since $B[m]=Z_B \oplus \dZ_B$, that $Z_B$ and $\dZ_B$ are isotropic and that $e_{B,m}$ is
	perfect, the map $\dZ_B \rightarrow \mathrm{Hom}(Z_B, \overk^*)$, $w \mapsto e_{B,m}(\cdot,
	w)$, is an isomorphism, and it maps $\dZ_B[d_0]$ onto $\mathrm{Hom}(Z_B, \mu_{d_0})$. Hence,
	for $j=1, \ldots, g$, there exists $w_j \in \dZ_B[d_0]$ such that for $i=1, \ldots, g$:
	\begin{equation}
		e_{B,m}(z_i, w_j)= c_{ij}^{-1} \text{ if } i < j, \quad e_{B,m}(z_i, w_j)= 1 \text{ if } i
		\geq j.
	\end{equation}
	Let $u'_i=u_i+w_i$. As $d_0$ divides $d$ and $m$, $(u'_i)$ still satisfies
	(\ref{eq:constructG2}). By (\ref{eq:weilcompat}), for all $i, j$, $e_{B,n}(u_i,
	w_j)=e_{B,m}(du_i,w_j)=e_{B,m}(z_i, w_j)$ and $e_{B,n}(w_i,w_j)=e_{B,m}(dw_i,w_j)=1$. So, for
	$i < j$:
	\begin{equation}
		e_{B,n}(u'_i,u'_j)=c_{ij}\, e_{B,m}(z_i,w_j)\, e_{B,m}(z_j,w_i)^{-1}=c_{ij}c_{ij}^{-1}
		\cdot 1=1,
	\end{equation}
	and, $e_{B,n}$ being alternating, $e_{B,n}(u'_i,u'_j)=1$ for all $i,j$. From now on, we
	replace $u_i$ by $u'_i$ and let $G$ be the subgroup of $B[n]$ generated by $u_1, \ldots,
	u_g$; it is isotropic for $e_{B,n}$.

	\emph{Step 4: structure of $G$.} Let $a_1, \ldots, a_g \in \Z$ be such that $\sum_i a_i
	u_i=0$. Multiplying by $d$, we get $\sum_i a_i z_i=0$ so that $m | a_i$ for all $i$. Write
	$a_i=m b_i$; then $\sum_i b_i k_i=0$ so that $d | b_i$ and $n | a_i$ for all $i$. Thus,
	$(u_1, \ldots, u_g)$ is a basis of $G$ and $G \simeq \Zn$. In particular, $G[m]=dG$ is
	generated by $du_i=z_i$, $i=1, \ldots, g$, so $G[m]=Z_B$, and $G[d]=mG$ is generated by
	$mu_i=k_i$, so $G[d]=K$.

	Finally, let $g_f: \Zn \rightarrow G$ be the group isomorphism such that $g_f(\epsilon'_i)=u_i$,
	where $(\epsilon'_1, \ldots, \epsilon'_g)$ is the canonical basis of $\Zn$. For $i=(i_1,
	\ldots, i_g) \in Z(m)$, we have $\mu_{m,n}(i)= \sum_j d\,i_j \epsilon'_j$, thus
	$g_f(\mu_{m,n}(i))=\sum_j i_j du_j=\sum_j i_j z_j=\Thetabpolbar((i,0))$.
\end{proof}

\begin{remark}~
	\label{rk:constructG}
	\begin{enumerate}
		\item If $d | m$, then $d_0=d$ and condition (2) of Lemma~\ref{lem:constructG} reads
			$K=Z_B[d]$; in particular $K \subset Z_B$. If $\gcd(d,m)=1$, condition (2) is empty.
		\item Condition (2) of Lemma~\ref{lem:constructG} can always be obtained by changing the
			theta structure $\Thetabpol$. Indeed, let $K \subset B[d]$ be isomorphic to $Z(d)$
			and isotropic for $e_{B,d}$. Then $K[d_0]=(d/d_0)K$ is isomorphic to $Z(d_0)$ and, by
			(\ref{eq:weilcompat}), for all $x,y \in K$,
			$e_{B,d_0}((d/d_0)x,(d/d_0)y)=e_{B,d}(x,y)^{d/d_0}=1$, so that $K[d_0]$ is isotropic
			for $e_{B,d_0}$. Let $\zeta=e_{B,m}(\Thetabpolbar((\epsilon_1,0)),
			\Thetabpolbar((0,\hat{\epsilon}_1)))^{m/d_0}$, where $\epsilon_1$ (resp.
			$\hat{\epsilon}_1$) is the first element of the canonical basis of $Z(m)$ (resp.
			$\dZ(m)$); it is a primitive $d_0$-th root of unity. We say that a basis $(b_1, \ldots,
			b_{2g})$ of $B[d_0]$ is symplectic if $e_{B,d_0}(b_i,b_j)=1$, $e_{B,d_0}(b_{g+i},
			b_{g+j})=1$ and $e_{B,d_0}(b_i,b_{g+j})=\zeta^{\delta_{ij}}$ for $i,j=1, \ldots, g$.

			We first show that there exists a symplectic basis of $B[d_0]$ whose first $g$
			elements form a basis $(\kappa_1, \ldots, \kappa_g)$ of $K[d_0]$. As the torsion subgroup of
			$\overk^*$ is divisible, every morphism $K[d_0] \rightarrow \mu_{d_0}$ extends to a
			morphism $B[d_0] \rightarrow \mu_{d_0}$, which is of the form $e_{B,d_0}(\cdot, y)$ for
			some $y \in B[d_0]$ since $e_{B,d_0}$ is perfect. So there exist $\lambda_1,
			\ldots, \lambda_g \in B[d_0]$ such that $e_{B,d_0}(\kappa_i, \lambda_j)=\zeta^{\delta_{ij}}$.
			Let $b_{ij} \in \Z/d_0\Z$ be such that $e_{B,d_0}(\lambda_i, \lambda_j)=\zeta^{b_{ij}}$ and set
			$\lambda'_j=\lambda_j+\sum_{i<j} b_{ij}\kappa_i$. Since $K[d_0]$ is isotropic, we have
			$e_{B,d_0}(\kappa_i, \lambda'_j)=\zeta^{\delta_{ij}}$ for all $i,j$ and, for $i<j$,
			$e_{B,d_0}(\lambda'_i, \lambda'_j)=\zeta^{b_{ij}}e_{B,d_0}(\lambda_i, \kappa_i)^{b_{ij}}=1$.
			If $\sum_i \alpha_i\kappa_i+\sum_i \beta_i\lambda'_i=0$, pairing with $\kappa_j$ and
			$\lambda'_j$ shows that $\alpha_j=\beta_j=0 \bmod d_0$, so that $(\kappa_1, \ldots,
			\kappa_g, \lambda'_1, \ldots, \lambda'_g)$ generates a subgroup of $B[d_0]$ of cardinality
			$d_0^{2g}$: it is a symplectic basis of $B[d_0]$.

			By (\ref{eq:weilcompat}), the multiplication by $m/d_0$ maps the symplectic basis
			$\Thetabpolbar(\mB)$ of $B[m]$, where $\mB$ is the canonical basis of $K(m)$, to a
			symplectic basis of $B[d_0]$, and the two symplectic bases of $B[d_0]$ differ by an
			element $M_0 \in \Sp_{2g}(\Z/d_0\Z)$. As the reduction map $\Sp_{2g}(\Z) \rightarrow
			\Sp_{2g}(\Z/d_0\Z)$ is surjective \cite[Theorem~1]{NewmanSmart}, so is the reduction
			map $\Sp_{2g}(\Z/m\Z) \rightarrow \Sp_{2g}(\Z/d_0\Z)$, and $M_0$ lifts to $M \in
			\Sp_{2g}(\Z/m\Z)$. The theta structure $\Thetabpol'$ obtained by acting with $M$ on
			$\Thetabpol$ satisfies $\Thetabpolbar'(\mu_{d_0,m}(Z(d_0))\times\{0\})=K[d_0]$.
			The theta null point of $(B, \bpol, \Thetabpol')$ can be computed with Algorithm
			\ref{algo:trans}.
	\end{enumerate}
\end{remark}

\section{Isogeny computation with change of level}\label{sec:isogchange}
Let $n,m,d > 1$ be integers such that $n=dm$.
Let $(B, \bpol, \Thetabpol)$ be a marked abelian variety of type $K(m)$ given by its theta null
point $\tildenullbpol=(\theta_i^\Thetabpol(0_\Thetabpol))$ and $K \subset B[d]$ be a subgroup
isomorphic to $Z(d)$, isotropic for $e_{B,d}$ and such that
$K[d_0]=\Thetabpolbar(\mu_{d_0,m}(Z(d_0))\times\{0\})$. Let $A=B/K$, $\hatf:B \rightarrow A$ be the quotient isogeny and
$f: A \rightarrow B$ be its contragredient isogeny.
The condition on $K[d_0]$ is necessary for the existence of a theta structure on $A$ as below,
by Proposition~\ref{prop:charac} and Lemma~\ref{lem:constructG}, and by Remark
\ref{rk:constructG} it can always be obtained by changing the theta structure $\Thetabpol$.
By Lemma~\ref{lem:constructG}, there exists a subgroup $G \subset B[n]$ isomorphic to $\Zn$,
isotropic for $e_{B,n}$, such that $G[m]=\Thetabpolbar(Z(m)\times\{0\})$ and $G[d]=K$.
We would like to be able to compute the theta null point of $(A, \pol, \Thetapol)$ of type $K(n)$ such that $(B, \bpol, \Thetabpol)$ and $(A, \pol,
\Thetapol)$ are isog-$f$-compatible from the
knowledge of $\tildenullbpol$, $K$ and $B[n]$.

First, by Remark~\ref{rk:affinefcompat}, we may assume that we are given a rigidified abelian variety
$(B, \bpol, \Thetabpol, \theta_0^\Thetabpol, \rho^\bpol_{0_{\Thetabpol}})$ of type $K(m)$ up to equivalence (see
Remark~\ref{sec:rmequiv}) by its
affine theta null point $\tildenullbpol$ and that we want to compute a rigidified abelian
variety $(A, \pol, \Thetapol, \theta_0^\Thetapol, \rho^\pol_{0_A})$ which is
isog-$f$-compatible. If $\tildenullpol$ is a theta null point of $A$, by Corollary~\ref{sec1:cor1}, we
have $\tildef(\tildenullpol)=\tildenullbpol$.

In this section, if $x=(x_1, \ldots, x_n) \in \Aff^n(\overk)$ and $\lambda \in \overk^*$,
we denote by $\lambda * x \in \Aff^n(\overk)$ the point with coordinates $(\lambda x_1,
\ldots, \lambda x_n)$. For any $\kappa$ positive integer, we recall that
$\pi_{\proj^{Z(\kappa)}}: \Aff^{Z(\kappa)}-\{0\} \rightarrow \proj^{Z(\kappa)}$ is the canonical
projection. The strategy we follow is analogous to that of \cite{DRmodular}.
In order to recover $(A, \pol, \Thetapol, \theta_0^\Thetapol,
\rho^\pol_{0_A})$ we have to fix $\Thetapol$. As $\Thetapol(\{1\}\times\{0\}\times\dZn)$ is essentially defined by the descent data of $\pol = f^*(\bpol)$ to $\bpol$, it remains to set
$\Thetapol(\{1\}\times\Zn\times\{0\})$. For $x \in \Zn$, let $\Thetapol((1, x, 0))= \lambda_x
g_x$ for $g_x \in \Gpol$, we have to determine $\lambda_x$. But as $\tildef(\lambda_x g_x
\tilde0_A)$ is an affine lift of $f(x) \in B[n]$, we see that computing
$\Thetapol$ boils down to computing affine lifts of points of $B[n]$.

Precisely, the algorithm that we describe comes from two crucial remarks. 
The first one is contained in the following proposition:
\begin{proposition}\label{sec2:lemmalambdai}
 Let $G \subset B[n]$ be a subgroup of $B[n]$ isomorphic to $Z(n)$
 containing $\Thetabpolbar(Z(m)\times\{0\})$, isotropic for $e_{B,n}$ and such that for all
$x \in G(\overk)$, $x$ is symmetric compatible with $\Thetabpol( \{1\}\times Z(m) \times\{0\}
)$. We choose a numbering of the elements of $G$ by writing $G=\{g_f(i), i \in \Zn\}$
such that the map $i \mapsto g_f(i)$ is a group morphism and for all $i \in Z(m)$,
$g_f(\mu_{m,n}(i))=\Thetabpolbar((i,0))$. For all $i \in \Zn$, denote by
$\tildeg_f(i) \in \Aff^{Z(m)}$ an affine lift of $g_f(i)$. We suppose that for all $i \in
Z(m)$, $\tildeg_f(\mu_{m,n}(i))=\Thetabpol((1, i,0)).\tilde0_{\Thetabpol}$ (see Lemma~\ref{sec1:lemact}).

 There exists a rigidified abelian variety $(A, \pol, \Thetapol,
 \theta_0^\Thetapol, \rho^\pol_{0_A})$, isog-$f$-compatible with $(B, \bpol, \Thetabpol)$ for an
 isogeny $f: A \rightarrow B$, with affine theta null point
 $\tildenullpol =
 (\theta_i^\Thetapol(0_A))_{i \in \Zn}$ such that for all $i \in
 \Zn$, there exists $\lambda_i \in \overk^*$ such that 
 \begin{equation}\label{sec2:lemmaeq21}
	\tildeg_f(i)= \lambda_i * (\theta^\Thetapol_{\mu_{m,n}(j)+i}(0_A))_{ j\in Z(m)}.
 \end{equation}

 Moreover, let $z \in A(\overk)$, recall that $G(z)=f(z)+ G$ and denote by $g_f^z$ the map $Z(n)
 \rightarrow G(z)$, $i \mapsto f(z)+g_f(i)$. We suppose that we have chosen a
 rigidification $\rho_z^\pol$ of $\pol$ at $z$.
 For all $i \in \Zn$, denote by
 $\tildeg_f^z(i)$ an affine lift of $g_f^z(i)$, then for all $i \in
 \Zn$, there exists $\lambda_i^z \in \overk^*$ such that:
 \begin{equation}\label{sec2:lemmaeq212}
	\tildeg^z_f(i)= \lambda^z_i * (\theta^\Thetapol_{\mu_{m,n}(j)+i}(z))_{ j\in Z(m)}.
 \end{equation}
\end{proposition}
\begin{proof}
 The existence of 
$(A, \pol, \Thetapol,
 \theta_0^\Thetapol, \rho^\pol_{0_A})$ isog-$f$-compatible with 
$(B, \bpol, \Thetabpol, \theta_0^\Thetabpol, \rho^\bpol_{0_{\Thetabpol}})$
 comes from Theorem~\ref{prop:cond}.
 By Proposition~\ref{sec3:prop9}, there exists a group isomorphism $\sigma: \Zn
 \rightarrow \Zn$ such that the restriction of $\sigma$ to $\mu_{m,n}(Z(m))$ is the
 identity and for all $i \in \Zn$:
 \begin{equation}
	\tildeg_f(\sigma(i))=\lambda_i * (\theta^\Thetapol_{\mu_{m,n}(j)+i}(0_A))_{ j\in Z(m)}.
 \end{equation}
 By acting on $(A, \pol, \Thetapol,
 \theta_0^\Thetapol, \rho^\pol_{0_A})$ by the subgroup $\subG_1$ of Proposition
~\ref{prop:action2}
 (2), we can suppose that $\sigma$ is the identity of $Z(n)$.

 As by hypothesis for all $i \in \Zn$, $\tildeg_f^z(i)$ is an affine lift of $f(z) + g_f(i)$, we can apply
 Proposition~\ref{sec3:prop9} to obtain (\ref{sec2:lemmaeq212}).
\end{proof}

This motivates the following definition:
\begin{definition}\label{def:excel}
Let $(B, \bpol, \Thetabpol, \theta_0^\Thetabpol, \rho^\bpol_{0_{\Thetabpol}})$ be a rigidified
 abelian variety of type $K(m)$ with affine theta null point $\tildenullbpol$.
 Let $G \subset B[n]$ be a subgroup of $B[n]$ isomorphic to $Z(n)$
 containing $\Thetabpolbar(Z(m)\times\{0\})$, isotropic for $e_{B,n}$ and such that for all
$x \in G(\overk)$, $x$ is symmetric compatible with $\Thetabpol( \{1\}\times Z(m) \times\{0\}
)$. We choose a numbering of the elements of $G$ by writing $G=\{g_f(i), i \in \Zn\}$
such that the map $i \mapsto g_f(i)$ is a group morphism and for all $i \in Z(m)$,
$g_f(\mu_{m,n}(i))=\Thetabpolbar((i, 0))$. 
 We say that $\tildeG= \{ \tildeg_f(i), i \in \Zn \}$ is an excellent lift of $G$ with
 respect to $\tildenullbpol$ if 
 there exists a rigidified abelian variety
$(A, \pol, \Thetapol,
 \theta_0^\Thetapol, \rho^\pol_{0_A})$ with affine theta null point
 $\tildenullpol$ 
 of type $K(n)$ isog-$f$-compatible with 
$(B, \bpol, \Thetabpol, \theta_0^\Thetabpol, \rho^\bpol_{0_{\Thetabpol}})$ such that
 for all $i \in \Zn$:
 \begin{equation}\label{sec2:lemmaeq23}
	\tildeg_f(i)= \tildef(\Thetapol((1,i,0)).\tildenullpol).
 \end{equation}
 where $\tildef$ is given by Definition~\ref{def:affinemap}.
\end{definition}

The second key observation underlying our approach allows one to interpret the modular Riemann equations for the theta
null point as relations among points of the variety $B$. 

\begin{proposition}\label{sec2:lemm6}
Let $(A, \pol, \Thetapol, \theta_0^\Thetapol, \rho^\pol_{0_A})$ be a rigidified abelian
variety satisfying the hypotheses of Proposition~\ref{sec2:lemmalambdai}.
 Suppose that we have chosen $\tildeg_f(i)$ for $i \in \Zn$ so that 
 $\tildeg_f(i)= (\theta^\Thetapol_{\mu_{m,n}(j)+i}(0_A))_{ j\in Z(m)}$
 ($\lambda_i=1$ in (\ref{sec2:lemmaeq21})). We let $\Aff^{Z(m)}=\Spec(k[x_i,i \in Z(m)])$
 so that for $i \in Z(m)$, $x_i$ is the $i$-th coordinate function.

 \begin{enumerate}
	\item 
	 Let $\vx= (y_1, \ldots, y_4; y_5, \ldots, y_8) \in \Zn^8$ and $\vi=(i_1, \ldots, i_4; i_5, \ldots,
i_8) \in Z(m)^8$ be elements in Riemann position, then we have a Riemann equation:
 \begin{equation}\label{eq:riemaneq2}
	\sum_{\eta \in Z(2)} \prod_{j=1}^4 x_{i_j +
	\eta}(\tildeg_f(y_j)) =\sum_{\eta \in Z(2)} \prod_{j=5}^8 x_{i_j +
	\eta}(\tildeg_f(y_j)).
 \end{equation}
\item
For all $i \in Z(m)$ and $j \in \Zn$, we have the following symmetry relation:
\begin{equation}\label{eq:symrel}
 x_{i}(\tildeg_f(j))=x_{-i}(\tildeg_f(-j)).
\end{equation}
\item
 For all $\kappa,i \in Z(m)$ and $j \in \Zn$, we have
\begin{equation}
 x_{i + \kappa}(\tildeg_f(j))=x_i (\tildeg_f(j+\mu_{m,n}(\kappa))).
\end{equation}
\end{enumerate}
\end{proposition}
\begin{proof}
 Under the hypotheses of the proposition, we have for all $i \in Z(m)$ and $j \in \Zn$: 
\begin{equation}
 x_{i}(\tildeg_f(j))= \theta^\Thetapol_{\mu_{m,n}(i)+j}(0_A),
\end{equation}
 so (1) is an immediate consequence of Theorem~\ref{sec:thriemann},
 (2) comes from the symmetry relations of Proposition~\ref{sec:propsym} and (3) is given by
 the action on $\tildeg_f(j)$ of $\Thetabpol((1, \kappa, 0))$, for $\kappa \in
 Z(m)$ (see equation (\ref{eq:actthetagroup})).
\end{proof}
Proposition~\ref{sec2:lemmalambdai} shows how to recover the theta null point of a rigidified abelian variety
$(A, \pol, \Thetapol,
 \theta_0^\Thetapol, \rho^\pol_{0_A})$
isog-$f$-compatible with $(B, \bpol, \Thetabpol, \theta_0^\Thetabpol, \rho^\bpol_{0_{\Thetabpol}})$
from the knowledge of well-chosen $\tildeg_f(i)$ for $i \in \Zn$, and
Proposition~\ref{sec2:lemm6} provides necessary conditions for $\tildeg_f(i)$ to be well-chosen.

The rigidified abelian variety $(B, \bpol, \Thetabpol, \theta_0^\Thetabpol, \rho^\bpol_{0_{\Thetabpol}})$, given up to equivalence by its affine theta null point $\tilde0_{\Thetabpol}$, comes equipped with important arithmetic
operations defined by the Riemann equations, described in \cite{DRoptimal,DRarithmetic} for
projective and affine points.

We recall those that we use (by abuse of notation, we write $\tildex \in B(\overk) \subset
\Aff^{Z(m)}(\overk)$ for an affine point above a point $x \in B(\overk)$):
\begin{itemize}
 \item Normal addition: $x+y=\normadd(x,y, 0_\Thetabpol)$ takes $x, y \in B(\overk) \subset \proj^{Z(m)}(\overk)$ and
	returns $x+y \in B(\overk) \subset \proj^{Z(m)}(\overk)$;
 \item Differential addition: $\tildexy=\diffadd(\tildex,\tildey,
	\tildexmy, \tildenullbpol)$ takes $\tildex, \tildey, \tildexmy,
	\tildenullbpol \in B(\overk) \subset \Aff^{Z(m)}(\overk)$ and returns
	$\tildexy \in B(\overk) \subset \Aff^{Z(m)}(\overk)$;
 \item Three-way addition: $\widetilde{x+y+z}=\threeway(\tildexy, \widetilde{y+z},
	\widetilde{x+z}, \tildex,
	\tildey, \tildez, \tildenullbpol)$ takes $\tildexy$,
 $\widetilde{y+z}, \widetilde{x+z},
	\tildex, \tildey, \tildez, \tildenullbpol \in B(\overk) \subset \Aff^{Z(m)}(\overk)$
 and returns $\widetilde{x+y+z}
	\in B(\overk) \subset \Aff^{Z(m)}(\overk)$.
\end{itemize}
We can chain differential additions in a Montgomery-ladder type algorithm in order to
compute scalar multiplication $\scalmult(\ell, \tildexy, \tildex, \tildey,
\tildenullbpol)$, which
takes as input a positive integer $\ell$, affine points $\tildexy$, $\tildex$,
$\tildey$, $\tildenullbpol \in B(\overk) \subset \Aff^{Z(m)}(\overk)$ and returns
$\widetilde{\ell x + y} \in B(\overk) \subset \Aff^{Z(m)}(\overk)$. 

There are two more operations that will be useful. The first one comes from the action of the
Heisenberg group on affine points (\ref{eq:trans}) and the second one is an immediate
consequence of the symmetry relations of Proposition~\ref{sec2:lemm6}.
\begin{definition}
We denote by:
 \begin{itemize}
	\item
 $\thetaact(\tildex,i)$ the operation that takes as input
 $\tildex \in B(\overk) \subset \Aff^{Z(m)}(\overk)$ and $i \in Z(m) \cup \dZ(m)$ and outputs $(1,i,0).\tildex$ or $(1,0,i).\tildex$ depending on
 whether $i \in Z(m)$ or $i \in \dZ(m)$. 
\item $\inv(\tildex)$ the operator that takes as input $\tildex=(\tildex_j)_{j \in Z(m)} \in B(\overk) \subset \Aff^{Z(m)}(\overk)$
 which is a lift of $x \in B$
 and returns $\tildey=(\tildey_j)_{j \in Z(m)} \in B(\overk) \subset \Aff^{Z(m)}(\overk)$ where
 $\tildey_j=\tildex_{-j}$, which is an affine lift of
 $-x$.
 \end{itemize}
\end{definition}
The fact that $\inv$ is well defined on $G$ is an immediate consequence of the inverse
formula given by Lemma~\ref{sec1:leminv}. Note however that $\inv$ acts on affine points,
whereas the inverse formula deals with projective points.

\begin{remark}[The case $m=2$]\label{rk:level2}
 The results of this paper hold for $m=2$, with the following adaptations. When $m=2$,
 $e_\Thetabpol$ is not an embedding of $B$ but of its Kummer variety (Definition
~\ref{def:thetanull}): $x$ and $-x$ have the same theta coordinates, and the points we compute
 with are points of $\mathrm{Km}(B)$.
 \begin{itemize}
	\item As $-j=j$ for $j \in Z(2)$, $\inv$ is the identity, and the symmetry relations become
	 linear conditions. For instance, for $d=2$ and $e$ in a basis of $\Zz(2)$, $\ell=2$ in
	 Definition~\ref{def:symcompat2} and $g_f(e)$ is symmetric compatible with $\nullbpol$ if
	 and only if $\tildeg_f(e)$ is an eigenvector of $\thetaact(\cdot, e')$ for the eigenvalue
	 $1$ (the other possible eigenvalue being $-1$), where $\mu_{2,4}(e')=2e$.
	\item The operations $\diffadd$, $\threeway$ and $\scalmult$ are still given by the Riemann
	 relations of Theorem~\ref{sec:thriemann}, which hold for $n=2$. They involve divisions by
	 the values $L(\Thetabpol, \chi, 0, 0, 0_B, 0_B)$, $\chi \in \dZ(2)$, and by the coordinates
	 of the difference (for instance of $\widetilde{x-y}$ for $\diffadd$), which we assume to
	 be nonzero. Points of $G_2$ may have vanishing coordinates: in dimension $1$, the halves of
	 $\Thetabpolbar((0,1))$ are $(1:0)$ and $(0:1)$. This is why the good lift of $G_2$ is
	 computed after the change of coordinates $\tildeH_{\Delta(m)}$ in Algorithm
	~\ref{algo:changelevel}.
	\item The normal addition is not available, and $g_f(e_i)+g_f(e_j)$ is not determined by
	 $g_f(e_i)$ and $g_f(e_j)$ on $\mathrm{Km}(B)$. In the algorithms, a subgroup $G$ is thus
	 given by the images in $\mathrm{Km}(B)$ of a chain basis $(g_f(e_i), g_f(e_i+e_j))$, which
	 determine its numbering up to $[-1]_B$, an automorphism of $(B, \bpol, \Thetabpol)$ when
	 $m=2$ (see the proof of Proposition~\ref{prop:phs}).
	\item For a point $x \neq 0_B$, the good lifts of $x+G$ in Algorithm~\ref{algo:goodliftg}
	 also require the points $x+g_f(e_i)$ of $\mathrm{Km}(B)$, and Theorem~\ref{th:main7} then
	 gives the coordinates of $x$ or of $-x$, according to these data. We do not develop this
	 case: the theta null points are obtained with $x=0_B$, for which no further data is
	 needed.
 \end{itemize}
\end{remark}

We recall \cite[Lemma~1]{DRoptimal}, which explains how the outputs of $\diffadd$ and $\scalmult$ change with the choice of input affine points, and we complete this result with
an analogous result for the other operations on affine points that we use:
\begin{lemma}\label{lem:scalar}
Let $x,y \in B(\overk)$ and let $\tildex$, $\tildey$, $\tildexmy$ be affine lifts of $x$,
$y$ and $x-y$. Let
$$\tilde{r} = \diffadd(\tildex, \tildey, \tildexmy, \tildenullbpol).$$
Let $\alpha, \beta, \gamma, \delta \in \overk^*$, we have:
\begin{equation}
 \diffadd(\alpha* \tildex, \beta*\tildey, \gamma*\tildexmy,
 \delta*\tildenullbpol)=\frac{\alpha^2 \beta^2}{\gamma \delta^2} \tilde{r}.
\end{equation}
Let $x, y \in B(\overk)$ and let
$\tildex, \tildey, \tildexy$ be affine lifts of $x$, $y$ and
 $x+y$. Let $$\widetilde{r} = \scalmult(\ell, \tildexy, \tildex, \tildey,
 \tildenullbpol).$$ Let $\alpha, \beta, \gamma, \delta \in
\overk^*,$ we have:
\begin{equation}\label{eq:factor1} \scalmult(\ell,\alpha*\tildexy,
 \beta*\tildex, \gamma*\tildey, \delta*\tildenullbpol)=(\alpha^{\ell}
\beta^{\ell(\ell-1)} / \gamma^{\ell-1}\delta^{\ell(\ell-1)})*\widetilde{r}, \end{equation}
\begin{equation}\label{eq:factor2} \scalmult(\ell,\alpha*\tildex,
\alpha*\tildex, \delta*\tildenullbpol,
 \delta*\tildenullbpol)=\frac{\alpha^{\ell^2}}{\delta^{\ell^2-1}} * \scalmult(\ell,\tildex,
\tildex, \tildenullbpol, \tildenullbpol).\end{equation}
For $\alpha \in \overk^*$ and all $\tildex$ affine point and $i \in
Z(n)$, we have 
\begin{equation}
\thetaact(\alpha*\tildex,i)=\alpha* \thetaact(\tildex,i).
\end{equation}
In the same way, for all $\tildex$ affine point which is a lift of $x \in G$ and $\alpha \in
\overk^*$:
\begin{equation}
 \inv(\alpha* \tildex)=\alpha*\inv(\tildex).
\end{equation}
\end{lemma}
The following lemma states that $\scalmult$ is compatible with $\tildef$ (see
Definition~\ref{def:affinemap}):
\begin{lemma}\label{lem:tildefscalmult}
Let $(A, \pol, \Thetapol)$ and $(B, \bpol, \Thetabpol)$ be isog-$f$-compatible or
dual-isog-$f$-compatible abelian
varieties of respective types $K(n)$ and $K(m)$ and associated affine theta null points
$\tildenullpol$ and $\tildenullbpol$. We suppose that
$\tildef(\tildenullpol)= \tildenullbpol$.

 Let $x, y \in A(\overk)$ and let
$\tildex, \tildey, \tildexmy$ be affine lifts of $x$, $y$ and
 $x-y$, we have:
 \begin{equation}
	\tildef(\diffadd(\tildex, \tildey, \tildexmy, \tildenullpol))=\diffadd(\tildef(\tildex),
	\tildef(\tildey), \tildef(\tildexmy), \tildef(\tildenullpol)).
 \end{equation}
 In particular for all $\ell$ positive integer, we have:
 \begin{equation}
	\tildef(\scalmult(\ell, \tildexy, \tildex, \tildey, \tildenullpol))=
	\scalmult(\ell, \tildef(\tildexy), \tildef(\tildex), \tildef(\tildey),
	\tildef(\tildenullpol))
 \end{equation}
\end{lemma}
\begin{proof}
 The second claim is an immediate consequence of the first. 
 Let $x,y \in A(\overk)$, let $U$ be an affine open subset of $A$ containing $x$, $y$, $x-y$,
 $0$ and $x+y$ and choose a local trivialization $\Gamma(U, \pol) \simeq \Gamma(U,
 \stsheaf_A)$. In the same manner, choose a local trivialization $\Gamma(f(U), \bpol)
 \simeq \Gamma(f(U), \stsheaf_B)$. If $s \in \Gamma(U, \pol)$ (resp. $s \in
 \Gamma(f(U), \bpol)$) and $x \in U(\overk)$ (resp. $x \in f(U)(\overk)$), denote by $s(x)$ the
 evaluation map. We can choose the trivializations so that $\tildef(
 (\theta_i^\Thetapol(0))_{i \in \Zn}) = (\theta_i^\Thetabpol(0))_{i \in Z(m)}$, thus
 for all $z \in A(\overk)$, 
 \begin{equation}\label{eq:tildefscalarmult}
 \tildef(
 (\theta_i^\Thetapol(z))_{i \in \Zn}) = (\theta_i^\Thetabpol(z))_{i \in Z(m)}.
\end{equation}

 As for all $\alpha \in \overk^*$, $\tildef(\alpha * \tildex)= \alpha *
 \tildef(\tildex)$, we can suppose that for $z=0, x,y, x-y$, 
 $\tildez = (\theta^\Thetapol_i(z))_{i \in
 \Zn}$. Then, as $\diffadd(\tildex, \tildey, \tildexmy, \tildenullpol)$ is computed using Riemann relations of Theorem~\ref{sec:thriemann} for the
 points $(x+y, x-y, 0, 0; x, x, y,y)$ in Riemann position, it returns the affine point
 $(\theta^\Thetapol_i(x+y))_{i \in \Zn}$. In the same manner, 
$\diffadd(\tildef(\tildex),
	\tildef(\tildey), \tildef(\tildexmy), \tildef(\tildenullpol))$ returns
	$(\theta_i^\Thetabpol(f(x+y)))_{i \in \Zn}$. The result is thus an immediate
	consequence of Equation (\ref{eq:tildefscalarmult}).
\end{proof}
The previous lemma shows that $\diffadd$ is well-behaved with respect to isog-$f$-compatible isogenies. We have a similar result for the action of $\Gpol$ on affine points. In order
to prove it, we first state a more general result for Riemann equations:
\begin{definition}\label{def:riemanngroup}
 Let $(A, \pol, \Thetapol)$ be a marked abelian variety of type $K(n)$. We say that $\vg=(g_1,
 \ldots, g_4; g_5, \ldots, g_8) \in \Gpol^8$ is in Riemann position if, writing
 $\Thetapol^{-1}(g_j)=(\alpha_j, u_j, v_j) \in G(n)$ for $j=1, \ldots, 8$, the elements
 $(u_1, \ldots, u_4; u_5, \ldots, u_8)$ of $\Zn$ and $(v_1, \ldots, v_4; v_5, \ldots, v_8)$ of
 $\dZn$ are in Riemann position (Definition~\ref{sec1:defrieman}) and $\prod_{j=1}^4
 \alpha_j=\prod_{j=5}^8 \alpha_j$. We denote by $\mathbf{1} \in \Gpol^8$ the $8$-tuple of neutral
 elements, which is in Riemann position.
\end{definition}
Note that the scalar part matters: if $g_j=\alpha_j \in \overk^*$ for all $j$, the left (resp.
right) hand side of the equations below is multiplied by $\prod_{j=1}^4 \alpha_j$ (resp.
$\prod_{j=5}^8\alpha_j$).

\begin{proposition}\label{prop:riemannthetagroupaction}
 Let $(A, \pol, \Thetapol)$ be a marked abelian variety of type $K(n)$. For $\vg=(g_1,
 \ldots, g_4; g_5, \ldots, g_8) \in \Gpol^8$, $\vx=(x_1, \ldots, x_4; x_5, \ldots,
 x_8) \in A(\overk)^8$ and $\vi=(i_1, \ldots, i_4; i_5, \ldots, i_8)\in \Zn^8$, we set:
 \begin{equation}\label{eq:riemaneq22}
	L'(\vi, \vx, \vg) = \sum_{\eta \in Z(2)} \prod_{j=1}^4 g_j(\theta_{i_j +
	\eta}^{\Thetapol})(x_j) -\sum_{\eta \in Z(2)} \prod_{j=5}^8 g_j(\theta_{i_j +
	\eta}^{\Thetapol})(x_j).
 \end{equation}
 Then the relations $L'(\vi, \vx, \mathbf{1})=0$, for $\vx$ and $\vi$ in Riemann position,
 span the vector space of Riemann equations. Moreover, if $\vg$ is in Riemann position and
 $\Thetapol^{-1}(g_j)=(\alpha_j, u_j, v_j)$, then $\vi+\overline{u}=(i_1+u_1, \ldots, i_8+u_8)$ is in
 Riemann position for all $\vi$ in Riemann position, and
 \begin{equation}\label{eq:riemanngroup}
	L'(\vi, \vx, \vg)=F\, L'(\vi+\overline{u}, \vx, \mathbf{1}), \quad \text{with } F=\prod_{j=1}^4 \alpha_j
	v_j(-i_j-u_j).
 \end{equation}
 In particular, if $\vx$, $\vi$ and $\vg$ are in Riemann position, then $L'(\vi, \vx, \vg)=0$.
\end{proposition}
\begin{proof}
 For the first claim, taking the sum over all the characters $\chi \in
 \dZ(2)$ of $(\ref{eq:riemaneq})$, we obtain the relation (\ref{eq:riemaneq22}) with
 $\vg=\mathbf{1}$. Conversely, we can recover the equations (\ref{eq:riemaneq}) from a linear
 combination of equations (\ref{eq:riemaneq22}).

 We first note two facts on elements in Riemann position. Let $(w_1, \ldots, w_4; w_5,
 \ldots, w_8)$ be elements of $\Zn$ in Riemann position, with $z_w$ as in Definition
~\ref{sec1:defrieman}, and $(v_1, \ldots, v_4; v_5, \ldots, v_8)$ be elements of $\dZn$ in Riemann
 position, with $z_v$. Then:
 \begin{itemize}
	\item $\sum_{j=1}^4 v_j=2(v_1+z_v)$ and $\sum_{j=5}^8 v_j=\sum_{j=1}^4 v_j-2z_v=2v_1$, so that
	 $(\prod_{j=1}^4 v_j)(\eta)=(\prod_{j=5}^8 v_j)(\eta)=1$ for all $\eta \in Z(2)$;
	\item writing the group law of $\dZn$ additively and $\langle v, w \rangle=v(w)$, we have
	 \begin{equation}
		\begin{split}
		 \sum_{j=5}^8 \langle v_j, w_j \rangle & = \langle v_1+z_v, w_1+z_w \rangle + \sum_{j=2}^4
		 \langle v_j-z_v, w_j-z_w \rangle \\
		 & = \sum_{j=1}^4 \langle v_j, w_j \rangle + \langle v_1-\sum_{j=2}^4 v_j, z_w
		 \rangle + \langle z_v, w_1-\sum_{j=2}^4 w_j \rangle + 4 \langle z_v, z_w \rangle \\
		 & =\sum_{j=1}^4 \langle v_j, w_j \rangle - 2\langle z_v, z_w \rangle - 2 \langle z_v,
		 z_w \rangle + 4 \langle z_v, z_w \rangle=\sum_{j=1}^4 \langle v_j, w_j \rangle,
		\end{split}
	 \end{equation}
	 that is $\prod_{j=5}^8 v_j(w_j)=\prod_{j=1}^4 v_j(w_j)$.
 \end{itemize}
 Also, if $\vi$ and $\overline{u}$ are in Riemann position with $z_i$ and $z_u$, then $\vi+\overline{u}$ is in
 Riemann position with $z_i+z_u$.

 Now let $\vg$ be in Riemann position, $\Thetapol^{-1}(g_j)=(\alpha_j, u_j, v_j)$, and $\vi$ in
 Riemann position, and let $w_j=i_j+u_j$. By (\ref{eq:trans}), for $j=1, \ldots, 8$ and $\eta \in
 Z(2)$,
 \begin{equation}
	g_j(\theta_{i_j+\eta}^\Thetapol)=\alpha_j v_j(-w_j-\eta)\theta_{w_j+\eta}^\Thetapol=\alpha_j
	v_j(-w_j)v_j(\eta)^{-1}\theta_{w_j+\eta}^\Thetapol.
 \end{equation}
 By the first fact, $\prod_{j=1}^4 v_j(\eta)=\prod_{j=5}^8 v_j(\eta)=1$, so that
 \begin{equation}
	L'(\vi, \vx, \vg)=\Big(\prod_{j=1}^4 \alpha_j v_j(-w_j)\Big) \sum_{\eta \in Z(2)}
	\prod_{j=1}^4 \theta_{w_j+\eta}^\Thetapol(x_j)-\Big(\prod_{j=5}^8 \alpha_j v_j(-w_j)\Big)
	\sum_{\eta \in Z(2)} \prod_{j=5}^8 \theta_{w_j+\eta}^\Thetapol(x_j).
 \end{equation}
 As $\prod_{j=1}^4\alpha_j=\prod_{j=5}^8 \alpha_j$ and, by the second fact applied to $(v_j)$
 and $(w_j)$, $\prod_{j=1}^4 v_j(-w_j)=\prod_{j=5}^8 v_j(-w_j)$, the two factors are equal to $F$,
 which gives (\ref{eq:riemanngroup}). The last claim follows from Theorem~\ref{sec:thriemann}
 since $\vi+\overline{u}$ is in Riemann position.
\end{proof}
\begin{corollary}\label{cor:diffaddaction}
 Let $(A, \pol, \Thetapol)$ be a marked abelian variety of type $K(n)$ and let $\tildeH$ be one
 of the level subgroups $\Thetapol(\{1\}\times \Zn \times\{0\})$ or $\Thetapol(\{1\}\times
 \{0\}\times\dZn)$. For $g \in \Gpol$ and $\tildex=(\rho^\pol_x(\theta^\Thetapol_i))_{i \in \Zn}$
 an affine lift of $x \in A(\overk)$, we denote by $g\tildex=(\rho^\pol_x(g.\theta^\Thetapol_i))_{i
 \in \Zn}$ the affine lift of $x+\pigpol(g)$ given by Lemma~\ref{sec1:lemact}. Let $x,y \in
 A(\overk)$ and let $\tildex$, $\tildey$, $\tildexmy$ be affine lifts of $x$, $y$, $x-y$ and
 $\tildexpy=\diffadd(\tildex, \tildey, \tildexmy, \tildenullpol)$. For all $g_x, g_y \in
 \tildeH$, we have:
 \begin{equation}\label{eq:diffaddaction1}
	(g_x g_y) \tildexpy=\diffadd(g_x \tildex, g_y \tildey, (g_xg_y^{-1})\tildexmy,
	\tildenullpol).
 \end{equation}
 As a consequence, for all $\ell \in \N^*$, $x,y \in
 A(\overk)$, $\tildex$, $\tildey$, $\tildexpy$ affine lifts of $x$, $y$, $x+y$ and $g_x,
 g_y \in \tildeH$:
 \begin{align}
	\scalmult(\ell, (g_x g_y) \tildexpy, g_x \tildex, g_y \tildey,
	\tildenullpol) &=(g_x^\ell g_y) \scalmult(\ell, \tildexpy, \tildex, \tildey,
	\tildenullpol) \label{eq:diffaddaction2},\\
	\scalmult(\ell, g_x \tildex, g_x \tildex, \tildenullpol,
	\tildenullpol)&=g_x^\ell \scalmult(\ell, \tildex, \tildex, \tildenullpol,
	\tildenullpol). \label{eq:diffaddaction3}
 \end{align}
\end{corollary}
\begin{proof}
 By Lemma~\ref{lem:scalar}, both sides of (\ref{eq:diffaddaction1}) are multiplied by the same
 factor when $\tildex$, $\tildey$, $\tildexmy$ and $\tildenullpol$ are multiplied by
 scalars, since $g(\lambda \tildez)=\lambda g\tildez$ for $\lambda \in \overk^*$. As in the
 proof of Lemma~\ref{lem:tildefscalmult}, we can thus suppose that, for a local
 trivialization of $\pol$ on an affine open subset containing all the points involved,
 $\tildez=(\theta^\Thetapol_i(z))_{i \in \Zn}$ for $z=x, y, x-y, 0$, so that
 $\tildexpy=(\theta^\Thetapol_i(x+y))_{i \in \Zn}$, and $g\tildez=(g(\theta^\Thetapol_i)(z))_{i \in
 \Zn}$ for $g \in \Gpol$.

 Recall that $\diffadd(\tildex, \tildey, \tildexmy, \tildenullpol)$ is the affine point
 $\tilde{r}$ determined by the Riemann relations (\ref{eq:riemaneq1}) for the $8$-tuple of affine
 points $(\tilde{r}, \tildexmy, \tildenullpol, \tildenullpol; \tildex, \tildex, \tildey, \tildey)$,
 which lie above $\vx=(x+y, x-y, 0, 0; x, x, y, y)$, in Riemann position (with $z=-y$), and all
 $\vi$ in Riemann position. Let $\vg=(g_xg_y, g_xg_y^{-1}, 1, 1; g_x, g_x, g_y, g_y)$. If
 $\tildeH=\Thetapol(\{1\}\times \Zn \times\{0\})$, write $g_x=\Thetapol((1,u,0))$ and
 $g_y=\Thetapol((1,u',0))$; then $g_xg_y=\Thetapol((1,u+u',0))$ and
 $g_xg_y^{-1}=\Thetapol((1,u-u',0))$, all scalar parts are $1$, the $\dZn$-components are $0$, and
 $(u+u', u-u', 0, 0; u, u, u', u')$ is in Riemann position (with $z=-u'$): $\vg$ is in Riemann
 position. The same holds if $\tildeH=\Thetapol(\{1\}\times \{0\}\times\dZn)$. By Proposition
~\ref{prop:riemannthetagroupaction}, $L'(\vi, \vx, \vg)=0$ for all $\vi$ in Riemann position, which
 means that the $8$-tuple of affine points $((g_xg_y)\tildexpy, (g_xg_y^{-1})\tildexmy,
 \tildenullpol, \tildenullpol; g_x\tildex, g_x \tildex, g_y\tildey, g_y \tildey)$ satisfies the
 Riemann relations (\ref{eq:riemaneq1}). Hence $\diffadd(g_x \tildex, g_y \tildey,
 (g_xg_y^{-1})\tildexmy, \tildenullpol)=(g_xg_y)\tildexpy$, which is (\ref{eq:diffaddaction1}).

 Equation (\ref{eq:diffaddaction2}) follows by induction along the Montgomery ladder computing
 $\scalmult$. Every affine point computed by the ladder lies above $ax+by$ for some integers $a, b$
 and is either one of the inputs $\tildexpy$, $\tildex$, $\tildey$, $\tildenullpol$ (for
 $(a,b)=(1,1), (1,0), (0,1), (0,0)$) or the output of a differential addition $\diffadd(\tilde{p},
 \tilde{q}, \tilde{s}, \tildenullpol)$ of previously computed points $\tilde{p}$, $\tilde{q}$, $\tilde{s}$
 above $ax+by$, $a'x+b'y$ and $(a-a')x+(b-b')y$. We show by induction that the point computed by the
 ladder with the inputs of the left hand side at the same step is $g_x^ag_y^b \tilde{p}$, where
 $\tilde{p}$ is the corresponding point computed with the inputs of the right hand side. This is
 true for the inputs, and, if it holds for $\tilde{p}$, $\tilde{q}$, $\tilde{s}$, then, as
 $\tildeH$ is commutative, $g_x^{a-a'}g_y^{b-b'}=(g_x^ag_y^b)(g_x^{a'}g_y^{b'})^{-1}$, and
 (\ref{eq:diffaddaction1}) shows that it holds for the output. Equation (\ref{eq:diffaddaction3}) is the particular case $y=0$, $\tildey=\tildenullpol$,
 $g_y=1$ of Equation (\ref{eq:diffaddaction2}).
\end{proof}

We measure the complexity of algorithms by counting operations in $k$, where $k$ is
the compositum of the fields of definition of $(B, \bpol, \Thetabpol, \theta_0^\Thetabpol,
\rho^\bpol_{0_{\Thetabpol}})$ and $G \subset B[n]$. An element of $B(\overk)$ can be
represented with $\sharp Z(m)=m^g$ coordinates, each an element of $k$. It is clear that, on $B$,
$\diffadd$ takes $O(m^g)$ base field operations, and $\scalmult(\ell,\dots)$, $O(m^g\log(\ell))$
base field operations (and $O(n^g)$ and $O(n^g\log(\ell))$ operations on an abelian variety of type
$K(n)$). The algorithms below also require some root extractions (for instance to compute
good lifts, see Lemma~\ref{lem:scalar} and the following discussion) and roots of unity of order
dividing $2n$: we assume that $k$ contains these roots, which may require replacing $k$ by a
finite extension, and we do not include the cost of the root extractions, whose number does
not depend on $m$ and $n$, in the complexity estimates.

Let $\Zz(d)$ denote a set of representatives of the classes of $Z(n)/\mu_{m,n}(Z(m))\simeq Z(d)$,
and let $\pi_{n,d}: \Zn \rightarrow Z(d)\simeq \Zn/\mu_{m,n}(Z(m))$ be the canonical projection.
\begin{definition}
 We say that
 $(e_1, \ldots, e_g) \in \Zz(d)^g$ is a basis of $\Zz(d)$ if $(\pi_{n,d}(e_1), \ldots,
 \pi_{n,d}(e_g))$ is a basis of $Z(d)$. Then $(e_i, e_i + e_j)_{i,j=1,\ldots, g}$ is
 called a chain basis of $\Zz(d)$.
\end{definition}
Using Riemann equations, we can recover an excellent lift of $G$ from the knowledge of
lifts $\tildeg_f(\kappa)$ for $\kappa \in \Bb$ for $\Bb$ a chain basis of $\Zz(d)$:
\begin{proposition}\label{sec2:proprecons}
 Keeping the hypotheses and notation of Proposition~\ref{sec2:lemmalambdai}, let 
 $\Bb_0=(e_i)_{i=1, \ldots, g}$ and
 $\Bb = (e_i, e_i + e_j)_{i,j=1,
 \ldots, g}$ be respectively a basis and a chain basis of
 $\Zz(d)$, suppose that we have chosen affine lifts such that for all $\kappa \in \Bb$
 \begin{equation}\label{sec2:lemmaeq29}
	\tildeg_f(\kappa)= (\theta^\Thetapol_{\mu_{m,n}(j)+\kappa}(0_A))_{ j\in Z(m)}.
 \end{equation}
 
 Then, there exists a unique set of affine
 lifts $\tildeG= \{ \tildeg'_f(e), e \in \Zn \}$ such that 
 for all $e \in \Bb$, $\tildeg'_f(e)=\tildeg_f(e)$ and satisfying all possible relations
 provided by $\scalmult$, $\threeway$, $\inv$ and $\thetaact$ operations.
 In particular, there exists a unique
 rigidified abelian variety $(A, \pol, \Thetapol,
 \theta_0^\Thetapol, \rho^\pol_{0_A})$ satisfying (\ref{sec2:lemmaeq29}).

 Suppose moreover that we have chosen affine lifts such that for all $\kappa \in \Bb_0
 \cup \{ 0 \}$:
 \begin{equation}\label{sec2:lemmaeq30}
	\tildeg^z_f(\kappa)= (\theta^\Thetapol_{\mu_{m,n}(j)+\kappa}(z))_{ j\in Z(m)}.
 \end{equation}
 Then, there exists a unique set of affine
 lifts $\tildeG(z)= \{ {\tildeg}^{\prime z}_f(e), e \in \Zn \}$ such that 
 for all $e \in \Bb_0 \cup \{ 0 \}$, ${\tildeg}^{\prime z}_f(e)=\tildeg^z_f(e)$ and satisfying all possible relations
 provided by $\scalmult$, $\threeway$ and $\thetaact$ operations. In
 particular, there exists a unique $z \in A(\overk)$ satisfying (\ref{sec2:lemmaeq30}).
\end{proposition}
\begin{proof}
 From the knowledge of $\tildeg_f(\kappa)$ for $\kappa \in \Zz(d)$, because of
 Proposition~\ref{sec2:lemm6} (3), one can recover
 $\tildeg_f(\kappa)$ for all $\kappa \in \Zn$ using the operator $\thetaact$. So, by
 Proposition~\ref{sec2:lemm6} (1) and (2), it is
 enough to show that from the knowledge of $\tildeg_f(\kappa)$ for $\kappa \in \Bb$, we can recover $\tildeg_f(\kappa)$ for
 all $\kappa$ in $\Zz(d)$ using $\scalmult$, $\threeway$, $\inv$ and $\thetaact$.

 But, by a repeated use of $\threeway$, we can compute $\tildeg_f(\kappa)$ for $\kappa=
 \sum_{i \in I} e_i$ for all $I \subset \{1, \ldots, g\}$. Then we can recover all
 $\tildeg_f(\kappa)$ for $\kappa \in \Zz(d)$ with a repeated use of $\diffadd$ 
 as explained in \cite[Theorem~4.4]{DRisogenies}.

 For the second claim of the proposition, we use the same argument as before to reduce the
 claim to showing that, from the knowledge of $\tildeg^z_f(\kappa)$ for $\kappa \in \Bb_0
 \cup \{ 0 \}$,
 we can recover $\tildeg^z_f(\kappa)$ for
 all $\kappa$ in $\Zz(d)$ using $\scalmult$, $\threeway$ and $\thetaact$
 operations.
 For
 this, it is enough to explain that for $i,j \in \Zn$, from the knowledge of
 $\tildeg^z(i)$, $\tildeg^z(j)$, $\tildeg^z(0)$, $\tildeg(i)$, $\tildeg(j)$, $\tildeg(i+j)$
 we can recover the unique affine lift $\tildeg^z(i + j)$ compatible with the
 former lifts and Riemann equations. But this is exactly what $\threeway(\tildeg^z(i),
 \tildeg(i+j), \tildeg^z(j), \tildeg^z(0), \tildeg(i), \tildeg(j), \tildenullbpol)$ does.
\end{proof}

\begin{remark}
 In the proof of the proposition, we use a sequence of operations by $\scalmult$,
 $\threeway$, $\inv$ and $\thetaact$ to recover $\tildeG$ without specifying it
 precisely. The fact that all these operations are based on relations that are satisfied by
 $\tildenullpol$ as explained in Proposition~\ref{sec2:lemm6} guarantees that
 whatever sequence of operations we choose in order to do the computation, we will
 obtain the same result.

 The time complexity of the algorithm obtained from the proof of Proposition
~\ref{sec2:proprecons}, to recover $\tildeg_f(\kappa)$ for $\kappa \in \Zz(d)$ from the
 knowledge of $\tildeg_f(\kappa)$ for $\kappa \in \Bb$ is $O(n^g)$ operations in
 $k$: each of the $d^g$ affine points $\tildeg_f(\kappa)$, $\kappa \in \Zz(d)$, has $m^g$
 coordinates and is obtained by one $\diffadd$ (of cost $O(m^g)$) from previously computed
 points, apart from the $2^g$ points $\tildeg_f(\sum_{i \in I} e_i)$ obtained by $\threeway$.
\end{remark}

Proposition~\ref{sec2:proprecons} shows that in order to compute the theta null point of 
$(A, \pol, \Thetapol, \theta_0^\Thetapol, \rho^\pol_{0_A})$ which is
isog-$f$-compatible to $(B, \bpol, \Thetabpol, \theta_0^\Thetabpol, \rho^\bpol_{0_{B}})$, it
is enough to compute lifts of $\tildeg_f(\kappa)$ for $\kappa \in \Bb$ a chain basis of
$\Zz(d)$. By Proposition~\ref{sec2:proptors} and Theorem~\ref{prop:cond}, such lifts can always be computed when $d$ is odd. But when $d$ is even, the
symmetric compatibility condition becomes non-trivial, and it is not always possible to find
an excellent lift of $G$. Theorem~\ref{prop:cond} suggests two ways to address this obstruction and
satisfy the symmetric compatibility condition:
\begin{itemize}
 \item change the theta structure $\Thetabpol$ by acting on it by $K(2)$, which is always
	possible (Proposition~\ref{prop:changesym1});
 \item change the basis of $G$ by adding points of $\Thetabpolbar(\{0\}\times\dZ(m))$ of order
	$2^{v_2(d)}$, where $v_2(d)$ is the $2$-adic valuation of $d$; this is possible if and
	only if $2^{v_2(d)}$ divides $m$ (Lemma~\ref{lem:transsym} and Proposition
	\ref{prop:changesym2}). Note that points of $B[2]$ are not enough when $4 \mid d$.
\end{itemize}
We examine this obstruction and ways to remedy it by direct
computation on affine lifts. In the affine lifts approach, the difference between the even
and odd cases is that $G/\Thetabpolbar(Z(m) \times\{0\}) \simeq Z(d)$ has trivial
$2$-torsion in the odd case and non-trivial $2$-torsion in the even case. Thus, in the odd case,
the inversion operation acts freely on $G/\Thetabpolbar(Z(m) \times\{0\})-\{0\}$, but has non-trivial fixed points in the even case.

To make this argument precise, we consider
$\Sinv = \{ t \in Z(n), -t = t \mod Z(m)\}$. Recall that $\pi_{n,d}: \Zn \rightarrow
Z(d)\simeq \Zn/\mu_{m,n}(Z(m))$ is the canonical projection, then $\Sinv=
\pi_{n,d}^{-1}(Z(d)[2])$. For $e \in \Sinv$, we look at the relation of projective points:
\begin{equation}\label{eq:invol1}
 g_f(e) + g_f(e)= (1, 2e, 0).\nullbpol.
\end{equation}
From the previous relation, we deduce for $\tildee$ an affine point above $g_f(e)$ there
exists $\lambda \in \overk^*$ such that: 
\begin{equation}\label{eq:aobsrtruc1}
 \tildee= \lambda* \thetaact(\inv(\tildee), 2e).
\end{equation}
Recall that
$\pi_{\proj^{Z(m)}}: \Aff^{Z(m)}-\{0\} \rightarrow
\proj^{Z(m)}$ is the canonical projection, and for $e \in \Sinv$, we let $T(e) =
\pi_{\proj^{Z(m)}}^{-1}(g_f(e))
\subset \Aff^{Z(m)}$. We consider the map:
\begin{align*}
 \invd(g_f(e)): T(e)(\overk) & \rightarrow T(e)(\overk) \\
 \tildee & \mapsto \thetaact( \inv (\tildee), 2e).
\end{align*}
Because of (\ref{eq:invol1}), this map is well defined.
\begin{lemma}
 For all $e \in \Sinv$, 
 the map $\invd(g_f(e))$ is involutive, i.e. $\invd(g_f(e)) \circ \invd(g_f(e))=1$ so that there exists
 $\kappa(g_f(e)) \in \{-1, 1\}$ such that for all $\tildee \in T(e)(\overk)$,
 $\invd(g_f(e))(\tildee)=\kappa(g_f(e)) * \tildee$. Moreover, for all $t \in Z(m)$, we have
 $\kappa(g_f(e+t))=\kappa(g_f(e))$.
\end{lemma}
\begin{proof}
Let $\tildex=(x_i)_{i \in Z(m)}$ be an affine point. We have successively
$\inv((x_i))=(x_{-i})$ and, as $(1,t,0).(x_i)=(x_{i+t})$ by (\ref{eq:trans}) and Lemma
\ref{sec1:lemact}, $\thetaact( (x_{-i}), 2e)=(x_{-i-2e})$. Then $\thetaact( \inv ((x_{-i-2e})),
2e)=\thetaact((x_{i-2e}), 2e)=(x_i)$ and $\invd$ is involutive.
 By Lemma~\ref{lem:scalar}, for all $\lambda \in \overk^*$, $\invd(g_f(e))(\lambda * \tildee)=\lambda
 *(\invd(g_f(e))(\tildee))$, from which we deduce that there exists $\kappa(g_f(e)) \in \{-1, 1\}$
 such that $\invd(g_f(e))(\tildee)=\kappa(g_f(e)) * \tildee$.
 For the last claim, let $t \in Z(m)$. If $\tildeg_f(e)$ is an affine point above
 $g_f(e)$, then $(1, t, 0)\tildeg_f(e)$ is an affine point above $g_f(e+t)$.
 So starting from $(1, 2e, 0) \inv (\tildee)=\kappa(g_f(e)) * \tildee$, we get
 $(1, -t, 0) (1, 2(e+t), 0). \inv(\tildee)=\kappa(g_f(e))*(1, t, 0)\tildee$, so that
 $(1, 2(e+t), 0).\inv ((1,t,0) \tildee)=\kappa(g_f(e))*(1, t, 0)\tildee$ and we have
 proved that $\kappa(g_f(e))=\kappa(g_f(e+t))$.
\end{proof}
We remark that the definition of $\kappa(g_f(e))$ only depends on $\nullbpol$ and on
the class of 
$g_f(e)$ modulo $\Thetabpolbar((i,0))_{i \in Z(m)}$.
This allows us to state the following definition:
\begin{definition}\label{def:symcompat2}
    Let $e \in Z(n)$ and let $\ell = \min \{ t \in \N^* | te \in \mu_{m,n}(Z(m)) \}$. We say
    that $g_f(e)$ is symmetric compatible with
 $\nullbpol$ if either $\ell$ is odd or $\ell$ is even and, writing $\ell=2\ell'$, we have
 $\kappa(g_f(\ell' e))=1$. By the previous Lemma $\kappa(g_f(\ell' e))$ only depends on the class
 of $g_f(\ell' e)$ modulo $\Thetabpolbar((i,0))_{i \in Z(m)}$.
 We say that $G=\{ g_f(e), e \in \Zn \}$ is symmetric compatible with
 $\nullbpol$ if for all $e \in \Sinv$, $g_f(e)$ is symmetric compatible
 with $\nullbpol$.
\end{definition}
\begin{remark}
    We note that $\kappa(g_f(\ell' e))$ is independent of the choice of an affine lift
    $\tildeg_f(\ell' e)$ above $g_f(\ell' e)$. Thus, even if the computation of
    $\tildeg_f(\ell' e)$ with a $\scalmult$ depends on the choice of an affine lift of
    $g_f(e)$ and $\tildenullbpol$, $\kappa(g_f(\ell' e))$ depends only on $\nullbpol$.

    When $e$ belongs to a basis of $\Zn$, we have $\ell=d$ and, with $d'=d/2$, $\kappa(g_f(d'e))$ is
    obtained by computing $\tildee=\scalmult(d', \tildeg_f(e), \tildeg_f(e), \tildenullbpol,
    \tildenullbpol)$ and $\lambda \in \overk^*$ such that $\tildee=\lambda*(1,e',0).\inv(\tildee)$,
    where $e' \in Z(m)$ is such that $\mu_{m,n}(e')=de$: then $\kappa(g_f(d'e))=\lambda$. This is the
    test used in Algorithms~\ref{algo:changesym1} and~\ref{algo:changesym2}.
\end{remark}

\begin{example}\label{ex:hand48}
 Let $(B, \bpol, \Thetabpol)$ be a marked abelian variety of dimension $1$ and of type $K(4)$.
 Because of the symmetry relations, the general form of its theta null point is
 $0_{\Thetabpol}= (b_0:b_1:b_2:b_1) \in \proj^{Z(4)}(k)$. By acting on $0_{\Thetabpol}$ by $G(\bpol)$ we
 obtain the $4$-torsion point $P_0=(b_1:b_2:b_1:b_0)$. Let $P=(a_0:a_1:a_2:a_3) \in B[8]$ be such that
 $2P=P_0$. Lift it to an affine point $\tildeP=(a_0, a_1, a_2, a_3)$ and we suppose
 that we have chosen $a_0$ such that $\scalmult(2, \tildeP, \tildeP,
 \tildenullbpol,
 \tildenullbpol)=(b_1, b_2, b_1, b_0)$.
 Then we have $\widetilde{-P}=(a_0,a_3,a_2,a_1)$ and $(1, 1,
 0).(\widetilde{-P})=(a_3,a_2,a_1,a_0)$.
 The theta structure $\Thetabpol$ is symmetric compatible with $P$ if and only if
 $a_0=a_3$, and in this case we also have $a_2=a_1$. We can then form the level $8$ theta
 null point $(b_0:a_0:b_1:a_1:b_2:a_1:b_1:a_0)$, which satisfies the symmetry relations.
\end{example}

\begin{example}\label{ex:small48}
 We make Example~\ref{ex:hand48} explicit over the field $\F_{137}$; it can be checked by hand, or with a few lines of code. The point
 $0_{\Thetabpol}=(1:93:98:93)$ satisfies the Riemann relations of level $4$, and the corresponding
 curve $B$ (given by these relations) is an elliptic curve with $128$ rational points. The
 $4$-torsion point $P_0=(93:98:93:1)$ has four halves in $B(\F_{137})$, computed with the doubling
 formula $\scalmult(2, \cdot, \cdot, \tildenullbpol, \tildenullbpol)$:
 \begin{equation*}
	(1:31:31:1), \quad (1:84:84:1), \quad (1:53:84:136), \quad (1:106:31:136).
 \end{equation*}
 By Example~\ref{ex:hand48}, the first two are symmetric compatible with $\Thetabpol(\{1\}\times
 Z(4)\times\{0\})$ ($a_0=a_3$), and the last two are not ($a_3=-a_0$). For $P=(1:31:31:1)$, the affine
 lift $\tildeP=(1,31,31,1)$ already satisfies $\scalmult(2, \tildeP, \tildeP, \tildenullbpol,
 \tildenullbpol)=(93,98,93,1)$, and we obtain the level $8$ theta null point of the elliptic curve $A=B/K$ of Section~\ref{sec:isogchange}, where $K$ is generated by $4P$,
 \begin{equation*}
	(1:1:93:31:98:31:93:1),
 \end{equation*}
 which satisfies the Riemann relations of level $8$. The other good lifts $\lambda\tildeP$, with
 $\lambda^4=1$, that is $\lambda \in \{37, 100, 136\}$, give the theta null points of the other
 symmetric theta structures of level $8$ compatible with $\Thetabpol$, for instance
 $(1:136:93:106:98:106:93:136)$ for $\lambda=-1$. On the other hand, the same construction applied to
 the half $(1:106:31:136)$ gives $(1:1:93:106:98:31:93:136)$, which does not even satisfy the
 symmetry relations $\theta_j=\theta_{-j}$: no symmetric theta structure of level $8$ compatible with
 $\Thetabpol$ can be built from this point.
\end{example}

In order to check whether $G$ is symmetric compatible with $\nullbpol$, it is
enough to check that $g_f(e)$ is symmetric compatible with $\nullbpol$ for
all $e$ in a basis of $G$. To prove this, we can rely on Definition
\ref{def:symcompat} by first proving that the two
notions of symmetric compatibility given in Definition~\ref{def:symcompat} and Definition~\ref{def:symcompat2}
in fact agree.
\begin{proposition}\label{prop:propclef}
    Definition~\ref{def:symcompat}  and Definition
~\ref{def:symcompat2} are equivalent, i.e. for $e \in \Zn$, the point $g_f(e) \in G(\overk)$
 is symmetric compatible with $\Thetabpol(\{1\}\times Z(m) \times\{0\})$
 if and only if
 it is symmetric compatible with $\nullbpol$.
\end{proposition}
\begin{proof}
Let $\ell= \min \{ t \in \N^* | t e \in \Thetabpolbar(Z(m)\times
\{0\})(\overk) \}$.
If $\ell$ is odd, the proposition is trivial. Thus, we suppose that $\ell$ is even and
 write $\ell=2\ell'$ for $\ell'$ an integer.
 Following Definition~\ref{def:symcompat}, consider:
 \begin{itemize}
	\item $f_0: B_0
 \rightarrow B$ a separable isogeny with kernel $K_0$, $e_0 \in
	 K(f_0^*(\bpol))(\overk)$ such that $f_0(e_0)=e$;
	\item a symmetric level subgroup $\tildeH$ of $G(f_0^*(\bpol))$ above
	 $H=f_0^{-1}(\Thetabpolbar(Z(m)\times\{0\}))$
	 such that $\tildeK_0 \subset \tildeH$ (where $\tildeK_0$ is the descent data
	 of $f_0^*(\bpol)$ to $\bpol$) and
	 $f_0^\sharp(\tildeH)=\Thetabpol(\{1\}\times Z(m)\times\{0\})$;
     \item $g_{e_0} \in G(f_0^*(\bpol))$ symmetric such that
	 $\pi_{G(f_0^*(\bpol))}(g_{e_0})=e_0$ (so that $g_{e_0}^{\ell} \in \pm\tildeH$).
 \end{itemize}
 Let $h_0 \in H$ and $g_{h_0} \in
 \tildeH$ so that $g_{e_0}^{\ell} = \kappa_0(e_0, \tildeH) g_{h_0}$, following Definition
~\ref{def:kappa}.
 We have:
 \begin{equation}\label{prop19:eq3}
     g_{\ell' e_0}  = \kappa_0(e_0, \tildeH) g_{h_0} g_{\ell' e_0}^{-1},
\end{equation} 
where $g_{\ell' e_0}=g_{e_0}^{\ell'}$. 

Choose a rigidification $\rho_{0_{B_0}}^{f_0^*(\bpol)}: f_0^*(\bpol)(0_{B_0})
\rightarrow \stsheaf_{B_0}(0_{B_0})$. If $s \in \Gamma(B_0,f_0^*(\bpol))$, it allows us to
evaluate $s$ in $0_{B_0}$ by taking $\rho_{0_{B_0}}^{f_0^*(\bpol)}(s) \in
\overk$. 
We claim that the affine point $\tildede = (\rho^{f_0^*(\bpol)}_{0_{B_0}}\circ g_{\ell'
e_0}(f_0^*(\theta_i^{\Thetabpol})))_{i \in Z(m)} \in \Aff^{Z(m)}(\overk)$ is an affine lift
of $\ell'e \in \proj^{Z(m)}(\overk)$ (see Definition~\ref{def:affinepoint}).
Indeed, by Lemma~\ref{sec1:lemact}, $\rho^{f_0^*(\bpol)}_{0_{B_0}}\circ g_{\ell'
e_0}$ is a rigidification of $f_0^*(\bpol)$ in $\ell' e_0$. Moreover, if $\rho^{\bpol}_{\ell'e} :
\bpol(\ell' e) \rightarrow \overk$ is any rigidification of $\bpol$ in $\ell'e$, because there is a
canonical isomorphism $\bpol(\ell'e) \rightarrow f_0^*(\bpol)(\ell' e_0)$, $s \mapsto f_0^*(s)$, there exists
$\lambda \in \overk^*$ such that
for all $s \in \Gamma(B, \bpol)$, $\lambda
\rho^{\bpol}_{\ell'e}(s)=\rho^{f_0^*(\bpol)}_{\ell'e_0}(f_0^*(s))$. This completes the proof of the claim.
In the same way $\tildemde = (\rho^{f_0^*(\bpol)}_{0_{B_0}}\circ g_{\ell'
e_0}^{-1}(f_0^*(\theta_i^{\Thetabpol})))_{i \in Z(m)} \in \Aff^{Z(m)}(\overk)$ is an affine
lift of $-\ell'e\in \proj^{Z(m)}(\overk)$.

 Let $h=\ell e \in \Thetabpolbar(Z(m)\times\{0\})$ and $g_h \in \Thetabpol(\{1\}\times Z(m)
 \times\{0\})$ such that $\pigbpol(g_h)=h$. Remark that $f_0^\sharp(g_{h_0})=g_h$. 
We note that:
$\thetaact((\tildemde), h) = \tildede$ if and only if relation (\ref{prop19:eq3}) is satisfied with
$\kappa_0(e_0, \tildeH)=1$.
Indeed, $g_h (\tildemde)_i = 
\rho^{f_0^*(\bpol)}_{0_{B_0}}\circ g_{h_0} \circ g_{\ell' e_0}^{-1}
(f_0^*(\theta_i^{\Thetabpol}))
= \kappa_0(e_0, \tildeH)\rho^{f_0^*(\bpol)}_{0_{B_0}}\circ g_{\ell' e_0} (f_0^*(\theta_i^{\Thetabpol}))
= \kappa_0(e_0, \tildeH)(\tildede)_i$. 

It remains to prove that, for all $i \in Z(m)$,
$(\tildemde)_i = (\tildede)_{-i}$. Let $g_i =
\Thetabpol((1,i,0))$, let $i_0 \in H$ be such
that $f_0(i_0)=i$ and denote $g_{i_0} \in \tildeH$ such that
$\pi_{G(f_0^*(\bpol))}(g_{i_0})=i_0$. Then, $f_0^\sharp(g_{i_0})=g_i$.

We recall that $g_{-1}$ is defined in Remark~\ref{rq:rqgi}.
By using the fact that $g_{i_0}$ and $g_{\ell' e_0}$ are symmetric, and moreover that  
$\rho^{f_0^*(\bpol)}_{0_{B_0}}\circ g_{-1}=\rho^{f_0^*(\bpol)}_{0_{B_0}}$ (see Remark
\ref{rk:inverse}) and
that $g_{-1} ( f_0^*(\theta_0^{\Thetabpol}))= f_0^*(g_{-1} ( \theta_0^{\Thetabpol}))=
f_0^*(\theta_0^\Thetabpol)$ (see Remark~\ref{rq:rqgi}, where $\mu=1$ since $\bpol$ is totally symmetric, $(B,
\bpol, \Thetabpol)$ being a marked abelian variety), we get, writing $\rho_0=\rho^{f_0^*(\bpol)}_{0_{B_0}}$
and recalling that $g_{-1}^2=1$, $f_0^\sharp(g_{i_0})=g_i$ and $g_{i_0}^{-1}=g_{-i_0}$:

\begin{gather}
\begin{aligned}
    (\tildemde)_i & = \rho_0 \circ g_{\ell' e_0}^{-1}
    (f_0^*(\theta_i^{\Thetabpol})) & (\text{definition of } \tildemde) \\
& = \rho_0 \circ g_{\ell' e_0}^{-1}
    (f_0^*(g_i(\theta_0^{\Thetabpol}))) & (\text{Proposition~\ref{prop:base}}) \\
		  & =\rho_0 \circ g_{\ell' e_0}^{-1}\circ g_{i_0}
    (f_0^*(\theta_0^{\Thetabpol})) & (\text{Corollary~\ref{cor:canonical}}) \\
    &=\rho_0 \circ  g_{-1} \circ g_{\ell' e_0}\circ  g_{-1}\circ  g_{i_0}
    (f_0^*(\theta_0^{\Thetabpol})) & (g_{\ell'e_0} \text{ symmetric, Remark~\ref{rq:rqgi}}) \\
    & =  \rho_0 \circ g_{-1}
    \circ g_{\ell' e_0}\circ  g_{-i_0}\circ  g_{-1}
    (f_0^*(\theta_0^{\Thetabpol})) & (g_{i_0} \text{ symmetric, Remark~\ref{rq:rqgi}}) \\
    & = \rho_0 \circ g_{\ell' e_0} \circ g_{-i_0}
(f_0^*(\theta_0^{\Thetabpol})) & (\text{Remarks~\ref{rk:inverse} and~\ref{rq:rqgi}}) \\
				   & = (\tildede)_{-i} & (\text{Corollary~\ref{cor:canonical}, Proposition~\ref{prop:base}}).
\end{aligned}
\end{gather}
This concludes the proof.
\end{proof}

\begin{corollary}\label{cor:symbasis}
 Let $G \subset B[n]$ be a subgroup isomorphic to $Z(n)$, isotropic for $e_{B,n}$ and containing
 $\Thetabpolbar(Z(m)\times\{0\})$, and let $\tildeH=\Thetabpol(\{1\}\times Z(m)\times\{0\})$. Let $(e_i)_{i \in I}$ be a set of generators of $G(\overk)$ as a group
 such that $d/\expo(e_i, \tildeH)$ is odd for all $i \in I$ (for instance a basis of
 $G$). We have the equivalence:
 \begin{itemize}
	\item for all $e \in G(\overk)$, $e$ is symmetric compatible with
	 $\nullbpol$;
	\item for all $i \in I$, $e_i$ is symmetric compatible with
	 $\nullbpol$.
 \end{itemize}
\end{corollary}
\begin{proof}
 As $\Thetabpolbar(Z(m)\times\{0\}) \simeq Z(m)$ is contained in $G[m]=dG$, which is also
 isomorphic to $Z(m)$, we have $\Thetabpolbar(Z(m)\times\{0\})=dG$. The result is then an
 immediate consequence of Proposition~\ref{prop:propclef} and Corollary
~\ref{sec2:corgroup}.
\end{proof}

The following lemma explains how the symmetric compatibility of a point of $G$ changes
when it is translated by a point of $\Thetabpolbar(\{0\} \times \dZ(m))$.
\begin{lemma}\label{lem:transsym}
 Let $\tildeH=\Thetabpol(\{1\}\times Z(m)\times\{0\})$ and $Z_B=\Thetabpolbar(Z(m)\times\{0\})$.
 Let $G \subset B[n]$ be a subgroup isomorphic to $Z(n)$, isotropic for $e_{B,n}$ and
 containing $Z_B$. We suppose that $d$ is even. Let $x \in G$ be such that
 $\expo(x, \tildeH)=d$ and let $c \in Z(m)$ be such that $dx=\Thetabpolbar((c,0))$. Let $\hate \in
 \dZ(m)$ be such that $d\hate=0$ and let $x'=\Thetabpolbar((0,\hate))$. Then
 $\expo(x+x', \tildeH)=d$ and
 \begin{equation}\label{eq:transsym}
	\kappa_0(x+x', \tildeH)=\kappa_0(x, \tildeH)\, \hate(c)^{d/2}.
 \end{equation}
 Note that $\hate(c)^{d/2} \in \{-1, 1\}$ since $\hate(c)^d=1$.
\end{lemma}
\begin{proof}
 As $G \simeq Z(n)$, $G[m]=dG$ is isomorphic to $Z(m)$ and contains $Z_B$, so that
 $G[m]=Z_B$. As $Z_B \cap \Thetabpolbar(\{0\} \times \dZ(m))=\{0\}$, we get $G \cap
 \Thetabpolbar(\{0\} \times \dZ(m))=G[m] \cap \Thetabpolbar(\{0\} \times \dZ(m))=\{0\}$. Let $\lambda$
 be a positive integer such that $\lambda(x+x') \in Z_B \subset G$. Then $\lambda x' \in G \cap
 \Thetabpolbar(\{0\} \times \dZ(m))$, so that $\lambda x'=0$, and $\lambda x \in Z_B$, so that
 $d \mid \lambda$. As $d(x+x')=dx \in Z_B$, we have proved that $\expo(x+x', \tildeH)=d$.

 Let $f_0=[d]: B \rightarrow B$ and $\pol'=[d]^*(\bpol)$. As $\bpol$ is symmetric,
 $\pol' \simeq \bpol^{d^2}$, so that $K(\pol')=B[md^2]$. Let $y, y' \in B(\overk)$ be such that
 $dy=x$ and $dy'=x'$. Then $y \in B[nd] \subset K(\pol')$, $y' \in B[md] \subset K(\pol')$ and
 $y+y'$ is above $x+x'$. Let $\tildeH'$ be the symmetric level subgroup of $G(\pol')$ given by
 Proposition~\ref{prop:compat1} (2). By Proposition~\ref{prop:compat1} and Lemma
~\ref{lem:phikappa}, we can use $f_0$, $y$, $y'$, $y+y'$ and $\tildeH'$ to compute
 $\kappa_0(x, \tildeH)$, $\varphi_d(x', \tildeH)$ and $\kappa_0(x+x', \tildeH)$.

 Let $g$ (resp. $g'$) be a symmetric element of $G(\pol')$ above $y$ (resp. $y'$). As
 $\expo(y, \tildeH')=\expo(x, \tildeH)=d$ is even, we have $g^d=\kappa_0(x, \tildeH)\,h$ with $h \in
 \tildeH'$. Next, $x' \in K(\bpol)$ and $\Thetabpol((1,0,\hate))$ is a symmetric element above $x'$
 whose $d$-th power is $\Thetabpol((1,0,d\hate))=1 \in \tildeH$, so that $\varphi_d(x',
 \tildeH)=1$ and $g'^d=k$ with $k \in \tildeH'$.

 Let $\gamma \in \overk^*$ be such that $gg'=\gamma\, g'g$, so that $\gamma$ is, up to the
 convention for the commutator, $e_{\pol'}(y, y')^{\pm 1}$. For $\lambda_0 \in \overk^*$, we have
 $\delta_{-1}(\lambda_0 g g')=\lambda_0 g^{-1} g'^{-1}=\lambda_0 \gamma\, g'^{-1}g^{-1}=\lambda_0^2\gamma\,
 (\lambda_0 gg')^{-1}$, so that $u=\lambda_0 gg'$ is a symmetric element above $y+y'$ if
 $\lambda_0^2=\gamma^{-1}$. Moreover, an immediate induction gives $(gg')^d=\gamma^{-d(d-1)/2}
 g^d g'^d$, so that
 \begin{equation}
	u^d=\lambda_0^d \gamma^{-d(d-1)/2} g^d g'^d=\gamma^{-d^2/2}\,\kappa_0(x, \tildeH)\, h k,
 \end{equation}
 with $hk \in \tildeH'$. As $u^d$ and $hk$ are both symmetric elements above $d(y+y')$, we have
 $\gamma^{-d^2/2} \in \{-1, 1\}$, and as $\expo(x+x', \tildeH)=d$ is even, we obtain
 $\kappa_0(x+x', \tildeH)=\kappa_0(x, \tildeH)\, \gamma^{d^2/2}$.

 It remains to compute $\gamma^{d^2/2}$. By Proposition~\ref{def:weil} (4), $e_{\pol'}=e_{\bpol^{d^2}}$.
 By Proposition~\ref{def:weil} (3) applied to $\bpol^d$, as $y' \in K(\bpol^d)=B[n]$ and $dy=x \in
 K(\bpol^d)$, we have $e_{\bpol^{d^2}}(y', y)=e_{\bpol^d}(y', x)$. As $y'$ and $x$ are in
 $K(\bpol^d)$, by bilinearity and Proposition~\ref{def:weil} (3) applied to $\bpol$ (with $x' \in
 K(\bpol)$ and $dx \in K(\bpol)$), we get
 \begin{equation}
	e_{\bpol^d}(y', x)^d=e_{\bpol^d}(x', x)=e_\bpol(x', dx)=e_m((0,\hate),(c,0))=\hate(c)^{\pm 1}.
 \end{equation}
 Hence $\gamma^{d^2/2}=\left(e_{\bpol^d}(y',x)^{d}\right)^{\pm d/2}=\hate(c)^{\pm d/2}=\hate(c)^{d/2}$,
 the last equality because $\hate(c)^{d/2} \in \{-1, 1\}$. This proves (\ref{eq:transsym}).
\end{proof}

As noted above, if $G$ is not symmetric compatible with $\nullbpol$, one can
either change $G$ or $\nullbpol$ to restore symmetric compatibility. The following
proposition explains how to change $\nullpol$, while the next one handles the
case of changing $G$. Both should be compared with Theorem~\ref{prop:cond}:
\begin{proposition}\label{prop:changesym1}
 Let $\Bb=(e_i)_{i=1, \ldots, g}$ be a basis of $Z(n)$. For $i=1, \ldots, g$, denote by
 $c_i=d e_i \in Z(m)$ and
 by $\dc_i \in
 \dZ(m)[2]$ the character such that for $i,j=1, \ldots, g$, $\dc_i(c_i)=-1$ and
 $\dc_i(c_j)=1$ if $i \neq j$.

 Let $I \subset \{ 1, \ldots, g\}$ be the set of indices $i$ such that $g_f(e_i)$ is not
 symmetric compatible with $\nullbpol$. Let $c=\sum_{i \in I} \dc_i$.
 Let $g_c = \Phi(c) \in \Auts(G(m))$ where $\Phi$ is defined by (\ref{eq:defphi}).
 Then, for $i=1, \ldots, g$, $g_f(e_i)$ is symmetric compatible with ${0}_{\Thetabpol \circ g_c}$.
 Consequently, if $G$ is isotropic for $e_{B,n}$, ${0}_{\Thetabpol \circ g_c}$ is symmetric compatible
 with $G$ (Corollary~\ref{cor:symbasis}).
\end{proposition}
\begin{proof}
 Using Corollary~\ref{cor:symbasis}, it is enough to check that for all $i \in \{1,
 \ldots, g\}$, $g_f(e_i)$ is
 symmetric compatible with ${0}_{\Thetabpol \circ g_c}$. Let $d'=d/2$. 
 Let
 $\tildedei$ be an affine lift of $g_f(d' e_i)$. Then, we have
 $\Thetabpol((1, 2d'e_i, 0)). \inv (\tildedei)=\kappa(g_f(d'e_i)) * \tildedei$. 
 Using Equation (\ref{eq:defphi}), we have $(\Thetabpol\circ g_c)((1, 2d'e_i,
 0))=c(2d'e_i)\Thetabpol((1, 2d'e_i, 0))$ so that
 $(\Thetabpol\circ g_c)((1, 2d'e_i, 0)). \inv (\tildedei)=
 c(2d'e_i)\kappa(g_f(d'e_i))*\tildedei$. As $c(2d'e_i)\kappa(g_f(d'e_i))=1$, $e_i$ is
 symmetric compatible with ${0}_{\Thetabpol \circ g_c}$.
\end{proof}

If $G$ is not symmetric compatible with $\nullbpol$, the following
proposition shows how to modify it to restore symmetric compatibility with $\nullbpol$.
\begin{proposition}\label{prop:changesym2}
 We suppose that $G$ is isotropic for $e_{B,n}$, that $d$ is even and that $2^{v_2(d)}$ divides
 $m$, where $v_2(d)$ is the $2$-adic valuation of $d$. Let $\Bb=(e_i)_{i=1, \ldots, g}$ be a
 basis of $Z(n)$ and, for $i=1, \ldots, g$, let $c_i=d e_i \in Z(m)$, so that $(c_i)_{i=1,
 \ldots, g}$ is a basis of $Z(m)$. For $j=1, \ldots, g$, let $\zeta_j=1$ if $g_f(e_j)$ is
 symmetric compatible with $\nullbpol$, and let $\zeta_j$ be a root of unity of order
 $2^{v_2(d)}$ otherwise. Let $\hate_j \in \dZ(m)$ be such that $\hate_j(c_j)=\zeta_j$ and
 $\hate_j(c_i)=1$ if $i \neq j$. For $i=1, \ldots, g$, let $g'_f(e_i) = g_f(e_i)
 +\Thetabpolbar((0, \hate_i))$ and let $G'$ be the subgroup of $B[n]$ generated by the
 $g'_f(e_i)$.
 Then $G'$ is a subgroup of $B[n]$ isomorphic to $Z(n)$, containing
 $\Thetabpolbar(Z(m)\times\{0\})$, isotropic for $e_{B,n}$, and, for $i=1, \ldots, g$,
 $g'_f(e_i)$ is symmetric compatible with $\nullbpol$. Consequently (Corollary~\ref{cor:symbasis}), every $x \in
 G'(\overk)$ is symmetric compatible with $\Thetabpol( \{1\}\times Z(m) \times\{0\} )$.
\end{proposition}
\begin{proof}
 Let $Z_B=\Thetabpolbar(Z(m)\times\{0\})$ and, for $i=1, \ldots, g$, $x'_i=\Thetabpolbar((0,
 \hate_i))$. As the values of $\hate_i$ are $2^{v_2(d)}$-th roots of unity, we have $d\hate_i=0$, so
 that $d g'_f(e_i)=d g_f(e_i)=\Thetabpolbar((c_i,0))$. In particular $dG'=dG=Z_B \subset G'$.
 Let $(a_i) \in \Z^g$ be such that $\sum_i a_i g'_f(e_i)=0$. Then $\sum_i a_i g_f(e_i)=-\sum_i a_i
 x'_i \in G \cap \Thetabpolbar(\{0\} \times \dZ(m))=\{0\}$ (see the proof of Lemma
~\ref{lem:transsym}), so that $n \mid a_i$ for all $i$. Hence $G' \simeq Z(n)$.

 For $x \in B[n]$ and $y \in K(\bpol)=B[m]$, we have $e_{B,n}(x,y)=e_{B,m}(dx, y)$ and, by
 Proposition~\ref{def:weil} (1), $e_{B,m}=e_\bpol$ on $B[m]$. Thus, for $i \neq j$,
 $e_{B,n}(g_f(e_i), x'_j)=e_\bpol(\Thetabpolbar((c_i,0)), \Thetabpolbar((0,\hate_j)))=\hate_j(c_i)^{\pm
 1}=1$, and in the same way $e_{B,n}(x'_i, g_f(e_j))=1$. Moreover $e_{B,n}(x'_i,
 x'_j)=e_\bpol(d x'_i, x'_j)=1$. As $G$ is isotropic, we get $e_{B,n}(g'_f(e_i), g'_f(e_j))=1$
 for all $i,j$ (for $i=j$ this is clear), so that $G'$ is isotropic for $e_{B,n}$.

 Finally, let $i \in \{1, \ldots, g\}$. As $g_f(e_i)$ is an element of a basis of $G$ and $Z_B=dG$,
 we have $\expo(g_f(e_i), \tildeH)=d$ for $\tildeH=\Thetabpol(\{1\}\times Z(m)\times\{0\})$. By
 Lemma~\ref{lem:transsym}, $\kappa_0(g'_f(e_i), \tildeH)=\kappa_0(g_f(e_i),
 \tildeH)\,\zeta_i^{d/2}$. If $g_f(e_i)$ is symmetric compatible, $\zeta_i=1$; otherwise
 $\zeta_i^{d/2}=-1$ since the order of $\zeta_i$ is $2^{v_2(d)}$. In both cases,
 $\kappa_0(g'_f(e_i), \tildeH)=1$, and we conclude by Proposition~\ref{prop:propclef}. The last
 claim follows from Corollary~\ref{cor:symbasis}.
\end{proof}
\begin{remark}
 If $2^{v_2(d)}$ does not divide $m$ (for instance if $m=2$ and $4 \mid d$), then for $c$ an
 element of a basis of $Z(m)$ and $\hate \in \dZ(m)[d]$, $\hate(c)$ is a root of unity whose order
 divides $\gcd(m,d)$, hence has $2$-adic valuation smaller than $v_2(d)$, so that
 $\hate(c)^{d/2}=1$. By Lemma~\ref{lem:transsym}, no translation by a point of
 $\Thetabpolbar(\{0\}\times\dZ(m))$ killed by $d$ can then restore the symmetric compatibility, and one has to use Proposition
~\ref{prop:changesym1} instead.
\end{remark}

\begin{corollary}
From Proposition~\ref{prop:changesym1}, we deduce Algorithm~\ref{algo:changesym1}, the
running time of which is $O(g\,m^g\log(d))$ operations in the base field.
From Proposition~\ref{prop:changesym2}, we deduce Algorithm~\ref{algo:changesym2} with running time
$O(g^2 m^g \log(n))$ operations in the base field.
\end{corollary}
\begin{proof}
 In Algorithm~\ref{algo:changesym1}, the computation of $\kappa(d'e_j)$ requires the
 computation of $d'e_j$ with $\scalmult$, that is $O(m^g\log(d))$ operations, for $j=1, \ldots, g$,
 and applying $\Phi(c)$ to $g+1$ points costs $O(g\,m^g)$ operations. The dominant step of
 Algorithm~\ref{algo:changesym2} is the Gram-Schmidt algorithm, which uses $O(g^2)$
 Weil pairing computations. We can use the algorithm of \cite{cryptoeprint:2024/517} to compute the Weil
 pairing of points of $B[n]$ in $O(m^g \log(n))$ operations in the base field.
\end{proof}

\begin{algorithm}
\SetKwInOut{Input}{input}\SetKwInOut{Output}{output}

\SetKwComment{Comment}{/* }{ */}
\Input{
 \begin{itemize}
	\item the marked abelian variety $(B, \bpol, \Thetabpol)$ of type $K(m)$ given by its theta null
	 point $0_{\Thetabpol}$;
	\item $G_1 \subset B[n]$ such that:
	 \begin{itemize}
		\item $G_1$ is isomorphic to $Z(n)$;
		\item $G_1$ is isotropic for $e_{B,n}$;
		\item $G_1 \supset \Thetabpolbar(Z(m) \times \{ 0 \})$;
	 \end{itemize}
given by a basis $(e_i)_{i=1, \ldots, g}$.
 \end{itemize}
}
\Output{
 \begin{itemize}
	\item 
 $\Thetaubpol$ a theta structure for $(B, \bpol)$ such that
	for all $x \in G_1$, $x$ is symmetric compatible with $\Thetaubpol(\{1 \}
	 \times Z(m) \times \{ 0 \})$ given by its new theta null point
	 $0_{\Thetaubpol}$;
	\item $(e_i)_{i=1, \ldots, g}$ in the new coordinates provided by
	 $0_{\Thetaubpol}$.
 \end{itemize}
}
\BlankLine
$d'=d/2$, I = \{ \}\;
 \For{$j \in \{1, \ldots, g\}$}
 {
	\If {$\kappa(d' e_j)==-1$ (Definition~\ref{def:symcompat2})}
	{
	 I = I + \{ j \}\;
	}
 }
 Let $c = \sum_{i \in I} \dc_i$\;
 Let $\Phi(c)$ be the automorphism of $\Aff^{Z(m)}= \Spec(k[X_i, i \in Z(m)])$,
 $X_i \mapsto c(i) X_i$ (Lemma~\ref{lem:actionK})\;
 $0_{\Thetaubpol}=\Phi(c) (0_{\Thetabpol})$\;
 \For{$j \in \{1, \ldots, g\}$}
 {
	$e_j = \Phi(c)(e_j)$\; 
 }
 \Return $0_{\Thetaubpol}, (e_i)_{i=1, \ldots, g}$.
 \caption{Algorithm to compute $0_{\Thetabpol}$ symmetric compatible with a subgroup
 $G_1$ of $B[n]$.}
 \label{algo:changesym1}
\end{algorithm}

\begin{algorithm} 
\SetKwInOut{Input}{input}\SetKwInOut{Output}{output}

\SetKwComment{Comment}{/* }{ */}
\Input{
 \begin{itemize}
	\item the marked abelian variety $(B, \bpol, \Thetabpol)$ of type $K(m)$ given by its theta null
 point $0_{\Thetabpol}$, with $d$ even and $2^{v_2(d)} \mid m$ (otherwise, use Algorithm
	~\ref{algo:changesym1});
	\item a basis $(e_i)_{i=1, \ldots, 2g}$ of $B[n]$.
 \end{itemize}
}
\Output{
A basis $(w_i)_{i=1, \ldots, 2g}$ of $B[n]$ such that if $G_1=(w_1, \ldots, w_g)$ and
 $G_2=(w_{g+1}, \ldots, w_{2g})$:
 \begin{itemize}
	\item $G_1$ and $G_2$ are isotropic for $e_{B,n}$;
	\item for all $x \in G_1$ (resp. $x \in G_2$), $x$ is symmetric compatible with $\Thetabpol(\{1 \}
	 \times Z(m) \times \{ 0 \})$ (resp. with $\Thetabpol(\{1 \}
	 \times \{ 0 \} \times \dZ(m))$).
 \end{itemize}
}
\BlankLine
 From $(e_i)_{i=1, \ldots, 2g}$, using a Gram-Schmidt like algorithm, compute a basis
 $(w_1, \ldots, w_{2g})$ of $B[n]$ such that:
\begin{itemize}
    \item 
 $(dw_1, \ldots, dw_g)=(\Thetabpolbar((u_1,0)), \ldots, \Thetabpolbar((u_g,0)))$ where
 $(u_i)_{i=1, \ldots, g}$ is a basis of $Z(m)$, and $(dw_{g+1}, \ldots,
 dw_{2g})=(\Thetabpolbar((0,v_1)), \ldots, \Thetabpolbar((0,v_g)))$ where $(v_i)_{i=1, \ldots,
 g}$ is a basis of $\dZ(m)$\;
\item $e_{B,n}(w_i, w_j)=1$ if $i, j \le g$ or $i,j > g$, and $e_{B,n}(w_i, w_{j+g})=\zeta^{\delta_{ij}}$
 for $i,j \le g$ (where $\zeta$ is a primitive $n$-th root of unity).
\end{itemize}
 $d'=d/2$, $\nu=v_2(d)$\;
	\For{$i \in \{1, \ldots, g\}$}
	{
 \lIf{$\kappa(d'w_i)==1$ (Definition~\ref{def:symcompat2})}{$\zeta_i=1$}\lElse{let $\zeta_i$ be a root of unity of order $2^\nu$}
 Let $\hat u_i \in \dZ(m)$ be such that $\hat u_i(u_i)=\zeta_i$ and $\hat u_i(u_j)=1$ for $j \ne i$\;
 \lIf{$\kappa(d'w_{i+g})==1$}{$\zeta'_i=1$}\lElse{let $\zeta'_i$ be a root of unity of order $2^\nu$}
 Let $\hat v_i \in Z(m)$ be such that $v_i(\hat v_i)=\zeta'_i$ and $v_j(\hat v_i)=1$ for $j \ne i$\;
 }
	\For{$i \in \{1, \ldots, g\}$}
	{
 $w_i= w_i + \Thetabpolbar((0, \hat u_i))$\;
 $w_{i + g}= w_{i + g} + \Thetabpolbar((\hat v_i, 0))$\;
 }
 \Return $(w_i)_{i=1, \ldots, 2g}$.
 \caption{Algorithm to compute a decomposition of $B[n]=G_1 \times G_2$ symmetric compatible with
 $0_{\Thetabpol}$.}
 \label{algo:changesym2}
\end{algorithm}

Now that Proposition~\ref{prop:propclef} gives us an effective criterion to check that all elements of
$G(\overk)$ are symmetric compatible with $\Thetabpol(\{1\}\times Z(m) \times \{0\})$, we can
come back to our initial problem of finding the theta null point of $(A, \pol, \Thetapol)$
a marked abelian variety isog-$f$-compatible with $(B, \bpol, \Thetabpol)$.
We keep the hypotheses and notation of Proposition~\ref{sec2:lemmalambdai}. 
Let $(A, \pol, \Thetapol,
 \theta_0^\Thetapol, \rho^\pol_{0_A})$ be a rigidified abelian variety with
 affine theta null point $\tildenullpol =
 (\theta_i^\Thetapol(0_A))_{i \in \Zn}$ such that for all $i \in
 \Zn$, there exists $\lambda_i \in \overk^*$ such that 
 \begin{equation}\label{sec2:lemmaeq2}
	\tildeg_f(i)= \lambda_i * (\theta^\Thetapol_{\mu_{m,n}(j)+i}(0_A))_{ j\in Z(m)}.
 \end{equation}

Let $\Bb=(e_i, e_i + e_j)$ be a chain basis of $\Zz(d)$. For $e \in \Bb$, we want to
compute $\lambda_e$. For this, let $\ell$ be
the smallest positive integer such that $\ell e \in \mu_{m,n}(Z(m))$.

We collect information on $\lambda_e$ by writing the relations satisfied by affine points.
All these relations arise from trivial relations among projective points, which we translate
into non-trivial relations between affine points. For instance, the equality of projective points
$\ell g_f(e)= \Thetabpolbar((\ell e, 0)). 0_B$, can be rewritten with
operations on affine points as
\begin{equation}\label{eq:affinelift1}
\scalmult(\ell, \lambda_e * \tildeg_f(e), \lambda_e * \tildeg_f(e), \tildenullbpol,
\tildenullbpol) = (1, \ell e, 0). \tildenullbpol.
\end{equation}
Following Lemma~\ref{lem:scalar}, we obtain an expression for $\lambda_e^{\ell^2}$.
We then need a relation that takes the symmetry relations into account. For this,
consider the equality of projective points $(\ell -1) g_f(e) + g_f(e)= \Thetabpolbar((\ell e,0 )).
	0_B$. We deduce from it the relation on affine points:
	\begin{equation}\label{eq:affinelift2}
	 \scalmult(\ell-1, \lambda_e * \tildeg_f(e), \lambda_e * \tildeg_f(e),
	 \tildenullbpol,
	 \tildenullbpol) = (1, \ell e,0). \inv(\lambda_e * \tildeg_f(e)).
	\end{equation}
	Because of Lemma~\ref{lem:scalar}, we get an expression for $\lambda_e^{(\ell-1)^2
	-1 }=\lambda_e^{\ell^2 - 2\ell}$, but as we already know $\lambda_e^{\ell^2}$ from
	Equation (\ref{eq:affinelift1}), we finally get $\lambda_e^{2\ell}$. Note that in
	the case where $\ell$ is odd, we in fact obtain $\lambda_e^\ell$ by taking the
	unique square root $t$ of $\lambda_e^{2\ell}$ such that $t^\ell=\lambda_e^{\ell^2}$.

\begin{remark}
 When $\ell$ is odd, we can get $\lambda_e^\ell$ without computing a square root
 in the following manner: we write $\ell = 2\ell'+1$, then from the equality of
	projective points $\ell' g_f(e) + (\ell'+1) g_f(e)=\Thetabpolbar((\ell e,
	0)).0_B$, we obtain the operations on affine points:
	\begin{equation}\label{eq:affineliftold}
	 \scalmult(\ell', \lambda_e * \tildeg_f(e), \lambda_e * \tildeg_f(e),
	 \tildenullbpol,
	 \tildenullbpol) = (1, \ell e, 0).\inv (\scalmult(\ell'+1, \lambda_e *
	 \tildeg_f(e), \lambda_e * \tildeg_f(e), \tildenullbpol,
	 \tildenullbpol)).
	\end{equation}
	Using Lemma~\ref{lem:scalar}, we obtain an expression for
	$\lambda_e^{\ell'^2- (\ell'+1)^2}$ and thus we get directly $\lambda_e^\ell$.
\end{remark}
As the $e_i$ form a basis of $\Zz(d)$, we have $\ell=d$ for $e=e_i$, $i=1, \ldots, g$. From the
knowledge of $\lambda_{e_i}, \lambda_{e_j}$ for
$i,j \in \{1, \ldots, g\}$, $i \ne j$, one can get additional information (which is needed when $d$
is even) on $\lambda_{e_i +
e_j}$ using the projective relation $g_f(\ell e_i + e_j) =\Thetabpolbar((\ell e_i, 0)).g_f(e_j)$. With
affine points, we get:
\begin{equation}\label{eq:affinelift3}
 \scalmult(\ell, \lambda_{e_i + e_j}* \tildeg_f(e_i + e_j), \lambda_{e_i}*\tildeg_f(e_i),
 \lambda_{e_j}*\tildeg_f(e_j), \tildenullbpol)= (1, \ell e_i, 0) (\lambda_{e_j}*
 \tildeg_f(e_j)).
\end{equation}

Using Lemma~\ref{lem:scalar}, we obtain an expression for $\frac{\lambda_{e_i +
e_j}^\ell \lambda_{e_i}^{\ell (\ell -1)}}{\lambda_{e_j}^{\ell}}$ and as we know
$\lambda_{e_i}^{\ell^2}$ from
Equation (\ref{eq:affinelift1}), we finally get an expression for $\frac{\lambda_{e_i +
e_j}^\ell}{\lambda_{e_i}^\ell \lambda_{e_j}^{\ell}}$.

\begin{definition}\label{def:good}
 Let $\Zz(d)$ be a set of representatives of classes of $Z(n)/Z(m)$ and $\Bb=(e_i,
 e_i+e_j)$ be a chain basis of $\Zz(d)$.
We say that:
 \begin{itemize}
	\item 
	 $\{\lambda_e * \tildeg_f(e), e \in \Bb \}$ is a good lift with respect to
	 $\tildenullbpol$ of $\{ g_f(e), e \in \Bb\}$ if
 for all $i, j \in \{1, \ldots, g\}$, $i \ne j$, $\lambda_{e_i}$ satisfies the relations
 (\ref{eq:affinelift1}), (\ref{eq:affinelift2}) and $\lambda_{e_i}, \lambda_{e_j},
 \lambda_{e_i + e_j}$ satisfy the relations (\ref{eq:affinelift3}).
 \item 
 $\{\lambda_e * \tildeg_f(e), e \in \Zz(d) \}$ is a good lift of $\{ g_f(e), e \in \Zz(d)
	 \}$ with respect to $\tildenullbpol$ if for all $e \in \Bb$, $\lambda_e * \tildeg_f(e)$ is a good lift of $g_f(e)$ and
	 all the $\tildeg_f(e)$ for $e\in \Zz(d)$ are computed from $\{\lambda_e * \tildeg_f(e), e \in \Bb \}$
 with the algorithm described in the proof of Proposition~\ref{sec2:proprecons}. 
\item $\tildeG= \{ \tildeg_f(e), e \in \Zn \}$ is a good lift of $G$ with respect to
 $\tildenullbpol$ if $\{\lambda_e * \tildeg_f(e), e
	 \in \Zz(d) \}$ is a good lift of $\{ g_f(e), e \in \Zz(d) \}$ and if, for all $e \in
	 \Zn$, if we write $e=e_d + e_m$ with $e_d \in \Zz(d)$ and $e_m \in \mu_{m,n}(Z(m))$, we
	 have $\tildeg_f(e)=\Thetabpol((1, e_m, 0))\tildeg_f(e_d)$.
 \end{itemize}
\end{definition}
\begin{remark}
 In order to prove that the algorithm described in the proof of Proposition
~\ref{sec2:proprecons} outputs the same result independently of the particular sequence
 of operations that we perform, we have used the hypothesis that $\{ \lambda_e *
 \tildeg_f(e), e \in \Zz(d) \}$ are excellent lifts. We see in Theorem~\ref{th:main1} that this
 hypothesis is actually always true up to an action of the metaplectic group on $(A,
 \pol, \Thetapol)$ for good lifts.
\end{remark}

Note that if $\tildeG= \{ \tildeg_f(e), e \in \Zn \}$ is a good lift of $G$ with respect to
 $\tildenullbpol$, we have in particular, for all $e_m \in Z(m)$,
 $\tildeg_f(e_m)= \Thetabpol((1, e_m, 0)) \tildenullbpol$.
The following lemma will be used below:
\begin{lemma}\label{lem:goodliftdiff}
 Let $\Bb = \{ e_i , e_i + e_j \}_{i,j= 1, \ldots, g}$ be a chain basis of $\Zz(d)$.
 If $\{ \tildeg_f(e), e \in \Bb \}$ and $\{ \tildeg'_f(e), e \in \Bb\}$ are two good
 lifts of $\{ g_f(e), e \in \Bb \}$,
 then for $i=1, \ldots, g$, there exists an $\ell$-th root of unity $\mu_{e_i}$, with $\ell=d$ if $d$ is odd and
 $\ell=2d$ if $d$ is even such that:
 \begin{equation}
	\tildeg_f(e_i)=\mu_{e_i} * \tildeg'_f(e_i).
 \end{equation}
 Moreover, for all $i,j =1, \ldots, g$,
 $i\ne j$, we have that:
 \begin{equation}\label{eq:diffgoodlift}
	\frac{\mu_{e_i+ e_j}}{\mu_{e_i} \mu_{e_j}}
 \end{equation}
 is a $d$-th root of unity.
 
\end{lemma}
\begin{proof}
 Write $\tildeg_f(e)=\lambda_e*\tildeg'_f(e)$ for $e \in \Bb$, so that $\mu_e=\lambda_e$. By Equations
 (\ref{eq:affinelift1}) and (\ref{eq:affinelift2}) (with $\ell=d$) and Lemma~\ref{lem:scalar}, we have
 $\mu_{e_i}^{d^2}=\mu_{e_i}^{2d}=1$, hence $\mu_{e_i}^{\ell}=1$ with $\ell=\gcd(d^2, 2d)$, which is $d$ if
 $d$ is odd and $2d$ if $d$ is even. For $i \ne j$, Equation (\ref{eq:affinelift3}) and Lemma
~\ref{lem:scalar} give similarly $\big(\frac{\mu_{e_i + e_j}}{\mu_{e_i} \mu_{e_j}}\big)^d=1$.
\end{proof}

\begin{remark}
 If $\tildeG=\{ \tildeg_f(e), e \in \Zn \}$ is a good lift of $G=\{ g_f(e), e \in \Zn \}$,
 then $\tildeG$ is an excellent lift by Theorem~\ref{th:main1} below, so that
 $\inv(\tildeg_f(e))=\tildeg_f(-e)$ for all $e \in \Zn$ by Proposition~\ref{sec2:lemm6} (2). In
 particular, $\inv(\tildeG)=\tildeG$, although Definition~\ref{def:good} imposes the symmetry
 relation (\ref{eq:affinelift2}) only to the $e_i$.
\end{remark}

We now compare the notions of good lift and excellent lift from
Definitions~\ref{def:good} and~\ref{def:excel}. Keeping the hypotheses of Definition~\ref{sec1:def1}, the obvious direction holds: excellent lifts are good lifts.
\begin{proposition}\label{prop:excelgood}
 Suppose that $\tildeG = \{ \tildeg_f(i), i \in \Zn \}$ is an excellent lift of $G$,
 then for $i \in \Zz(d)$, $\tildeg_f(i)$ is a good lift of $g_f(i)$. 
\end{proposition}
\begin{proof}
 The relations
 defining good lifts are given by composing Riemann and symmetry relations and the
 action of the theta group $G(\bpol)$, which are all satisfied by excellent lifts because
 of Proposition~\ref{sec2:lemm6}.
\end{proof}

Proposition~\ref{prop:excelgood} shows that excellent lifts are good lifts. The following
theorem, one of the main results of this section, shows that conversely good lifts are
excellent lifts up to a change of the isog-$f$-compatible theta structure of $(A, \pol,
\Thetapol)$.
\begin{theorem}\label{th:main1}
Let $(B, \bpol, \Thetabpol, \theta_0^\Thetabpol, \rho^\bpol_{0_{\Thetabpol}})$ be a rigidified
 abelian variety of type $K(m)$.
 Let $G \subset B[n]$ be a subgroup of $B[n]$ isomorphic to $Z(n)$
 containing $\Thetabpolbar(Z(m)\times\{0\})$, isotropic for $e_{B,n}$ and such that for all
$x \in G(\overk)$, $x$ is symmetric compatible with $\Thetabpol( \{1\}\times Z(m) \times\{0\}
)$. We choose a numbering of the elements of $G$ by writing $G=\{g_f(i), i \in \Zn\}$
such that the map $i \mapsto g_f(i)$ is a group morphism and that for all $i \in Z(m)$,
$g_f(\mu_{m,n}(i))=\Thetabpolbar((i,0))$. Suppose that $\tildeG= \{ \tildeg_f(i), i \in \Zn
 \}$ is a good lift of $G$. Then there exists $(A, \pol, \Thetapol,
 \theta_0^\Thetapol, \rho^\pol_{0_A})$ a rigidified abelian variety such that for all $i \in
 \Zn$: 
 \begin{equation}\label{sec2:lemmaeq2main1}
	\tildeg_f(i)= (\theta^\Thetapol_{\mu_{m,n}(j)+i}(0_A))_{ j\in Z(m)}.
 \end{equation}
 Equivalently, $\tildeG$ is an excellent lift of $G$ with respect to $\tildenullbpol$.
\end{theorem}
\begin{proof}
Denote by $\zeta$ a primitive $n$-th root of unity 
 and let $\Bb_0 = \{ e_k, \hate_{k} \}$ be a $\zeta$-symplectic basis of $K(n)=Z(n) \times \dZn$
 such that for $k=1, \ldots, g$,
 $\hate_k(e_k)=\zeta$.
 Denote by $\Zz(d)$ a set of representatives of classes of $Z(n)/Z(m)$ containing $e_k$
 for $k=1, \ldots, g$. Let $\Bb=(e_i, e_i
 + e_j)$ be a chain basis of $\Zz(d)$. For $i=1, \ldots, g$, we let $e_{ii}=e_i$ and for
 $i,j=1, \ldots, g$, $i\ne j$, we set $e_{ij}=e_i + e_j$.
 Using Proposition~\ref{sec2:lemmalambdai}, we can suppose that for all $e_{ij} \in \Bb$,
 there exists $\lambda_{e_{ij}} \in \overk^*$ such that
 \begin{equation}
	\tildeg_f(e_{ij})= \lambda_{e_{ij}} *
	(\theta^\Thetapol_{\mu_{m,n}(\kappa)+e_{ij}}(0_A))_{ \kappa\in Z(m)}.
 \end{equation}
 By Proposition~\ref{prop:excelgood}, we know that for all $e_{ij} \in \Bb$,
 $(\theta^\Thetapol_{\mu_{m,n}(\kappa)+e_{ij}}(0_A))_{ \kappa \in Z(m)}$ is a good lift of
 $\tildeg_f(e_{ij})$ with respect to $\tildenullbpol$. Using Lemma
~\ref{lem:goodliftdiff}, it means
 that $\lambda_{e_{ii}}$ is an $\ell$-th root of unity with $\ell=d$ if $d$ is odd and
 $\ell=2d$ if $d$ is even and, if $i \ne j$, $\lambda_{e_{ij}}/(\lambda_{e_{ii}}\lambda_{e_{jj}})$ is a $d$-th root of
 unity. 

 We consider a morphism $\gamma_C \colon \Zn \rightarrow \dZn$ given in the bases $\{e_1,
 \ldots, e_g\}$ and $\{\hate_1, \ldots, \hate_g\}$ by a symmetric matrix $C = (c_{ij})_{1
 \le i,j \le g} \in M_g(\Z / n \Z)$. 
Our goal is to choose the entries of $C$ so that, by applying
Proposition~\ref{prop:action2}(3), we can choose $\gamma \in \Psi^{-1}(\psimbz(S_g(C)))$
satisfying: \begin{equation} \tildeg_f(e_{ij}) = \big( \theta^{\Thetapol\circ
    \gamma}_{\mu_{m,n}(\kappa)+e_{ij}}(0_A) \big)_{\kappa\in Z(m)}, \end{equation} where
    $\psimbz(M)$ denotes the element of $\Sp(K(n))$ whose matrix in the basis $\Bb_0$ is $M$ (as in the paragraph just before Proposition~\ref{prop:trans} for the canonical basis). 
Then, we conclude by the very definition of a good lift.

 For all $k=1, \ldots, g$, $\gamma_C(e_k)(e_k)= \zeta^{c_{kk}}$ and we choose $c_{kk} \in
 \Z / n \Z$ so
 that $\zeta^{c_{kk}}=\lambda^{-2}_{e_{kk}}$. As $\lambda_{e_{kk}}$ is an
 $\ell$-th root of unity, we have that $c_{kk} \in m \Z / n \Z$, and moreover
 $\gamma_C(e_k)(e_k)^{-1/2}=\lambda_{e_{kk}}$: when $d$ is odd the square root is
 unique and when $d$ is even we choose the sign of the square root to obtain the
 equality which is allowed by Proposition~\ref{prop:trans} (2).

 Next, we have $\gamma_C(e_i + e_j)(e_i + e_j)=\zeta^{c_{ii} + c_{ij} + c_{ji} + c_{jj}}$
 and we choose $c_{ij}=c_{ji} \in \Z / n \Z$ so that $\zeta^{2 c_{ij}}=\lambda_{e_{ij}}^{-2} \zeta^{-
	(c_{ii} + c_{jj})}$. Thus, $\zeta^{2 c_{ij}}=
	\left(\frac{\lambda_{e_{ij}}}{\lambda_{e_{ii}} \lambda_{e_{jj}}}\right)^{-2}$. By
	Lemma~\ref{lem:goodliftdiff}, we know that
	$\frac{\lambda_{e_{ij}}}{\lambda_{e_{ii}} \lambda_{e_{jj}}}$ is a $d$-th root of
	 unity. It means that $\zeta^{c_{ij}}$ is a $d$-th root of unity: if $d$ is
	 odd, the square and the square root of a $d$-th root of unity are $d$-th roots
	 of unity, and if $d$ is even, it is because $\zeta^{2 c_{ij}}$ is a
	 $(d/2)$-th root of unity. Thus we have
	 $c_{ij} \in m \Z / n \Z$. Moreover, $\gamma_C(e_i + e_j)(e_i +
	 e_j)^{-1/2}= \lambda_{e_{ij}}$ where when $\ell$ is even we have
	 chosen the sign of $(\zeta^{2 c_{ij}})^{-1/2}$ so that $\gamma_C(e_i + e_j)(e_i
	 + e_j)^{-1/2}=\gamma_C(e_i)(e_i)^{-1/2} \gamma_C(e_j)(e_j)^{-1/2}
	 \gamma_C(e_i)(e_j)= \zeta^{-1/2(c_{ii} + 2 c_{ij} + c_{jj})}$, thus satisfying
	 condition (\ref{eq:choicesign}) of Proposition~\ref{prop:trans}.
\end{proof}

\begin{corollary}
Suppose that we are given $(B, \bpol, \Thetabpol, \theta_0^\Thetabpol, \rho^\bpol_{0_{B}})$ a rigidified
 abelian variety of type $K(m)$ up to equivalence by its affine theta null point
 $\tildenullbpol$.
 Let $G \subset B[n]$ be a subgroup of $B[n]$ isomorphic to $Z(n)$
 containing $\Thetabpolbar(Z(m)\times\{0\})$ and isotropic for $e_{B,n}$. There exists an
 algorithm whose running time is $O(n^g \log(d))$ which outputs:
 \begin{itemize}
	\item $\widetilde{0}_{\Thetapbpol}$ the theta null point of $(B, \bpol, \Thetapbpol,
	 \theta_0^\Thetabpol, \rho^\bpol_{0_{B}})$ a rigidified abelian variety
	 such that for all $x \in G(\overk)$, $x$ is symmetric compatible with
	 $\Thetapbpol(\{1\}\times Z(m) \times \{0\})$;
	\item $\tildenullpol$ the theta null point of $(A, \pol, \Thetapol,
	    \theta_0^\Thetapol, \rho^\pol_{0_A})$ which is
isog-$f$-compatible.
 \end{itemize}
\end{corollary}
\begin{proof}
 This is an immediate consequence of Proposition~\ref{prop:changesym1} and Theorem~\ref{th:main1}.
\end{proof}

The preceding theorem gives an algorithm to compute a theta null point of type $K(n)$
of an abelian variety which is isog-$f$-compatible with $(B, \bpol, \Thetabpol,
\theta_0^\Thetabpol, \rho^\bpol_{0_{B}})$ a rigidified
 abelian variety of type $K(m)$. Thus, it gives us a way to compute a point
 of the fiber of the map $\pi^0_{n,m}: \Modu_n \rightarrow \Modu_m$. The theorem tells us
 that the choices of the
 roots of unity that appear in the algorithm correspond to a choice of a point in
 the fiber.

 Our goal is now to complete the picture by computing the fiber of the full map
 $\pi_{n,m}: \univAb_n \rightarrow \univAb_m$,
 i.e., to compute the isogeny $f: A\rightarrow B$. For that, since $\tildenullbpol$ and
 $\tildenullpol$ are known,
 we explain how to compute a point in the fiber $f^{-1}(z_0)$ for $z_0 \in B(\overk)$.
 Here too, choices of roots of unity arise in the course of the algorithm,
 and we explain how they correspond to a choice of point in $f^{-1}(z_0)$.

\begin{definition}\label{def:excelz}
 We keep the notation of Definition~\ref{def:excel}, let $(A, \pol, \Thetapol,
 \theta_0^\Thetapol, \rho^\pol_{0_A})$ be as in that definition, let $z \in A(\overk)$, $z_0=f(z)$, and
 let $\tildez$ be the affine lift of $z$ given by a rigidification $\rho^\pol_z$ of $\pol$ at $z$.
 Let $g^z_f: \Zn \rightarrow G(z)$, $i \mapsto z_0 + g_f(i)$. We say that $\tildeG(z)= \{
 \tildeg^z_f(i), i \in \Zn \}$ is an excellent lift of $G(z)$ if there exists $\lambda \in
 \overk^*$ (independent of $i$) such that for all $i \in \Zn$:
 \begin{equation}\label{sec2:excelz}
	\lambda *\tildeg^z_f(i)= \tildef(\Thetapol((1,i,0)) \tildez).
 \end{equation}
 For $z=0_A$ and $\tildez=\tildenullpol$, we recover Definition~\ref{def:excel} (with $\lambda=1$).
\end{definition}
Note that the choice of $\rho^\pol_z$ in the preceding definition will only affect the
global factor $\lambda$ in (\ref{sec2:excelz}).
We have seen in Lemma~\ref{sec1:lemact} that given a rigidification $\rho^\pol_z$ 
and $g' \in \Gpol$, we can
deduce a rigidification $\rho^\pol_{z'}$ for any $z' = z + \pigpol(g')$. Thus, we can define
an excellent lift $\tildeG(z')$ of $G(z')$ for any $z' \in z + \Kpol$. The following lemma
explains that $\tildeG(z')$ does not depend on the choice of $g'$ such that $\pigpol(g') = z'-z$ and
gives a way to compute $\tildeG(z')$ from the knowledge of $\tildeG(z)$.

\begin{lemma}\label{lem:excelchange}
 Let $z \in A(\overk)$, choose $\rho_z^\pol$ a rigidification of $\pol$ at $z$ and let 
$\tildeG(z)= \{
 \tildeg^z_f(i), i \in \Zn \}$ be an excellent lift of $G(z)$. Let $g' \in \Gpol$,
let $(\lambda', \alpha, \beta) \in G(n)$ such that $g' = \Thetapol((\lambda', \alpha,
 \beta))$. Let $z'=z + \pigpol(g')$ and let $\rho_{z'}^\pol$ be the rigidification of $\pol$ at $z'$
 such that $\rho_{z'}^\pol(s)=\rho_z^\pol(g'.s)$ for $s \in \Gamma(A, \pol)$ (Lemma~\ref{sec1:lemact}). If $s
 \in \Gamma(A,\pol)$, we let $s(z')= \rho_{z'}^\pol(s)$.
 Finally, let $\tildeG(z')= \{ \tildeg^{z'}_f(i), i \in \Zn
 \}$ be an excellent lift of $G(z')$. There exists $\lambda \in
 \overk^*$ such that, for all $i \in \Zn$:
 \begin{equation}
	\tildeg^{z'}_f(i) = \lambda \beta(- i)\, \Thetabpol((1,0,\drho_{n,m}(\beta))).\tildeg_f^{z}(i+\alpha),
 \end{equation}
 where $\drho_{n,m}(\beta) \in \dZ(m)$ is the restriction of $\beta$ to $\mu_{m,n}(Z(m))$. In particular,
 if $\beta \in \dnu_{d,n}(\dZ(d))$, that is if $\beta$ is trivial on $\mu_{m,n}(Z(m))$, then
 $\tildeg^{z'}_f(i) = \lambda \beta(- i) \tildeg_f^{z}(i+\alpha)$.
\end{lemma}
\begin{proof}
 For all $i \in \Zn$, we have, by definition and by (\ref{eq:trans}), $\theta^\Thetapol_i(z')= ((\lambda',
 \alpha, \beta). \theta^\Thetapol_i)(z)=\lambda' \beta(-\alpha-i) \theta^\Thetapol_{i+\alpha}(z)$. Let
 $\lambda_z, \lambda_{z'} \in \overk^*$ be such that for all $i \in \Zn$, $\tildeg^{z}_f(i)= \lambda_{z}
 *(\theta^\Thetapol_{\mu_{m,n}(j)+i}(z))_{ j\in Z(m)}$ and $\tildeg^{z'}_f(i)= \lambda_{z'}
 *(\theta^\Thetapol_{\mu_{m,n}(j)+i}(z'))_{ j\in Z(m)}$. For $i \in \Zn$ and $j \in Z(m)$, the coordinate of index
 $j$ of $\tildeg^{z'}_f(i)$ is
 \begin{equation}
	\lambda_{z'}\lambda'\beta(-\alpha-\mu_{m,n}(j)-i)\theta^\Thetapol_{\mu_{m,n}(j)+i+\alpha}(z)=
	\lambda \beta(-i)\,\beta(-\mu_{m,n}(j))\, (\tildeg^{z}_f(i+\alpha))_j,
 \end{equation}
 where $\lambda=\lambda_{z'}\lambda'\beta(-\alpha)/\lambda_z$. As $\beta(-\mu_{m,n}(j))=\drho_{n,m}(\beta)(-j)$ and,
 by (\ref{eq:trans}) at level $m$, $\Thetabpol((1,0,y))$ multiplies the coordinate of index $j$ of an
 affine point by $y(-j)$, we get the result.
\end{proof}

For $e \in \Zn$, let $\ell_e$ be the smallest integer such that $\ell_e g_f(e) \in
K(\bpol)$. As before, we exploit trivial relations among projective
points to obtain non-trivial conditions on affine lifts. For instance, the relation
$\ell_e g_f(e) + z_0 = \Thetabpolbar((\ell_e e,0)).z_0$, taking into account that $g_f^z(e)=g_f(e) +
z_0$, can be rewritten with operations on affine points as:
\begin{equation}\label{eq:affineliftz}
 \scalmult(\ell_e, \lambda_e * \tildeg^z_f(e), \tildeg_f(e), \tildez_0,
 \tildenullbpol) = (1, \ell_e e, 0). \tildez_0.
\end{equation}
Using Lemma~\ref{lem:scalar}, we obtain an expression for $\lambda_e^{\ell_e}$ so that $\lambda_e$ is known
up to an $\ell_e$-th root of unity.

This motivates the following definition of a good lift of $G(z)$.
\begin{definition}\label{def:goodz}
 We keep the notation of Definition~\ref{def:excelz}, let $\tildeG=\{\tildeg_f(e), e \in \Zn\}$ be a
 good lift of $G$ with respect to $\tildenullbpol$ (Definition~\ref{def:good}) and $\tildez_0$ an
 affine lift of $z_0$. Let $\Zz(d)$ be a set of representatives of classes of $Z(n)/Z(m)$ and
 $\Bb=(e_i)_{i=1, \ldots, g}$ a basis of $\Zz(d)$. We say that $\tildeG(z)= \{ \tildeg^z_f(e), e \in
 \Zn \}$ is a good lift of $G(z)$ with respect to $\tildez_0$ and $\tildeG$ if:
 \begin{itemize}
	\item for all $e \in \Bb$, $\tildeg^z_f(e)$ satisfies relation (\ref{eq:affineliftz}) with
	 $\lambda_e=1$;
	\item the $\tildeg^z_f(e)$, $e\in \Zz(d)$, are computed from the $\tildeg^z_f(e)$, $e \in \Bb$,
	 with the algorithm described in the proof of Proposition~\ref{sec2:proprecons};
	\item for all $e=e_d + e_m$ with $e_d \in \Zz(d)$ and $e_m \in \mu_{m,n}(Z(m))$,
	 $\tildeg^z_f(e)=\Thetabpol((1, e_m, 0))\tildeg^z_f(e_d)$.
 \end{itemize}
\end{definition}

From Definitions~\ref{def:good} and~\ref{def:goodz}, we immediately deduce Algorithm~\ref{algo:goodlift} for computing a good lift
of $G \subset B[n]$ and Algorithm~\ref{algo:goodliftg} for computing a good lift of $x + G$
for $x \in B(\overk)$.

\begin{algorithm}
\SetKwInOut{Input}{input}\SetKwInOut{Output}{output}

\SetKwComment{Comment}{/* }{ */}
\Input{
 \begin{itemize}
	\item the marked abelian variety $(B, \bpol, \Thetabpol)$ of type $K(m)$ given by its theta null
	 point $0_{\Thetabpol}$;
	\item $G \subset B[n]$ a subgroup isomorphic to $Z(n)$ such that:
	 \begin{itemize}
		\item $G$ is isotropic for $e_{B,n}$;
		\item $\Thetabpolbar(Z(m) \times \{0\}) \subset G$;
		\item for all $P \in G$, $P$ is symmetric compatible with
		 $\Thetabpol(\{1\}\times Z(m)\times\{0\})$.
	 \end{itemize}
	\item $\Bb=(e_i, e_i + e_j)$ a chain basis of $\Zz(d)$.
 \end{itemize}
}
\Output{
 \begin{itemize}
	\item $\tGBb=(\tgfei, \tgfeij)$ a good lift of $(\gfei, \gfeij)$ with respect to $\tildenullbpol$.
 \end{itemize}
}
\BlankLine
\For{$i \in \{1, \ldots, g\}$}
{
 Let $\ell_{e_i} = \min \{ \ell \in \N^* | \ell \gfei \in K(\bpol) \}$\;
 Let $e \in K(m)$ be such that $\ell_{e_i} \gfei = \Thetabpolbar(e)$\;
Fix $\tgfei$ arbitrary affine lift of $\gfei$\;
Compute $\lambda_{e_i}$ such that:\nonl \\
 \begin{minipage}{13cm}
 \begin{itemize}
	\item $\scalmult(\ell_{e_i}, \lambda_{e_i}* \tgfei, \lambda_{e_i} * \tgfei, \tildenullbpol,
\tildenullbpol) = (1, e). \tildenullbpol$;
\item $\scalmult(\ell_{e_i}-1, \lambda_{e_i} * \tgfei, \lambda_{e_i} * \tgfei,
	 \tildenullbpol,
	 \tildenullbpol) = (1, e). \inv(\lambda_{e_i} * \tgfei)$;
 \end{itemize}
 \end{minipage} \\
 \, \\
}

\For{$i,j \in \{1, \ldots, g\}$, $i > j$}
{
 Fix $\tgfeij$ arbitrary affine lift of $\gfeij$\;
Compute $\lambda_{e_i + e_j}$ such that: 
$\scalmult(\ell_{e_i}, \lambda_{e_i + e_j}* \tgfeij, \lambda_{e_i}*\tildeg_f(e_i),
 \lambda_{e_j}*\tildeg_f(e_j), \tildenullbpol)= (1, \ell_{e_i} e_i, 0). (\lambda_{e_j}*
 \tildeg_f(e_j))$\;

}
\Return $\tGBb=(\lambda_{e_i} * \tgfei, \lambda_{e_i + e_j} * \tgfeij)$.
\caption{Algorithm to compute a good lift of a chain basis of $G \subset B[n]$.}
 \label{algo:goodlift}
\end{algorithm}

\begin{algorithm}
\SetKwInOut{Input}{input}\SetKwInOut{Output}{output}

\SetKwComment{Comment}{/* }{ */}
\Input{
 \begin{itemize}
	\item the marked abelian variety $(B, \bpol, \Thetabpol)$ of type $K(m)$ given by its theta null
	 point $0_{\Thetabpol}$;
	\item $G \subset B[n]$ a subgroup isomorphic to $Z(n)$ such that:
	 \begin{itemize}
		\item $G$ is isotropic for $e_{B,n}$;
		\item $\Thetabpolbar(Z(m) \times \{0\}) \subset G$;
		\item for all $P \in G$, $P$ is symmetric compatible with
		 $\Thetabpol(\{1\}\times Z(m)\times\{0\})$.
	 \end{itemize}
	\item $z \in B(\overk)$.
 \end{itemize}
}
\Output{
 \begin{itemize}
 \item $\tildeG$ a good lift of $G$ with respect to $\tildenullbpol$;
	\item $\widetilde{z + G}$ a good lift of $z+G$ with respect to $\tildeG$.
 \end{itemize}
}
\BlankLine
 Fix a basis $\Bb_0=(e_1, \ldots, e_g)$ of $\Zz(d)$, compute $\Bb=(e_i, e_i + e_j)$ a chain
 basis of $\Zz(d)$\;
 Compute $\tGBb = (\tgfei, \tgfeij)$ using Algorithm~\ref{algo:goodlift}\;
Fix $\tildez$ an arbitrary affine lift of $z$\;
 \For{$i \in \{1, \ldots, g\}$}
	{
 Let $\ell_{e_i} = \min \{ \ell \in \N^* | \ell \gfei \in K(\bpol) \}$\;
 Let $e \in K(m)$ be such that $\ell_{e_i} \gfei = \Thetabpolbar(e)$\;
Fix $\tgzfei$ an arbitrary affine lift of $\gzfei$\;
Compute $\lambdazei$ such that: $\scalmult(\ell_{e_i}, \lambdazei * \tgzfei,
\tgfei, \tildez, \tildenullbpol) = (1, e, 0). \tildez$\;
\For{$j \in \{1, \ldots, g\}$, $i > j$}
	{
Compute $\tgzfeij$ as: $\tgzfeij = \threeway(\lambdazei * \tgzfei,
\tgfeij, \lambdaz_{e_j} * \tgzfej, \tildez, \tgfei, \tgfej, \tildenullbpol)$\;
 }
}
From $\tGBb$ and $(\lambdazei * \tgzfei, \tgzfeij)_{i,j=1, \ldots, g}$,
compute $\tildeG$ and $\widetilde{z+G}$ using $\scalmult$, $\threeway$ and
the action of $\Thetabpol$ on affine points\;

 \Return $\tildeG, \widetilde{z+G}$.
 \caption{Algorithm to compute a good lift of $G$ and $z+G$ for $G \subset B[n]$.}
 \label{algo:goodliftg}
\end{algorithm}

As before, we show that excellent lifts are good lifts:
\begin{lemma}\label{lem:goodexcel2}
 Keeping the hypotheses of Definition~\ref{def:excelz}, if $\tildeG(z)$ is an excellent lift of
 $G(z)$ then it is also a good lift.
\end{lemma}
\begin{proof}
We would like to show that excellent lifts are good lifts. For this, we need to show that
$\tildeG(z)=\{ \tildeg_f^z(i), i \in \Zn \}$ satisfies the Riemann equations, and we cannot
use Proposition~\ref{sec2:lemm6} which only deals with $\tildeG$.

Let $\rho^\pol_z$ be a rigidification of $\pol$ at $z$.
 In the following, for $s \in \Gamma(A,\pol)$, we let $s(0_A)=\rho^\pol_{0_A}(s)$ and
 $s(z)=\rho^\pol_z(s)$. Let $(i_1, \ldots, i_4; i_5, \ldots, i_8)$ be elements of
 $Z(n)$ in Riemann position. By Theorem~\ref{sec:thriemann}, we have:
 \begin{equation}\label{eq:riemannexcel}
	L(\Thetapol, \chi,i_1, i_2, z, z) L(\Thetapol, \chi, i_3, i_4, 0_A,
	0_A) =L(\Thetapol, \chi,i_5, i_6, z, z)
	L(\Thetapol, \chi, i_7, i_8, 0_A, 0_A).
 \end{equation}
We remark that the preceding relation, for homogeneity reasons, does not depend on the
 choice of $\rho^\pol_z$. 

 Let $\tildeG= \{ \tildeg_f(i), i \in \Zn \}$ be an excellent lift of $G$. We let
 $\Aff^{Z(m)}=\Spec(k[x_i, i \in Z(m)])$
 so that for $i \in Z(m)$, $x_i$ is the $i$-th coordinate function.
	 Let $\vx= (y_1, \ldots, y_4; y_5, \ldots, y_8) \in \Zn^8$ and $\vi=(i_1, \ldots, i_4; i_5, \ldots,
i_8) \in Z(m)^8$ be elements in Riemann position. Then we have a Riemann equation:
 \begin{equation}\label{eq:riemaneq2proof}
	\sum_{\eta \in Z(2)} \prod_{j=1}^2 x_{i_j +
	\eta}(\tildeg^z_f(y_j)) \prod_{j=3}^4 x_{i_j +
	\eta}(\tildeg_f(y_j)) = \sum_{\eta \in Z(2)} \prod_{j=5}^6 x_{i_j +
	\eta}(\tildeg^z_f(y_j))\prod_{j=7}^8 x_{i_j +
	\eta}(\tildeg_f(y_j)).
 \end{equation}
These relations are enough to be able to compute $\scalmult(\ell_e, \lambda_e * \tildeg^z_f(e), \tildeg_f(e), \tildez_0,
 \tildenullbpol)$ which is used to define a good lift. We can proceed in the
 same manner to obtain the relations used for $\threeway$.
\end{proof}

We can state the second main result of this section:
\begin{theorem}
 We suppose that $(B, \bpol, \Thetabpol, \theta_0^\Thetabpol, \rho^\bpol_{0_{\Thetabpol}})$
 and $(A, \pol, \Thetapol,
 \theta_0^\Thetapol, \rho^\pol_{0_\Thetapol})$ are isog-$f$-compatible rigidified
 abelian varieties with respective affine theta null points $\tildenullbpol$ and
 $\tildenullpol$. Let $z \in A(\overk)$ and $z_0 =f(z)$. Consider $g_f: \Zn
 \rightarrow G(0_A)$ a group morphism such
 that there exists
 $\tildeG(0_A)=\{ \tildeg_f(i), i \in \Zn \}$ an excellent
 lift of $G(0_A)$ (see Definition~\ref{def:excel}). This group morphism exists and
 is unique by Proposition~\ref{sec3:prop9}.
 Consider the map $g^z_f: \Zn \rightarrow G(z)$, $i \mapsto z_0 + g_f(i)$. Let
 $\tildeG'(z)=\{ \tildeg^{\prime z}_f(i), i \in \Zn \}$ be a good lift of $G(z)$. Then there
 exists $z' \in f^{-1}(z_0)$ such that $\tildeG'(z)$ is an excellent lift of $G(z')$ (with respect to $z'$). In
 other words, there exists $\lambda \in \overk^*$ (independent of $i$) such that:
 \begin{equation}\label{sec2:excelzi}
	\tildeg^{\prime z}_f(i)= \lambda *(\theta^\Thetapol_{\mu_{m,n}(j)+i}(z'))_{ j\in Z(m)}.
 \end{equation}
\end{theorem}
\begin{proof}
 Let $\Bb=(e_i)_{i=1, \ldots, g}$ be a basis of $\Zz(d)$. Let
 $\tildeG(z)=\{ \tildeg^z_f(i), i \in \Zn \}$ be an excellent lift of $G(z)$. By Lemma
~\ref{lem:goodexcel2}, $\tildeG(z)$ is a good lift of $G(z)$. So for $i=1, \ldots, g$,
 there exist a $d$-th root of unity $\zeta_i$ and $\lambda_0 \in \overk^*$ such that:
 \begin{equation}
	\lambda_0 * \tildeg^{\prime z}_f(e_i)= \zeta_i* \tildeg^{z}_f(e_i).
 \end{equation}
 As the classes of $e_1, \ldots, e_g$ form a basis of $\Zn/\mu_{m,n}(Z(m)) \simeq Z(d)$ and the
 $\zeta_i$ are $d$-th roots of unity, there exists a character $\beta \in \dZn$, trivial on
 $\mu_{m,n}(Z(m))$ (that is $\beta \in \dnu_{d,n}(\dZ(d))$), such that $\beta(-e_i) = \zeta_i$ for $i=1,\ldots,g$.
 Let $z' = z+\Thetapolbar((0, \beta))$; as $\Thetapolbar((0,\beta)) \in \Thetapolbar(\{0\}\times\dnu_{d,n}(\dZ(d)))=\Ker f$,
 we have $z' \in f^{-1}(z_0)$. Let
 $\tildeG(z')= \{ \tildeg^{z'}_f(i), i \in \Zn
 \}$ be an excellent lift of $G(z')$.
 By Lemma~\ref{lem:excelchange} (with $\alpha=0$ and $\beta$ trivial on $\mu_{m,n}(Z(m))$), there exists
 $\lambda_1 \in \overk^*$ such that, for $i
 =1, \ldots, g$:
 \begin{equation}
	\tildeg^{z'}_f(e_i) = \lambda_1 \zeta_i *\tildeg_f^{z}(e_i) = \lambda_1 \lambda_0
*	\tildeg^{\prime z}_f(e_i).
 \end{equation}
 By definition of a good lift, it means that for all $i \in \Zn$, we have:
 \begin{equation}
	\tildeg^{z'}_f(i) = \lambda_1 \lambda_0 * \tildeg^{\prime z}_f(i).
 \end{equation}
 By setting $\lambda=\lambda_1 \lambda_0$ in the preceding expression, we obtain exactly (\ref{sec2:excelzi}).
\end{proof}

\section{Change of level algorithms and isogeny computation}\label{sec:main}
This section addresses two closely related problems: change of level
algorithms and isogeny computation algorithms. 
Let $m,n,d>1$ be integers such that $n=md$.
A change of level algorithm going up in level takes as input the theta null point of
$(B, \bpol, \Thetabpol)$ a marked abelian variety of type $K(m)$ and $B[n]$, and outputs the theta
null point of $(B, \bpol^d, \Thetabpold)$, a marked abelian variety of type $K(n)$.
Conversely, a change of level algorithm going
down in level takes as input the theta null point of
$(B, \bpol^d, \Thetabpold)$ a marked abelian variety of type $K(n)$ and outputs the theta
null point of $(B, \bpol,
\Thetabpol)$ a marked abelian variety of type $K(m)$. In addition, an isogeny computation algorithm takes
as input the theta null point of $(B, \bpol, \Thetabpol)$ a marked abelian variety of type
$K(m)$, and $K \subset B[d]$ an isotropic subgroup for $e_{B,d}$ isomorphic to $Z(d)$,
and computes the theta null point of $(A, \pol, \Thetapol)$ a marked abelian variety of
type $K(m)$, where $A=B/K$, and the isogeny $f: B \rightarrow A$ (note that, in Theorem
\ref{th:mainisog} and Algorithm~\ref{algo:isogcomp}, contrary to the rest of the paper, $f$ goes
from $B$ to $A$). The case $d$ odd and coprime to
$m$ was treated in \cite{lubicz:hal-03738315}.

We begin by making precise the relationship between the theta structures
$\Thetabpol$ and $\Thetabpold$ for a change of level algorithm between $(B, \bpol, \Thetabpol)$
and $(B, \bpol^d, \Thetabpold)$ (whatever the direction of the change of
level). In fact, if we look at the simpler case of symplectic structures,
$\Thetabpolbar: K(m) \rightarrow K(\bpol)$ and $\Thetabpoldbar: K(n) \rightarrow K(\bpol^d)$,
as $K(\bpol^d) = \{ x \in B(\overk), dx \in K(\bpol) \}$, it is clear that
$K(\bpol) \subset K(\bpol^d)$, and we would like the symplectic structures
to be compatible in the sense that for all $(x,y) \in Z(m) \times \dZ(m)$,
$\Thetabpoldbar((\mu_{m,n}(x),0)) =\Thetabpolbar((x,0))$ and $\Thetabpoldbar((0,\dnu_{m,n}(y)))
=\Thetabpolbar((0,y))$. In Mumford's work, one finds an analogous notion of compatible theta structures, capturing the idea that going down in level amounts to forgetting part of the information available at the higher level. We recall it briefly.

\begin{definition}\label{def:pid}
Let $(B, \bpol)$ be a $g$-dimensional abelian variety together with a level $m$ ample
symmetric line bundle. As $\bpol$ is symmetric, there is an isomorphism
$\psi_d(\bpol^{d^2}): \bpol^{d^2} \rightarrow [d]^*(\bpol)$. 
Let $\tildeK_d$ be the level subgroup of $G([d]^*(\bpol))$ above $B[d]=\Ker([d])$ associated to
the descent data of $[d]^*(\bpol)$ to $\bpol$, let $G^*([d]^*(\bpol))$ be its centralizer in
$G([d]^*(\bpol))$ and let $[d]^\sharp: G^*([d]^*(\bpol)) \rightarrow G(\bpol)$ be the quotient by
$\tildeK_d$. Let $G^*(\bpol^{d^2}) \subset G(\bpol^{d^2})$ be the subgroup corresponding to
$G^*([d]^*(\bpol))$ under the isomorphism $\psi_d(\bpol^{d^2})$.

Following Definition~\ref{def:fsharp}, we define the group morphism:
\begin{equation}
 \begin{split}
	\pi_d(\bpol^{d^2}): G^*(\bpol^{d^2}) & \rightarrow G(\bpol) \\
	(\tau_x, \psi_x ) & \mapsto [d]^{\sharp} ( (\tau_x, \tau^*_x(\psi_d(\bpol^{d^2}))\circ
	\psi_x \circ \psi_d(\bpol^{d^2})^{-1})),
 \end{split}
\end{equation}
\end{definition}

For $m,n,d >1$ integers such that $n=md$, Mumford defines the morphisms (see \cite[p.~308--311]{MumfordOEDAV1}):
\begin{itemize}
 \item $\epsilon_d(\bpol): G(\bpol) \rightarrow
 G(\bpol^d)$, $(\tau_x, \psi_x) \mapsto (\tau_x, \psi_x^{\otimes d})$;
	\item $\eta_d(\bpol^d): G(\bpol^d)
	 \rightarrow G(\bpol)$, $(\tau_x, \psi_x) \mapsto \pi_d(\bpol^{d^2}) \circ
	 \epsilon_d(\bpol^d)(\tau_x, \psi_x)$;
 \item $\delta_d(\bpol): G(\bpol) \rightarrow G(\bpol)$, $h \mapsto
 h^{(d^2+d)/2} .\delta_{-1}(h^{(d^2-d)/2})$. 
\end{itemize}

On the Heisenberg group, Mumford defines $D_d$ for all $d$, but $E_d$ and $H_d$ only for $d=2$ (see
\cite[p.~316]{MumfordOEDAV1}). For $m,n,d>1$ integers such that $n=md$, we use the following
straightforward generalizations:
\begin{itemize}
 \item $E_d(m): G(m) \rightarrow G(n)$, $(\ell, x, y) \mapsto (\ell^d, \mu_{m,n}(x),
	\dnu_{m,n}(y))$;
 \item $H_d(n): G(n) \rightarrow G(m)$, $(\ell, x, y) \mapsto (\ell^d,\nu_{n,m}(x),
	\dmu_{n,m}(y))$;
 \item $D_d(m): G(m) \rightarrow G(m)$, $(\ell, x, y) \mapsto (\ell^{d^2}, dx, dy)$.
\end{itemize}
We will also need the following map, which is not a group morphism on $G(m)$ (because of the
factor $y_2(x_1)$ in the group law), but whose restrictions to the subgroups $\overk^* \times Z(m)$ and
$\overk^* \times \dZ(m)$ are group morphisms; it will only be used on these subgroups:
\begin{equation}\label{eq:epd}
    E'_d(m) : G(m) \rightarrow G(m), (\ell, x, y) \mapsto (\ell^d, x, y).
\end{equation}

We gather the results from \cite[Proposition~5]{MumfordOEDAV1} that we use:
\begin{proposition}\label{prop:mumfordtech}
 The maps $\epsilon_d(\bpol)$, $\eta_d(\bpol^d)$, $\delta_d(\bpol)$ are morphisms of theta groups considered as
 central extensions. We have:
 \begin{itemize}
	\item $\delta_{-1}(\bpol^d) \circ \epsilon_d(\bpol)=\epsilon_d(\bpol) \circ
	 \delta_{-1}(\bpol)$;
	\item $\delta_{-1}(\bpol) \circ \eta_d(\bpol^d)= \eta_d(\bpol^d) \circ
	 \delta_{-1}(\bpol^d)$;
	\item $\delta_{d}(\bpol^d)= \epsilon_d(\bpol) \circ \eta_d(\bpol^d)$;
	\item $\delta_{d}(\bpol)=\eta_d(\bpol^d) \circ \epsilon_d(\bpol)$.
 \end{itemize}
\end{proposition}
One checks directly that we have the same properties for
$E_d(m)$, $H_d(n)$ and $D_d(m)$ and that for all $h\in G(m)$, $D_d(m)(h)=h^{(d^2+d)/2}.
D_{-1}(h^{(d^2-d)/2})$. The following definition generalizes to any $d$ the definition of a compatible
pair of symmetric theta structures for $(\bpol, \bpol^2)$ given in \cite[p.~317]{MumfordOEDAV1}:

\begin{definition}\label{def:mumfordcomp}
 Let $m,n,d >1$ be integers such that $n=md$.
 Let $\Thetaubpol: \overk^* \times Z(m) \rightarrow G(\bpol)$ and
 $\Thetaubpold: \overk^* \times \Zn \rightarrow G(\bpol^d)$
 be two partial symmetric theta structures of respective types $Z(m)$ and $Z(n)$. We say that
 $\Thetabpol$ and $\Thetabpold$ are a (Mumford) compatible pair of theta structures if:
 \begin{enumerate}
 \item $\Thetaubpold \circ E_d(m) = \epsilon_d(\bpol) \circ \Thetaubpol$;
 \item $\Thetaubpol \circ H_d(n) = \eta_d(\bpol^d) \circ \Thetaubpold$.
	\end{enumerate}
 We say that $\Thetabpol$ and $\Thetabpold$ are a partial symmetric compatible pair of theta
 structures if moreover they are symmetric. We have the same definition if we replace $Z(m)$ (resp. $Z(n)$) by $\dZ(m)$ (resp. $\dZ(n)$).

 We say that the theta structures $\Thetabpol: G(m) \rightarrow G(\bpol)$ and
 $\Thetabpold: G(n) \rightarrow
 G(\bpol^d)$ of respective types $K(m)$ and $K(n)$ are (Mumford) compatible (resp. symmetric
 compatible) if the pairs of partial theta
 structures obtained by restriction of $\Thetabpol$ on $\overk^* \times Z(m)$ and
 $\overk^* \times \dZ(m)$ and
 $\Thetabpold$ on $\overk^* \times \Zn$ and $\overk^* \times \dZ(n)$ are compatible (resp. symmetric compatible).
\end{definition}
\begin{remark}\label{rk:thm6tech}
 We note that our definition of symmetric theta structure given in Section~\ref{sec:notations} is
 trivially equivalent to that given by Mumford \cite[p.~317]{MumfordOEDAV1},
 which says that $\Thetabpol$ is symmetric if:
 \begin{equation}
	\Thetabpol \circ D_{-1}(m) = \delta_{-1}(\bpol) \circ \Thetabpol.
 \end{equation}
 It is immediate to see that if $\Thetabpol$ is symmetric, then for all $d$ positive
 integer:
 \begin{equation}
	\Thetabpol \circ D_{d}(m) = \delta_{d}(\bpol) \circ \Thetabpol.
 \end{equation}
\end{remark}
Unfortunately, to the best of our understanding, this definition of compatible theta structures
is not easily amenable to computation. We can nevertheless build on the results of the previous section to obtain an equivalent, more effective definition.

For this, let $m,n,d >1$ be integers such that $n=md$, and consider a pair $(B, \bpol, \Thetabpol)$, $(A, \pol, \Thetapol)$ of
isog-$f$-compatible marked abelian varieties of respective levels $m$ and $n$. Let $\tildeK$ be the level subgroup of $\Gpol$
corresponding to the descent data of $(\pol, \psi)$ where $\psi : f^*(\bpol) \rightarrow
\pol$ is an isomorphism. The following basic
lemma is an important tool to establish the link between isog-$f$-compatibility and Mumford
compatibility:
\begin{lemma}\label{lem:commepsi}
Let $\tildeK_0:=\epsilon_d(\pol)(\tildeK)$ and let $G^*(\pol^d)$ be the centralizer of $\tildeK_0$ in $G(\pol^d)$. 
 Then:
 \begin{equation}
	\epsilon_d(\pol)(\Thetapol(\{1\}\times \Zn\times\{0\})) \subset G^*(\pol^d).
 \end{equation}
Moreover,
$\tildeK_0$ is the descent data of $(\pol^d, \psi^{\otimes d})$.
\end{lemma}
\begin{proof}
 Let $h \in \epsilon_d(\pol)(\Thetapol(\{1\}\times \Zn\times\{0\}))$. We have to prove that $h$ commutes
 with every element of $\tildeK_0$. For this, let $h' \in \tildeK_0$. We have to prove
 that $e_{\pol^d}(h,h')=1$. Let $h_0 \in \Thetapol(\{1\}\times \Zn\times\{0\})$ be such that
 $h=\epsilon_d(\pol)(h_0)$ and $h'_0 \in \tildeK$ such that $h'=\epsilon_d(\pol)(h'_0)$.
 We have $e_{\pol^d}(h,h')=e_\pol(h_0, h'_0)^{d}$, and moreover $e_\pol(h_0, h'_0)$ is a
 $d^{th}$-root of unity because $K \subset A[d]$. We are done for the first claim.

 For the second claim, let $\psi : f^*(\bpol) \rightarrow \pol$ be the isomorphism such
 that $(\pol, \psi)$ is the pair associated to $\tildeK$. We
 have seen that $(\tau_x, \psi_x) \in \tildeK$ if and only if $\psi_x$ and $\tau_x$ make Diagram
 (\ref{eq:descentdata}) commutative. Note that we have an isomorphism $\psi^{\otimes d} :
 f^*(\bpol)^d = f^*(\bpol^d) \rightarrow \pol^d$. Let $\tildeK_1$ be the descent data
 of $\pol^d$ to $\bpol^d$ associated to $(\pol^d, \psi^{\otimes d})$. Let $x \in K$, then $(\tau_x, \psi_x^1) \in \tildeK_1$ if
 and only if the following diagram commutes:
\begin{equation}\label{eq:descentdata1}
 \begin{tikzpicture}
 \matrix [column sep={1cm}, row sep={1cm}]
 { 
 \node(a){$f^*(\bpol^d)$}; & \node(b){$\pol^d$}; \\
 \node(c){$\tau_x^*(f^*(\bpol^d))$}; & \node(d){$\tau_x^*(\pol^d)$}; \\
 };
 \draw [->] (a) -- (b) node[above, midway]{$\psi^{\otimes d}$};
 \draw [->] (c) -- (d) node[above, midway]{$\tau^*_x(\psi^{\otimes d})$};
 \draw [->] (b) -- (d) node[right, midway] {$\psi_x^{1}$};
 \draw [double distance = 2 pt] (a) -- (c);
 \end{tikzpicture}
\end{equation}
But we see immediately that $\psi^1_x = \psi_x^{\otimes d}$, thus $\tildeK_1 =
\epsilon_d(\pol)(\tildeK)=\tildeK_0$.
\end{proof}
Denote by $\fsharp(\pol^d): G^*(\pol^d) \rightarrow G(\bpol^d)$ the map given by Definition~\ref{def:fsharp}
and the descent data $\epsilon_d(\pol)(\tildeK)$. Note that, by the previous lemma, we have that $\epsilon_d(\pol) \circ
\Thetapol(\{1\}\times \Zn\times\{0\})$ is a subset of the domain of $\fsharp(\pol^d)$, so that 
$\fsharp(\pol^d) \circ \epsilon_d(\pol) \circ \Theta^1_\pol: \overk^* \times \Zn \rightarrow G(\bpol^d)$ is a well-defined
group morphism.
\begin{definition}\label{def:fcompatibleB}
 Let $m,n,d > 1$ be integers such that $n=md$.
 Let $\Thetaubpol: \overk^* \times Z(m) \rightarrow G(\bpol)$ and $\Thetaubpold: \overk^* \times \Zn \rightarrow
 G(\bpol^d)$ be partial symmetric theta structures. Recall that $E'_d(n)$ is defined by
 Equation (\ref{eq:epd}).

 We say that they are $f$-compatible if there
 exists $(A, \pol, \Thetapol)$ an isog-$f$-compatible marked abelian variety such that we
 have the equality of maps $\overk^* \times \Zn \rightarrow G(\bpol^d)$:
 \begin{equation}\label{eq:maps}
	\Thetaubpold \circ E'_d(n)= \fsharp(\pol^d) \circ \epsilon_d(\pol) \circ \Thetaupol.
 \end{equation}
 Similarly, if $\Theta^2_\bpol: \overk^* \times \dZ(m) \rightarrow G(\bpol)$ and
 $\Theta^2_{\bpol^d}: \overk^* \times \dZ(n) \rightarrow G(\bpol^d)$ are partial theta
 structures, we say that they are $f$-compatible if there exists $(A, \pol, \Thetapol)$ a
 dual-isog-$f$-compatible marked abelian variety such that we
 have the same equation as in Equation (\ref{eq:maps}) considered as maps $\overk^* \times
 \dZn \rightarrow G(\bpol^d)$.

 We say that the theta structures $\Thetabpol: G(m) \rightarrow G(\bpol)$ and
 $\Thetabpold: G(n) \rightarrow
 G(\bpol^d)$ are $f$-compatible if the induced partial theta structures on $Z(m)$, $Z(n)$ and
 $\dZ(m)$, $\dZ(n)$ are $f$-compatible. 
\end{definition}

At first sight, Mumford compatibility and $f$-compatibility are different properties:
Mumford compatibility uses the map $\eta_d(\bpol^d)$, which is constructed from the morphism
$\psi_d(\bpol^{d^2})$ using the symmetry of $\bpol$. We will see in the following theorem that the two definitions are in fact equivalent. 
To prove it, we need the following technical lemma:

\begin{lemma}\label{lem:thm6lemtech}
 Keeping the notation of Definition~\ref{def:fcompatibleB}, we have:
 \begin{enumerate}
	\item $\fsharp(\pol^d) \circ \epsilon_d(\pol) = \epsilon_d(\bpol) \circ
	    \fsharp(\pol)$ as functions $G^*(\pol) \rightarrow G(\bpol^d)$;
	\item $\eta_d(\bpol^d) \circ \fsharp(\pol^d) = \fsharp(\pol) \circ \eta_d(\pol^d)$
	    as functions $G^*(\pol^d) \rightarrow G(\bpol)$.
 \end{enumerate}
\end{lemma}
\begin{proof}

We first prove (1). 
Let $\tildeK$ be the level subgroup defining $f^\sharp(\pol)$.
Let $\psi : f^*(\bpol) \rightarrow \pol$ be the isomorphism such
that the pair $(\pol, \psi)$ is associated to the descent data $\tildeK$. Denote
by $G^*(\pol)$ the centralizer of $\tildeK$ in $\Gpol$.
Let $(\tau_y, \psi_y) \in G^*(\pol)$, let $x=f(y)$ and set $(\tau_x, \psi_x) =
f^\sharp(\pol)((\tau_y,
\psi_y))$. 

Then, by Definition~\ref{def:fsharp} of $f^\sharp(\pol)$, we have the
commutative diagram
\begin{equation}\label{eq:descentdata2}
 \begin{tikzpicture}
 \matrix [column sep={1.5cm}, row sep={1cm}]
 { 
 \node(a){$f^*(\bpol)$}; & \node(b){$f^*(\tau_x^*(\bpol))=\tau_y^*(f^*(\bpol))$}; \\
 \node(c){$\pol$}; & \node(d){$\tau_y^*(\pol)$}; \\
 };
 \draw [->] (a) -- (b) node[above, midway]{$f^*(\psi_x)$};
 \draw [->] (c) -- (d) node[above, midway]{$\psi_y$};
 \draw [->] (b) -- (d) node[right, midway] {$\tau^*_y(\psi)$};
 \draw [->] (a) -- (c) node[left, midway] {$\psi$};
 \end{tikzpicture}
\end{equation}
from which we deduce the commutative diagram:
\begin{equation}\label{eq:descentdata3}
 \begin{tikzpicture}
 \matrix [column sep={1.5cm}, row sep={1cm}]
 { 
 \node(a){$f^*(\bpol)^d$}; & \node(b){$f^*(\tau_x^*(\bpol))^d=\tau_y^*(f^*(\bpol))^d$}; \\
 \node(c){$\pol^d$}; & \node(d){$\tau_y^*(\pol)^d$}; \\
 };
 \draw [->] (a) -- (b) node[above, midway]{$f^*(\psi_x^{\otimes d})$};
 \draw [->] (c) -- (d) node[above, midway]{$\psi_y^{\otimes d}$};
 \draw [->] (b) -- (d) node[right, midway] {$\tau^*_y(\psi)^{\otimes d}$};
 \draw [->] (a) -- (c) node[left, midway] {$\psi^{\otimes d}$};
 \end{tikzpicture}
\end{equation}
By definition, the map $f^\sharp(\pol^d)$ is given by the descent data
$\epsilon_d(\pol)(\tildeK)$, which, by Lemma~\ref{lem:commepsi}, is associated to the pair
$(\pol^d, \psi^{\otimes d})$.
Thus,
this last diagram shows, following Definition~\ref{def:fsharp}, that $(\tau_x, \psi_x^{\otimes d})=\epsilon_d(\bpol)((\tau_x, \psi_x))$ is the image by
$f^\sharp(\pol^d)$ of $\epsilon_d(\pol)((\tau_y, \psi_y))$.

For (2), by definition of $\eta_d(\bpol^d)$ and $\eta_d(\pol^d)$, we have to prove:
\begin{equation}
\pi_d(\bpol^{d^2}) \circ \epsilon_d(\bpol^d) \circ \fsharp(\pol^d) = \fsharp(\pol) \circ
\pi_d(\pol^{d^2}) \circ \epsilon_d(\pol^d).
\end{equation}
Denote by $\fsharp(\pol^{d^2}): G(\pol^{d^2})^* \rightarrow G(\bpol^{d^2})$ the quotient
map defined by the level subgroup $\epsilon_{d^2}(\pol)(\tildeK)$ where 
$G(\pol^{d^2})^*$ is the centralizer of  $\epsilon_{d^2}(\pol)(\tildeK)$ in
$G(\pol^{d^2})$. If we prove that
\begin{equation}\label{eq:lemtech314}
 \pi_d(\bpol^{d^2}) \circ
 \fsharp(\pol^{d^2})=\fsharp(\pol) \circ \pi_d(\pol^{d^2}),
\end{equation}
then the result stems from Lemma~\ref{lem:thm6lemtech} (1) which says that $\fsharp(\pol^{d^2}) \circ
\epsilon_d(\pol^d)= \epsilon_d(\bpol^d) \circ \fsharp(\pol^d)$.

In order to prove Equation (\ref{eq:lemtech314}), recall the isomorphisms
$\psi_d(\bpol^{d^2}): \bpol^{d^2} \rightarrow [d]^*(\bpol)$ and 
$\psi : f^*(\bpol) \rightarrow \pol$. Then, using the fact that $[d]\circ f= f \circ[d]$,
we can define the following diagram:

\begin{equation}\label{eq:lemtech315}
 \begin{tikzpicture}
 \matrix [column sep={3cm}, row sep={1cm}]
 { 
     \node(a){$f^*(\bpol^{d^2})=f^*(\bpol)^{d^2}$}; & \node(b){$f^*([d]^*(\bpol))=[d]^*(f^*(\bpol))$}; \\
	 \node(c){$\pol^{d^2}$}; & \node(d){$[d]^*(\pol)$}; \\
 };
 \draw [->] (a) -- (b) node[above, midway]{$f^*(\psi_d(\bpol^{d^2}))$};
 \draw [->] (c) -- (d) node[above, midway]{$\psi_d(\pol^{d^2})$};
 \draw [->] (b) -- (d) node[right, midway] {$[d]^*(\psi)$};
 \draw [->] (a) -- (c) node[left, midway] {$\psi^{d^2}$};
 \end{tikzpicture}
\end{equation}
We define $\psi_d(\pol^{d^2}): \pol^{d^2} \rightarrow [d]^*(\pol)$ making the preceding
diagram commutative. We remark that this diagram gives descent data, respectively for the
isogeny $[d]$ for the horizontal arrows and for the isogeny $f$ for the vertical arrows,
from which the maps in Equation (\ref{eq:lemtech314}) are defined. Precisely, $\psi^{d^2}$
allows one to define $f^{\sharp}(\pol^{d^2})$, $\psi_d(\pol^{d^2})$ defines
$\pi_d(\pol^{d^2})$, and so on. So Equation~\ref{eq:lemtech314} is a direct consequence of
the commutativity of Diagram (\ref{eq:lemtech315}) and the functoriality of Definition
\ref{def:fsharp} with respect to descent data.
\end{proof}
\begin{remark}\label{rk:thm6lemtechN}
 The proof of Lemma~\ref{lem:thm6lemtech} only uses the descent data $\tildeK$ of $\pol$ to
 $\bpol$ along $f$ and the symmetry of $\bpol$: it remains valid when $d$ is replaced by any
 positive integer $N$. We use it in this generality in the proof of Proposition
~\ref{prop:mumfordequiv}.
\end{remark}
The following proposition sheds some light on the meaning of the mysterious map $\eta_d(\bpol^d)$
used to define Mumford compatibility, by uncovering the link between this map and the
notion of symmetric compatibility of Definition~\ref{def:symcompat}.

\begin{proposition}\label{prop:mumfordequiv}
 Let $(B, \bpol, \Thetabpol)$ be a marked abelian variety and denote by
 $\pigbpol:G(\bpol) \rightarrow K(\bpol)$ the projection.
Let $\tildeH$ be a
 symmetric level subgroup of $G(\bpol)$ and let $H=\pigbpol(\tildeH)$.
 Let $x \in B(\overk)$ be a torsion point, let $\expo=\expo(x, \tildeH)$ and let $N$ be a positive
 integer such that $Nx \in H$ (so that $\expo | N$). Note that $x \in K(\bpol^N)$ (because of
 \cite[Proposition~4]{MumfordOEDAV1}). Let $g_x \in G(\bpol^N)$ be a symmetric element above $x$.
 Then $\eta_N(\bpol^N)(g_x)$ is a symmetric element of $G(\bpol)$ above $Nx$, and:
 \begin{enumerate}
	\item if $N$ is even, $\eta_N(\bpol^N)(g_x) \in \tildeH$ if and only if $\varphi_N(x,
	 \tildeH)=1$ (Definition~\ref{def:phiN}); this condition does not depend on the choice of the
	 symmetric element $g_x$ above $x$. In particular, if $N/\expo$ is odd (for instance if
	 $N=\expo$), $\eta_N(\bpol^N)(g_x) \in \tildeH$ if and only if $x$ is symmetric compatible
	 with $\tildeH$;
	\item if $N$ is odd, $x$ is symmetric compatible with $\tildeH$ and exactly one of the two
	 symmetric elements $g_x$ above $x$ satisfies $\eta_N(\bpol^N)(g_x) \in \tildeH$.
 \end{enumerate}
\end{proposition}
\begin{proof}
 By Proposition~\ref{prop:mumfordtech}, $\eta_N(\bpol^N)$ commutes with $\delta_{-1}$, so that
 $\eta_N(\bpol^N)(g_x)$ is symmetric; it lies above $Nx$ since $\epsilon_N(\bpol^N)$ lies above the
 identity and $\pi_N(\bpol^{N^2})$ above $[N]$. Moreover, $\eta_N(\bpol^N)$ raises the elements of
 $\overk^*$ to the power $N$, so that $\eta_N(\bpol^N)(-g_x)=(-1)^N\eta_N(\bpol^N)(g_x)$.

 Let $f_0=[N]$ with kernel $K_0 = B[N]$ and let $\pol = f_0^*(\bpol)$. Let $y
 \in B(\overk)$ be such that $f_0(y)=x$; as $\pol \simeq \bpol^{N^2}$ and $N^2y=Nx \in K(\bpol)$, we have $y \in
 \Kpol$. Denote by $\tildeK_0$ the descent data of
 $\pol$ to $\bpol$ and let $f_0^\sharp(\pol): G^*(\pol) \rightarrow G(\bpol)$ be the map
 defined as in Definition~\ref{def:fsharp}.
 Let $H'=f_0^{-1}(H)$ and $\tildeH'=(f_0^\sharp(\pol))^{-1}(\tildeH)$, which is the unique
 symmetric level subgroup of $\Gpol$ above $H'$ such that $\tildeK_0 \subset
 \tildeH'$ and $f_0^\sharp(\tildeH')=\tildeH$ (Proposition~\ref{prop:compat1} (2)). Let
 $g_y=(\tau_y, \psi_y) \in \Gpol$ be a symmetric element above $y$.

 Let $g_z = \epsilon_N(\pol)(g_y) \in G(\pol^N)$; it is symmetric by Proposition
~\ref{prop:mumfordtech}. Let $K_0^N = \epsilon_N(\pol) (\tildeK_0)$, which is the descent data of
 $\pol^N$ to $\bpol^N$, and denote by
 $f_0^\sharp(\pol^N): G^*(\pol^N) \rightarrow G(\bpol^N)$ the corresponding map. For $k \in K_0$,
 by Proposition~\ref{def:weil} (3), $e_{\pol^N}(y,k)=e_\pol(y,Nk)=1$, so that $g_z \in G^*(\pol^N)$, and
 $g'_x=f_0^\sharp(\pol^N)(g_z)$ is a symmetric element of $G(\bpol^N)$ above $x$ (Lemma
~\ref{lem:fsharpprops} (1)). As $g_y^N$ lies above $Ny=x$ and $f_0(x)=Nx \in K(\bpol)$, $g_y^N \in
 G^*(\pol)$. We have:
 \begin{gather}
	\begin{aligned}\label{eq:propequiveta:eq1}
	 \eta_N(\bpol^N)(g'_x) & = \eta_N(\bpol^N)(f_0^\sharp(\pol^N)(g_z)) & \text{}\\
	 & = f_0^\sharp(\pol)(\eta_N(\pol^N)(g_z)) & (\text{by Lemma
	~\ref{lem:thm6lemtech} and Remark~\ref{rk:thm6lemtechN}}) \\
	 & = f_0^\sharp(\pol)(\eta_N(\pol^N) \circ \epsilon_N(\pol)(g_y)) & \\
	 & = f_0^\sharp(\pol)(\delta_N(\pol)(g_y)) & (\text{because of Proposition
	~\ref{prop:mumfordtech}}) \\
	 & = f_0^\sharp(\pol)(g_y^N). & (\text{because $g_y$ is symmetric}) \\
	\end{aligned}
 \end{gather}
 As $\tildeH'=(f_0^\sharp(\pol))^{-1}(\tildeH)$, we get that $\eta_N(\bpol^N)(g'_x) \in \tildeH$ if and
 only if $g_y^N \in \tildeH'$. Replacing $g_y$ by $-g_y$ multiplies $g_z$, $g'_x$ and $g_y^N$ by
 $(-1)^N$.

 Suppose that $N$ is even. By Definition~\ref{def:phiN}, $g_y^N \in \tildeH'$ if and only if
 $\varphi_N(x, \tildeH)=1$. As the symmetric elements above $x$ are $\pm g'_x$ and
 $\eta_N(\bpol^N)(-g'_x)=\eta_N(\bpol^N)(g'_x)$, we get the first claim of (1). If $N/\expo$ is odd, then
 $\varphi_N(x, \tildeH)=\kappa_0(x, \tildeH)$ by Lemma~\ref{lem:phikappa}, and $\kappa_0(x, \tildeH)=1$
 if and only if $x$ is symmetric compatible with $\tildeH$ (Definition~\ref{def:kappa}).

 Suppose that $N$ is odd. Then $\expo$ is odd and $x$ is symmetric compatible with $\tildeH$ by
 Lemma~\ref{lem:symcompatchar}. As $g_y^N$ is a symmetric element above $x \in H'$, we have $g_y^N=\pm h$
 where $h$ is the element of $\tildeH'$ above $x$, and replacing $g_y$ by $-g_y$ changes this sign since
 $N$ is odd: for exactly one of the two choices of $g_y$, $g_y^N \in \tildeH'$, that is
 $\eta_N(\bpol^N)(g'_x) \in \tildeH$. The two choices of $g_y$ give the two symmetric elements $\pm
 g'_x$ above $x$, which proves (2).
\end{proof}
The proof also clarifies the role of $\eta_d(\bpol^d)$ in Mumford compatibility: it can be
replaced by the condition that every $x \in \Thetaubpoldbar(Z(n))$ is symmetric compatible with $\Thetaubpol(\{1 \} \times Z(m))$. We introduce the following definition.

\begin{definition}\label{def:mumfordcomp1}
 Let $m,n,d >1$ be integers such that $n=md$.
 Let $\Thetaubpol: \overk^* \times Z(m) \rightarrow G(\bpol)$ and
 $\Thetaubpold: \overk^* \times \Zn \rightarrow G(\bpol^d)$
 be two partial symmetric theta structures of respective types $Z(m)$ and $Z(n)$. We say that
 $\Thetabpol$ and $\Thetabpold$ are a (Mumford) compatible pair of theta structures if:
 \begin{enumerate}
 \item $\Thetaubpold \circ E_d(m) = \epsilon_d(\bpol) \circ \Thetaubpol$;
 \item for all $x \in \Thetaubpoldbar(Z(n))$, $x$ is symmetric
compatible with $\Thetaubpol(\{1 \} \times Z(m))$.
	\end{enumerate}
\end{definition}
Note that the second condition can be dropped when $d$ is odd. By what has been explained,
we have the following corollary:
\begin{corollary}
 Definitions~\ref{def:mumfordcomp} and~\ref{def:mumfordcomp1} are equivalent.
\end{corollary}
\begin{proof}
 Let $\Thetaubpol$ and $\Thetaubpold$ be partial symmetric theta structures satisfying condition
 (1) of Definitions~\ref{def:mumfordcomp} and~\ref{def:mumfordcomp1} (which is the same). Let
 $\tildeH=\Thetaubpol(\{1\}\times Z(m))$, $H=\pigbpol(\tildeH)$ and $G=\Thetaubpoldbar(\Zn)$. For $j \in
 Z(m)$, $E_d(m)((1,j))=(1,\mu_{m,n}(j))$, so that $\Thetaubpold((1,\mu_{m,n}(j)))=\epsilon_d(\bpol)(\Thetaubpol((1,j)))$;
 as $\epsilon_d(\bpol)$ lies above the identity, $\Thetaubpoldbar(\mu_{m,n}(j))=\Thetaubpolbar(j)$, and
 $H=\Thetaubpoldbar(\mu_{m,n}(Z(m)))=\Thetaubpoldbar(d\Zn)=dG$. For $i \in \Zn$, let $x_i=\Thetaubpoldbar(i)$ and
 $g_i=\Thetaubpold((1,i))$, which is a symmetric element of $G(\bpol^d)$ above $x_i$; as
 $di=\mu_{m,n}(\nu_{n,m}(i))$, we have $dx_i=\Thetaubpolbar(\nu_{n,m}(i)) \in H$. Since
 $H_d(n)((\ell,i))=(\ell^d, \nu_{n,m}(i))$, since $\eta_d(\bpol^d)$ raises the elements of $\overk^*$ to the
 power $d$ and since both sides are group morphisms, condition (2) of Definition
~\ref{def:mumfordcomp} is equivalent to $\eta_d(\bpol^d)(g_i)=\Thetaubpol((1,\nu_{n,m}(i)))$ for all $i \in
 \Zn$. As $\eta_d(\bpol^d)(g_i)$ lies above $dx_i$ and $\Thetaubpol((1,\nu_{n,m}(i)))$ is the element of
 $\tildeH$ above $dx_i$, this is equivalent to: $\eta_d(\bpol^d)(g_i) \in \tildeH$ for all $i \in \Zn$.

 Suppose that $d$ is odd. For all $x \in G$, $\expo(x, \tildeH)$ divides $d$, so it is odd and $x$ is
 symmetric compatible with $\tildeH$ by Lemma~\ref{lem:symcompatchar}: condition (2) of Definition
~\ref{def:mumfordcomp1} holds. On the other hand, by Proposition~\ref{prop:mumfordequiv},
 $\eta_d(\bpol^d)(g_i)=s(i)\Thetaubpol((1,\nu_{n,m}(i)))$ with $s(i) \in \{-1, 1\}$, and $s: \Zn
 \rightarrow \{-1,1\}$ is a group morphism since $i \mapsto \eta_d(\bpol^d)(g_i)$ and $i \mapsto
 \Thetaubpol((1,\nu_{n,m}(i)))$ are. For $j \in Z(m)$, $\eta_d(\bpol^d)(g_{\mu_{m,n}(j)})=\eta_d(\bpol^d)\circ
 \epsilon_d(\bpol)(\Thetaubpol((1,j)))=\delta_d(\bpol)(\Thetaubpol((1,j)))=\Thetaubpol((1,j))^d=\Thetaubpol((1,dj))$
 by Proposition~\ref{prop:mumfordtech}, since $\Thetaubpol((1,j))$ is symmetric; as
 $\nu_{n,m}(\mu_{m,n}(j))=dj$, $s$ is trivial on $\mu_{m,n}(Z(m))=\Zn[m]$. Thus $s$ factors through
 $\Zn/\Zn[m] \simeq (\Z/d\Z)^g$, which has odd order, so that $s=1$: condition (2) of Definition
~\ref{def:mumfordcomp} holds. Hence, when $d$ is odd, both definitions reduce to condition (1).

 Suppose that $d$ is even. By Proposition~\ref{prop:mumfordequiv} (1) with $N=d$, condition (2) of
 Definition~\ref{def:mumfordcomp} is equivalent to $\varphi_d(x, \tildeH)=1$ for all $x \in G$. If
 every $x \in G$ is symmetric compatible with $\tildeH$, then $\varphi_d(x, \tildeH)=\kappa_0(x,
 \tildeH)^{d/\expo(x, \tildeH)}=1$ for all $x \in G$ by Lemma~\ref{lem:phikappa}. Conversely, suppose that
 $\varphi_d(x, \tildeH)=1$ for all $x \in G$. Let $(e_k)_{k=1, \ldots, g}$ be the image by
 $\Thetaubpoldbar$ of the canonical basis of $\Zn$; it is a basis of $G \simeq \Zn$ and, as $H=dG$,
 $\expo(e_k, \tildeH)=d$, so that $\kappa_0(e_k, \tildeH)=\varphi_d(e_k, \tildeH)=1$ by Lemma
~\ref{lem:phikappa}. The group $G \subset B[n]$ is isomorphic to $\Zn$, contains $H=\Thetaubpolbar(Z(m))$
 and is isotropic for $e_{B,n}=e_{\bpol^d}$ (Proposition~\ref{def:weil} (1)), since $\Thetaubpold(\overk^*
 \times \Zn)$ is commutative. By Corollary~\ref{sec2:corgroup}, every $x \in G$ is symmetric compatible with $\tildeH$: condition (2) of Definition
~\ref{def:mumfordcomp1} holds.
\end{proof}
\begin{remark}\label{rk:symvsmumford}
 Proposition~\ref{prop:mumfordequiv} and Lemma~\ref{lem:phikappa} make precise the relation between
 symmetric compatibility and Mumford's map $\eta_N$. Let $x \in B(\overk)$ be a torsion point and $N$
 an even integer such that $Nx \in H$. The test $\eta_N(\bpol^N)(g_x) \in \tildeH$ is the condition
 $\varphi_N(x, \tildeH)=1$, and $\varphi_N(x, \tildeH)=\kappa_0(x, \tildeH)^{N/\expo(x, \tildeH)}$: it
 gives $\kappa_0(x, \tildeH)$ when $N/\expo(x, \tildeH)$ is odd, and no information on $x$ when
 $N/\expo(x, \tildeH)$ is even. Conversely, when $\expo(x, \tildeH)$ is even, $\kappa_0(x,
 \tildeH)=\varphi_{\expo(x, \tildeH)}(x, \tildeH)$: symmetric compatibility is Mumford's test
 performed with an exponent adapted to each point. It is therefore a finer, pointwise notion, but it is
 not additive (see the example after Definition~\ref{def:kappa}), whereas, once the exponent $N$ is
 fixed, $\varphi_N(\cdot, \tildeH)$ is a character (Proposition~\ref{sec2:propgroup}).

 For a subgroup $G \simeq \Zn$ such that $H=dG$, as in the preceding corollary, the two points of
 view give the same global condition: every point of $G$ is symmetric compatible with $\tildeH$ if
 and only if $\varphi_d(\cdot, \tildeH)$ is trivial on $G$. This can also be seen directly: every $x
 \in G$ can be written $x=(d/\expo(x, \tildeH))\, y$ with $y \in G$, and if $g_y$ is a symmetric element
 above (a preimage of) $y$, then $g_y^{d/\expo(x, \tildeH)}$ is a symmetric element above (a preimage
 of) $x$ whose $\expo(x, \tildeH)$-th power is $g_y^d$, so that $\kappa_0(x, \tildeH)=\varphi_d(y,
 \tildeH)$ when $\expo(x, \tildeH)$ is even. The pointwise notion is used to test and correct the
 elements of a basis one at a time (Algorithm~\ref{algo:changesym2}), and the additivity at fixed
 exponent to reduce the verification to a basis.
\end{remark}

\begin{lemma}\label{lem:techmainphi}
 Let $\Thetaubpol: \overk^* \times Z(m) \rightarrow G(\bpol)$ and
 $\Thetaubpold: \overk^* \times \Zn \rightarrow G(\bpol^d)$
 be two partial symmetric theta structures of respective types $Z(m)$ and $Z(n)$. We
 suppose that they are Mumford compatible. 
 \begin{enumerate}
	\item For any $g \in \dnu_{d,n}(\dZ(d)[2])$, $\Thetaubpol$ and $\Thetaubpold \circ \Phi(g)$ are
	 also Mumford compatible, with $\Phi(g):\overk^* \times \Zn\rightarrow \overk^* \times \Zn$, $(\alpha, x)\mapsto(\alpha g(x),x)$ (see Equation (\ref{eq:defphi}));
	\item If $(\Thetaubpol, \Thetaubpold_{,0})$ and $(\Thetaubpol, \Thetaubpold_{,1})$ are
	 two pairs of Mumford compatible partial symmetric theta structures such that
	 ${\Thetaubpoldbar}_{,0}$ and ${\Thetaubpoldbar}_{,1}$ coincide on $\Zn$, then there exists $g \in
	 \dnu_{d,n}(\dZ(d)[2])$ such that ${\Thetaubpold}_{,0}={\Thetaubpold}_{,1} \circ \Phi(g)$.
 \end{enumerate}
\end{lemma}
 \begin{proof}
	We prove (1) following Definition~\ref{def:mumfordcomp}. For $(\alpha,x) \in \overk^*\times Z(m)$,
	$\Thetaubpold \circ \Phi(g) \circ E_d(m)((\alpha, x))=\Thetaubpold \circ \Phi(g)((\alpha^d,
	\mu_{m,n}(x)))=\Thetaubpold((\alpha^dg(\mu_{m,n}(x)),
	\mu_{m,n}(x)))=\Thetaubpold((\alpha^d,\mu_{m,n}(x)))=\Thetaubpold \circ E_d(m)((\alpha, x)) $. We have checked the first condition
	of Definition~\ref{def:mumfordcomp}. Next, we have $\eta_d(\bpol^d) \circ
	\Thetaubpold \circ \Phi(g)= \pi_d(\bpol^{d^2}) \circ \epsilon_d(\bpol^d) \circ
	\Thetaubpold \circ \Phi(g)$. As $\epsilon_d(\bpol^d)$ raises the elements of $\overk^*$ to the power $d$ and $g(x)^d=1$ for
	all $x \in \Zn$ ($g$ is killed by $2$ and by $d$), we have $\epsilon_d(\bpol^d) \circ
	\Thetaubpold \circ \Phi(g) = \epsilon_d(\bpol^d) \circ
	\Thetaubpold$ so that $\eta_d(\bpol^d) \circ
	\Thetaubpold \circ \Phi(g)=\eta_d(\bpol^d) \circ
	\Thetaubpold$ and we are done for (1).

	For (2), as ${\Thetaubpold}_{,0}$ and ${\Thetaubpold}_{,1}$ are group morphisms which are the
	identity on $\overk^*$ and lie above the same map $\Zn \rightarrow B[n]$, there exists a
	character $g \in \dZn$ such that ${\Thetaubpold}_{,0}={\Thetaubpold}_{,1} \circ \Phi(g)$. As both
	are symmetric, $g(-x)=g(x)$ for all $x \in \Zn$, that is $2g=0$. By the first condition of
	Definition~\ref{def:mumfordcomp}, ${\Thetaubpold}_{,0}$ and ${\Thetaubpold}_{,1}$ coincide on
	$E_d(m)(\overk^* \times Z(m))$, so that $g$ is trivial on $\mu_{m,n}(Z(m))=d\Zn$, that is $dg=0$.
	Hence $g \in \dZn[2] \cap \dZn[d]=\dnu_{d,n}(\dZ(d)[2])$.
 \end{proof}

 \begin{theorem}\label{thm:compatthetas}
 Let $\Thetaubpol: \overk^* \times Z(m) \rightarrow G(\bpol)$ and $\Thetaubpold:
 \overk^* \times \Zn \rightarrow
 G(\bpol^d)$ be a pair of partial theta structures. They are $f$-compatible if and
 only if they are
 Mumford compatible.
\end{theorem}
\begin{proof}
First, we prove that $f$-compatible implies Mumford compatible.
Suppose that $\Thetaubpol: \overk^* \times Z(m) \rightarrow \Gpol$ and $\Thetaubpold: \overk^* \times \Zn \rightarrow
 G(\bpol^d)$ are $f$-compatible. This means that there exists $(A, \pol, \Thetapol)$ 
 a marked abelian variety isog-$f$-compatible with $(B, \bpol, \Thetabpol)$ such that:
 \begin{equation}\label{eq:th6:eq3}
	\Thetaubpold\circ E'_d(n) = \fsharp(\pol^d) \circ \epsilon_d(\pol) \circ \Thetaupol,
 \end{equation}
	keeping the notation
 of Definition~\ref{def:fcompatibleB}.
 Consider the map $M_{m,n}: \overk^* \times Z(m) \rightarrow \overk^* \times \Zn$, $(\alpha, x)
 \mapsto (\alpha, \mu_{m,n}(x))$. Note that:
 \begin{equation}\label{eq:th6:eq4}
 M_{m,n} \circ E'_d(m) = E'_d(n) \circ M_{m,n}=E_d(m).
 \end{equation}
 By Definition~\ref{sec1:def1} (3), we have, for
 all $g \in \overk^* \times Z(m)$:
 \begin{equation}\label{eq:th6:eq1}
	\Thetaubpol(g) = \fsharp(\pol) \circ \Thetaupol\circ M_{m,n}(g). 
 \end{equation}
The following diagram shows all the objects and maps that we consider; it will
be useful to understand the proof:
\begin{equation}\label{eq:diag3d}
 \begin{tikzpicture}
 \matrix [column sep={1cm}, row sep={1cm}]
 { 
 & \node(gmd){$G_1(n)$}; & & & \node(gnd){$G_1(nd)$}; \\
 \node(gbpold){$G(\bpol^d)$}; & & & \node(gpold){$G(\pol^d)$}; & \\
 \\
 & \node(gm){$G_1(m)$}; & & & \node(gn){$G_1(n)$}; \\
 \node(gbpol){$G(\bpol)$}; & & & \node(gpol){$\Gpol$}; & \\
 };
 \draw [->] (gpold) -- (gbpold) node[above right, midway]{$\fsharp(\pol^d)$};
 \draw [->] (gpol) -- (gbpol) node[above, midway]{$\fsharp(\pol)$};
 \draw [->] (gmd) -- (gbpold) node[xshift=0.5ex, yshift=-0.5ex, above left, midway]{$\Thetaubpold$};
\draw [->] (gnd) -- (gpold) node[xshift=-0.5ex, yshift=0.5ex, below right, midway]{$\Thetaupold$};
 \draw [->] (gm) -- (gn) node[below left, midway]{$M_{m,n}$};
 \draw [->] (gn) -- (gpol) node[xshift=-0.5ex, yshift=0.5ex, below right, midway]{$\Thetaupol$};
\draw [->] (gm) -- (gbpol) node[xshift=-0.5ex, yshift=0.5ex, below right, midway]{$\Thetaubpol$};
 \draw [->] (gpol) -- (gpold) node[above right, midway]{$\epsilon_d(\pol)$};
\draw [->] (gn) -- (gnd) node[right, midway]{$E_d(n)$};
 \draw [->] (gnd) -- (gmd) node[above, midway]{$\fsharp(nd)$};
 \draw [->, transform canvas={xshift=1ex}] (gbpol) -- (gbpold) node[right,
 midway]{$\epsilon_d(\bpol)$};
\draw [->, transform canvas={xshift=-1ex}] (gbpold) -- (gbpol) node[left,
 midway]{$\eta_d(\bpol^d)$};
\draw [->, transform canvas={xshift=1ex}] (gm) -- (gmd) node[below right, midway]{$E_d(m)$};
\draw [->, transform canvas={xshift=-1ex}] (gmd) -- (gm) node[below left,midway]{$H_d(n)$};
 \end{tikzpicture}
\end{equation}
For $\xi\in\{m,n,nd\}$, we have denoted $G_1(\xi)=\overk^* \times Z(\xi)$. All the arrows in this diagram have already been defined (or are restrictions of such maps) except $\fsharp(nd)$ which is the
analog for Heisenberg groups of the map $\fsharp(\pol^d)$. Precisely, let $M_{n, nd}:
\overk^* \times \Zn \rightarrow \overk^* \times Z(nd)$, $(\alpha, x) \mapsto (\alpha,
\mu_{n,nd}(x))$. Remark that $M_{n,nd}$ is injective, and denote by $G^*_1(nd)$ its image
in $G_1(nd)$. Then $\fsharp(nd): G^*_1(nd) \rightarrow G_1(n)$ is a left inverse of
$M_{n,nd}$. We remark that $\fsharp(nd) \circ E_d(n)= E'_d(n)$: this explains why we have
introduced this a priori strange map. Note that the diagram has the shape of a cube,
and we can interpret some of the previous results as the commutativity of
its faces. For instance, Lemma~\ref{lem:thm6lemtech} states the commutativity of the front face of the cube.

	Now, we can compute (and follow the paths in the diagram):
	\begin{gather}
 \begin{aligned}\label{eq:th6:eq2}
	\Thetaubpold \circ E_d(m) & = \Thetaubpold \circ E'_d(n) \circ M_{m,n} &
	(\text{because of Equation (\ref{eq:th6:eq4})}) \\ 
	& = \fsharp(\pol^d) \circ \epsilon_d(\pol) \circ \Thetaupol \circ M_{m,n} & (\text{because of
	(\ref{eq:th6:eq3})}) \\
	& = \epsilon_d(\bpol) \circ \fsharp(\pol) \circ \Thetaupol \circ M_{m,n} & (\text{thanks
	to Lemma~\ref{lem:thm6lemtech}}) \\
	& = \epsilon_d(\bpol) \circ \Thetaubpol. & (\text{because of Equation
	(\ref{eq:th6:eq1})})
 \end{aligned}
	\end{gather}
 Next, we want to prove that $\Thetaubpol \circ H_d(n) = \eta_d(\bpol^d) \circ
 \Thetaubpold$ (still under the $f$-compatibility assumption). We can write:
 \begin{gather}
	\begin{aligned}\label{eq:th6:eq5}
	\eta_d(\bpol^d) \circ \Thetaubpold \circ E'_d(n) & = \eta_d(\bpol^d) \circ
	\fsharp(\pol^d) \circ \epsilon_d(\pol) \circ \Thetaupol & (\text{because of
	(\ref{eq:th6:eq3})}) \\
	& = \fsharp(\pol) \circ \eta_d(\pol^d)\circ \epsilon_d(\pol) \circ \Thetaupol & (\text{thanks
	to Lemma~\ref{lem:thm6lemtech}}) \\
	& = \fsharp(\pol) \circ \delta_d(\pol) \circ \Thetaupol & (\text{because of Proposition
	\ref{prop:mumfordtech}}) \\
	& = \fsharp(\pol) \circ \Thetaupol \circ D_d(n) & (\text{using Remark~\ref{rk:thm6tech}}) \\
	& = \fsharp(\pol) \circ \Thetaupol \circ M_{m,n} \circ E'_d(m) \circ H_d(n) & (\text{by
	definition of $D_d(n)$})\\
	& = \Thetaubpol \circ E'_d(m) \circ H_d(n) & (\text{because of Equation
	(\ref{eq:th6:eq1})}) \\
	& = \Thetaubpol \circ H_d(n) \circ E'_d(n).
\end{aligned}
 \end{gather}
Next, we prove that Mumford compatible implies $f$-compatible. Suppose that $\Thetaubpol
: \overk^* \times Z(m) \rightarrow \Gpol$ and $\Thetaubpold: \overk^* \times \Zn \rightarrow
 G(\bpol^d)$ are Mumford compatible. Let $G= \Thetaubpoldbar(Z(n))$. 
 For all $x_1, x_2 \in G$, $e_{B,n}(x_1, x_2)= e_{\bpol^d}(x_1, x_2)=1$, so $G$ is
 isotropic for $e_{B,n}$. We show that for all $x \in G$, $x$ is
 symmetric compatible with $\Thetaubpol(\{1\} \times Z(m))$. As $x \in
 \Thetaubpoldbar(Z(n))$, there exists $e_x \in G_1(n)$,
 such that $g_x = \Thetaubpold(e_x) \in G(\bpol^d)$ is a symmetric element above $x$.
 Then using Mumford compatibility, we have $\eta_d(\bpol^d)(g_x) =
 \eta_d(\bpol^d)(\Thetaubpold(e_x))=\Thetaubpol(H_d(n)(e_x)) \in \Thetaubpol(\{1\} \times
 Z(m))$. Thus using Proposition~\ref{prop:mumfordequiv}, we get that $x$ is
 symmetric compatible with $\Thetaubpol(\{1\} \times Z(m))$.

 From the above, $G$ is a subgroup of $B[n]$ isomorphic to $Z(n)$
 containing $\Thetaubpolbar(Z(m)\times\{0\})$. Moreover, $G$ is isotropic for $e_{B, n}$ and for all $x \in G(\overk)$, $x$ is
 symmetric compatible with
 $\Thetabpol( \{1\}\times Z(m) \times\{0\} )$, so by applying Theorem
~\ref{prop:cond}, we get that there exists $(A, \pol, \Thetapol)$ isog-$f$-compatible to $(B,
 \bpol, \Thetabpol)$ such that $f(\Thetapolbar((i,0)))=\Thetaubpoldbar((i,0))$ for all $i \in \Zn$
 (the numbering $i \mapsto \Thetaubpoldbar((i,0))$ of $G$ satisfies
 $\Thetaubpoldbar((\mu_{m,n}(j),0))=\Thetabpolbar((j,0))$ for $j \in Z(m)$ by the first condition of
 Definition~\ref{def:mumfordcomp}). Let $K$ be the kernel of $f$,
 and denote by $\tildeK$ the unique level subgroup above $K$ which is the descent data
 of $\pol=f^*(\bpol)$.

 Let $G^*(\pol^d)$ be the centralizer of  $\epsilon_d(\pol)(\tildeK)$ in $G(\pol^d)$.
By Lemma~\ref{lem:commepsi}, we know that $\epsilon_d(\pol) (\Thetapol(\{1\}\times
Z(n)\times\{0\})) \subset G^*(\pol^d)$. Let $\fsharp(\pol^d): G^*(\pol^d) \rightarrow
G(\bpol^d)$ be defined by $\epsilon_d(\pol)(\tildeK)$. We consider the partial theta structure $\Thetaalphabpold: \overk^* \times
Z(n) \rightarrow G(\bpol^d)$ such that 
 \begin{equation}
	\Thetaalphabpold \circ E'_d(n)= \fsharp(\pol^d) \circ \epsilon_d(\pol) \circ \Thetaupol.
 \end{equation}
Then, by definition, $\Thetaalphabpold$ and $\Thetaubpol$ are $f$-compatible, and
$\overline{\Theta}^{\alpha}_{\bpol^d}$ and $\Thetaubpoldbar$ coincide on $\Zn$ since $\fsharp(\pol^d)$ lies above $f$ and
$\epsilon_d(\pol)$ above the identity. By the first
part of the proof, it implies that they are Mumford compatible. Then, by Lemma
\ref{lem:techmainphi},
there exists $c \in \dnu_{d,n}(\dZ(d)[2]) \subset K(n)$ such that $\Thetaalphabpold \circ
g_c = \Thetaubpold$ where $g_c=\Phi(c)$.

It remains to prove that $\Thetaubpold=\Thetaalphabpold \circ g_c$ and $\Thetaubpol$ are also
$f$-compatible. If $d$ is odd, $\dZ(d)[2]=\{0\}$, so that $c=0$, $g_c$ is the identity and there
is nothing to prove. We suppose that $d$ is even. Write $c=(0,\hat{c})$ with $\hat{c} \in \dnu_{d,n}(\dZ(d)[2])$, so that the
map $x \mapsto e_n(c,(x,0))$ is a character of $\Zn$ with values in $\{-1,1\}$. We show that
there exists $g_0 \in \subG(n)$ such that, as maps $\overk^* \times \Zn \rightarrow G(\bpol^d)$,
\begin{equation}\label{eq:pullbackhei}
 \Thetaalphabpold \circ g_c \circ E'_d(n) = \fsharp(\pol^d) \circ \epsilon_d(\pol) \circ \Thetapol
 \circ g_0.
\end{equation}
This proves the claim since $(A, \pol, \Thetapol \circ g_0)$ is isog-$f$-compatible with $(B,
\bpol, \Thetabpol)$ by Proposition~\ref{prop:phs} (1).

Let $(e_i)_{i=1, \ldots, g}$ be the canonical basis of $\Zn$, $(\hate_i)_{i=1, \ldots, g}$ the
dual basis of $\dZn$ and $\zeta=\hate_1(e_1)$, a primitive $n$-th root of unity, so that
$\hate_i(e_j)=\zeta^{\delta_{ij}}$. For $i=1, \ldots, g$, let $c_i \in \{0,1\}$ be such that
$e_n(c, (e_i,0))=(-1)^{c_i}$, and let $\lambda_i \in \overk^*$ be such that
$\lambda_i^2=\zeta^{mc_i}$. Let $\psigz: \Zn \rightarrow \dZn$ be the morphism such that
$\psigz(e_i)=mc_i\hate_i$; its image is contained in $m\dZn=\dnu_{d,n}(\dZ(d))$. Let
$\chi_{\gbar_0}: \Zn \rightarrow \overk^*$, $\sum_i a_ie_i \mapsto \prod_{i=1}^g \lambda_i^{a_i^2}$ (for
$a_i \in \Z$); it is well defined since $\lambda_i^{(a_i+n)^2}=\lambda_i^{a_i^2}\zeta^{mc_ia_in}\zeta^{mc_in^2/2}$
and $n$ is even. For $x=\sum_i a_ie_i$ and $x'=\sum_i a'_ie_i$, we have
\begin{equation}
 \chi_{\gbar_0}(x+x')=\chi_{\gbar_0}(x)\chi_{\gbar_0}(x')\prod_{i=1}^g \zeta^{mc_ia_ia'_i}=
 \chi_{\gbar_0}(x)\chi_{\gbar_0}(x')\,\psigz(x')(x),
\end{equation}
and $\chi_{\gbar_0}(-x)=\chi_{\gbar_0}(x)$. Let $g_0: G(n) \rightarrow G(n)$, $(\alpha, x, y) \mapsto (\alpha
\chi_{\gbar_0}(x), x, y+\psigz(x))$. A direct computation using the previous relation and the group
law of $G(n)$ shows that $g_0$ is a group morphism; it is the identity on $\overk^*$, it is bijective
(its inverse is $(\alpha,x,y) \mapsto (\alpha \chi_{\gbar_0}(x)^{-1}, x, y-\psigz(x))$), and it
commutes with $D_{-1}$ since $\chi_{\gbar_0}$ is even and $\psigz$ is linear: $g_0 \in \Auts(G(n))$.
Moreover $g_0$ is compatible with $G(m)$ (Definition~\ref{def:compat}): it is the identity on
$\{1\}\times\{0\}\times\dZn$, and for $x \in \mu_{m,n}(Z(m))=\Zn[m]$, we have $x=\sum_i
da'_ie_i$, so that $\psigz(x)=\sum_i mdc_ia'_i\hate_i=0$ and $\chi_{\gbar_0}(x)=\prod_i
\zeta^{mc_id^2a_i'^2/2}=1$ since $d$ is even; hence $g_0((1,x,0))=(1,x,0)$. By Lemma
\ref{lem:action}, $g_0 \in \subG(n)$.

Let $x \in \Zn$ and $\alpha \in \overk^*$. In $G(n)$, $g_0((\alpha, x,0))=\alpha\chi_{\gbar_0}(x)(1,0,\psigz(x))(1,x,0)$
and, by Definition~\ref{sec1:def1} (2), $\Thetapol((1,0,\psigz(x))) \in \tildeK$. As $\epsilon_d(\pol)$ is a
group morphism which raises the elements of $\overk^*$ to the power $d$ and maps $\tildeK$ to the
descent data of $\pol^d$, which is the kernel of $\fsharp(\pol^d)$, we get
\begin{equation}
 \begin{split}
 \fsharp(\pol^d) \circ \epsilon_d(\pol) \circ \Thetapol \circ g_0((\alpha,x,0)) & =\chi_{\gbar_0}(x)^d\,
 \fsharp(\pol^d) \circ \epsilon_d(\pol) \circ \Thetapol((\alpha,x,0)) \\
 & =\chi_{\gbar_0}(x)^d\, \Thetaalphabpold\circ E'_d(n)((\alpha, x, 0)).
 \end{split}
\end{equation}
On the other hand, $\Thetaalphabpold \circ g_c \circ E'_d(n)((\alpha,x,0))=e_n(c,(x,0))\,
\Thetaalphabpold\circ E'_d(n)((\alpha,x,0))$. Now $\chi_{\gbar_0}^d$ is a character of $\Zn$, since
$\psigz(x')(x)^d=1$ for all $x,x' \in \Zn$, and $\chi_{\gbar_0}(e_i)^d=\lambda_i^d=\zeta^{mc_id/2}=(-1)^{c_i}=e_n(c,(e_i,0))$
for $i=1, \ldots, g$. Thus $\chi_{\gbar_0}^d=e_n(c, (\cdot, 0))$, which proves (\ref{eq:pullbackhei}).
\end{proof}
The two definitions of Mumford compatibility and $f$-compatibility complement each other:
Mumford compatibility is more intrinsic and shows in particular that $f$-compatibility does
not depend on $f$. 

By Lemma~\ref{lem:techmainphi}, the data of the partial theta structure $\Thetaubpol$ and a
subgroup $G \subset B[n]$ together with a numbering $g_f: \Zn \rightarrow G(\overk)$
satisfying $g_f(\mu_{m,n}(i))= \Thetabpolbar((i, 0))$ for all $i \in Z(m)$,
determines $\Thetaubpold$ (if it exists) Mumford compatible to $\Thetaubpol$, up to an action of $\Phi(g)$ for $g \in \dnu_{d,n}(\dZ(d)[2])$.
By the equivalence between Mumford compatibility and $f$-compatibility established above, the same holds for $f$-compatibility:
\begin{proposition}\label{prop:uniqueness}
 Let $m,n,d>1$ be integers such that $n=md$.
 Let $(B, \bpol, \Thetabpol)$ be a marked abelian variety of type $K(m)$. 
Denote by $\Thetaubpol: \overk^* \times Z(m) \rightarrow G(\bpol)$ the partial theta
 structure obtained by restriction from $\Thetabpol$.
 Let $G= \{ g_f(i), i \in \Zn \}$ be a subgroup of $B[n]$ isomorphic to $Z(n)$
 containing $\Thetabpolbar(Z(m)\times\{0\})$, such that $G$ is isotropic for $e_{B, n}$
 and that for all $x \in G(\overk)$, $x$ is symmetric compatible with
 $\Thetabpol( \{1\}\times Z(m) \times\{0\} )$. 
We suppose moreover that for all $i \in Z(m)$, $g_f(\mu_{m,n}(i))= \Thetabpolbar((i, 0))$.
 Denote by $\sgood$ the set of good lifts
 $\tildeG = \{ \tildeg_f(i), i \in \Zn \}$ of $G$ and by $\sfcompat$ the set of 
partial theta structures $\Thetaubpold: \overk^* \times \Zn \rightarrow G(\bpol^d)$
 which are
 $f$-compatible with $\Thetaubpol$ and such that for all $i \in \Zn$,
 $\Thetaubpoldbar((i, 0))=g_f(i)$.
 Then:
 \begin{enumerate}
	\item $\sfcompat \ne \emptyset$;
	\item let $\Thetaubpold, \Thetadbpold \in \sfcompat$, there exists $g_0 \in \dnu_{d,n}(\dZ(d)[2])$ such that
 $\Thetaubpold= \Thetadbpold
 \circ \Phi(g_0)$;
\item there is a well-defined map $\gtf : \sgood \rightarrow \sfcompat$;
\item let $(e_1, \ldots, e_g)$ be the canonical basis of $Z(n)$, and for $\alpha = 1, 2$, let
 $\tildeG^\alpha = \{ \tildeg^\alpha_f(i), i \in \Zn \} \in \sgood$. For $i=1,
 \ldots, g$, let $\omega_i$ be such that $\tildeg_f^1(e_i) = \omega_i * \tildeg_f^2(e_i)$. We have, for all $i=1,
 \ldots, g$, $\omega_i^d \in \{-1, 1\}$. Moreover, $\gtf(\tildeG^1)=\gtf(\tildeG^2)$ if and only if for all $i=1, \ldots, g$,
 $\omega_i^d=1$.
\end{enumerate}
\end{proposition}
\begin{proof}
To prove (1), we have to prove the existence of a partial theta structure $\Thetaubpold: \overk^* \times \Zn \rightarrow G(\bpol^d)$
 which is
 $f$-compatible with $\Thetaubpol$ and such that for all $i \in \Zn$,
 $\Thetaubpoldbar((i, 0))=g_f(i)$. By Theorem~\ref{prop:cond} and the hypothesis, there exists $(A,
 \pol, \Thetapol)$ isog-$f$-compatible with $(B, \bpol, \Thetabpol)$ and such that for
 all $i \in \Zn$,
 $f(\Thetapolbar((i,0)))=g_f(i)$. Let $\Thetaubpold$ be the partial theta
 structure defined by $\Thetaubpold \circ E'_d(n)= \fsharp(\pol^d) \circ
 \epsilon_d(\pol) \circ \Thetaupol$. By definition, $\Thetaubpold$ is $f$-compatible
 with $\Thetabpol$, and for all $i \in \Zn$, $\Thetaubpoldbar((i, 0))=g_f(i)$. We have
 proved the first claim.

 Claim (2) is an immediate consequence of Theorem~\ref{thm:compatthetas} and Lemma~\ref{lem:techmainphi}.

 In order to define $\gtf$, let $\tildeG = \{ \tildeg_f(i), i \in \Zn \}$ be any good
 lift of $G$. Then by Theorem
~\ref{th:main1}, there exists a unique $(A, \pol, \Thetapol)$ isog-$f$-compatible with
 $(B, \bpol, \Thetabpol)$ such that for all $i \in \Zn$, $\tildeg_f(i)=
 (\theta^\Thetapol_{\mu_{m,n}(j)+i}(0_\Thetapol))_{ j\in Z(m)}$. Let 
$\Thetaubpold \circ E'_d(n)= \fsharp(\pol^d) \circ
 \epsilon_d(\pol) \circ \Thetaupol$, we can set $\gtf(\tildeG)= \Thetaubpold$ and we have
 proved claim (3).

Next, for $\alpha = 1, 2$, let
 $\tildeG^\alpha = \{ \tildeg^\alpha_f(i), i \in \Zn \} \in \sgood$. For $i=1, \ldots,
 g$, $(\tildeg_f^1(e_i)/\tildeg_f^2(e_i))^{2d}=1$ because of Lemma~\ref{lem:goodliftdiff}, thus
 $(\tildeg_f^1(e_i)/\tildeg_f^2(e_i))^{d} \in \{ -1, 1 \}$. Note that if $d$ is odd, still by Lemma
~\ref{lem:goodliftdiff}, $(\tildeg_f^1(e_i)/\tildeg_f^2(e_i))^{d}=1$ and by claim (2) 
there is a unique partial theta structure $\Thetaubpold: \overk^* \times \Zn \rightarrow G(\bpol^d)$
 which is
 $f$-compatible with $\Thetaubpol$ and such that for all $i \in \Zn$,
 $\Thetaubpoldbar((i, 0))=g_f(i)$. So claim (4) is true if $d$ is odd. We suppose that
 $d$ is even. For $\alpha=1, 2$, let $(A, \pol, \Thetaalphapol)$
isog-$f$-compatible with
 $(B, \bpol, \Thetabpol)$ be such that for all $i \in \Zn$, $\tildeg^\alpha_f(i)=
 (\theta^\Thetaalphapol_{\mu_{m,n}(j)+i}(0_\Thetaalphapol))_{ j\in Z(m)}$ (Theorem~\ref{th:main1});
 by Proposition~\ref{prop:charac}, $A=B/G[d]$ and $f$ is the contragredient of the quotient
 isogeny in both cases. By the proof of Proposition~\ref{prop:phs} (2), there exists $g_0 \in \subG(n)$
 such that $\Thetaupol = \Thetadpol \circ g_0$.
 For $x \in \Zn$, $f(\Theta^1_\pol((x,0)))=f(\Theta^2_\pol((x,0)))=g_f(x)$ (where, for
 $\alpha=1,2$, we write $\Theta^\alpha_\pol((x,0))$ for the image of $(x,0)$ in $\Kpol$), so that
 $\Theta^1_\pol((x,0))-\Theta^2_\pol((x,0)) \in \Ker f=\Theta^2_\pol(\{0\}\times\dnu_{d,n}(\dZ(d)))$.
 Moreover, for $\alpha=1,2$, $\Theta^\alpha_\pol(\{0\}\times\dZn)=f^{-1}(\Thetabpolbar(\{0\}\times\dZ(m)))$,
 which is isotropic and contains $\Ker f$, so that the elements $\Theta^\alpha_\pol((0,y))$, $y \in
 \dZn$, which are determined by their pairings with the $\Theta^\alpha_\pol((x,0))$, do not depend on
 $\alpha$; and $\Thetaalphapol(\{1\}\times\{0\}\times\dZn)$ is the unique level subgroup above this
 group containing $\tildeK$ and mapped by $\fsharp(\pol)$ onto $\Thetabpol(\{1\}\times\{0\}\times\dZ(m))$
 (Definition~\ref{sec1:def1} (2) and (4), and Proposition~\ref{prop:compat1} (2)). Hence $g_0$ is the
 identity on $\{1\}\times\{0\}\times\dZn$ and there exist a morphism $\psigz: \Zn \rightarrow
 \dnu_{d,n}(\dZ(d))$ and a map $\chi_{\gbar_0}: \Zn \rightarrow \overk^*$ such that, for all $(\alpha,
 x, y) \in G(n)$, $g_0((\alpha,x,y))=(\alpha\chi_{\gbar_0}(x), x, y+\psigz(x))$. As $g_0$ is a group
 morphism, $\chi_{\gbar_0}(x+x')=\chi_{\gbar_0}(x)\chi_{\gbar_0}(x')\psigz(x')(x)$, so that
 $\chi_{\gbar_0}^d$ is a character of $\Zn$ ($\psigz$ takes its values in $\dZn[d]$).

 Exactly as at the end of the proof of Theorem~\ref{thm:compatthetas}, we have, as maps $\overk^* \times
 \Zn \rightarrow G(\bpol^d)$,
 $\fsharp(\pol^d) \circ \epsilon_d(\pol) \circ \Thetaupol = \chi_{\gbar_0}^d \cdot (\fsharp(\pol^d) \circ
 \epsilon_d(\pol) \circ \Thetadpol)$, so that $\gtf(\tildeG^1)=\gtf(\tildeG^2)$ if and only if
 $\chi_{\gbar_0}^d=1$, that is, if and only if $\chi_{\gbar_0}(e_i)^d=1$ for $i=1, \ldots, g$.

 Finally, both $\theta_0^{\Theta^1_\pol}$ and $\theta_0^{\Theta^2_\pol}$ span the line of sections
 invariant under $\Thetaupol(\{1\}\times\{0\}\times\dZn)=\Thetadpol(\{1\}\times\{0\}\times\dZn)$, so that
 by (\ref{eq:trans}), there exists $\kappa \in \overk^*$ (taking also into account the rigidifications at
 $0_A$) such that for all $k \in \Zn$,
 $\theta^{\Theta^1_\pol}_k(0_A)=\kappa\, \Thetadpol(g_0((1,k,0)))(\theta^{\Theta^2_\pol}_0)(0_A)=\kappa\,
 \chi_{\gbar_0}(k)\psigz(k)(-k)\,\theta^{\Theta^2_\pol}_k(0_A)$. Comparing the coordinates of index
 $\mu_{m,n}(j)+e_i$ and $\mu_{m,n}(j)$ (for $j \in Z(m)$ such that these coordinates do not vanish) of
 $\tildeg^1_f(e_i)=\omega_i*\tildeg^2_f(e_i)$ and of $\tildeg^1_f(\mu_{m,n}(j))=\tildeg^2_f(\mu_{m,n}(j))$
 (see Definition~\ref{def:good}), and raising to the power $d$ (recall that $\psigz(k)(-k)^d=1$ and that
 $\chi_{\gbar_0}^d$ is a character), we get $\omega_i^d=\kappa^d\chi_{\gbar_0}(e_i)^d\chi_{\gbar_0}(\mu_{m,n}(j))^d$ and
 $1=\kappa^d\chi_{\gbar_0}(\mu_{m,n}(j))^d$, so that $\omega_i^d=\chi_{\gbar_0}(e_i)^d$. This proves claim (4).
\end{proof}
\begin{remark}\label{rk:uniqueness}
 Recall from Definition~\ref{def:fcompatibleB} that we say that 
$\Thetaubpol: \overk^* \times Z(m) \rightarrow G(\bpol)$ and $\Thetaubpold: \overk^* \times \Zn \rightarrow
 G(\bpol^d)$ partial symmetric theta structures are $f$-compatible if there exists $(A,
 \pol, \Thetapol)$ an isog-$f$-compatible marked abelian variety such that we
 have $\Thetaubpold \circ E'_d(n)= \fsharp(\pol^d) \circ \epsilon_d(\pol) \circ
 \Thetaupol$.

 It is remarkable that, while $(A, \pol, \Thetapol)$ isog-$f$-compatible with $(B, \bpol,
 \Thetabpol)$ is not unique since, by Lemma~\ref{lem:action}, there is an action of
 $\subG(n)$ on $(A, \pol, \Thetapol)$, the preceding proposition states that all these
 choices collapse to $\dZ(2)$ for $\Thetaubpold$.

 It has an important algorithmic consequence that was pointed out in \cite{DRisogenies}: the
 non-uniqueness of $(A, \pol, \Thetapol)$ materializes algorithmically in the choices of
 roots of unity in the computations of good lifts. The uniqueness of $(B, \bpol^d,
 \Thetabpold)$ up to the action of $\dZ(d)[2]$ shows that we can expect, and this can be verified directly in the
 formulas, that most choices of roots of unity in good lifts will cancel out in the algorithms to
 compute the theta null point of $(B, \bpol^d, \Thetabpold)$.
 So, in practice, we should be able to fix $\Thetabpold$ just by extracting
 $g$ square roots.
\end{remark}

In order to compute elements of
$\Gamma(B, \bpol^d)$, we can just take the product of $d$ sections of $\Gamma(B,\bpol)$.
Let $s= \prod_{i=1}^d s_i$ for $s_i \in \Gamma(B,\bpol)$ be such a section. If $g_d \in
G(\bpol^d)$ is of the form $g_d= \epsilon_d(\bpol)(g_0)$ for $g_0 \in G(\bpol)$, then $g_d(s)=
\prod_{i=1}^d g_0(s_i)$. If $g_d \notin \epsilon_d(\bpol)(G(\bpol))$, we can suppose that $g_d$ is in the
image of $\fsharp(\pol^d) \circ \epsilon_d(\pol)$; write $g_d = \fsharp(\pol^d) \circ \epsilon_d(\pol)(g_0)$
for $g_0 \in G(\pol)$. Then, $\fsharp(\pol^d) \circ \epsilon_d(\pol)(\prod_{i=1}^d s_i)=
\prod_{i=1}^d g_0(f^*(s_i))$, and it is straightforward to see that this section is invariant by
$\epsilon_d(\pol)(\tildeK)$, the descent data of $\pol^d$ to $\bpol^d$, so that it is of the form
$f^*(s_1)$ and we can set $g_d(s)=s_1$.

The isogeny and change of level algorithms share the same three-step structure:
\begin{enumerate}
 \item From sections of $\bpol$, compute sections of $\bpol^d$ and compute a map
	$G(\bpol) \rightarrow G(\bpol^d)$;
 \item Compute certain level subgroups, say $\tildeK_i$ for $i=1, 2$ of $G(\bpol^d)$;
 \item Compute the action of $\tildeK_i$ on $\Gamma(B, \bpol^d)$ to obtain sections of
	$\bpol^d/ \tildeK_1$ and a theta
	structure for $A$ or a theta structure for $(B, \bpol^d)$.
\end{enumerate}

Let $(B, \bpol, \Thetabpol)$ be a marked abelian variety. Suppose that $\Thetabpold$ is a
theta structure such that $\Thetabpol$ and $\Thetabpold$ are $f$-compatible. What we have
to do is:
\begin{itemize}
 \item compute sections in $\Gamma(B, \bpol^d)$;
 \item compute the action of $g \in \Thetabpold((1, e, 0))$ for $e \in \Zn$ on an element of $\Gamma(B, \bpol^d)$.
\end{itemize}


Alternatively, if $d=i_0^2$ for $i_0 \in \N$, we can use the fact that, as $\bpol$ is symmetric,
$[i_0]^*(\bpol)\simeq \bpol^{i_0^2}$ so that $\Gamma(B, \bpol^d)=\Gamma(B, [i_0]^*(\bpol))$.

More generally, following \cite{DRniveau}, we can write $d
= \sum_{j=1}^r a_j^2$ where $a_j$ are positive integers, so that (still under the
assumption that $\bpol$ is symmetric):
\begin{equation}
 \otimes_{j = 1}^r [a_j]^*(\bpol) \simeq
\bpol^{\sum_{j=1}^r a_j^2}= \bpol^d. 
\end{equation}
Recall that we have a morphism $\epsilon_d(\pol): \Gpol \rightarrow G(\pol^d)$, $(\tau_x,
\psi_x)\mapsto (\tau_x, \psi_x^{\otimes d})$. This result generalizes immediately:

For $(B, \bpol)$ an abelian variety together with an ample symmetric line bundle of type $K(m)$
and $a$ a positive integer prime to $m$, the multiplication by $a$ is a bijection of $K(\bpol)=B[m]$,
and for $z \in B[m]$ we denote by $a^{-1}z$ the unique $x \in B[m]$ such that $ax=z$. Recall that
$\delta_a(\bpol): G(\bpol) \rightarrow G(\bpol)$, $h \mapsto h^{(a^2+a)/2}.\delta_{-1}(h^{(a^2-a)/2})$
is a group morphism (Proposition~\ref{prop:mumfordtech}) above the multiplication by $a$ on
$K(\bpol)$, that $\delta_a(\bpol)(c)=c^{a^2}$ for $c \in \overk^*$, and that $\delta_a(\bpol)(h)=h^a$
if $h$ is symmetric (since then $\delta_{-1}(h)=h^{-1}$).

\begin{lemma}\label{lem:prodiso}
 Let $(B, \bpol)$ be an abelian variety together with an ample symmetric line bundle of
 type $K(m)$ and let $a$ be a positive integer such that $\gcd(a,m)=1$.
 \begin{enumerate}
	\item The map
	 \begin{gather}
		\begin{aligned}
		 P_a(\bpol): G(\bpol) & \rightarrow G([a]^*(\bpol))\\
		 (\tau_z, \psi_z) & \mapsto (\tau_{a^{-1}z}, [a]^*(\psi_z))
		\end{aligned}
	 \end{gather}
	 is an injective group morphism, which is the identity on $\overk^*$. Its image is
	 contained in the centralizer $G^*([a]^*(\bpol))$ of the descent data $\tildeK_a$ of
	 $[a]^*(\bpol)$ to $\bpol$, and $[a]^\sharp \circ P_a(\bpol)$ is the identity of $G(\bpol)$,
	 where $[a]^\sharp: G^*([a]^*(\bpol)) \rightarrow G(\bpol)$ is the map of Definition
	~\ref{def:fsharp}.
	\item The map $\ea(\bpol)=P_a(\bpol) \circ \delta_a(\bpol): G(\bpol) \rightarrow
	 G([a]^*(\bpol))$ is a group morphism. For $g=(\tau_x, \psi_x) \in G(\bpol)$, we have
	 $\ea(\bpol)(g)=(\tau_x, [a]^*(\psi'))$ where $\delta_a(\bpol)(g)=(\tau_{ax}, \psi')$. Moreover
	 $\ea(\bpol)(c)=c^{a^2}$ for $c \in \overk^*$, the kernel of $\ea(\bpol)$ is the group $\mu_{a^2}$
	 of $a^2$-th roots of unity in $\overk^*$, and $\ea(\bpol)$ is injective on every level subgroup
	 of $G(\bpol)$. If $g$ is symmetric and $g^a=(\tau_{ax}, \psi_{ax})$, then $\ea(\bpol)(g)=(\tau_x,
	 [a]^*(\psi_{ax}))$.
	\item The morphism $\ea(\bpol)$ is compatible with the action on sections: for all $s \in
	 \Gamma(B,\bpol)$ and $g \in G(\bpol)$, we have
	 \begin{equation}
		\ea(\bpol)(g) ([a]^*(s))=[a]^* (\delta_a(\bpol)(g)(s)),
	 \end{equation}
	 and in particular $\ea(\bpol)(g) ([a]^*(s))=[a]^* (g^a(s))$ if $g$ is symmetric.
 \end{enumerate}
\end{lemma}
\begin{proof}
 Let $(\tau_z, \psi_z) \in G(\bpol)$ and $x=a^{-1}z$. As $[a] \circ \tau_x=\tau_{ax} \circ
 [a]=\tau_z \circ [a]$, $[a]^*(\psi_z)$ is an isomorphism $[a]^*(\bpol) \rightarrow [a]^*(
 \tau^*_{z} \bpol)= \tau^*_x( [a]^* \bpol)$. As $[a]^*(\bpol) \simeq \bpol^{a^2}$, we have
 $K([a]^*(\bpol))=B[a^2m] \ni x$, so that $(\tau_x, [a]^*(\psi_z)) \in G([a]^*(\bpol))$ and
 $P_a(\bpol)$ is well defined.

 We show that $P_a(\bpol)$ is a group morphism. Let $(\tau_{z_\nu}, \psi_{z_\nu}) \in G(\bpol)$ and
 $x_\nu=a^{-1}z_\nu$ for $\nu=1,2$. The composition $(\tau_{z_1},
 \psi_{z_1})\circ (\tau_{z_2}, \psi_{z_2})$ is given by
 $$\bpol \stackrel{\psi_{z_1}}{\longrightarrow} \tau^*_{z_1} \bpol
 \stackrel{\tau^*_{z_1}(\psi_{z_2})}{\longrightarrow}
 \tau^*_{z_1}(\tau^*_{z_2}\bpol)=\tau^*_{z_1 + z_2}(\bpol),$$
 and, as $a^{-1}(z_1+z_2)=x_1+x_2$, we deduce from the functoriality of $[a]^*$ the following
 commutative diagram, which shows that $P_a(\bpol)$ is a group morphism:
	 \begin{center}
	 \begin{tikzpicture}[ampersand replacement=\&]
		\matrix (m) [matrix of math nodes,row sep=2em,column sep=7em,minimum
		width=2em]
	 {
		 [a]^*(\bpol) \& {[a]}^*(\tau^*_{z_1} \bpol) \&
		 {[a]}^*(\tau^*_{z_1}(\tau^*_{z_2} \bpol)) \\
		{[a]}^*(\bpol) \& \tau^*_{x_1}({[a]}^*(\bpol)) \& \tau^*_{x_1+x_2}
		({[a]}^*(\bpol)) \\
	 };
		\path[->] (m-1-1) edge node[above]{$[a]^*(\psi_{z_1})$} (m-1-2)
		(m-1-2) edge node[above]{$[a]^*(\tau^*_{z_1}(\psi_{z_2}))$} (m-1-3)
		(m-2-1) edge node[above]{$[a]^*(\psi_{z_1})$} (m-2-2)
		(m-2-2) edge node[above]{$\tau^*_{x_1}([a]^*(\psi_{z_2}))$} (m-2-3);
		\draw[double distance=1pt] (m-1-1) -- (m-2-1);
 \draw[double distance=1pt] (m-1-2) -- (m-2-2);
 \draw[double distance=1pt] (m-1-3) -- (m-2-3);
	 \end{tikzpicture}
	 \end{center}
 For $c \in \overk^*$, $[a]^*(c)=c$, so that $P_a(\bpol)$ is the identity on $\overk^*$; as $a$
 is a bijection of $B[m]$, $P_a(\bpol)$ is injective. The elements of $\tildeK_a$ are the
 $(\tau_k, \mathrm{can}_k)$ for $k \in B[a]$, where $\mathrm{can}_k: [a]^*(\bpol) \rightarrow
 \tau_k^*([a]^*(\bpol))=[a]^*(\bpol)$ is the canonical identification coming from $[a] \circ
 \tau_k=[a]$. As an isomorphism of the form $[a]^*(\psi_z)$ is invariant under these
 identifications, $P_a(\bpol)(G(\bpol))$ commutes with $\tildeK_a$. Finally, by Definition
~\ref{def:fsharp} (applied with $\psi$ the identity of $[a]^*(\bpol)$), $[a]^\sharp((\tau_x,
 [a]^*(\psi_z)))=(\tau_{ax}, \psi_z)=(\tau_z, \psi_z)$, which proves (1).

 For (2), $\ea(\bpol)$ is a group morphism as a composition of group morphisms. If
 $g=(\tau_x, \psi_x)$ with $x \in B[m]$, then $\delta_a(\bpol)(g)$ is above $ax$ and $a^{-1}(ax)=x$,
 which gives the formula for $\ea(\bpol)(g)$. For $c \in \overk^*$, $\ea(\bpol)(c)=
 P_a(\bpol)(c^{a^2})=c^{a^2}$. If $\ea(\bpol)(g)=1$, then $x=0$, so that $g=c \in \overk^*$ and
 $c^{a^2}=1$; conversely $\mu_{a^2}$ is contained in the kernel. As a level subgroup meets
 $\overk^*$ trivially, $\ea(\bpol)$ is injective on it. The last claim of (2) follows from
 $\delta_a(\bpol)(g)=g^a$ for $g$ symmetric.

 For (3), let $s \in \Gamma(B,\bpol)$, $g \in G(\bpol)$ and $\delta_a(\bpol)(g)=(\tau_{ax}, \psi')$.
 We have:
 \begin{gather}
	\begin{aligned}\label{eq:fonctorial}
	 [a]^*(\delta_a(\bpol)(g)(s)) & = [a]^* (\psi'^{-1} \circ \tau^*_{ax}(s)) \\
	 & = [a]^*(\psi'^{-1}) \circ [a]^* (\tau_{ax}^*(s)) \\
	 & = (([a]^*(\psi'))^{-1}) \circ (\tau^*_x ([a]^*(s))) = \ea(\bpol)(g)([a]^*(s)).
	\end{aligned}
 \end{gather}
\end{proof}

\begin{remark}\label{rk:prodiso}
 The map $g \mapsto (\tau_x, [a]^*(\psi_{ax}))$, where $g^a=(\tau_{ax}, \psi_{ax})$, which coincides
 with $\ea(\bpol)$ on symmetric elements, is not a group morphism in general, because $g \mapsto
 g^a$ is not a group morphism of the non-commutative group $G(\bpol)$: if $g, h \in G(\bpol)$,
 then $(gh)^a=[h,g]^{a(a-1)/2} g^a h^a$, where $[h,g]=hgh^{-1}g^{-1} \in \overk^*$ (for instance,
 for $a=3$, $m=4$ and $g$, $h$ above $x, y$ with $e_\bpol(x,y)$ a primitive $4$-th root of unity,
 the factor is non-trivial). This is why we use $\delta_a(\bpol)$. In the explicit computations below, $\ea$ is
 only evaluated on symmetric elements (elements of level subgroups of symmetric theta
 structures, or of symmetric descent data), for which we can use the formula $\ea(\bpol)(g)=
 P_a(\bpol)(g^a)$.
\end{remark}

From the preceding lemma, we get the following corollary which gives a generalization of
Mumford's map $\epsilon_d$ defined in \cite{MumfordOEDAV1} to line bundles built as tensor products of
pullbacks by isogenies:
\begin{corollary}\label{cor:genmum}
Let $(B, \bpol, \Thetabpol)$ be a marked abelian variety of type $K(m)$.
 Let $d=\sum_{j=1}^r a_j^2$ where $a_j$ are positive integers such that $\gcd(a_j,m)=1$. There
 exists a morphism of theta groups
 \begin{gather}
	\begin{aligned}
	 \eaj(\bpol): G(\bpol) & \rightarrow
	 G(\otimesbpol), \\ 
	 g=(\tau_x, \psi_x) & \mapsto (\tau_x,
	 \otimes_{j=1}^r [a_j]^*(\psi'_j)),
	\end{aligned}
 \end{gather}
 where $\delta_{a_j}(\bpol)(g)=(\tau_{a_jx}, \psi'_j)$ for $j=1, \ldots, r$. It satisfies
 $\eaj(\bpol)(c)=c^d$ for $c \in \overk^*$, it is injective on every level subgroup of $G(\bpol)$,
 and it is compatible with the action on products of sections: for all $s_1, \ldots, s_r \in
	 \Gamma(B,\bpol)$ and $g \in G(\bpol)$, we have:
	 \begin{gather}
		\begin{aligned}
		\eaj(\bpol)(g) (\otimes_{j=1}^r [a_j]^*(s_j))
		 & =\otimes_{j=1}^r \eajs(\bpol)(g)([a_j]^*(s_j)) \\
		 &= \otimes_{j=1}^r [a_j]^* (\delta_{a_j}(\bpol)(g)(s_j)).
		\end{aligned}
	 \end{gather}
 If $g$ is symmetric, $\delta_{a_j}(\bpol)(g)=g^{a_j}$ in these formulas.
\end{corollary}
\begin{proof}
 For $j=1, \ldots, r$, the elements $\eajs(\bpol)(g)=(\tau_x, [a_j]^*(\psi'_j))$ given by Lemma
~\ref{lem:prodiso} are all above the same point $x$. As the map $(\tau_x, \varphi_1), \ldots,
 (\tau_x, \varphi_r) \mapsto (\tau_x, \otimes_{j=1}^r \varphi_j)$ from the fibered product of the
 $G([a_j]^*(\bpol))$ over $K(\bpol)$ to $G(\otimesbpol)$ is a group morphism, $\eaj(\bpol)$ is a
 group morphism. On $\overk^*$, it is $c \mapsto \prod_j c^{a_j^2}=c^d$, and it is injective on
 level subgroups because each $\eajs(\bpol)$ is. The compatibility with the action on sections
 follows from Lemma~\ref{lem:prodiso} (3) applied componentwise.
\end{proof}
We also deduce immediately:
\begin{corollary}\label{cor:productsection}
With the same hypotheses as in the preceding corollary, set
 $\Thetabpoldij = \eaj \circ \Thetabpol$.
 Let $s_0 \in \Gamma(B,\bpol)$ then $s = \prod_{j=1}^r
	 [a_j]^*(s_0) \in \Gamma(B, \otimesbpol)$. 

	 Moreover, we have for all $t \in G(m)$:
	 \begin{equation}
		\Thetabpoldij(t)(s)=\prod_{j=1}^r \eajs(\bpol)(\Thetabpol(t))[a_j]^*(s_0).
	 \end{equation}
\end{corollary}
\begin{lemma}\label{lem:tensorquo}
 Let $(A, \pol)$ and $(B, \bpol)$ be abelian varieties and suppose that there exists an
 isogeny $f: A \rightarrow B$ with kernel $K$ isotropic for $e_\pol$. Suppose that
 $\tildeK$ is a level subgroup which is the descent data of $\pol$ to $\bpol$ then
 $\eaj(\pol)(\tildeK)$ is the descent data of 
 $\polaj$ to $\bpolaj$. Let $\psi: f^*(\bpol) \rightarrow \pol$ be the isomorphism
 associated to $\tildeK$; then $\otimes_{j=1}^r [a_j]^* \psi: f^*((\otimesbpol))
 \rightarrow \otimespol$ is the isomorphism associated to $\eaj(\tildeK)$.
 In particular, if $s \in \Gamma(B,\bpol)$, we have 
 \begin{equation}
 (\otimes_{j=1}^r [a_j]^*\psi) (f^*(\otimes_{j=1}^r [a_j]^*(s)))= \otimes_{j=1}^r [a_j]^*(f^*(s)).
 \end{equation}
\end{lemma}
\begin{proof}
 As $\bpol$ is symmetric, $\tildeK$ is symmetric (Remark~\ref{rm:sym}), so that, by Lemma
~\ref{lem:prodiso} (2), for $g=(\tau_x, \psi_x) \in \tildeK$ we have $\eajs(\pol)(g)=(\tau_x,
 [a_j]^*(\psi_{a_jx}))$ where $g^{a_j}=(\tau_{a_jx}, \psi_{a_jx}) \in \tildeK$.
 For $x \in K$, by definition of $\psi$ (see Diagram (\ref{eq:descentdata})) we have
 $(\tau_{a_j x}, \psi_{a_jx}) \in \tildeK$ if and only if $\psi_{a_j x}
 =\tau^*_{a_jx}(\psi) \circ \psi^{-1}$. But then $[a_j]^*(\psi_{a_jx})= [a_j]^* (\tau^*_{a_j
 x}(\psi) \circ
 \psi^{-1})=\tau^*_{x}([a_j]^*(\psi)) \circ [a_j]^*(\psi)^{-1}$. By taking the tensor
 product, we obtain
 $\otimes_{j=1}^r [a_j]^* (\psi_{a_j x})=
 \otimes_{j=1}^r [a_j]^*(\tau^*_{a_j x}(\psi) \circ \psi^{-1})=
 \otimes_{j=1}^r \tau^*_x([a_j]^*(\psi)) \circ [a_j]^*(\psi)^{-1}=
 \tau^*_x(\otimes_{j=1}^r [a_j]^*(\psi)) \circ (\otimes_{j=1}^r [a_j]^*(\psi)^{-1})$.
\end{proof}
Now let $(B, \bpol, \Thetabpol)$ and $(A, \pol, \Thetapol)$ be isog-$f$-compatible
marked abelian varieties
of respective types $K(m)$ and $K(n)$. Let $K$ be the kernel of $f:A \rightarrow B$ and
let $\tildeK$ be the level subgroup above $K$ which is the descent data of
$f^*(\bpol)=\pol$ to $\bpol$. Then it is clear that $\eaj(\tildeK)$ is the descent data of
$\polaj$ to $\bpolaj$. This allows us to state the following definition:
\begin{definition}\label{def:fcompatibleC}
 Let $m,n,d > 1$ be integers such that $n=md$.
 Let $a_j$ for $j=1, \ldots, r$ be positive integers such that $d=\sum_{j=1}^r a_j^2$
 and we suppose that for $j=1, \ldots, r$, $\gcd(a_j, n)=1$.

 We keep the same notation as in Definition~\ref{def:fcompatibleB} for $E'_d(n)$. Let
$\Thetaubpol: \overk^* \times Z(m) \rightarrow G(\bpol)$ and $\Thetaubpoldij
 : \overk^* \times \Zn \rightarrow G(\otimesbpol)$ be 
 partial symmetric theta structures.

 We say that they are $f$-compatible if there
 exists $(A, \pol, \Thetapol)$ an isog-$f$-compatible marked abelian variety such that we
 have the equality of maps $\overk^* \times \Zn \rightarrow G(\otimesbpol)$:
 \begin{equation}
	\Thetaubpoldij \circ E'_d(n) = \fsharp(\polaj) \circ \eaj(\pol) \circ
	\Thetaupol,
 \end{equation}
 where $f^{\sharp}(\polaj): G(\polaj)^* \rightarrow G(\bpolaj)$ is the map defined by
 $\eaj(\tildeK)$ in Definition~\ref{def:fsharp}.
\end{definition}
\begin{remark}\label{rk:extension}
One verifies that the results proved for $f$-compatible theta structures following
Definition~\ref{def:fcompatibleB} extend
immediately to $f$-compatible theta structures in the sense of Definition
\ref{def:fcompatibleC}. In particular, Proposition~\ref{prop:uniqueness} applies mutatis
mutandis to $f$-compatible theta structures in general.
\end{remark}

Let $(B, \bpol, \Thetabpol)$ be a marked abelian variety of type $K(m)$.
The following proposition explains how we can use the computation with affine lifts to
compute the action of $G(n)$ via a theta structure of type $K(n)$ on sections of
$\bpol^d$. This is a key ingredient for the main theorems of this section.
\begin{proposition}\label{prop:thetadaction}
 Let $m,n,d >1$ be integers such that $n=md$. Suppose that there exist
 positive integers $(a_j)_{j=1, \ldots, r}$ such that $d= \sum_{j=1}^r a_j^2$ and $\gcd(a_j, n)=1$.
 Let $(B, \bpol, \Thetabpol)$ be a marked abelian variety of type $K(m)$.
 We consider the following hypotheses:
 \begin{enumerate}
	\item[(H1)] $G_1$ is a subgroup of $B[n]$ isomorphic to $Z(n)$, isotropic for $e_{B,n}$,
	 containing $\Thetabpolbar(Z(m) \times \{ 0 \})$, whose elements are symmetric compatible with
	 $\Thetabpol(\{ 1 \} \times Z(m) \times \{0 \} )$, together with a group isomorphism $g_1: \Zn
	 \rightarrow G_1$ such that $g_1(\mu_{m,n}(i))=\Thetabpolbar((i,0))$ for $i \in Z(m)$, and a good
	 lift $\tildeG_1$ of $G_1$;
	\item[(H2)] $G_2$ is a subgroup of $B[n]$ isomorphic to $Z(n)$, isotropic for $e_{B,n}$,
	 containing $\Thetabpolbar(\{ 0 \} \times \dZ(m))$, whose elements are symmetric compatible with
	 $\Thetabpol(\{ 1 \} \times \{ 0 \} \times \dZ(m))$, together with a group isomorphism $g_2:
	 \dZn \rightarrow G_2$ such that $g_2(\dnu_{m,n}(\nu))=\Thetabpolbar((0,\nu))$ for $\nu \in
	 \dZ(m)$, and a good lift $\tildeG_2$ of $G_2$;
	\item[(H3)] (H1) and (H2) hold, and the numberings are dual: for all $i \in \Zn$ and $\nu \in
	 \dZn$, $e_{B,n}(g_1(i), g_2(\nu))=e_n((i,0),(0,\nu))$.
 \end{enumerate}
 Let $U$ be an affine open subset of $B$ and choose an isomorphism $\bpol(U) \simeq
 \stsheaf_B(U)$, so that for all $s \in \Gamma(B,\bpol)$ and all $z \in U(\overk)$ we can
 evaluate $s$ in $z$; we denote by $s(z) \in \overk$ the evaluation. If $z$ and $a_jz$ are in
 $U$ for $j=1, \ldots, r$, we evaluate the sections of $\otimes_{j=1}^r [a_j]^*\bpol$ at $z$
 using these trivializations at the points $a_jz$. For $\beta=(\beta_1, \ldots, \beta_r) \in
 Z(m)^r$, let $s_\beta=\prod_{j=1}^r [a_j]^*(\theta^\Thetabpol_{\beta_j}) \in \Gamma(B,
 \otimesbpol)$. For such a $z$ and an affine lift $\tildez$ of $z$, the affine point
 $a_j(\tildez)=\scalmult(a_j, \tildez, \tildez, \tildenullbpol, \tildenullbpol)$ is an affine lift of
 $a_jz$, so that there is a unique linear form $\mathrm{ev}_{\tildez}$ on $\Gamma(B,
 \otimesbpol)$ which is a nonzero multiple of the evaluation at $z$ and satisfies, for all
 $\beta \in Z(m)^r$:
 \begin{equation}\label{eq:defev}
	\mathrm{ev}_{\tildez}(s_\beta)=\prod_{j=1}^r (a_j(\tildez))_{\beta_j}.
 \end{equation}
 In what follows, all the points at which we evaluate (and their multiples by $a_j$) are
 supposed to be in $U$.
 \begin{enumerate}
	\item Under (H1), let $\hatf: B \rightarrow A=B/G_1[d]$ be the quotient isogeny and $f: A
	 \rightarrow B$ its contragredient isogeny. There exists a partial symmetric theta structure
	 $\Thetabpoldij: \overk^* \times \Zn \rightarrow G(\otimesbpol)$, $f$-compatible with
	 $\Thetabpol$ (Definition~\ref{def:fcompatibleC}), such that $\Thetabpoldijbar((i,0))=g_1(i)$
	 for $i \in \Zn$, and such that, for all $z \in B(\overk)$, all affine lifts $\tildez$ of
	 $z$ and all $j_1 \in \Zn$, if $\widetilde{z+g_1(j_1)}$ is the good lift of $z+g_1(j_1)$ with
	 respect to $\tildez$ and $\tildeG_1$:
	 \begin{equation}\label{eq:thetadaction1}
		\mathrm{ev}_{\widetilde{z+g_1(j_1)}}=\mathrm{ev}_{\tildez} \circ
		\Thetabpoldij((1,j_1,0)).
	 \end{equation}
	\item Under (H2), there exists a partial symmetric theta structure $\Thetabpoldij:
	 \overk^* \times \dZn \rightarrow G(\otimesbpol)$ such that $\Thetabpoldijbar((0,\nu))=g_2(\nu)$
	 for $\nu \in \dZn$, and such that, for all $z \in B(\overk)$, all affine lifts $\tildez$ of
	 $z$ and all $j_2 \in \dZn$, if $\widetilde{z+g_2(j_2)}$ is the good lift of $z+g_2(j_2)$ with
	 respect to $\tildez$ and $\tildeG_2$:
	 \begin{equation}\label{eq:thetadaction2}
		\mathrm{ev}_{\widetilde{z+g_2(j_2)}}=\mathrm{ev}_{\tildez} \circ
		\Thetabpoldij((1,0,j_2)).
	 \end{equation}
	\item Under (H3), the partial theta structures of (1) and (2) are the restrictions of a
	 unique theta structure $\Thetabpoldij$ of type $K(n)$ for $\otimesbpol \simeq \bpol^d$.
	 Let $x \in B(\overk)$ and $\tildex$ be an affine lift of $x$; for $j_1 \in \Zn$, let
	 $\widetilde{x+g_1(j_1)}$ be the good lift of $x+g_1(j_1)$ with respect to $\tildex$ and
	 $\tildeG_1$, and for $(j_1, j_2) \in \Zn \times \dZn$, let $\widetilde{x+g_1(j_1)+g_2(j_2)}$
	 be the good lift of $x+g_1(j_1)+g_2(j_2)$ with respect to $\widetilde{x+g_1(j_1)}$ and
	 $\tildeG_2$. Then there exists a constant $C \in \overk^*$, independent of $\beta$, $j_1$ and
	 $j_2$, such that for all $\beta \in Z(m)^r$ and $(j_1, j_2) \in \Zn \times \dZn$:
	 \begin{equation}\label{prop26:eq12}
		\begin{split}
		\prod_{j=1}^r(a_j(\widetilde{x+g_1(j_1)+g_2(j_2)}))_{\beta_j} & =C\,
		\big(\Thetabpoldij((1,j_1,0))\Thetabpoldij((1,0,j_2))(s_\beta)\big)(x) \\
		& =C\, j_2(j_1)\,\big(\Thetabpoldij((1,j_1,j_2))(s_\beta)\big)(x).
		\end{split}
	 \end{equation}
 \end{enumerate}
\end{proposition}
\begin{proof}
  \emph{Proof of (1).} To lighten the notation, we write $x$ for the point $z$ of the
 statement and $\tildex$ for $\tildez$.

 By Theorem~\ref{th:main1}, as we have chosen a good lift $\tildeG_1$ of $G_1$ with
 respect to $\tildenullbpol$, we have fixed
 $(A, \pol, \Thetapol)$ $f$-compatible with $(B, \bpol, \Thetabpol)$. By Proposition
~\ref{prop:charac} and Proposition~\ref{sec2:lemmalambdai}, the kernel of $\hatf$ is
 $f(A[d])=mG_1=G_1[d]$, so that $f$ is the isogeny of the statement.
Let $K$ be the kernel of $f$ and $\tildeK$ be the level subgroup above $K$ defined by the
 descent data $f^*(\bpol)=\pol$.
  Let $\Thetabpoldij: \overk^* \times \Zn \rightarrow G(\otimesbpol)$ be the partial theta
 structure which is $f$-compatible with $\Thetabpol$ through $(A, \pol, \Thetapol)$
 (Definition~\ref{def:fcompatibleC}), that is such that, when considered as maps $\overk^* \times \Zn \subset G(n)
 \rightarrow G(\otimes_{j=1}^r [a_j]^*(\bpol))$, we have the equality:
 \begin{equation}
 \Thetabpoldij \circ E'_d(n)= \fsharp(\polaj) \circ \eaj(\pol) \circ
	\Thetapol,
 \end{equation}
  where $\fsharp(\polaj)$ is defined by $\eaj(\pol)(\tildeK)$ following Lemma
~\ref{lem:tensorquo}. As $\eaj(\pol)$ lies above the identity of $\Kpol$ (Lemma
~\ref{lem:prodiso}), $\fsharp(\polaj)$ lies above $f$ and $f(\Thetapolbar((i,0)))=g_1(i)$ for $i \in \Zn$
 (Definition~\ref{sec1:defG} and Proposition~\ref{sec2:lemmalambdai}), we have
 $\Thetabpoldijbar((i,0))=g_1(i)$ for $i \in \Zn$.

 Let $y \in A(\overk)$ be such that $f(y)=x$. Let $\tildey$ be an affine lift of $y$
	such that $\tildef(\tildey)= C*\tildex$ for $C \in \overk^*$. Then, for $j_1 \in
	Z(n)$, $\tildef(
	\Thetapol((1, j_1,
	0))\tildey)= C* (\widetilde{x + g_1(j_1)})^\Thetabpol$ by Definition~\ref{def:excelz} of an excellent lift. We
	thus have:
	\begin{gather}\label{eq:prop26:tech1}
	 \begin{aligned}
		(a_j (\widetilde{x + g_1(j_1)}^\Thetabpol))_\alpha
		& = C a_j \tildef(\Thetapol((1,j_1, 0)) \tildey)_\alpha & \\
		& = C \tildef(a_j (\Thetapol((1,j_1, 0)) \tildey))_\alpha & (\text{because of Lemma
		\ref{lem:tildefscalmult}}) \\
		& = C (a_j (\Thetapol((1,j_1, 0)) \tildey))_{\mu_{m,n}(\alpha)} & (\text{by
		definition of $\tildef$}) \\
	& = C ((a_j \Thetapol((1,j_1, 0))) (a_j \tildey))_{\mu_{m,n}(\alpha)} &
	(\text{because of Corollary~\ref{cor:diffaddaction}}) \\
		& = C [a_j]^* (\Thetapol((1,a_j (\mu_{m,n}(j_{1,m})+j_{1,n}), 0)) (\theta^{\Thetapol}_{\mu_{m,n}(\alpha)}))(y) &
		(\text{by definition})\\
		& = C \eajs(\pol)(\Thetapol((1, j_{1,n},0)))
		[a_j]^*(\theta^{\Thetapol}_{\mu_{m,n}(a_j j_{1,m}+\alpha)})(y). &
		(\text{because of Lemma~\ref{lem:prodiso}})
	 \end{aligned}
	\end{gather}
For the last equality, we used the fact that $\Thetapol((1, j_{1,n}, 0))$ is symmetric, so
 that, by Lemma~\ref{lem:prodiso} (2), $\eajs(\pol)$ can be computed with $\Thetapol((1, j_{1,n},
 0))^{a_j}= \Thetapol((1, a_j j_{1,n}, 0))$. As a consequence, we have:
 \begin{gather}\label{eq:prop26eq16}
	\begin{aligned}
	 \prod_{j=1}^r(a_j(\widetilde{x+g_1(j_1)}^\Thetabpol))_\alpha &= C^r
	 \prod_{j=1}^r \eajs(\pol)(\Thetapol((1, j_{1,n},0)))
	 [a_j]^*(\theta^{\Thetapol}_{\mu_{m,n}(a_j j_{1,m}+\alpha)})(y) & (\text{because of Equation
		(\ref{eq:prop26:tech1})}) \\
		&= C^r \Thetapoldij((1,j_{1,n},0))(\prod_{j=1}^r
		[a_j]^*(\theta^\Thetapol_{\mu_{m,n}(a_j j_{1,m}+\alpha)}))(y) & (\text{because of Corollary
		\ref{cor:productsection}})\\
		&= C^r \Thetabpoldij((1,j_{1,n},0))(\prod_{j=1}^r
		[a_j]^*(\theta^\Thetabpol_{a_j j_{1,m}+\alpha}))(x).
		& \\
	\end{aligned}
 \end{gather}
  We get the last equality by applying Lemma~\ref{lem:tensorquo} and Corollary
		\ref{cor:canonical}.

 The computation (\ref{eq:prop26:tech1}) is done factor by factor, so that
 (\ref{eq:prop26eq16}) still holds when the index $\alpha$ is replaced by $\beta_j$ in the
 $j$-th factor; moreover, the last two lines of (\ref{eq:prop26:tech1}) show that
 $\Thetabpoldij((1,\mu_{m,n}(j_{1,m}),0))(s_\beta)=s_{(a_jj_{1,m}+\beta_j)_j}$. As
 $\Thetabpoldij((1,j_{1,n},0))\Thetabpoldij((1,\mu_{m,n}(j_{1,m}),0))=\Thetabpoldij((1,j_1,0))$, we
 get, for all $\beta \in Z(m)^r$ and $j_1 \in \Zn$:
 \begin{equation}\label{eq:prop26eq16b}
	\mathrm{ev}_{\widetilde{x+g_1(j_1)}}(s_\beta)=C^r\big(\Thetabpoldij((1,j_1,0))(s_\beta)\big)(x).
 \end{equation}
 For $j_1=0$, we have $\tildef(\tildey)=C*\tildex$ and $\tildef(\tildey)=C*\widetilde{x+g_1(0)}$, so
 that $\widetilde{x+g_1(0)}=\tildex$, and (\ref{eq:prop26eq16b}) gives $\mathrm{ev}_{\tildex}(s_\beta)=C^r
 s_\beta(x)$ for all $\beta$. Choosing $\beta_j$ such that $\theta^\Thetabpol_{\beta_j}(a_jx) \neq 0$
 for $j=1, \ldots, r$, we get $s_\beta(x) \neq 0$; as $\mathrm{ev}_{\tildex}$ is a multiple of the
 evaluation at $x$, $\mathrm{ev}_{\tildex}(s)=C^rs(x)$ for all $s \in \Gamma(B, \otimesbpol)$. Now,
 the two linear forms $\mathrm{ev}_{\widetilde{x+g_1(j_1)}}$ and $s \mapsto
 \mathrm{ev}_{\tildex}(\Thetabpoldij((1,j_1,0))(s))$ are both multiples of the evaluation at
 $x+g_1(j_1)$ (for the second one by Lemma~\ref{sec1:lemact}, since $\Thetabpoldij((1,j_1,0))$ lies
 above $g_1(j_1)$). By (\ref{eq:prop26eq16b}), they coincide on the sections $s_\beta$, among which
 there is one which does not vanish at $x+g_1(j_1)$, so that they are equal: this is
 (\ref{eq:thetadaction1}).

 \emph{Proof of (2).} Again, we write $x$ for $z$ and $\tildex$ for $\tildez$. Let $(e_i,
 \de_i)_{i=1,\dots,g} \in (Z(m) \times \dZ(m))^g$ be the canonical symplectic basis of $K(m)$ and
 consider the automorphism $\Delta(m): G(m) \rightarrow G(m)$ (corresponding to $H_g$
 of Section~\ref{sec:trans}) such that for all $i=1,
 \ldots, g$, $\Delta(m)((1, e_i,
 0))= (1, 0, -\de_i)$ and $\Delta(m)((1, 0, \de_i))=(1, e_i, 0)$. We then consider the
 theta structure $\Thetabpoldel = \Thetabpol \circ \Delta(m)$. It is clear that
 $\Delta(m)$
 exchanges the roles of $Z(m)$ and $\dZ(m)$ in the theta structure, so that we can do
 the same thing as before, replacing $G_1$ by $G_2$. We denote by $\Deltambar:
 K(m) \rightarrow K(m)$ the map induced by $\Deltam$ on $K(m)$.
 Note that by Proposition~\ref{prop:base}, we have for $\nu \in Z(m)$ and $C \in
 \overk^*$,
 \begin{equation}
	\theta^{\Thetabpoldel}_{\nu} = C \frac{1}{m^g}\sum_{\nu' \in Z(m)} \Deltambar(\nu)(\nu')
	\theta^{\Thetabpol}_{\nu'}.
 \end{equation}
Consider the morphism of the affine spaces:
\begin{gather}
\begin{aligned}
\tildeH_{\Delta(m)}: \Aff^{Z(m)} & \rightarrow \Aff^{Z(m)} \\
(\alpha_\nu)_{\nu\in Z(m)} & \rightarrow \frac{1}{m^g} (\sum_{\nu' \in Z(m)} \Deltambar(\nu)(\nu') \alpha_{\nu'})_{\nu\in Z(m)}.
\end{aligned}
\end{gather}
We denote by $H_{\Delta(m)}: \proj^{Z(m)} \rightarrow \proj^{Z(m)}$ the projective
morphism induced by $\tildeH_{\Delta(m)}$. It is clear that $H_{\Delta(m)}$ maps points of
$e_{\Thetabpol}(B)$ to points of $e_{\Thetabpoldel}(B)$. In particular, we have
$0_{\Thetabpoldel}=H_{\Delta(m)}(0_{\Thetabpol})$. So we can set
$\widetilde{0}_{\Thetabpoldel}=\tildeH_{\Delta(m)}(\tildenullbpol)$. Then $\tildeH_{\Delta(m)}$ maps the good lifts of $G_2$ with respect to $\tildenullbpol$ (for
$\Thetabpol$) to the good lifts of $G_2$ with respect to $\widetilde{0}_{\Thetabpoldel}$ (for
$\Thetabpoldel$). This is an immediate
consequence of Definition~\ref{def:good} of a good lift and the facts:
\begin{enumerate}
 \item if $\tildex, \tildey,
	\tildexmy \in \Aff^{Z(m)}$ are affine points for
	$\tildenullbpol$ (lift of projective points of $e_{\Thetabpol}(B)$) then 
	\begin{equation}
	 \tildeH_{\Delta(m)}(\diffadd ( \tildex, \tildey, \tildexmy,
	 \tildenullbpol))=
	 \diffadd ( \tildeH_{\Delta(m)}(\tildex), \tildeH_{\Delta(m)}(\tildey),
	 \tildeH_{\Delta(m)}(\tildexmy), \tildeH_{\Delta(m)}(\tildenullbpol)) \in \Aff^{Z(m)};
	\end{equation}
 \item $\inv(\tildeH_{\Delta(m)}(\tildex))=\tildeH_{\Delta(m)}(\inv(\tildex))$;
 \item if $t \in G(m)$, $\tildeH_{\Delta(m)}(t.\tildex)= \Delta(m)^{-1}(t).
	\tildeH_{\Delta(m)}(\tildex)$. 
\end{enumerate}
Indeed, fact (1) is an immediate consequence of the fact that $\diffadd$ is computed with
Riemann equations, which are satisfied for any choice of theta functions associated to a totally
symmetric theta structure. Fact (2) comes from the fact that $\Delta(m)(-\nu)(\nu')=\Delta(m)(\nu)(-\nu')$.
Finally, for fact (3), for $t \in G(m)$, we remark that $\Thetabpol(t) =
\Thetabpoldel(\Delta(m)^{-1}(t))$. Thus
$\tildeH_{\Delta(m)}(t.\tildex)=\tildeH_{\Delta(m)}(\Thetabpol(t)
\tildex)=\Thetabpoldel(\Delta(m)^{-1}(t)) \tildeH_{\Delta(m)}(\tildex)=\Delta(m)^{-1}(t).
\tildeH_{\Delta(m)}(\tildex)$.

By Theorem~\ref{th:main1}, as we have chosen $\tildeG_2^\Thetabpoldel \subset \Aff^{Z(m)}$ a good lift of $G_2$ with
 respect to $\widetilde{0}_{\Thetabpoldel}$, we have fixed
 $(A', \polp, \Thetapolpdel)$ isog-$f'$-compatible with $(B, \bpol, \Thetabpoldel)$,
 $f'$ being the contragredient isogeny of $\hatf':
 B \rightarrow A'$, where $\hatf'$ has kernel $G_2[d]=mG_2$. Indeed,
 $\Thetabpoldelbar(Z(m)\times\{0\})=\Thetabpolbar(\{0\}\times\dZ(m)) \subset G_2$, so that
 $G_2$ plays for $\Thetabpoldel$ the role of $G_1$ for $\Thetabpol$, and by Proposition
~\ref{prop:charac} the kernel of $\hatf'$ is $f'(A'[d])=mG_2$ (when $d \mid m$, it is
 $\Thetabpolbar(\{0\}\times\dnu_{d,m}(\dZ(d)))$).
Let $K'$ be the kernel of $f'$ and $\tildeK'$ be the level subgroup above $K'$ defined by the
 descent data $f'^*(\bpol)=\polp$.
  We define the partial theta structure $\Thetabpoldij$ on $\overk^* \times \dZn$ by requiring
 that $\Thetabpoldijdel =\Thetabpoldij \circ \Delta(n)$ satisfies, as maps $\overk^* \times Z(n)
 \subset G(n) \rightarrow G(\otimes_{j=1}^r [a_j]^*(\bpol))$, the equality (this makes sense since
 $\Delta(n)$ maps $\overk^*\times \Zn \times\{0\}$ onto $\overk^* \times\{0\}\times\dZn$):
 \begin{equation}
 \Thetabpoldijdel \circ E'_d(n) = f'^\sharp(\polajp) \circ \eaj(\polp) \circ
	\Thetadpolp,
 \end{equation}
  where $f'^\sharp(\polajp)$ is defined by $\eajs(\polp)(\tildeK')$ following Lemma
~\ref{lem:tensorquo}. Here, Theorem~\ref{th:main1} is applied to $\Thetabpoldel$ and to $G_2$
 numbered by $i \mapsto g_2(-\iota_n(i))$, where, for $k=m,n$, $\iota_k: Z(k) \rightarrow \dZ(k)$ is
 the isomorphism mapping the canonical basis of $Z(k)$ to its dual basis: as
 $\Delta(m)((1,i,0))=(1,0,-\iota_m(i))$ and $\iota_n \circ \mu_{m,n}=\dnu_{m,n}\circ \iota_m$, this
 numbering satisfies the hypothesis of Theorem~\ref{th:main1} for $\Thetabpoldel$ by (H2). Let
 $\Thetapolp=\Thetadpolp \circ \Delta(n)^{-1}$. As $\Delta(n)((1,i,0))=(1,0,-\iota_n(i))$, the
 excellent lifts given by Theorem~\ref{th:main1} satisfy the relation below with
 $\Thetapolp((1,0,j_2))$, and we get $\Thetabpoldijbar((0,\nu))=g_2(\nu)$ for $\nu \in \dZn$.

 Let $y \in A'(\overk)$ be such that $f'(y)=x$. Let $\tildey$ be an affine lift of $y$
	such that $C'*\tildef'(\tildey)= \tildex$ for $C' \in \overk^*$. Then, for $j_2 \in
	\dZ(n)$, $C'*\tildef'(
	\Thetapolp((1, 0,
	j_2))\tildey)= \widetilde{x + g_2(j_2)}^{\Thetabpol}$ by Definition~\ref{def:excelz} of an
	excellent lift.

	Note that, as 
$(A', \pol', \Thetapolpdel)$ and $(B, \bpol,
	\Thetabpoldel)$ are isog-$f'$-compatible by construction,
	$(A', \pol', \Thetapolp)$ and $(B, \bpol, \Thetabpol)$ are
	dual-isog-$f'$-compatible.
	Recall from Definition~\ref{def:notationmaps} that $\rho_{n,m}: \Zn \rightarrow Z(m) \simeq \Zn/ \mu_{d,n}(Z(d))$ is
	the canonical projection.
	Doing the same computations as in Equation (\ref{eq:prop26:tech1}), we obtain
	for $\alpha \in Z(m)$ and $C' \in \overk^*$:
	\begin{gather}\label{eq:prop26:tech2}
	 \begin{aligned}
		(a_j (\widetilde{x + g_2(j_2)}^\Thetabpol))_\alpha
		& = C' a_j \tildef'(\Thetapolp((1,0, j_2)) \tildey)_\alpha & \\
		& = C' \tildef'(a_j (\Thetapolp((1,0, j_2)) \tildey))_\alpha & (\text{by Lemma
		\ref{lem:tildefscalmult}}) \\
		& = C' \sum_{\nu \in \rho_{n,m}^{-1}(\alpha)} (a_j (\Thetapolp((1,0, j_2))
		\tildey))_{\nu} & (\text{by 
		definition of $\tildef'$}) \\
 & = C' \sum_{\nu \in \rho_{n,m}^{-1}(\alpha)} ((a_j
\Thetapolp((1,0, j_2))) (a_j \tildey))_{\nu} &
(\text{by Corollary~\ref{cor:diffaddaction}}) \\
 & = C' \sum_{\nu \in
 \rho_{n,m}^{-1}(\alpha)} [a_j]^* (\Thetapolp((1,0, a_j j_2)) (\theta^{\Thetapolp}_{\nu}))(y) &
(\text{by definition})\\
 & = C'\sum_{\nu\in\rho_{n,m}^{-1}(\alpha)}
 \eajs(\polp)(\Thetapolp((1,0,j_2)))
 ([a_j]^*(\theta^{\Thetapolp}_{\nu}))(y) &
		(\text{by Lemma~\ref{lem:prodiso}}) \\
 & = C'
 \eajs(\polp)(\Thetapolp((1,0,j_2)))
 ([a_j]^*(\sum_{\nu \in
 \rho_{n,m}^{-1}(\alpha)}  \theta^{\Thetapolp}_{\nu}))(y). &
		(\text{by linearity})
	 \end{aligned}
	\end{gather}

	Thus, we have:
 \begin{gather}\label{eq:prop26eq17}
	\begin{aligned}
	 \prod_{j=1}^r(a_j(\widetilde{x+g_2(j_2)}^\Thetabpol))_\alpha &= C'^r
	 \prod_{j=1}^r \eajs(\polp)(\Thetapolp((1,0,j_2)))
	 [a_j]^*(\sum_{\nu \in
 \rho_{n,m}^{-1}(\alpha)}  \theta^{\Thetapolp}_{\nu})(y) & (\text{because of Equation
		(\ref{eq:prop26:tech2})}) \\
		&= C'^r \Thetapoldij((1,0,j_2))(\prod_{j=1}^r
		[a_j]^*( \sum_{\nu \in
 \rho_{n,m}^{-1}(\alpha)}  \theta^{\Thetapolp}_{\nu}   ))(y) & (\text{because of Corollary
		\ref{cor:productsection}})\\
		&= C'^r \Thetabpoldij((1,0,j_2))(\prod_{j=1}^r
		[a_j]^*(\theta^\Thetabpol_{\alpha}))(x).
		& \\
	\end{aligned}
 \end{gather}

 As before, the computation (\ref{eq:prop26:tech2}) is done factor by factor, so that
 (\ref{eq:prop26eq17}) still holds when $\alpha$ is replaced by $\beta_j$ in the $j$-th factor:
 for all $\beta \in Z(m)^r$ and $j_2 \in \dZn$,
 $\mathrm{ev}_{\widetilde{x+g_2(j_2)}}(s_\beta)=C'^r(\Thetabpoldij((1,0,j_2))(s_\beta))(x)$. For $j_2=0$,
 $C'*\tildef'(\tildey)=\tildex$ and $C'*\tildef'(\tildey)=\widetilde{x+g_2(0)}$, so that
 $\widetilde{x+g_2(0)}=\tildex$ and, as in the proof of (1), $\mathrm{ev}_{\tildex}(s)=C'^rs(x)$ for all
 $s$. Comparing the two multiples of the evaluation at $x+g_2(j_2)$ given by
 $\mathrm{ev}_{\widetilde{x+g_2(j_2)}}$ and $s \mapsto \mathrm{ev}_{\tildex}(\Thetabpoldij((1,0,j_2))(s))$ as
 in the proof of (1), we get (\ref{eq:thetadaction2}).

 \emph{Proof of (3).} Suppose (H3) and let $i \in \Zn$, $\nu \in \dZn$. The commutator of
 $\Thetabpoldij((1,i,0))$ and $\Thetabpoldij((1,0,\nu))$ is
 $e_{\otimesbpol}(g_1(i), g_2(\nu))=e_{\bpol^d}(g_1(i), g_2(\nu))=e_{B,n}(g_1(i), g_2(\nu))$ (Proposition
~\ref{def:weil} (4) and (1)), which is equal by (H3) to $e_n((i,0),(0,\nu))$, the commutator of
 $(1,i,0)$ and $(1,0,\nu)$ in $G(n)$. By Remark~\ref{rk:partialtheta}, the partial theta structures of
 (1) and (2) are thus the restrictions of a unique theta structure $\Thetabpoldij: G(n)
 \rightarrow G(\otimesbpol)$, which is symmetric since its restrictions are. Let $x$, $\tildex$,
 the good lifts and $(j_1,j_2)$ be as in (3), and $\beta \in Z(m)^r$. Applying
 (\ref{eq:thetadaction2}) with $z=x+g_1(j_1)$ and $\tildez=\widetilde{x+g_1(j_1)}$, and then
 (\ref{eq:thetadaction1}) with $z=x$ to the section $\Thetabpoldij((1,0,j_2))(s_\beta)$, we get
 \begin{equation}
	\begin{split}
	\mathrm{ev}_{\widetilde{x+g_1(j_1)+g_2(j_2)}}(s_\beta) & =\mathrm{ev}_{\widetilde{x+g_1(j_1)}}
	(\Thetabpoldij((1,0,j_2))(s_\beta)) \\
	& =\mathrm{ev}_{\tildex}(\Thetabpoldij((1,j_1,0))\Thetabpoldij((1,0,j_2))(s_\beta)).
	\end{split}
 \end{equation}
 By the proof of (1), $\mathrm{ev}_{\tildex}(s)=Cs(x)$ for all $s$, with $C \in \overk^*$ depending
 only on $\tildex$, which gives the first equality of (\ref{prop26:eq12}) by (\ref{eq:defev}). The
 second one follows from the equality $(1,j_1,0)(1,0,j_2)=(j_2(j_1),j_1,j_2)$ in $G(n)$.
\end{proof}

The ability to act on sections of $\bpol^d$ by $G(n)$ via a theta structure $\Thetabpoldij$ of type $K(n)$ allows one to recover the unique theta basis defined by $\Thetabpoldij$. We thereby obtain the following change of level theorem (which should be compared to \cite{lubicz:hal-03738315}):
\begin{theorem}\label{th:main7}
 Let $m,n,d>1$ be positive integers such that $n=md$.
 Let $(B, \bpol, \Thetabpol)$ be a marked abelian variety of type $K(m)$
 given by its (affine) theta null point $\tildenullbpol$. Suppose given
 a decomposition $B[n]=G_1 \oplus G_2$ into subgroups isomorphic to $Z(n)$,
 isotropic for the Weil pairing $e_{B,n}$, such that $\Thetabpolbar(Z(m)
 \times \{0\}) \subset G_1$ and $\Thetabpolbar(\{ 0 \}
 \times \dZ(m)) \subset G_2$, and suppose that for all $x \in G_1$ (resp. for
 all $x \in G_2$), $x$ is symmetric compatible with $\Thetabpol( \{1 \} \times Z(m)
 \times \{ 0 \})$ (resp. with $\Thetabpol( \{1 \} \times \{0\}
 \times \dZ(m))$). Let $g_1: \Zn \rightarrow G_1$ be a group isomorphism such that
 $g_1(\mu_{m,n}(i))=\Thetabpolbar((i,0))$ for $i \in Z(m)$, and let $g_2: \dZn \rightarrow G_2$ be
 the dual numbering, that is the group isomorphism such that $e_{B,n}(g_1(i), g_2(\nu))=e_n((i,0),(0,\nu))$
 for all $i \in \Zn$ and $\nu \in \dZn$ (then $g_2(\dnu_{m,n}(\nu))=\Thetabpolbar((0,\nu))$ for
 $\nu \in \dZ(m)$, see the proof).

 Suppose that there exist positive integers $(a_j)_{j=1, \ldots, r}$ such that $d=
 \sum_{j=1}^r a_j^2$ and $\gcd(a_j, n)=1$. Fix good lifts $\tildeG_1$ and
 $\tildeG_2$ of respectively $G_1$ and $G_2$ with respect to $\tildenullbpol$. For $x \in B(\overk)$, fix an affine lift
 $\tildex$, and for all $(P,Q) \in G_1
 \times G_2$, let $\widetilde{x + P}$ be the good lift of $x+P$ with respect to $\tildex$ and $\tildeG_1$
 and $\widetilde{x + P + Q}$ the good lift of $x+P+Q$ with respect to $\widetilde{x + P}$ and
 $\tildeG_2$.
 Compute $a_i (\widetilde{x + P + Q})$ using $\scalmult$.

 Let $U$ be an affine open subset of $B$ containing $G_1+G_2$,
 $\lambda x+G_1+G_2$ for $\lambda =1, \ldots, d$ and choose an
 isomorphism $\bpol(U) \simeq \stsheaf_B(U)$ so that for all $s \in \Gamma(B,\bpol)$ and
 all $x \in U(\overk)$ we can evaluate $s$ in $x$: we denote by $s(x) \in \overk$ the
 evaluation. Then, there exists a theta structure $\Thetabpoldij$ of type $K(n)$ for $\otimes_{j=1}^r
 [a_j]^*\bpol\simeq \bpol^d$,
 $f$-compatible with $\Thetabpol$, such that for all $i \in \Zn$,
 $\Thetabpoldijbar((i,0))=g_1(i)$, and for all $\nu \in \dZ(n)$,
 $\Thetabpoldijbar((0,\nu))=g_2(\nu)$, and such that the following holds. For $s \in \Gamma(B,
 \otimesbpol)$, denote by $(\lambda_k(s))_{k \in \Zn}$ the coordinates of $s$ in the basis
 $(\theta^{\Thetabpoldij}_k)_{k \in \Zn}$ (Proposition~\ref{prop:base}) and, for $\beta \in Z(m)^r$,
 let $s_\beta=\prod_{i=1}^r [a_i]^*\theta^\Thetabpol_{\beta_i}$. There exists $C \in \overk^*$, which
 only depends on $\tildex$, such that for all $\beta \in Z(m)^r$ and all $j \in \Zn$:
 \begin{equation}\label{eq:th7:eqbeta}
	\sum_{Q \in G_2} \prod_{i=1}^r (a_i(\widetilde{x + g_1(j) + Q}))_{\beta_i}=C_\beta\,
	\theta^{\Thetabpoldij}_j(x), \qquad \text{with } C_\beta=n^g C \lambda_0(s_\beta).
 \end{equation}
 In particular, for $(\alpha, j_0) \in Z(m)^2$, taking $\beta=(a_ij_0+\alpha)_{i=1, \ldots, r}$, we
 get, with $C_{\alpha,j_0}=n^g C \lambda_{-\mu_{m,n}(j_0)}(s_{(\alpha, \ldots, \alpha)})$, for all $j
 \in \Zn$:
 \begin{equation}\label{eq:th7:eq2}
	\sum_{Q \in G_2} \prod_{i=1}^r
	(a_i(\widetilde{x + g_1(j) + Q}))_{a_i j_0+\alpha}=C_{\alpha, j_0}\,\theta^{\Thetabpoldij}_j(x).
 \end{equation}
 In particular, for $j=0$ and $j_0=0$:
 \begin{equation}\label{eq:th7:eq1}
	\sum_{Q \in G_2} \prod_{i=1}^r (a_i
	(\widetilde{x + Q}))_\alpha=C_{\alpha,0}\, \theta_0^{\Thetabpoldij}(x).
 \end{equation}
 Moreover:
 \begin{enumerate}
	\item whether $C_\beta$ vanishes only depends on $\beta$, and not on $x$ and $\tildex$: $C_\beta
	 \neq 0$ if and only if the vector $\big(\sum_{Q \in G_2} \prod_{i=1}^r (a_i(\widetilde{g_1(j) +
	 Q}))_{\beta_i}\big)_{j \in \Zn}$, computed as above with $x=0_B$ and $\tildex=\tildenullbpol$, is
	 nonzero. In this case, Equation (\ref{eq:th7:eqbeta}) determines
	 $(\theta^{\Thetabpoldij}_j(x))_{j \in \Zn}$ up to a common factor;
	\item if $m \geq 3$ and $a_1=\cdots=a_r=1$ (so that $r=d$), there exists $\beta \in Z(m)^r$ such
	 that $C_\beta \neq 0$;
	\item if $m=d=2$ and $a_1=a_2=1$, then for $u \in Z(2)$, $C_{(u,u)} \neq 0$ if and only if
	 $\theta^{\Thetabpoldij}_{\mu_{2,4}(u)}(0_B) \neq 0$, and there exists $u \in Z(2)$ with this
	 property.
 \end{enumerate}
\end{theorem}
\begin{proof}
 We first check that the hypotheses (H1), (H2) and (H3) of Proposition~\ref{prop:thetadaction}
 are satisfied. It only remains to show that $g_2(\dnu_{m,n}(\nu))=\Thetabpolbar((0,\nu))$ for $\nu
 \in \dZ(m)$. Let $\nu \in \dZ(m)$, $w_1=g_2(\dnu_{m,n}(\nu))$ and $w_2=\Thetabpolbar((0,\nu))$. Both
 are in $G_2[m]=\Thetabpolbar(\{0\}\times\dZ(m)) \subset B[m]$. For $i \in \Zn$, $dg_1(i)=g_1(di)$ and $di
 \in \mu_{m,n}(Z(m))$ is the image of $\nu_{n,m}(i)$, so that $dg_1(i)=\Thetabpolbar((\nu_{n,m}(i),0))$.
 By (\ref{eq:weilcompat}), $e_{B,n}(g_1(i),w)=e_{B,m}(dg_1(i),w)$ for $w \in B[m]$, so that
 $e_{B,m}(\Thetabpolbar((\nu_{n,m}(i),0)), w_1)=e_n((i,0),(0,\dnu_{m,n}(\nu)))=e_m((\nu_{n,m}(i),0),(0,\nu))=
 e_{B,m}(\Thetabpolbar((\nu_{n,m}(i),0)), w_2)$, since $\dnu_{m,n}$ is the dual of $\nu_{n,m}$ and
 $\Thetabpolbar$ is symplectic. As $\nu_{n,m}$ is surjective and the pairing $e_{B,m}$ between
 $\Thetabpolbar(Z(m)\times\{0\})$ and $\Thetabpolbar(\{0\}\times \dZ(m))$ is perfect, $w_1=w_2$.

 Let $\Thetabpoldij$ be the theta structure given by Proposition~\ref{prop:thetadaction} (3),
 $C$ the constant of Equation (\ref{prop26:eq12}), $\beta \in Z(m)^r$ and $j \in \Zn$. By Equation
 (\ref{prop26:eq12}) applied with $j_1=j$, writing $Q=g_2(\nu)$:
 \begin{equation}
	\sum_{Q \in G_2} \prod_{i=1}^r (a_i(\widetilde{x + g_1(j) + Q}))_{\beta_i}
	=C\, \Big(\Thetabpoldij((1,j,0)) \sum_{\nu \in \dZn}
	\Thetabpoldij((1,0,\nu))(s_\beta)\Big)(x).
 \end{equation}
 By (\ref{eq:trans}), $\Thetabpoldij((1,0,\nu))(\theta^{\Thetabpoldij}_k)=\nu(-k)\theta^{\Thetabpoldij}_k$ and
 $\Thetabpoldij((1,j,0))(\theta^{\Thetabpoldij}_k)=\theta^{\Thetabpoldij}_{k+j}$. As $\sum_{\nu \in
 \dZn} \nu(-k)$ is equal to $n^g$ if $k=0$ and to $0$ otherwise, the right hand side is
 $C n^g \lambda_0(s_\beta)\theta^{\Thetabpoldij}_j(x)$, which proves (\ref{eq:th7:eqbeta}). Let $\alpha,
 j_0 \in Z(m)$, $u=\mu_{m,n}(j_0)$ and $\beta=(a_ij_0+\alpha)_{i=1, \ldots, r}$. As in the computation
 (\ref{eq:prop26:tech1}), $s_\beta=\Thetabpoldij((1,u,0))(s_{(\alpha, \ldots, \alpha)})$, so that
 $\lambda_0(s_\beta)=\lambda_{-u}(s_{(\alpha, \ldots, \alpha)})$, which gives (\ref{eq:th7:eq2}) and
 (\ref{eq:th7:eq1}).

 (1) As $C \neq 0$, $C_\beta \neq 0$ if and only if $\lambda_0(s_\beta) \neq 0$, which does not
 depend on $x$. Applying (\ref{eq:th7:eqbeta}) with $x=0_B$ and $\tildex=\tildenullbpol$, the vector
 of the statement is $C'n^g\lambda_0(s_\beta)(\theta^{\Thetabpoldij}_j(0_B))_{j \in \Zn}$ with $C'
 \in \overk^*$, and $(\theta^{\Thetabpoldij}_j(0_B))_{j \in \Zn}$ is nonzero since
 $(\theta^{\Thetabpoldij}_j)_{j \in \Zn}$ is a basis of $\Gamma(B, \otimesbpol)$ and $\otimesbpol \simeq
 \bpol^d$ is generated by its global sections (\cite[Section~17]{MumfordAV}, as $d \geq 2$).

 (2) As $K(\bpol)=B[m]$, $\bpol \simeq \bpol_0^m$ for an ample line bundle $\bpol_0$ on $B$
 (\cite[Section~23]{MumfordAV}). Since $m \geq 3$, $\bpol$ is projectively normal
 (\cite{Koizumi76}, \cite[Theorem~7.3.1]{BiLaCAV} over $\C$, and \cite{Sekiguchi76} in arbitrary
 characteristic), so that the multiplication map $\Gamma(B, \bpol)^{\otimes d} \rightarrow \Gamma(B,
 \bpol^d)$ is surjective. As $a_1=\cdots=a_r=1$, the sections $s_\beta=\prod_{i=1}^d
 \theta^\Thetabpol_{\beta_i}$, $\beta \in Z(m)^d$, span the image of this map, that is $\Gamma(B,
 \bpol^d)$. As $\lambda_0$ is a nonzero linear form on $\Gamma(B, \bpol^d)$, there exists $\beta$
 such that $\lambda_0(s_\beta) \neq 0$, that is $C_\beta \neq 0$.

 (3) When $m=2$, $\bpol$ is not projectively normal, and the sections $s_\beta$ only span the
 even sections of $\bpol^2$, so that the argument of (2) does not apply. By Theorem
~\ref{thm:compatthetas}, $\Thetabpol$ and $\Thetabpoldij$ are Mumford compatible, and the
 duplication formula of \cite{MumfordOEDAV1} (see (\ref{eq:dup601})), applied to the points $x$
 and $0_B$, gives $C' \in \overk^*$ such that, for $u, v \in Z(2)$ lifted to $\{0,1\}^g \subset
 Z(4)$:
 \begin{equation}\label{eq:dup24}
	\theta^\Thetabpol_u \theta^\Thetabpol_v=C' \sum_{t \in \mu_{2,4}(Z(2))}
	\theta^{\Thetabpoldij}_{u-v+t}(0_B)\, \theta^{\Thetabpoldij}_{u+v+t}.
 \end{equation}
 The coefficient of $\theta^{\Thetabpoldij}_0$ is nonzero only if $u+v+t=0$ in $Z(4)$ for some $t$,
 that is if $u=v$ and $t=\mu_{2,4}(u)$, so that $\lambda_0(s_{(u,v)})=0$ if $u \neq v$ and
 $\lambda_0(s_{(u,u)})=C'\theta^{\Thetabpoldij}_{\mu_{2,4}(u)}(0_B)$. As $C_\beta=n^g C
 \lambda_0(s_\beta)$, this proves the first claim. If $\theta^{\Thetabpoldij}_{\mu_{2,4}(u)}(0_B)=0$
 for all $u \in Z(2)$, then (\ref{eq:dup24}) evaluated at $0_B$ with $v=u$ gives
 $\theta^\Thetabpol_u(0_B)^2=0$ for all $u$, since each term contains a factor
 $\theta^{\Thetabpoldij}_t(0_B)$ with $t \in \mu_{2,4}(Z(2))$; this contradicts the fact that the
 theta null point of $B$ is nonzero.
\end{proof}
\begin{remark}\label{rk:followmain7}
 Note that from the knowledge of $(B, \bpol^d, \Thetabpold)$, one recovers immediately
 $B[n]=K(\bpol^d)$ using the action of the theta group on $0_{\Thetabpold}$. Thus in
 the course of the computation from $(B, \bpol, \Thetabpol)$ to $(B, \bpol^d,
 \Thetabpold)$, one necessarily needs to compute $B[n]$ from $(B,
 \bpol, \Thetabpol)$. The hypothesis in the theorem that $B[n]$ is given is therefore quite natural. 

 In theta coordinates, a way to do that is to solve the algebraic
 system in $k[x_i, i \in Z(m) ]$ made of the equations:
 \begin{itemize}
 \item defining the embedding of $B$ into $\proj^{Z(m)}$, which is given by
 $O(m^g-g)$ quadratic Riemann equations parametrized by $0_{\Thetabpol}$;
 \item $\scalmult(d, P, P, 0_{\Thetabpol}, 0_{\Thetabpol})= \Thetabpol((1,
 i)).0_{\Thetabpol}$ for $i \in K(m)$, where
 $P$ is a generic projective point.
 \end{itemize}
 Each of these systems is a $0$-dimensional algebraic system in $m^g$ variables of degree
 at most $O(5^{\log(d)})$. It can be solved by computing a triangular system,
 which can be obtained by computing the reduced Groebner basis for the lexicographic order.
 An efficient way to do so is to first compute a Groebner basis for the degree-reverse-lexicographical
 ordering, and then change the monomial order to the lexicographical one using \cite{FAUGERE1993329}.
 This Groebner basis step can be performed in time $O(5^{m^g\log(d)})$, 
 we refer to \cite{DRisogenykernel} for the use of the triangular system.
 Theorem~\ref{th:main7} then gives, once these algebraic systems have been solved (or $B[n]$ has been obtained by any other means), an efficient algorithm for computing $(B, \bpol^d, \Thetabpold)$.
\end{remark}
From Theorem~\ref{th:main7}, we immediately deduce the change of level Algorithm~\ref{algo:changelevel}, as well as the following corollary.
\begin{corollary}\label{cor:changeoflevel}
 Let $m,n,d >1$ be integers such that $n=md$.
 There exists a deterministic algorithm that takes as input
 the theta null point $0_{\Thetabpol}$ of a $g$-dimensional marked abelian variety $(B,
	 \bpol, \Thetabpol)$ of type $K(m)$, a basis of $B[n]$,
	 $(\theta_i^\Thetabpol(x))_{i\in Z(m)}$ for $x \in B(\overk)$ and outputs
	 $(\theta_i^\Thetabpold(x))_{i \in \Zn}$ where $\Thetabpold$ is a theta
	 structure of type $K(n)$ in time $O(N n^{3g} \log(d))$ operations
	 in the base field of $B$, where $N$ is the number of values of $\beta$ tried by
	 Algorithm~\ref{algo:changelevel}, provided that one of them satisfies $C_\beta \neq 0$;
	 this is always possible when $m \geq 3$, with the decomposition $d=1^2+\cdots+1^2$
	 (Theorem~\ref{th:main7} (2)), and when $m=d=2$ (Theorem~\ref{th:main7} (3); see Remark
	~\ref{rk:level2} for the data in level $2$).
\end{corollary}

\begin{algorithm}
\SetKwInOut{Input}{input}\SetKwInOut{Output}{output}

\SetKwComment{Comment}{/* }{ */}
\Input{
 \begin{itemize}
	\item $m,n,d>1$ integers such that $n=md$;
	\item $(a_i)_{i=1,\dots,r} \in \N^r$ such that $d= \sum_{i=1}^r a_i^2$ and $\gcd(a_i, n)=1$;
	\item the marked abelian variety $(B, \bpol, \Thetabpol)$ of type $K(m)$ given by its theta null
	 point $0_{\Thetabpol}$;
	\item $B[n]$ given by a basis $(e_i)_{i=1, \ldots, 2g}$;
	\item $(\theta^{\Thetabpol}_i(x))_{i \in Z(m)}$, for $x \in B(\overk)$.
 \end{itemize}
}
\Output{
 \begin{itemize}
	\item $\left(\theta^{\Thetabpoldij}_{j}(x)\right)_{j \in \Zn}$.
 \end{itemize}
}
\BlankLine
Call Algorithm~\ref{algo:changesym2} to obtain a symplectic basis $(e_i, e_{i+g})$ of
$B[n]$ such that for all $i=1, \ldots, g$ (resp. for all $i=g+1, \ldots, 2g$), $e_i$
 is symmetric compatible with $\Thetabpol(\{1\}\times Z(m)\times\{0\})$ (resp. with
 $\Thetabpol(\{1\}\times\{0\}\times\dZ(m))$)\;
Let $G_1= (e_i)_{i=1, \ldots, g}$ and $G_2=(e_i)_{i=g+1, \ldots, 2g}$, choose a numbering $G_1=\{ g_1(i), i \in \Zn \}$ such that for $i \in
 Z(m)$, $g_1(\mu_{m,n}(i))= \Thetabpolbar((i,0))$, and let $G_2=\{ g_2(\nu), \nu \in \dZn\}$ be the
 dual numbering ($e_{B,n}(g_1(i), g_2(\nu))=e_n((i,0),(0,\nu))$)\;
  Call Algorithm~\ref{algo:goodliftg} to compute good lifts $\tildeG_1$, $\widetilde{x + G_1}$\;
 Compute a good lift $\tildeG_2$ of $G_2$, once and for all, by calling Algorithm~\ref{algo:goodlift}
 with $\Thetabpol$ replaced by $\Thetabpoldel=\Thetabpol\circ\Delta(m)$, that is after applying to all
 affine points the change of coordinates $\tildeH_{\Delta(m)}$ of the proof of Proposition
~\ref{prop:thetadaction} and numbering $G_2$ by $i \mapsto g_2(-\iota_n(i))$, and then going back with
 $\tildeH_{\Delta(m)}^{-1}$\;
 Find $\beta \in Z(m)^r$ (trying for instance $\beta=0$ first) such that the vector
 $\big(\sum_{\tildeQ \in \tildeG_2} \prod_{i=1}^r(a_i(\widetilde{g_1(j) + Q}))_{\beta_i}\big)_{j \in
 \Zn}$, computed as in the loop below with $x=0_B$ and $\tildex=\tildenullbpol$, is
 nonzero\;\Comment{this holds if and only if $C_\beta \neq 0$ (Theorem~\ref{th:main7} (1)); such a
 $\beta$ exists if $m \geq 3$ and $a_1=\cdots=a_r=1$ (Theorem~\ref{th:main7} (2)), and if $m=d=2$, of
 the form $\beta=(u,u)$ (Theorem~\ref{th:main7} (3))}
 \For{$j \in \Zn$}{
  Set $P=g_1(j)$\;
 Compute, in the same way, the good lifts $\widetilde{x + P + G_2}$
 with respect to $\widetilde{x+P}$ and $\tildeG_2$ (Algorithm~\ref{algo:goodliftg} applied with
 $\Thetabpoldel$)\;
 Compute $\theta^{\Thetabpoldij}_j(x) = \sum_{\tildeQ \in \tildeG_2}
		\prod_{i=1}^r(a_i(\widetilde{x + P + Q}))_{\beta_i}$\;\Comment{we use Equation
	(\ref{eq:th7:eqbeta}); the result is $C_\beta\theta^{\Thetabpoldij}_j(x)$}
 }
 \Return $\left(\theta^{\Thetabpoldij}_j(x)\right)_{j \in \Zn}$.
 \caption{Change of level algorithm.}
 \label{algo:changelevel}
\end{algorithm}

The duplication formula is also a change of level formula. The following example, computed with
the SageMath library \cite{thetAV}, compares it with Algorithm~\ref{algo:changelevel}.

\begin{example}\label{ex:dup601}
 Let $(B, \bpol, \Thetabpol)$ be a marked abelian variety of dimension $1$ and of type $K(4)$, and
 let $\Thetabpold$ be a symmetric theta structure of type $K(8)$ such that $\Thetabpol$ and
 $\Thetabpold$ are Mumford compatible (Definition~\ref{def:mumfordcomp}). Write $b_i=\theta^\Thetabpol_i(0_B)$ for $i \in Z(4)$ and
 $\vartheta_j=\theta^\Thetabpold_j(0_B)$ for $j \in Z(8)$. Mumford's duplication formula
 \cite{MumfordOEDAV1}, evaluated at $0_B$, gives $C \in \overk^*$ such that, for all $i, j \in
 Z(8)$ with $i+j$ even:
 \begin{equation}\label{eq:dup601}
	b_{(i+j)/2}\, b_{(i-j)/2}=C\sum_{t \in \{0, 4\}}\vartheta_{i+t}\,\vartheta_{j+t},
 \end{equation}
 where $(i \pm j)/2$ is taken in $Z(4)$. Read from right to left, (\ref{eq:dup601}) changes the
 level from $8$ to $4$. Read from left to right, it is a system of quadratic equations in the
 $\vartheta_j$: for $j=i$ and $j=i+4$, it gives $\vartheta_i^2+\vartheta_{i+4}^2=b_0b_i/C$ and
 $2\vartheta_i\vartheta_{i+4}=b_2b_{i+2}/C$, so that solving it requires square root extractions
 and sign choices.

 Let $k=\F_{601}$ and $B: y^2=x(x-1)(x-9)$; then $B(\F_{601}) \simeq \Z/8\Z \times \Z/80\Z$
 contains $B[8]$, and $0_\Thetabpol=(1:545:117:545)$ is a theta null point of level $4$ of $B$.
 With $\vartheta_0=1$ and the symmetry relations $\vartheta_j=\vartheta_{-j}$, the system
 (\ref{eq:dup601}) has exactly $8$ solutions over $\overk$. They are all in $\F_{601}$ and all
 satisfy the Riemann relations of level $8$: with $\vartheta=(1,34,548,40,507,40,548,34)$, they are
 \begin{equation*}
	(\vartheta_j)_j, \quad ((-1)^j\vartheta_j)_j, \quad (\vartheta_{j+4})_j, \quad
	((-1)^j\vartheta_{j+4})_j,
 \end{equation*}
 and their images by $(\vartheta_j)_j \mapsto (\vartheta_{3j})_j$ (as projective points; for
 instance $(\vartheta_{j+4})_j=(1:575:365:38:390:38:365:575)$). The duplication formula does not
 tell which of them corresponds to a given symplectic basis of $B[8]$.

 Let us now apply Algorithm~\ref{algo:changelevel} with $d=2=1^2+1^2$ and $\beta=0$: for $j \in
 Z(8)$, it computes $\sum_{Q \in G_2}\big((\widetilde{g_1(j)+Q})_0\big)^2$, which is
 $C_0\,\theta^\Thetabpold_j(0_B)$ by (\ref{eq:th7:eqbeta}) with $x=0_B$. The point
 $\Thetabpolbar((1,0))=(1:395:1:440)$ has four halves:
 $(1:221:221:1)$ and $(1:223:223:1)$, which are symmetric compatible with $\Thetabpol(\{1\} \times Z(4) \times
 \{0\})$, and $(1:378:223:600)$ and $(1:380:221:600)$, which are not (Example~\ref{ex:hand48}).
 Similarly, the symmetric compatible halves of $\Thetabpolbar((0,1))=(1:389:484:212)$ are
 $(1:426:0:362)$ and $(1:175:0:239)$. Hence $G_1$ is generated by $P=(1:221:221:1)$, and $5P=(1:223:223:1)$; $G_2$ is
 generated by $Q=(1:426:0:362)$, and $5Q=(1:175:0:239)$. There are exactly two numberings
 $(g_1, g_2)$ satisfying the conditions of Theorem~\ref{th:main7}, exchanged by $(g_1,g_2) \mapsto
 (5g_1, 5g_2)$. Moreover, by Lemma~\ref{lem:goodliftdiff}, the good lifts of $g_1(1)$ and $g_2(1)$
 are only determined up to fourth roots of unity $\mu_1$ and $\mu_2$, which change the lifts of the
 $g_1(j)+g_2(\nu)$ with $j$ odd (resp. $\nu$ odd) by a factor $\mu_1$ (resp. $\mu_2$), since
 $\Zz(2)=\{0,1\}$ in Definition~\ref{def:good}. For one numbering, the four choices $\mu_1, \mu_2 \in \{1, 125\}$ (where
 $125^2=-1$) give the four points displayed above, and for the other one, their images by $j \mapsto 3j$. Thus every solution of
 (\ref{eq:dup601}) is an output of Algorithm~\ref{algo:changelevel}.

 The two methods have the same ambiguity, but Algorithm~\ref{algo:changelevel} resolves it through
 explicit data: by Theorem~\ref{th:main7}, the symplectic structure of the output satisfies
 $\Thetabpoldbar((1,0))=g_1(1)$ and $\Thetabpoldbar((0,1))=g_2(1)$, and the remaining choice is
 that of the roots of unity $\mu_1, \mu_2$ in the good lifts. As recalled in the introduction,
 with the duplication formula one can also resolve the ambiguity with a subgroup of $B[8]$, as in
 \cite{dartois_fast}.
\end{example}

\begin{example}\label{ex:dup601iter}
 We keep the notation of Example~\ref{ex:dup601} and the numbering with $g_1(1)=(1:223:223:1)$ and
 $g_2(1)=(1:426:0:362)$, for which the four theta structures of level $8$ compatible with
 $\Thetabpol$ have the theta null points $(\vartheta_j)_j$, $((-1)^j\vartheta_j)_j$,
 $(\vartheta_{j+4})_j$ and $((-1)^j\vartheta_{j+4})_j$. By \cite[\S~7]{MumfordOEDAV2} (see also
 \cite[p.~319]{MumfordOEDAV1}), the symmetric theta structure of level $8$ is determined by points of
 order $16$ above $g_1(1)$ and $g_2(1)$. Unlike $B[8]$, the group $B[16]$ is not defined over
 $\F_{601}$, but over $\F_{601^2}=\F_{601}[u]$, with $u^2=3u-7$. The halves of $g_1(1)$ are
 \begin{equation*}
	h_1 \in \{(1:231:275:15),\ (1:459:271:97),\ (1:142:271:504),\ (1:370:275:586)\},
 \end{equation*}
 and those of $g_2(1)$ are the points $h_2=(1:x_1:x_2:x_3)$ with
 \begin{equation*}
	(x_1, x_2, x_3) \in \{(130u+406, 80, 160u+361),\ (471u+195, 80, 441u+240),\ (2u+598, 293,
	227u+561),\ (599u+3, 293, 374u+40)\}.
 \end{equation*}
 We change the level from $4$ to $16$, either directly with Algorithm~\ref{algo:changelevel} and
 $d=4=1^2+1^2+1^2+1^2$, or in two steps $4 \to 8 \to 16$ with $d=2$ at each step. For each of the
 $16$ pairs $(h_1, h_2)$, the direct change of level gives four theta null points of level $16$,
 one for each choice of good lifts (Theorem~\ref{th:main1}). In two steps, the second change of
 level, applied to the images of $h_1$ and $h_2$, is possible for exactly one of the four theta
 structures of level $8$: for the three others, $h_1$ or $h_2$ is not symmetric compatible with the
 corresponding level subgroup. For the remaining one, the four theta null points of level $16$ are
 exactly those of the direct computation. The theta structure of level $8$ selected in this way is
 the one given by Mumford's theory: the sign $(-1)^j$ is absent if and only if $h_1$ is one of the
 first two points above, and the theta null point is $(\vartheta_j)_j$ or $((-1)^j\vartheta_j)_j$ if
 $x_2=293$, and $(\vartheta_{j+4})_j$ or $((-1)^j\vartheta_{j+4})_j$ if $x_2=80$. Fixing it in this
 way requires the points $h_2$, hence the field $\F_{601^2}$, whereas Algorithm
~\ref{algo:changelevel} computes the four theta structures of level $8$, and the corresponding
 coordinates of points, over $\F_{601}$, the remaining choice being that of the good lifts.
\end{example}

\begin{theorem}\label{th:mainisog}
 Let $m,n,d>1$ be integers such that $n=md$.
Suppose that there exist positive integers $(a_j)_{j=1, \ldots, r}$
such that $d= \sum_{j=1}^r a_j^2$ and $\gcd(a_j, n)=1$.

 Let $(B, \bpol, \Thetabpol)$ be a
 marked abelian variety of type $K(m)$ given by its (affine) theta null point
 $\tildenullbpol$. Let $K \subset B(\overk)$ be a subgroup isomorphic to $Z(d)$.
Let $A=B/K$ and $f: B\rightarrow
 A$ be the isogeny. 

 Let $G_1$ be a subgroup of $B[n]$
 isomorphic to $Z(n)$ isotropic for the Weil pairing $e_{B,n}$ and such that
 $\Thetabpolbar(Z(m)
 \times \{0\}) \cup K \subset G_1$ (since $G_1[d] \simeq Z(d)$ has the same cardinality as $K$, we
 then have $K=G_1[d]=mG_1$, so that $K$ is isotropic for $e_{B,d}$ and
 $K[d_0]=\Thetabpolbar(\mu_{d_0,m}(Z(d_0))\times\{0\})$, see the proof of Theorem
~\ref{prop:cond}). We suppose moreover that for all $x \in G_1$, $x$ is symmetric compatible with $\Thetabpol( \{1 \} \times Z(m)
 \times \{ 0 \})$. We fix a numbering $G_1=\{ g_1(i), i \in \Zn\}$ such that for all
 $i \in Z(m)$, $g_1(\mu_{m,n}(i))=\Thetabpolbar((i,0))$. We fix a good lift
 $\tildeG_1 = \{\tildeg_1(i), i \in \Zn\}$ of $G_1$ with respect to
 $\tildenullbpol$.

 Let $x \in B(\overk)$ and let $\tildex$ be an affine lift of $x$. For $P \in G_1$, let
 $\widetilde{x+P}$ be the good lift of $x+P$ with respect to $\tildex$ and $\tildeG_1$.
 Let $U$ be an affine open subset of $B$ containing $G_1$,
 $0_{\Thetabpol}$, $\lambda x+G_1$ for $\lambda =1, \ldots, d$, and choose an
 isomorphism $\bpol(U) \simeq \stsheaf_B(U)$ so that for all $s \in \Gamma(B,\bpol)$ and
 all $x \in U(\overk)$ we can evaluate $s$ in $x$: we denote by $s(x) \in \overk$ the
 evaluation.

 There exist a symmetric line bundle and theta structure $(\pol, \Thetapol)$ of type
 $K(m)$ for $A$ and, for all $\beta=(\beta_1, \ldots, \beta_r) \in Z(m)^r$ such that
 $\sum_{i=1}^r a_i\beta_i=0$ in $Z(m)$, a constant $C_\beta \in \overk$ such that for $j_0 \in
 Z(m)$, if we choose
$j_1 \in \Zn$ and $j_2 \in Z(m)$ such that $\rho_{n,m}(j_1 + \mu_{m,n}(j_2))=j_0$,
 we have:
 \begin{equation}\label{eq:mainisog}
	\sum_{P \in K} \prod_{i=1}^r (a_i
	(\widetilde{x+P+ g_1(j_1)}))_{\beta_i+a_i j_2}=C_\beta\, \theta_{j_0}^{\Thetapol}(f(x)).
 \end{equation}
 Moreover, $C_\beta=c\,\gamma(\beta)$, where $c \in \overk^*$ only depends on $\tildex$ and
 $\gamma(\beta) \in \overk$ does not depend on $x$ (see the proof). Thus $C_\beta \neq 0$ if and
 only if the vector given by the left hand side of (\ref{eq:mainisog}) for $x=0_B$,
 $\tildex=\tildenullbpol$ and $j_0 \in Z(m)$ is nonzero, and in this case (\ref{eq:mainisog})
 determines $(\theta_{j_0}^{\Thetapol}(f(x)))_{j_0 \in Z(m)}$ up to a common factor. If $m \geq 3$
 and $a_1=\cdots=a_r=1$, there exists $\beta$ such that $C_\beta \neq 0$.
\end{theorem}
\begin{proof}
 Let
 $\Thetaubpoldij: \overk^* \times \Zn
 \rightarrow G(\otimesbpol)$ be the partial theta structure defined as the image by
 the map $\gtf$ of Proposition~\ref{prop:uniqueness} of the good lift $\tildeG_1$ (see
 also Remark~\ref{rk:extension}). Moreover,
 consider the
 partial theta structure $\Thetadbpoldij: \overk^* \times \dZ(m) \rightarrow G(\otimesbpol)$
 defined by
 \begin{equation}\label{proof:thm8:theta2}
 \Thetadbpoldij((1, e))=\eaj(\bpol) \circ \Thetabpol((1, 0 ,e)),
 \end{equation}
 for all $e \in \dZ(m)$. Note that $K= \Thetaubpoldijbar(\mu_{d,n}(Z(d)) \times \{0\})$ so that we can
 define the level subgroup $\tildeK=\Thetaubpoldij(\{1\} \times \mu_{d,n}(Z(d)) \times \{0\})$. Let
 $\pol = \otimesbpol/\tildeK$ which is a symmetric line bundle over $A$.

 Let $\fsharp(\otimesbpol): G^*(\otimesbpol) \rightarrow \Gpol$ be the map defined by
 $\tildeK$ and Definition~\ref{def:fsharp}.
 Let $\pi_{G(\otimesbpol)}: G(\otimesbpol) \rightarrow K(\otimesbpol)$ be the canonical projection. Remark that the centralizer 
 $G^*(\otimesbpol)$ of $\tildeK$ in $G(\otimesbpol)$ is
 $$\pi_{G(\otimesbpol)}^{-1}(\Thetabpoldijbar(Z(n) \times
 \dnu_{m,n}(\dZ(m)))).$$ Recall that $\pi_{G(n)}: G(n) \rightarrow K(n)$ is the canonical projection and
 let $G^*(n)= \pi_{G(n)}^{-1}(Z(n) \times \dnu_{m,n}(\dZ(m)))$. 
 Denote by $\Thetabpoldij: G^*(n) \rightarrow G(\otimesbpol)$ the partial theta structure
 such that for all $\nu \in \Zn$, $\Thetabpoldij((1, \nu, 0)) = \Thetaubpoldij((1,
 \nu))$ and for all $\dnu \in \dZ(m)$, $\Thetabpoldij((1, 0, \dnu_{m,n}(\dnu)))=
 \Thetadbpoldij((1, \dnu))$.

 We remark that $\Thetabpoldij(\{1\}\times\mu_{d,n}(Z(d))\times\{0\})$ is the kernel of
 $f^{\sharp}(\otimesbpol)$ and moreover $G(m)\simeq G^*(n) / (\{1\}\times\mu_{d,n}(Z(d))\times\{0\}),$
since $G^*(n)$ is the centralizer of $\{1\}\times\mu_{d,n}(Z(d))\times\{0\}$, so this last
isomorphism is just the counterpart of (\ref{eq:quotientthetagroup}) for Heisenberg groups.

 Hence, the map $\fsharp(\otimesbpol) \circ \Thetabpoldij: G^*(n) \rightarrow \Gpol$ 
 factors through a map $\Thetapol: G(m) \rightarrow \Gpol$ which is an
 isomorphism of Heisenberg groups, and thus is a theta structure for $(A, \pol)$.

 In order to compute an embedding of $(A, \pol)$, we first have to compute elements of
 $\Gamma(A,\pol)$. But as $f^*(\pol)\simeq\otimesbpol$ there is a bijection between elements of
 $H^0(\pol)$ and sections of $H^0(\otimesbpol)$ which are invariant by the action of
 $\tildeK$. Consider the map:
 \begin{gather}
 \begin{aligned}
	\pi_{\tildeK}: \Gamma(B, \otimesbpol) & \rightarrow \Gamma(B, \otimesbpol) \\
	s &\mapsto \sum_{\nu \in Z(d)} \Thetabpoldij((1, \mu_{d,n}(\nu), 0))(s).
 \end{aligned}
 \end{gather}
	It is clear that the image of $\pi_{\tildeK}$ is invariant by the action of $\tildeK$, so the
	image of $\pi_{\tildeK}$ is contained in $\Gamma(B, f^*(\pol))=\Gamma(A,\pol)$.

	By Proposition~\ref{prop:base}, for $s \in \Gamma(A,\pol)$, we have $\sum_{\dnu \in \dZ(m)}
	\Thetapol((1,0,\dnu))(s)=m^g\lambda^\Thetapol_0(s)\,\theta_0^{\Thetapol}$, where
	$\lambda_0^\Thetapol(s)$ is the coordinate of $s$ on $\theta_0^\Thetapol$ in the basis
	$(\theta_k^\Thetapol)_{k \in Z(m)}$; in particular, a section $s$ can be used to compute
	$\theta_0^\Thetapol$ only if $\lambda_0^\Thetapol(s) \neq 0$. For $\beta \in Z(m)^r$, let
	$s_\beta = \prod_{j=1}^r [a_j]^*(\theta^\Thetabpol_{\beta_j}) \in \Gamma(B, \otimesbpol)$ and
	$\gamma(\beta)=m^g\lambda_0^\Thetapol(\pi_{\tildeK}(s_\beta))$, where we identify $\Gamma(A,
	\pol)$ with the space of $\tildeK$-invariant sections of $\otimesbpol$.
 Using the facts that for all $\dnu \in \dnu_{m,n}(\dZ(m))$ and for all $\nu \in \mu_{d,n}(Z(d))$,
	 $e_n((1, \nu,0), (1, 0, \dnu))=1$ (because $(1, 0, \dnu)$ is in the
	 centralizer of $\mu_{d, n}(Z(d))$), we get:
	\begin{gather}
	 \begin{aligned}
		\gamma(\beta)\,\theta_0^{\Thetapol}&= \sum_{\dnu \in \dZ(m)} \Thetapol((1, 0,
		\dnu))(\pi_{\tildeK}(s_\beta)) & \\
		& = \sum_{\dnu \in \dZ(m)} \Thetapol((1, 0, \dnu))
		\sum_{\nu \in \mu_{d,n}(Z(d))} \Thetabpoldij((1, \nu, 0))(s_\beta) & \\
		& = \sum_{\dnu \in \dnu_{m,n}(\dZ(m))} \Thetabpoldij((1, 0, \dnu))
		\sum_{\nu \in \mu_{d,n}(Z(d))} \Thetabpoldij((1, \nu, 0))(s_\beta) & \text{(explanation below)}\\
		& = \sum_{\nu \in \mu_{d,n}(Z(d))} \Thetabpoldij((1, \nu, 0))
		\sum_{\dnu \in \dnu_{m,n}(\dZ(m))}
		\Thetabpoldij((1, 0, \dnu))(s_\beta).
	 \end{aligned}
	\end{gather}
To get the third equality, we have used the fact that $f^{\sharp}(\Thetabpoldij((1, 0,
\dnu_{m,n}(\dnu))))=\Thetapol((1, 0, \dnu))$ and Corollary~\ref{cor:canonical}.

In the following computation of the inner sum, we obtain the first equality by using
the definition of $\Thetabpoldij$ (Equation (\ref{proof:thm8:theta2})) and then Corollary
\ref{cor:genmum} for the second equality (the elements $\Thetabpol((1,0,\dnu))$ being symmetric,
$\delta_{a_j}(\bpol)$ acts on them as the $a_j$-th power):
\begin{gather}
 \begin{aligned}
	\sum_{\dnu \in \dZ(m)} \Thetabpoldij((1, 0, \dnu_{m,n}(\dnu))) (s_\beta) &= \sum_{\dnu \in \dZ(m)}
	\eaj(\bpol) \circ \Thetabpol((1, 0, \dnu))
	(s_\beta) \\
	& = \sum_{\dnu \in \dZ(m)} \prod_{j=1}^r
	[a_j]^* (\Thetabpol((1, 0, a_j\dnu))(\theta_{\beta_j}^{\Thetabpol})) \\
	& = \sum_{\dnu \in \dZ(m)} \dnu\Big(-\sum_{j=1}^r a_j\beta_j\Big)\, s_\beta.
 \end{aligned}
\end{gather}
This last expression is $0$ unless $\sum_{j=1}^r a_j\beta_j=0$ in $Z(m)$, in which case it is
equal to $m^g s_\beta$. Hence $\gamma(\beta)=0$ if $\sum_{j=1}^r a_j\beta_j \neq 0$ and, if
$\sum_{j=1}^r a_j\beta_j=0$:
\begin{equation}
 \gamma(\beta)\,\theta_0^\Thetapol = m^g \sum_{\nu \in \mu_{d,n}(Z(d))} \Thetabpoldij((1, \nu,
 0))(s_\beta).
\end{equation}
Now, for $j_0 \in Z(m)$, following Proposition~\ref{prop:base}, we have
$\theta_{j_0}^{\Thetapol}=\Thetapol((1, j_0, 0))(\theta_0^{\Thetapol})$. Let $j_1 \in \Zn$ and
$j_2 \in Z(m)$ be
such that $\rho_{n,m}(j_1 + \mu_{m,n}(j_2))=j_0$, note that $f^{\sharp}(\Thetabpoldij((1, j_1+\mu_{m,n}(j_2),
0)))=\Thetapol((1, j_0, 0))$ so that, by Corollary~\ref{cor:canonical}:
\begin{equation}
 \begin{split}
 \gamma(\beta)\,\theta_{j_0}^{\Thetapol} & = m^g \sum_{\nu \in \mu_{d,n}(Z(d))} \Thetabpoldij((1, \nu +
 {j_1}+ \mu_{m,n}(j_2), 0)) (s_\beta) \\
 & = m^g \sum_{\nu \in \mu_{d,n}(Z(d))} \Thetabpoldij((1, \nu +
 {j_1}, 0)) (s_{\beta'}),
 \end{split}
\end{equation}
where $\beta'=(\beta_i+a_ij_2)_{i=1, \ldots, r}$, since $\Thetabpoldij((1,\mu_{m,n}(j_2),0))(s_\beta)=s_{\beta'}$
(see the proof of Proposition~\ref{prop:thetadaction} (1)). Let $\nu \in \mu_{d,n}(Z(d))$ and
$P=g_1(\nu)$. By Equation (\ref{eq:prop26eq16b}) applied with $j_1$ replaced by $\nu+j_1$ and $\beta$
by $\beta'$, the evaluation at $x$ of the term of index $\nu$ is equal to $C^{-r}\prod_{i=1}^r
(a_i(\widetilde{x+P+g_1(j_1)}))_{\beta_i+a_ij_2}$, where $C \in \overk^*$ is the constant of
(\ref{eq:prop26eq16b}), which only depends on $\tildex$. Summing over $\nu$, and taking into account
the nonzero constant relating the evaluation at $x$ of the pull back by $f$ of a section of $\pol$
to the evaluation of this section at $f(x)$, we get Equation (\ref{eq:mainisog}) with
$C_\beta=c\,\gamma(\beta)$, where $c \in \overk^*$ only depends on $\tildex$ (and on the chosen
trivializations).

As $c \neq 0$, $C_\beta \neq 0$ if and only if $\gamma(\beta) \neq 0$, which does not depend on
$x$. For $x=0_B$ and $\tildex=\tildenullbpol$, the left hand side of (\ref{eq:mainisog}) is $c'
\gamma(\beta)\theta_{j_0}^\Thetapol(0_A)$ with $c' \in \overk^*$, and $(\theta_{j_0}^\Thetapol(0_A))_{j_0
\in Z(m)}$ is nonzero since $\pol$, of type $K(m)$ with $m \geq 2$, is generated by its global
sections (\cite[Section~17]{MumfordAV}). Finally, the linear form $s \mapsto
\lambda_0^\Thetapol(\pi_{\tildeK}(s))$ on $\Gamma(B, \otimesbpol)$ is nonzero since $\pi_{\tildeK}$
is $d^g$ times the identity on $\tildeK$-invariant sections, hence maps $\Gamma(B, \otimesbpol)$
onto $\Gamma(A, \pol)$. If $m \geq 3$ and $a_1=\cdots=a_r=1$, the sections $s_\beta$, $\beta \in
Z(m)^r$, span $\Gamma(B, \otimesbpol)$ (see the proof of Theorem~\ref{th:main7} (2)), so that
$\gamma(\beta) \neq 0$ for some $\beta$; by the above, such a $\beta$ satisfies $\sum_{j=1}^r
a_j\beta_j=0$.
\end{proof}
\begin{remark}
 In Equation (\ref{eq:mainisog}), the index of the coordinate is $\beta_i+a_i j_2$.
 When all the $a_i$ are congruent to $1$ modulo $m$ (for instance $d=2=1^2+1^2$), it is simply
 $\beta_i+j_2$. In any case, one can always choose $j_2=0$ and $j_1$ such that $\rho_{n,m}(j_1)=j_0$, in
 which case the coordinates of index $\beta_i$ are used.
\end{remark}
\begin{remark}\label{rk:betazero}
 The constant $C_0$, corresponding to $\beta=0$, may vanish. For instance, over $\C$, let $g=2$,
 $m=4$, $d=2=1^2+1^2$, $B=\C^2/(\Z^2+\Omega\Z^2)$ and let $\Thetabpol$ be the analytic theta
 structure, for which $\theta^\Thetabpol_i(z)=\vartheta[i/4, 0](4z, 4\Omega)$, $i \in Z(4)$. As
 $d | m$, we have $K=\Thetabpolbar(Z(4)[2] \times \{0\})$ and, by the proof of Proposition
~\ref{prop:thetadaction} (1), $\sum_{\nu \in \mu_{d,n}(Z(d))} \Thetabpoldij((1, \nu,0))(s_0)=\sum_{t
 \in Z(4)[2]}(\theta^\Thetabpol_t)^2$. By the duplication formula, the latter section is equal
 to $z \mapsto \vartheta[0,0](0, 2\Omega)\vartheta[0,0](4z, 2\Omega)$. It vanishes
 identically when $\vartheta[0,0](0, 2\Omega)=0$, which happens when the principally polarized
 abelian surface $\C^2/(\Z^2+2\Omega\Z^2)$ is a product of elliptic curves, for a suitable choice
 of the symplectic basis of $\Z^2+2\Omega\Z^2$. In this case $C_0=0$, whereas, in the
 numerical examples we computed, $\beta=(\alpha, -\alpha)$ with $\alpha \in Z(4)$ having an odd
 coordinate gives $C_\beta \neq 0$.
\end{remark}
From Theorem~\ref{th:mainisog}, we deduce Algorithm~\ref{algo:isogcomp} to compute an isogeny as well as the
following corollary. In Algorithm~\ref{algo:isogcomp}, the subgroup $G_1$ is not an input:
it is constructed from $K$, after the change of theta structure, by following the proof of
Lemma~\ref{lem:constructG}. This only requires linear algebra over $\Z/n\Z$ and $O(g^2)$
computations of Weil pairings, which can be done with theta functions \cite{DRpairing}, so that
it does not change the complexity of the algorithm.
\begin{corollary}\label{cor:isogcomp}
 Let $m,n,d >1$ be integers such that $n=md$.
 There exists a deterministic algorithm that takes as input
 the theta null point $0_{\Thetabpol}$ of a $g$-dimensional marked abelian variety $(B,
	 \bpol, \Thetabpol)$ of type $K(m)$, a basis of $B[n]$, a subgroup $K$ of
	 $B[d]$ isomorphic to $Z(d)$ and isotropic for the Weil pairing $e_{B,d}$
	 (without any further condition relating $K$ and $\Thetabpol$)
	 defining the isogeny $f:B \rightarrow A=B/K$,
	 $(\theta_i^\Thetabpol(x))_{i\in Z(m)}$ for $x \in B(\overk)$ and outputs
	 $(\theta_i^\Thetapol(x))_{i \in Z(m)}$ where $(A, \pol, \Thetapol)$ is a
	 marked abelian variety of type $K(m)$ in time $O(N n^{2g} \log(d))$ operations
	 in the base field of $B$, where $N$ is the number of values of $\beta$ tried by
	 Algorithm~\ref{algo:isogcomp}, provided that one of them satisfies $C_\beta \neq 0$;
	 this is always possible when $m \geq 3$, with the decomposition $d=1^2+\cdots+1^2$
	 (Theorem~\ref{th:mainisog}).
\end{corollary}

\begin{algorithm}
\SetKwInOut{Input}{input}\SetKwInOut{Output}{output}

\SetKwComment{Comment}{/* }{ */}
\Input{
 \begin{itemize}
	\item $m,n,d>1$ integers such that $n=md$;
	\item $(a_i)_{i=1,\dots,r} \in \N^r$ such that $d= \sum_{i=1}^r a_i^2$ and $\gcd(a_i, n)=1$;
	\item the marked abelian variety $(B, \bpol, \Thetabpolp)$ of type $K(m)$ given by its theta null
	 point $0_{\Thetabpolp}$;
	\item a basis of $B[n]$;
	\item $K \subset B[d]$ isomorphic to $Z(d)$ and isotropic for $e_{B,d}$, the kernel of the
	 isogeny $f: B \rightarrow A$;
	\item $(\theta^{\Thetabpolp}_i(x))_{i \in Z(m)}$, for $x \in B(\overk)$.
 \end{itemize}
}
\Output{
 \begin{itemize}
	\item the theta null point $0_{\Thetapol}$ of $(A, \pol, \Thetapol)$;
	\item $\left(\theta^{\Thetapol}_{j_0}(f(x))\right)_{j_0 \in Z(m)}$.
 \end{itemize}
}
\BlankLine
Let $d_0=\gcd(d,m)$. Compute a symplectic matrix $M \in \Sp_{2g}(\Z / m \Z)$ as in Remark
~\ref{rk:constructG} and call Algorithm~\ref{algo:trans} to compute $(B, \bpol, \Thetabpol)$
 given by its theta null point $0_{\Thetabpol}$ such that
 $K[d_0]=\Thetabpolbar(\mu_{d_0,m}(Z(d_0))\times\{0\})$, together with
 $(\theta^{\Thetabpol}_i(x))_{i \in Z(m)}$\;
 Following the proof of Lemma~\ref{lem:constructG}, compute a basis $(u_1, \ldots, u_g)$ of a
 subgroup $G_1 \subset B[n]$ isotropic for $e_{B,n}$ such that $du_i=\Thetabpolbar((\epsilon_i,0))$
 for $i=1, \ldots, g$, where $(\epsilon_i)$ is the canonical basis of $Z(m)$, and $(mu_1, \ldots,
 mu_g)$ is a basis of $K$; let $g_1: \Zn \rightarrow G_1$ be the numbering such that
 $g_1(\epsilon'_i)=u_i$, where $(\epsilon'_i)$ is the canonical basis of $\Zn$ (so that
 $g_1(\mu_{m,n}(i))=\Thetabpolbar((i,0))$ for $i \in Z(m)$ and $g_1(\mu_{d,n}(Z(d)))=K$)\;
 \If{$d$ is even}{
 Call Algorithm~\ref{algo:changesym1} to compute $(B, \bpol, \Thetabpol)$ given by its theta null
 point $0_{\Thetabpol}$ such that for all $x \in G_1$, $x$ is symmetric compatible with
 $\Thetabpol(\{1\}\times Z(m)\times\{0\})$, together with $(\theta^{\Thetabpol}_i(x))_{i \in
 Z(m)}$ and the $u_i$ in the new coordinates (this does not change $\Thetabpolbar$, so that the
 properties of $G_1$ and $g_1$ above still hold)\;
 }
 Call Algorithm~\ref{algo:goodliftg} to compute good lifts $\tildeG_1$ of $G_1$ with respect to
 $0_\Thetabpol$ and $\widetilde{x+G_1}$ of $G_1+x$ with respect to $\tildeG_1$\;
 Choose a numbering $\tildeG_1=\{ \tildeg_1(i), i \in \Zn \}$ such that $\tildeg_1(i)$ is a lift
 of $g_1(i)$ for $i \in \Zn$ and $\tildeg_1(\mu_{m,n}(i))= \widetilde{\Thetabpolbar((i,0))}$ for $i \in
 Z(m)$\;
 For $j_0 \in Z(m)$, choose $j_1 \in \Zn$ and $j_2 \in Z(m)$ such that $\rho_{n,m}(j_1 +
 \mu_{m,n}(j_2))=j_0$\;
 Find $\beta \in Z(m)^r$ such that $\sum_{i=1}^r a_i\beta_i=0$ in $Z(m)$ (trying for instance
 $\beta=0$ first) and such that the vector $0_{\Thetapol}$ defined by $\left(0_{\Thetapol}\right)_{j_0} =
 \sum_{P \in K} \prod_{i=1}^r (a_i (\widetilde{P+ g_1(j_1)}))_{\beta_i+a_i j_2}$ for $j_0 \in Z(m)$ is
 nonzero, where $K=g_1(\mu_{d,n}(Z(d)))$ and the good lifts $\widetilde{P+g_1(j_1)}$ are the elements
 of $\tildeG_1$\;\Comment{this holds if and only if $C_\beta \neq 0$; such a $\beta$ exists if $m
 \geq 3$ and $a_1=\cdots=a_r=1$ (Theorem~\ref{th:mainisog})}
 \For{$j_0 \in Z(m)$}{
 Compute $\theta_{j_0}^{\Thetapol}(f(x)) = \sum_{P \in K} \prod_{i=1}^r (a_i
 (\widetilde{x+P+ g_1(j_1)}))_{\beta_i+a_i j_2}$\;\Comment{the result is
 $C_\beta\theta_{j_0}^{\Thetapol}(f(x))$}
 }
 \Return $0_{\Thetapol}$ and $\left(\theta_{j_0}^{\Thetapol}(f(x))\right)_{j_0 \in Z(m)}$.
 \caption{Isogeny computation algorithm.}
 \label{algo:isogcomp}
\end{algorithm}

\begin{remark}
 \label{rk:grob}
 As for the change of level (Remark~\ref{rk:followmain7}), the isogeny algorithm needs, besides $K$,
 a subgroup $G_1 \subset B[n]$ isomorphic to $Z(n)$ with $K=G_1[d]$, in order to recover the theta
 structure $\Thetapol$ of $(A, \pol)$. If only $B$ and $K$ are given, the computation of $G_1$ (for
 instance with Gröbner bases as in Remark~\ref{rk:followmain7}) is the dominant step.
\end{remark}

The following example continues Example~\ref{ex:dup601}; it can be checked by hand, or with a few
lines of code.

\begin{example}\label{ex:isog601}
 We keep the notation of Example~\ref{ex:dup601}: $k=\F_{601}$, $B: y^2=x(x-1)(x-9)$ and
 $0_\Thetabpol=(b_0:b_1:b_2:b_1)=(1:545:117:545)$. We apply Algorithm~\ref{algo:isogcomp} with
 $m=4$, $d=2=1^2+1^2$, $n=8$ and $K=\{0, \Thetabpolbar((2,0))\}$, so that $A=B/K$, and with $G_1$
 generated by $P=(1:221:221:1)$, numbered by $g_1(1)=P$ (so that $g_1(2)=\Thetabpolbar((1,0))$);
 all the points of $G_1$ are symmetric compatible with $\Thetabpol(\{1\}\times Z(4)\times\{0\})$
 (Example~\ref{ex:dup601}).

 For $u=(1,221,221,1)$, we have
 \begin{equation*}
	\scalmult(2, u, u, \tildenullbpol, \tildenullbpol)=167\,(545, 117, 545, 1),
 \end{equation*}
 so that the good lifts of $P$ are the $\lambda u$ with
 $\lambda^4=1/167$, that is $\lambda \in \{\pm 199, \pm 234\}$ (Example~\ref{ex:hand48}). Write
 $\lambda u=(a_0, a_1, a_1, a_0)$. By Definition~\ref{def:good} (with $\Zz(2)=\{0,1\}$), the
 coordinate of index $0$ of $\tildeg_1(k)$, for $k \in Z(8)$, is the coordinate of index $k$ of
 \begin{equation*}
	\vartheta^A=(b_0, a_0, b_1, a_1, b_2, a_1, b_1, a_0),
 \end{equation*}
 which is the level $8$ theta null point of $A$ given by Example~\ref{ex:hand48}. Taking
 $x=0_B$, $\beta=0$, $j_1=j_0$ and $j_2=0$ in (\ref{eq:mainisog}), where $K=\{g_1(0), g_1(4)\}$,
 we get that $0_\Thetapol$ is
 \begin{equation*}
	\big(\vartheta^A_{j_0}{}^2+\vartheta^A_{j_0+4}{}^2\big)_{j_0 \in Z(4)}=(b_0^2+b_2^2 : a_0^2+a_1^2 :
	2b_1^2 : a_0^2+a_1^2).
 \end{equation*}
 In this case, Algorithm~\ref{algo:isogcomp} thus amounts to Example~\ref{ex:hand48} followed by the
 duplication formula (\ref{eq:dup601}) with $j=i$, applied on $A$ to go down from level $8$ to
 level $4$. Numerically:
 \begin{itemize}
	\item for $\lambda=\pm 199$, $\vartheta^A=(1:199:545:106:117:106:545:199)$ and
	 $0_\Thetapol=(1:395:111:395)$;
	\item for $\lambda=\pm 234$, $\vartheta^A=(1:234:545:28:117:28:545:234)$ and
	 $0_\Thetapol=(1:206:111:206)$.
 \end{itemize}
 These points satisfy the Riemann relations of level $8$ and $4$ respectively, and the elliptic
 curve defined by $0_\Thetapol$ has $j$-invariant $233$, which is the $j$-invariant of a curve
 $2$-isogenous to $B$ (whose $j$-invariant is $461$). The two theta null points of level $4$ only
 depend on $\lambda^2$, and since $234^2=-199^2$, they differ by the signs of the odd
 coordinates, that is by the action of $(1,0,2) \in G(4)$ (Equation (\ref{eq:trans})): this is
 the sign ambiguity discussed in the introduction, which is resolved here by the choice of the good
 lift of $P$. The implementation of \cite{thetAV} returns $(1:395:111:395)$.
\end{example}

The same computation one level down gives the change of level from $2$ to $4$ (Remark
\ref{rk:level2}).

\begin{example}\label{ex:dup24}
 We keep $k=\F_{601}$ and $B: y^2=x(x-1)(x-9)$, and we consider the marked abelian variety $(B,
 \bpol_0, \Theta_0)$ of type $K(2)$ whose theta null point $(b_0:b_1)$ is obtained from
 $(\vartheta_j)_j=(1,545,117,545)$ by the duplication formula (\ref{eq:dup601}) one level down:
 $(b_0:b_1)=(\vartheta_0^2+\vartheta_2^2:2\vartheta_1^2)=(1:111)$. Here $\mathrm{Km}(B)=\proj^1$, and
 the translation by $\bar\Theta_0((1,0))$ is $(x_0:x_1) \mapsto (x_1:x_0)$. The halves of
 $\bar\Theta_0((1,0))=(b_1:b_0)$ in $\mathrm{Km}(B)$ are the two points fixed by this translation,
 $(1:1)$ and $(1:-1)$; by Remark~\ref{rk:level2}, $P=(1:1)$ is symmetric compatible with
 $\Theta_0(\{1\}\times Z(2)\times\{0\})$ (eigenvalue $1$) and $(1:-1)$ is not (eigenvalue $-1$).
 The good lifts $\lambda(1,1)$ of $P$ are given by (\ref{eq:affinelift1}),
 $\scalmult(2, \lambda(1,1), \lambda(1,1), \tilde 0_{\Theta_0}, \tilde 0_{\Theta_0})=(b_1, b_0)$,
 and the Riemann relations of level $2$ give $\lambda^4=b_0b_1(b_0^2+b_1^2)/2=534$, that is
 $\lambda \in \{395, 93, 206, 508\}=395\cdot\{1, i, -1, -i\}$, with $i=125$. As in Example
~\ref{ex:hand48}, Theorem~\ref{th:main1} gives the theta null points of level $4$
 \begin{equation*}
	(b_0:\lambda:b_1:\lambda)=(1:\lambda:111:\lambda)
 \end{equation*}
 of the marked abelian varieties isog-$f$-compatible with $(B, \bpol_0, \Theta_0)$. The four of them
 satisfy the Riemann relations of level $4$ and define curves with $j$-invariant $233$, $2$-isogenous
 to $B$; for $\lambda=395$ and $\lambda=206$, they are the theta null points of Example
~\ref{ex:isog601}. This is the classical AGM step: going down from level $4$ to level $2$ keeps the
 even coordinates, and going up requires the fourth root $\lambda$, whose four values differ by
 fourth roots of unity (Lemma~\ref{lem:goodliftdiff}). The sign choices met when the duplication
 formula is solved from level $2$ to level $4$ are thus the choices of good lifts.
\end{example}

\section{Conclusion and perspectives}\label{sec:concl}
In \cite{lubicz:hal-03738315}, the change of level and isogeny algorithms are extended to
abelian varieties with real multiplication by an order $O$ of a totally real field: the
integer $d$ is replaced by a totally positive element $\gamma \in O$, the polarization $H$ by
$\gamma H$, and the decomposition $d=\sum_i a_i^2$ by a decomposition $\gamma=\sum_i
\overline{\alpha}_i \alpha_i$ with $\alpha_i \in O$, so that $\bigotimes_i \alpha_i^*\bpol$
replaces $\bigotimes_i [a_i]^*\bpol$. As in the case of integers, the norm of $\gamma$ is assumed
there to be prime to the level $m$. We briefly explain which parts of the present paper extend
to the general case, and which questions remain open.

The results of Sections~\ref{sec:isogcompat} and~\ref{sec:struct} on symmetric compatibility
(Definition~\ref{def:symcompat}, Proposition~\ref{prop:compat1}, Lemma~\ref{lem:symcompatchar},
the additivity of $\varphi_N$) only involve theta groups, separable
isogenies and symmetric line bundles, and apply verbatim to abelian varieties with real
multiplication. The new constraint is that the symplectic decomposition $B[\gamma m]=G_1 \oplus
G_2$ must be stable under $O$. A computation in the analytic model $B=(O \otimes \mathbb{R})^2/(O
\oplus O^\vee)$ suggests that this constraint does not create a new obstruction, provided that
$\Thetabpol$ is induced by an $O$-stable symplectic decomposition $\Lambda=\Lambda_1 \oplus
\Lambda_2$ of the period lattice, as assumed in \cite{lubicz:hal-03738315}. Indeed, the images of
$\frac{1}{\gamma m}\Lambda_1$ and $\frac{1}{\gamma m}\Lambda_2$ are then $O$-stable and their points
are symmetric compatible with $\Thetabpol(\{1\}\times Z(m)\times\{0\})$ and $\Thetabpol(\{1\}
\times\{0\}\times\dZ(m))$ respectively. Moreover, the $O$-stable candidates for $G_1$ are
parametrized by $O$-linear maps, on which the defect of symmetric compatibility depends
linearly, so that a candidate can be corrected as in Algorithm~\ref{algo:changesym2} without
breaking the $O$-stability. In genus $2$, for $O=\Z[\sqrt{2}]$ and $m=2$, the obstruction even
vanishes when $\gamma$ is exactly divisible by the ramified prime $(\sqrt{2})$ (for instance
$\gamma=2+\sqrt{2}$), whereas it appears for $\gamma=2$, as in the case of integers.

Several questions remain open:
\begin{itemize}
 \item an algebraic proof of the analogue of Theorem~\ref{prop:cond} for real multiplication,
	together with the analogue of Proposition~\ref{prop:mumfordequiv}, where $\eta_d$ is
	replaced by a map defined from the descent of $\gamma^*\bpol$ along $\gamma$;
 \item a description of the obstruction in terms of the decomposition of $2$ in $O$,
	explaining the phenomenon observed above when $\gamma$ is exactly divisible by a ramified
	prime above $2$;
 \item the computation of the maps $\alpha_i$ in theta coordinates, which is the main difficulty
	already pointed out in \cite{lubicz:hal-03738315} and \cite{robert_computing_nodate};
 \item the existence of $\beta \in Z(m)^r$ such that $C_\beta \neq 0$ in Theorems
	\ref{th:main7} and~\ref{th:mainisog}, which is still open in our setting when the $a_i$ are not
	all equal to $1$, and when $m=2$ except for the change of level with $d=2$ (Theorem
	\ref{th:main7} (3)), and whose analogue for decompositions
	$\gamma=\sum_i\overline{\alpha}_i\alpha_i$ remains to be studied.
\end{itemize}

\printbibliography
\end{document}